\documentclass[a4paper,onecolumn]{elsarticle}
\usepackage{exscale}
\usepackage{graphicx}
\usepackage{amsmath}
\usepackage{latexsym}
\usepackage{amsfonts}
\usepackage{amssymb}
\usepackage{bbm}
\usepackage{times}
\usepackage[T1]{fontenc}
\usepackage{lipsum}
\usepackage{amsthm}
\usepackage{fancyhdr,txfonts,bbm}

\usepackage[skip=0.333\baselineskip]{caption}
\usepackage{graphics}
\usepackage{bbold}
\usepackage{bm}
\usepackage{color}

\usepackage{subcaption}
\usepackage{graphicx}

\usepackage{tikz}
\usetikzlibrary{matrix,calc}

\usepackage{epsfig,pstricks}
\usepackage{changes}
\definechangesauthor[name={Markus}, color=red]{M}
\definechangesauthor[name={Dario}, color=magenta]{D}

\newtheorem{theorem}{Theorem}

\newtheorem*{ptheorem}{Pauli's Fundamental Theorem}

\newcommand{{\Cd}}{{\mathbb{C}^d}}
\newcommand{{\C}}{{\mathbb{C}}}

\begin{document}

\title{Low degree Lorentz invariant polynomials as potential entanglement invariants for multiple Dirac spinors}

\author{Markus Johansson}
\affiliation{organization={ICFO-Institut de Ciencies Fotoniques, The Barcelona Institute of Science and Technology},
postcode={08860}, city={ Castelldefels (Barcelona)}, country={Spain}}
\date{\today}
\begin{abstract}
A system of multiple spacelike separated Dirac particles is considered and a method for constructing polynomial invariants under the spinor representations of the local proper orthochronous Lorentz groups is described. The method is a generalization of the method used in [Phys. Rev. A {\bf 105}, 032402 (2022), arXiv:2103.07784] for the case of two Dirac particles. All polynomials constructed by this method are identically zero for product states.
The behaviour of the polynomials under local unitary evolution that acts unitarily on any subspace defined by fixed particle momenta is described.
By design all of the polynomials have invariant absolute values on this kind of subspaces if the evolution is locally generated by zero-mass Dirac Hamiltonians. Depending on construction some polynomials have invariant absolute values also for the case of nonzero-mass or additional couplings.
Because of these properties the polynomials are considered potential candidates for describing the spinor entanglement of multiple Dirac particles, with either zero or arbitrary mass or additional couplings. Polynomials of degree 2 and 4 are derived for the cases of three and four Dirac spinors. For three spinors no non-zero degree 2 polynomials are found but 67 linearly independent polynomials of degree 4 are identified. For four spinors 16 linearly independent polynomials of degree 2 are constructed as well as 26 polynomials of degree 4 selected from a much larger number.
The relations of these polynomials to the polynomial spin entanglement invariants of three and four non-relativistic spin-$\frac{1}{2}$ particles are described.
Moreover, it is described how degree 4 polynomials for five spinors can be constructed and how degree 2 polynomials can be constructed for any even number of spinors.

\end{abstract}

\maketitle
\section{Introduction}
The Dirac equation was originally introduced as a relativistic description of the electron \cite{dirac2,dirac}. As such it is used in relativistic quantum mechanics \cite{bjorken}, quantum electrodynamics \cite{peskin,schwartz} and in relativistic quantum chemistry \cite{pykk}. It has subsequently also been used in the Standard Model to describe other leptons and quarks \cite{schwartz} and in the Yukawa model of hadrons to describe baryons \cite{yukawa}. For the case of zero mass the Dirac equation admits solutions with definite chirality, so called Weyl particles \cite{weyl}. Dirac-like equations are used also to describe Dirac and Weyl quasi-particles in graphene and other solid state and molecular systems as well as in photonic crystals \cite{semenoff,novos,lu,bohm,liu,Xu,pirie}.

Quantum entanglement is a feature of quantum mechanics that permits action at a distance \cite{epr,bell,chsh,bell2}, i.e., nonlocal causation between spacelike separated systems. The state of a composite quantum system with multiple spacelike separated subsystems is entangled if it is a superposition where some property of one subsystem is conditioned on properties of one or several other subsystems. In this case the state cannot be fully described by only local variables.
The presence of entanglement allows for phenomena that are impossible without nonlocal causation such as quantum teleportation \cite{bennett} or violation of a Bell inequality \cite{bell,chsh,svet}. 
Entanglement in a system of multiple particles can in general exist in multiple forms that are not mutually convertible using local operations.
One often difficult problem in the theory of entanglement is the characterization of these different ways in which a multipartite system can be entangled \cite{popescu,carteret,coffman,toni,higuchi, dur,sud,wong,tarrach,verstraete2,luque,moor,toumazet}.

In non-relativistic quantum mechanics an extensively studied type of entanglement is that between the spins of spacelike separated spin-$\frac{1}{2}$ particles \cite{bell,bennett,popescu,carteret,coffman,toni,higuchi, dur,sud,wong,tarrach,verstraete2,luque,moor,toumazet,ghz,ekert,wootters,wootters2}.
A system of three or more such particles can be spin entangled in qualitatively different ways \cite{toni,higuchi,dur,tarrach,verstraete2,luque}. If two spin entangled states can be transformed into each other by physically allowed local unitary transformations they are said to be equivalently entangled with respect to local unitary transformations \cite{toni,higuchi,sud,tarrach,toumazet, ekert}. These transformations include the local unitary evolutions generated by the physically allowed Hamiltonian operators and the local changes of reference frame, such as spatial rotations. Equivalence with respect to larger groups of local transformations such as local $\mathrm{SL}(2,\mathbb{C})$, has also been considered \cite{dur,wong,verstraete2,luque,moor}. The number of parameters needed to describe the set of equivalence classes of this kind increases rapidly with the number of spin-$\frac{1}{2}$ particles \cite{popescu,carteret}.

One way to test if two states are inequivalently spin entangled with respect to a group of local transformations is to evaluate the polynomials in the state coefficients that are invariants of the action of the group \cite{hilbert,mumford,grassl}, the so called polynomial entanglement invariants. If the ratio of two such polynomials of the same degree takes different values for the different states they are necessarily inequivalently entangled. In particular, polynomials invariant under local $\mathrm{SL}(2,\mathbb{C})$ have been constructed for different numbers of non-relativistic spin-$\frac{1}{2}$ particles \cite{coffman,wong,luque,wootters,wootters2}. The polynomial invariants of this kind can be used to partially parametrize the set of spin entanglement equivalence classes \cite{popescu,carteret}.


In relativistic quantum mechanics the spinorial degree of freedom of a spin-$\frac{1}{2}$ particle is described by a four component Dirac spinor. Therefore the tools developed to describe the spin entanglement of non-relativistic spin-$\frac{1}{2}$ particles can not be used outside special cases.
This lead to a search for new tools and concepts suited for the description of entanglement of relativistic particles.
The entanglement and non-locality between Dirac spinors and other descriptions of entanglement between Dirac particles has been investigated and discussed previously in a number of works, see e.g. \cite{czachor,alsing,terno,adami,pachos,ahn,terno2,tera,tera2,mano,won,caban3,leon,caban,tessier,geng,delgado,moradi,caban2,spinorent}. 
As in the case of non-relativistic spin entanglement one may consider the question of whether spinor entangled states are equivalent with respect to a given set of physically allowed local transformations. The issue of constructing polynomial invariants for this purpose was considered in Ref. \cite{spinorent} for the case of two Dirac spinors.

In this work we investigate the description of entanglement between the spinorial degrees of freedom of multiple Dirac particles. The conceptualization of spinor entanglement properties that has been introduced in Ref. \cite{spinorent} is considered also here.
We explore the same general idea as Ref. \cite{spinorent} for constructing polynomial Lorentz invariants but adapt it to the case of multiple Dirac spinors. 
We describe how these polynomial Lorentz invariants can be used to partially characterize equivalence classes of spinor entangled states and how some of them reduce to the already known local $\mathrm{SL}(2,\mathbb{C})$ invariants \cite{coffman,wong,luque,moor} for the case of non-relativistic free particles in an energy eigenstate and for the case of Weyl particles.

As in Ref. \cite{spinorent} we make the assumption that for any number of spacelike separated Dirac spinors the state can be expanded in a basis that is formed from the tensor products of the local basis elements used to describe the individual spinors. 
Given this assumption a method for constructing polynomials invariant under the spinor representation of the proper orthochronous Lorentz groups for any number of spinors is described. 
These polynomials are also invariant, up to a U(1) phase, under local unitary transformations on a subspace spanned by the spinorial degrees of freedom that are generated by zero-mass Dirac Hamiltonians. Depending on construction some are invariant, up to a U(1) phase, also for local unitary transformations generated by arbitrary-mass Dirac Hamiltonians, and some for zero-mass Dirac Hamiltonians with a coupling to a Yukawa pseudo-scalar boson.
Therefore these polynomials are considered potential candidates for describing multi-spinor entanglement for different physical scenarios.

The computational difficulty of deriving polynomials of this kind increases with the number of spinors but polynomials of low degree can still be found with modest effort for the case of three and four Dirac spinors.
For the case of three Dirac spinors we derive polynomial invariants of degree 4 from which a set of 67 linearly independent polynomials can be selected. For four Dirac spinors we derive 16 linearly independent degree 2 polynomials.
We describe how to construct degree 4 polynomials for four spinors, but only calculate a few due to the large number of such polynomials. Further we describe how to construct degree 4 polynomials for five spinors but do not calculate them. Finally we describe how to construct degree 2 polynomials for any even number of spinors.
We discuss the relations of the derived polynomials to the Coffmann-Kundu-Wootters 3-tangle \cite{coffman} and the so called $2\times 2\times 4$ tangle described in Refs. \cite{moor,verstraete} for the case of three spinors and the relations to the polynomial invariants found by Luque and Thibon \cite{luque} for the case of four spinors. For the case of arbitrary even number of spinors we describe the relation of the degree 2 polynomials to the so called $N$-tangle introduced by Wong and Christensen \cite{wong}.

This work is organized as follows. In sections \ref{dir}-\ref{spen} the relevant background material is reviewed, the physical assumptions made are discussed and the tools used to construct the Lorentz invariant polynomials are described.
In particular, section \ref{dir} introduces the description of Dirac and Weyl particles and describes the fundamental assumptions made in this work. In section  \ref{rep} we describe the spinor representation of the Lorentz group and the charge conjugation transformation. Section \ref{invariants}  describes how to construct skew-symmetric bilinear forms that are invariant under the spinor representation of the local proper orthochronous Lorentz transformations. In section \ref{ham} the behaviour of the bilinear forms under local unitary evolution generated by Dirac-like Hamiltonians is described. Section \ref{hinn} describes the invariance groups of the bilinear forms and their relation to the physically allowed local operations. In Section \ref{spen} the conceptual framework used for characterizing entanglement properties with polynomial invariants is described.
Sections \ref{ent}-\ref{getto} contain the results.
In particular, section \ref{ent} describes the method for constructing candidate polynomial entanglement invariants. In section \ref{three} the case of three spinors is considered, in section \ref{four} the case of four spinors is considered, and in section \ref{five} the case of five spinors is considered. Section \ref{getto} describes the degree 2 polynomials for any even number of spinors. In Section \ref{lli} some general properties of polynomial invariants are reviewed in relation to the constructed polynomials and how they can be used to characterize spinor entanglement. Section \ref{diss} is the discussion and conclusions.
\section{Dirac spinors}\label{dir}

The Dirac equation was introduced in Ref. \cite{dirac2} as a relativistic description of a spin-$\frac{1}{2}$ particle, or Dirac particle. For a particle with mass $m$ and charge $q$ coupled to an electromagnetic four-potential $A_{\mu}(x)$ it can be expressed, with natural units $\hbar=c=1$, on the form
\begin{eqnarray}\label{vanilla}
\left[\sum_{\mu}\gamma^\mu(i\partial_{\mu}-qA_{\mu}(x)) -m\right]\psi(x)=0.
\end{eqnarray}
Here $\psi(x)$ is a four component Dirac spinor 

\begin{eqnarray}\label{spinor}
\psi(x)\equiv
\begin{pmatrix}
\psi_0(x) \\
\psi_1(x)\\
\psi_2(x) \\
\psi_3(x) \\
\end{pmatrix},
\end{eqnarray}
where each component is a function of the four-vector $x$, and $\gamma^0,\gamma^1,\gamma^2,\gamma^3$ are $4\times 4$ matrices
satisfying the anticommutator relations
\begin{eqnarray}\label{anti}
\gamma^\mu\gamma^\nu+\gamma^\nu\gamma^\mu=2g^{\mu\nu}I,
\end{eqnarray}
where $g^{\mu\nu}$ is the Minkowski metric with signature $(+---)$. For a derivation of the Dirac equation see e.g. Ref. \cite{dirac2} or Ref. \cite{dirac} Ch. XI. 
The matrices $\gamma^0,\gamma^1,\gamma^2,\gamma^3$ are not uniquely defined by Eq. (\ref{anti}) and can be chosen in different physically equivalent ways.
The choice we use here is the so called Dirac matrices or { \it gamma matrices} which are defined as
\begin{align}
\gamma^0&=
\begin{pmatrix}
I & 0 \\
0 & -I  \\
\end{pmatrix},
&\gamma^1=
\begin{pmatrix}
0  & \sigma^1 \\
-\sigma^1 &  0 \\
\end{pmatrix},\nonumber\\
\gamma^2&=
\begin{pmatrix}
0  & \sigma^2 \\
-\sigma^2 &  0 \\
\end{pmatrix},
&\gamma^3=
\begin{pmatrix}
0  & \sigma^3 \\
-\sigma^3 &  0 \\
\end{pmatrix},
\end{align}
where $I$ is the $2\times 2$ identity matrix and $\sigma^1,\sigma^2,\sigma^3$ are the Pauli matrices
\begin{eqnarray}
I=
\begin{pmatrix}
1 & 0 \\
0 & 1  \\
\end{pmatrix},\phantom{o}
\sigma^1=
\begin{pmatrix}
0  & 1 \\
1 &  0 \\
\end{pmatrix},\phantom{o}
\sigma^2=
\begin{pmatrix}
0  & -i \\
i &  0 \\
\end{pmatrix},\phantom{o}
\sigma^3=
\begin{pmatrix}
1  & 0 \\
0 &  -1 \\
\end{pmatrix}.
\end{eqnarray}
Two different choices of $4\times4$ matrices $\gamma^0,\gamma^1,\gamma^2,\gamma^3$ and $\gamma'^0,\gamma'^1,\gamma'^2,\gamma'^3$ that both satisfy the anticommutation relations in Eq. (\ref{anti}) are related by a similarity transformation.
By Pauli's Fundamental Theorem this similarity transformation is unique up to a constant factor.

\begin{ptheorem}
If for two sets of $4\times4$ matrices $\gamma^0,\gamma^1,\gamma^2,\gamma^3$ and $\gamma'^0,\gamma'^1,\gamma'^2,\gamma'^3$ we have that $\{\gamma^\mu,\gamma^\nu\}=2g^{\mu\nu}I=\{\gamma'^\mu,\gamma'^\nu\}$, then there exist an $S\in \mathrm{GL}(4,\mathbb{C})$ such that $\gamma'^\mu=S\gamma^\mu S^{-1}$, and $S$ is unique up to a multiplicative constant.
\end{ptheorem}
\begin{proof}
See e.g. Ref. \cite{pauli} or Ref. \cite{messiah}.
\end{proof}

From here on we regularly suppress the four-vector dependence of the spinor $\psi(x)$ where it is not essential and write it as $\psi$. 
The Dirac equation can be written on a form where the time derivative $\partial_0\psi$ is separated from the remaining expression 
\begin{eqnarray}\label{ggh}
\left[-\sum_{\mu=1,2,3}\gamma^0\gamma^\mu(i\partial_{\mu}-qA_{\mu})+qA_0I +m\gamma^0\right]\psi=i\partial_0\psi,\nonumber\\
\end{eqnarray}
and the Dirac Hamiltonian $H_D$ can be identified in Eq. (\ref{ggh}) as 
\begin{eqnarray}
H_D=-\sum_{\mu=1,2,3}\gamma^0\gamma^\mu(i\partial_{\mu}-qA_{\mu})+qA_0I +m\gamma^0.
\end{eqnarray}

We can identify two matrices with useful properties in the algebra generated by the gamma matrices that are frequently featured in the following.
The first is the matrix
\begin{eqnarray}
C\equiv i\gamma^1\gamma^3=
\begin{pmatrix}
-\sigma^2 & 0 \\
0 & -\sigma^2  \\
\end{pmatrix},
\end{eqnarray}
which has the property that for each gamma matrix $\gamma^\mu$ and its transpose $\gamma^{\mu T}$ we have that
\begin{eqnarray}\label{uub}
 C\gamma^\mu C=\gamma^{\mu T}.
\end{eqnarray}
Moreover, $C$ is Hermitian and its own inverse, i.e., $C=C^\dagger=C^{-1}$.
The second is the matrix
\begin{eqnarray}
\gamma^5\equiv i\gamma^0\gamma^1\gamma^2\gamma^3=
\begin{pmatrix}
0 & I \\
I & 0  \\

\end{pmatrix},
\end{eqnarray}
which anticommutes with each of the $\gamma^\mu$
\begin{eqnarray}
\gamma^5\gamma^\mu+\gamma^\mu\gamma^5=0.
\end{eqnarray}
Moreover, $\gamma^5$ is real Hermitian and its own inverse, i.e., $\gamma^5=\gamma^{5\dagger}=\gamma^{5T}=(\gamma^{5})^{-1}$.

For the case of zero particle momentum and zero four-potential there is an invariant subspace defined by the projector $P_+=\frac{1}{2}(I+\gamma^0)$ and an invariant subspace defined by the projector $P_-=\frac{1}{2}(I-\gamma^0)$. The former subspace is an eigenspace of the Dirac Hamiltonian with eigenvalue $+m$ and is often identified with a non-relativistic free spin-$\frac{1}{2}$ particle. The latter subspace is an eigenspace of the Dirac Hamiltonian with eigenvalue $-m$ and is often identified with a non-relativistic free spin-$\frac{1}{2}$ antiparticle (See e.g. Ref. \cite{peskin} Ch. 3.5.). The spinors $\psi_+$ in the subspace defined by $P_+$ and the spinors $\psi_-$ in the subspace defined by $P_-$ are of the form

\begin{eqnarray}\label{weyn}
\psi_+=
\begin{pmatrix}
{\psi}_0 \\
{\psi}_1\\
0 \\
0 \\
\end{pmatrix},\phantom{o}
\psi_-=
\begin{pmatrix}
0 \\
0\\
{\psi}_2 \\
{\psi}_3 \\
\end{pmatrix}.
\end{eqnarray}
We can see in Eq. (\ref{weyn}) that spinors in both these subspaces have only two nonzero spinor components.

The Dirac equation for the case of zero mass was considered by Weyl in Ref. \cite{weyl}
\begin{eqnarray}\label{weyl}
\left[-\sum_{\mu=1,2,3}\gamma^0\gamma^\mu(i\partial_{\mu}-qA_{\mu})+qA_0I\right]\psi=i\partial_0\psi.
\end{eqnarray}
For the zero mass equation there is an invariant subspace defined by the projector $P_L=\frac{1}{2}(I-\gamma^5)$ called the left-handed chiral subspace, and an invariant subspace defined by the projector $P_R=\frac{1}{2}(I+\gamma^5)$ called the right-handed chiral subspace. Solutions to Eq. (\ref{weyl}) that belong to the right-handed subspace are called right-handed Weyl particles $\psi_R$ and solutions that belong to the left-handed subspace are called left-handed Weyl particles $\psi_L$. These have the form

\begin{eqnarray}\label{wey}
\psi_R=
\begin{pmatrix}
{\psi}_0 \\
{\psi}_1\\
{\psi}_0 \\
{\psi}_1 \\
\end{pmatrix},\phantom{o}
\psi_L=
\begin{pmatrix}
{\psi}_0 \\
{\psi}_1\\
-{\psi}_0 \\
-{\psi}_1 \\
\end{pmatrix}.
\end{eqnarray}
We can see in Eq. (\ref{wey}) that spinors in both the left- and right-handed chiral subspaces have only two independent spinor components.

Dirac or Weyl particles in solid state and molecular systems are quasiparticle excitations. As such their physical interpretation is fundamentally different from that of Dirac or Weyl particles in relativistic quantum mechanics. Even so, they can be described by four component spinors and their evolutions are generated by Dirac-like Hamiltonians.
In the 2D Dirac semimetal graphene the evolution of a Dirac particle is generated by a Hamiltonian that can be written on the form 

\begin{eqnarray}\label{2d}
H_{2D}=iv_D\gamma^0\sum_{\mu=1,2}\gamma^\mu\partial_\mu-\mu_PI,
\end{eqnarray}
where $v_D$ is the Dirac velocity and $\mu_P$ is the deviation from the half-filling value of the chemical potential (See e.g. Ref. \cite{kotov} or \cite{vass}). 
Similarly, in a 3D Dirac semimetal the Hamiltonian for a massless Dirac particle can be written on the form \cite{bohm}
\begin{eqnarray}\label{3d}
H_{3D}=iv_D\gamma^0\sum_{\mu=1,2,3}\gamma^\mu\partial_\mu.
\end{eqnarray}
Two examples of experimentally realized 3D Dirac semimetals are sodium bismuthide (Na$_3$Bi) \cite{fang,liu} and cadmium arsenide (Cd$_3$As$_2$) \cite{fang2,neupane,cava}.

\subsection{Describing solutions of the Dirac equation}

A solution to the Dirac equation or a Dirac-like equation can be expanded for any time $t$ in a set of basis modes $\phi_je^{i\bold{k}\cdot\bold{x}}$ as
\begin{eqnarray}\label{bass}
\psi(t,\bold{x})=\int_{\bold{k}}d\bold{k}\sum_{j}\psi_{j,\bold{k}}(t)\phi_je^{i\bold{k}\cdot\bold{x}},
\end{eqnarray}
where $\bold{k}$ is a wave three-vector, $\bold{x}$ is a spatial three-vector, the $\psi_{j,\bold{k}}(t)$ are complex numbers, and the $\phi_j$ are four spinors that form a basis for the spinorial degree of freedom

\begin{eqnarray}\label{basis}
{\phi_0}=
\begin{pmatrix}
1 \\
0   \\
0   \\
0   \\
\end{pmatrix},\phantom{o}
{\phi_1}=
\begin{pmatrix}
0 \\
1   \\
0   \\
0   \\
\end{pmatrix},\phantom{o}
{\phi_2}=
\begin{pmatrix}
0 \\
0   \\
1   \\
0   \\
\end{pmatrix},\phantom{o}
{\phi_3}=
\begin{pmatrix}
0 \\
0   \\
0   \\
1   \\
\end{pmatrix}.
\end{eqnarray}

On the space of these solutions Dirac introduced an inner product defined as
\begin{eqnarray}\label{inn}
(\psi(t),\varphi(t))=\int_{\bold{x}}d\bold{x}\psi^\dagger(t,\bold{x})\varphi(t,\bold{x}).
\end{eqnarray}
The Dirac inner product of two basis modes is in general not well defined. This is because the basis modes extend over all three-space, i.e., they are infinite plane waves. Since the support of these infinite plane waves is not bounded the integral in Eq. (\ref{inn}) does not converge for the inner products of basis modes with the same basis spinor $\phi_j$, i.e., inner products of the form $(\phi_je^{i\bold{k}\cdot\bold{x}},\phi_je^{i\bold{k'}\cdot\bold{x}})$ for any $\bold{k}$ and $\bold{k'}$. 
Because of this non-convergence Dirac chose to disregard the definition of the inner product and instead ad hoc impose the orthogonality relations $(\phi_je^{i\bold{k}\cdot\bold{x}},\phi_le^{i\bold{k'}\cdot\bold{x}})=\delta_{jl}\delta(\bold{k}-\bold{k'})$ where $\delta_{jl}$ is the Kronecker delta and $\delta(\bold{k}-\bold{k'})$ is the Dirac delta (See Ref. \cite{dirac} Ch. IV {\S} 23). Subsequent works tried to find a mathematically meaningful interpretation of these orthogonality relations and this gave rise to the formalism of generalized eigenfunctions and rigged Hilbert spaces \cite{gelfand,maurin}.
In the rigged Hilbert space approach the momentum eigenmodes are not allowed as physical states. Instead only Schwartz functions are physically allowed. The Schwartz functions are the functions that are infinitely differentiable and such that for sufficiently large $\bold{x}$  the absolute values of the function and its derivatives to all orders decrease more rapidly than any inverse power of $|\bold{x}|$ (See e.g. Ref. \cite{stein}). Due to this rapid decrease of a Schwartz function outside a finite spatial region the Dirac inner product of any two Schwartz functions is well defined. Moreover, a Schwartz function can be said to be "localized" in a spatial region in the sense that it has negligible value outside it. Commonly used types of Schwartz functions are the Gaussian functions (See e.g. Ref. \cite{stein}) and the bump functions (See e.g. Ref. \cite{lee}).
Furthermore, a Schwartz function in spatial three-space is a Schwartz function also in momentum three-space. If the support of a Schwartz function in momentum space is contained in a sufficiently small closed ball it cannot be experimentally distinguished from the single point support of a momentum eigenmode (See  \ref{opp} for a discussion).

A different approach to make the inner product well defined is to only consider modes in a finite spatial rectangular volume and impose periodic boundary conditions, so called box quantization (See e.g. Refs. \cite{mandl} and \cite{wightman}). Then there is only a countable set of allowed momentum modes and these satisfy the orthogonality relations $(\phi_je^{i\bold{k}\cdot\bold{x}},\phi_le^{i\bold{k'}\cdot\bold{x}})=V\delta_{\bold{k},\bold{k'}}\delta_{jl}$ where $V$ is the volume of the box and $\delta_{\bold{k},\bold{k'}}$ and $\delta_{jl}$ are Kronecker deltas. If the box is made sufficiently large the discrete set of $\bold{k}$ cannot be experimentally distinguished from a continuous set (See \ref{opp} for a discussion).
Note that box quantization is closely related to the introduction of a large length scale cutoff, i.e., an Infrared Cutoff (See e.g. Refs.   \cite{wightman} and \cite{duncan}).

In this work we consider states with definite momenta as was done in References \cite{czachor,alsing,pachos,mano,caban3,caban,moradi,caban2,spinorent}. 
Moreover, we use the momentum eigenbasis and assume that the modes are orthogonal and normalizable. As an exact description this is possible only if we assume boundary conditions that allow for definite momenta in a finite spatial volume as in the box quantization approach. In the rigged Hilbert space approach where such boundary conditions are not allowed it can only be an approximate description.
Thus we assume that the physical scenario is such that it is warranted to treat a particle as having both a definite momentum and being contained in a finite spatial volume, whether as an exact or approximate description.

With these qualifying remarks we can consider the subspace spanned by only spinorial degrees of freedom that is defined by a fixed momentum $\bold{k}$, i.e., the subspace spanned by the modes $\phi_je^{i\bold{k}\cdot\bold{x}}$ for the given fixed $\bold{k}$. On such a four-dimensional subspace we can consider the inner product $(\cdot,\cdot)_{\bold{k}}$ given by

\begin{eqnarray}
(\psi(t),\varphi(t))_{\bold{k}}=\psi^\dagger(t)\varphi(t).
\end{eqnarray}

\subsection{Describing multiple spacelike separated Dirac spinors}

In this work we consider a scenario with multiple spacelike separated indistinguishable Dirac particles.
As in Ref. \cite{spinorent} we introduce a number of laboratories and assume that each laboratory contains a single Dirac particle. Each laboratory is spacelike separated from the other laboratories and comes with its own local description of spacetime. We assume that each of the laboratories
use a Minkowski space for this purpose. Thus we describe each different particle as belonging to a different Minkowski space. These different Minkowski spaces could be either the same Minkowski space described by different spacelike separated observers or alternatively the different Minkowski tangent-spaces of different spacelike separated points in a curved spacetime described by General Relativity (See e.g. Ref. \cite{wald}).

As was done in References \cite{alsing,pachos,caban3,moradi,caban2,spinorent} we assume that for any number of spacelike separated particles that have not interacted their state can be described as a tensor product  $\psi_1(t)\otimes \psi_2(t)\otimes \psi_3(t)\otimes\dots$ of single particle states.
Each single particle state has its support contained in a single laboratory. Note that the tensor product structure used here is with respect to the laboratories and not the particles. We assume that the particles are indistinguishable and thus the state $\psi_1(t)\otimes \psi_2(t)\otimes \psi_3(t)\otimes\dots$ expresses the presence of one particle in the single particle state $\psi_1(t)$ in the first laboratory, one particle in the single particle state $\psi_2(t)$ in the second laboratory, one particle in the single particle state $\psi_3(t)$ in the third laboratory and so on. The particles are not individuated and not given labels. Instead we refer to them by the laboratory they are found in. 
Further, we assume that a basis for the multi-particle states can be constructed as the tensor products of the elements of the single particle bases $\phi_{j_1}e^{i\bold{k_1}\cdot\bold{x_1}}\otimes \phi_{j_2}e^{i\bold{k_2}\cdot\bold{x_2}}\otimes \phi_{j_3}e^{i\bold{k_3}\cdot\bold{x_3}}\otimes\dots$.

The assumption that the state of the particles can be described by such a tensor product structure is not trivial. The reason why this assumption is often made is that operations on any given particle can be made jointly with operations on any other spacelike separated particle, i.e., such operations commute. However, it is not known if a description in terms of commuting operator algebras is always equivalent to a description in terms of tensor product spaces \cite{navascues,tsirelson,werner}. This open question is called Tsirelson's Problem \cite{tsirelson}. Nevertheless, if the algebra of operations is finite dimensional for each observer however it has been shown that a tensor product structure can be assumed without loss of generality \cite{tsirelson,werner}. In particular this holds if the Hilbert space of the shared system is finite dimensional. Thus it holds for operations on a subspace defined by fixed particle momenta of a finite number of particles.

The solutions to the Dirac equation are for some purposes reinterpreted as operators. This is done for example in the context of relativistic Quantum Field Theory formalisms where the Dirac spinor is reinterpreted as an operator valued Dirac field acting on a Hilbert space. Such an operator can be made a bosonic operator by imposing equal-time canonical commutation relations

\begin{eqnarray}
&&[\psi_a(\bold{x}),\psi^{\dagger}_b(\bold{x'})]=\delta(\bold{x}-\bold{x'})\delta_{ab},\nonumber\\
&&{[}\psi_a(\bold{x}),\psi_b(\bold{x'}){]}=0,\nonumber\\
&&{[}\psi_a^{\dagger}(\bold{x}),\psi_b^{\dagger}(\bold{x'}){]}=0,
\end{eqnarray}  
where $a,b$ label the spinor components, $\delta_{ab}$ is the Kronecker delta, and $\delta(\bold{x}-\bold{x'})$ is the Dirac delta (See e.g. Ref. \cite{peskin} Ch. 3.5.). Alternatively the operator can be made a fermionic operator by imposing equal-time canonical anticommutation relations 

\begin{eqnarray}
&&\{\psi_a(\bold{x}),\psi^{\dagger}_b(\bold{x'})\}=\delta(\bold{x}-\bold{x'})\delta_{ab},\nonumber\\
&&{\{}\psi_a(\bold{x}),\psi_b(\bold{x'}){\}}=0,\nonumber\\
&&{\{}\psi_a^{\dagger}(\bold{x}),\psi_b^{\dagger}(\bold{x'}){\}}=0,
\end{eqnarray}  
(See e.g. Ref. \cite{peskin} Ch. 3.5.).
If the support of two operator valued Dirac fields are spacelike separated the fields either commute, in the bosonic case, or anticommute, in the fermionic case. The property of commutation or anticommutation at spacelike separation ensures that causality is not violated, i.e., that no signals can be sent between spacelike separated events, and is one of the so called Wightman axioms for Quantum Field Theory (See e.g. Ref. \cite{all that} Ch. 3-1). This axiom originally formulated for Minkowski spacetime is straightforwardly generalized to curved spacetimes (See e.g. Ref. \cite{holl} for a discussion).

In this work we have assumed that the particles are all spacelike separated from each other.
We therefore assume that the spinors either commute or anticommute.
The only difference between imposing the commutation relations and imposing the anticommutation relations in this scenario is an overall sign that depends on the ordering of the operators. Such an overall sign has no physical meaning and is chosen by convention. Therefore we do not impose any bosonic or fermionic nature on the particles. For the purpose of this work it is irrelevant if the particles are are fermions or bosons. Note however that given some commonly made assumptions the spin statistics connection implies that Dirac particles must be fermions (See Ref. \cite{spin} and Ref. \cite{all that} Ch. 4-4).

Finally we comment on the case of distinguishable particles. If the particles are distinguishable they can be individuated based on intrinsic properties and given physically meaningful labels. These labels constitute a further degree of freedom that enlarges the Hilbert space. In a given laboratory in place of the single particle state $\psi_1(t)$ we would have a collection of single particle states $\psi_1(t)^J$ where $J$ indicate the particle species.
However if the assumption is made that in any given laboratory only one particle species can be found we can absorb the particle labels into the laboratory labels. With this restrictive assumption the Hilbert space is of a system of distinguishable particles is equivalent to the Hilbert space of a system of indistinguishable particles in the scenario considered in this work.


\section{Spinor representation of the Lorentz group and the charge conjugation}\label{rep}

In General Relativity a spacetime is described by a four-dimensional manifold and at every non-singular point one can define a four-dimensional tangent vector space. Any such tangent space is isomorphic to the Minkowski space (See e.g. Ref. \cite{wald}).
Here, as in Ref. \cite{spinorent}, we make the assumption that it is physically motivated to neglect the local curvature of spacetime and treat a Dirac particle as belonging to a Minkowski tangent space instead of the underlying spacetime manifold. Without this assumption the Dirac equation would have to be replaced by a curved spacetime counterpart such as that introduced by Weyl \cite{weyl} and Fock \cite{fock}.

A Lorentz transformation is a coordinate transformation on the local Minkowski tangent space to a spacetime point, but it also induces an action on the Dirac spinor in the point. This action is given by the spinor representation of the Lorentz transformation. Let $\Lambda$ be a Lorentz transformation and $S(\Lambda)$ be the spinor representation of $\Lambda$. Then the spinor transforms as $\psi(x)\to \psi'(x')=S(\Lambda)\psi(x)$ where $x'=\Lambda x$ (See e.g. Ref. \cite{zuber}), and the Dirac equation transforms as
\begin{eqnarray}
&&\left[\sum_{\mu}\gamma^\mu(i\partial_{\mu}-qA_{\mu}) -m\right]\psi(x)=0\nonumber\\
\to&&\left[\sum_{\mu,\nu}\gamma^\mu(\Lambda^{-1})^{\nu}_{\mu}(i\partial_{\nu}-qA_{\nu}) -m\right]S(\Lambda)\psi(x)=0.
\end{eqnarray}
The Lorentz invariance of the Dirac equation implies that
\begin{eqnarray}
S^{-1}(\Lambda)\gamma^\mu S(\Lambda)=\sum_\nu\Lambda^{\mu}_{\nu}\gamma^\nu.
\end{eqnarray}

The Lorentz group is a Lie group with four connected components. The connected component that contains the identity element is the proper orthochronous Lorentz group.
Likewise, the spinor representation of the Lorentz group is also a Lie group with four connected components. The connected component of this group that contains the identity element, the spinor representation of the proper orthochronous Lorentz group, can be generated by the exponentials of its Lie algebra. 
This Lie algebra has six generators $S^{\rho\sigma}$ defined by

\begin{eqnarray}\label{gene}
S^{\rho\sigma}=\frac{1}{4}[\gamma^\rho,\gamma^\sigma]=\frac{1}{2}\gamma^\rho\gamma^\sigma-\frac{1}{2}g^{\rho\sigma}I,
\end{eqnarray}
where as before $g^{\rho\sigma}$ is the Minkowski metric with signature $(+---)$.
The spinor representations of spatial rotations are generated by $S^{12},S^{13}$, and $S^{23}$ while the spinor representations of the Lorentz boosts are generated by $S^{01},S^{02}$, and $S^{03}$.
By taking the exponential of an element in the Lie algebra a finite transformation can be obtained
 
\begin{eqnarray}
S(\Lambda)=\exp\left(\frac{1}{2}\sum_{\rho,\sigma} \omega_{\rho\sigma}S^{\rho\sigma}\right),
\end{eqnarray}
where the coefficients $\omega_{\rho\sigma}$ are real numbers. The spinor representation of any proper orthochronous Lorentz transformation
can be described as a product of such finite transformations. See e.g. Ref. \cite{zuber}.

The four connected components of the Lorentz group are related to each other by the parity inversion P and the time reversal T transformations. Likewise, the four connected components of the spinor representation of Lorentz group are related to each other by the spinor representations of the parity inversion P and the time reversal T transformations. These spinor representations of P and T are only defined up to a multiplicative U(1) factor that can be chosen in different physically equivalent ways.
The spinor representation of the parity inversion transformation can be chosen as

\begin{eqnarray}
S(\textrm{P})=\gamma^0.
\end{eqnarray}
The spinor representation of the time reversal transformation T involves the matrix $C$ and the complex conjugation of the spinor and can be chosen as $\psi \to C\psi^*$. 

Besides the Lorentz group we can consider the charge conjugation transformation C, as well as the charge parity CP and charge parity time CPT transformation. As with P and T, the spinor representation of C is only defined up to a multiplicative U(1) factor that can be chosen in different ways.
The spinor representation of the charge conjugation, like the time reversal, involves complex conjugation of the spinor and can be chosen as $\psi \to i\gamma^2\psi^*$. It follows that the CP transformation is $\psi \to -i\gamma^0\gamma^2\psi^*= iC\gamma^5\psi^*$ and the CPT transformation is given by the matrix $-i\gamma^5$

\begin{eqnarray}
S(\textrm{CPT})=-i\gamma^5.
\end{eqnarray}
See e.g. Ref. \cite{bjorken} Ch. 5. In the following we use these choices of the spinor representations of the P, T and C transformations.

\section{Bilinear forms invariant under the spinor representation of the proper orthochronous Lorentz group}
\label{invariants}
A physical quantity that transforms under some representation of the Lorentz group is called a Lorentz covariant.
If it is also invariant under the action it is called a Lorentz invariant. A quantity that is invariant under the action of a representation of the proper orthochronous Lorentz group may not be invariant under the representation of the full Lorentz group but in the following we still refer to such a quantity as a Lorentz invariant for convenience.

Lorentz invariants can be constructed as bilinear forms on the Dirac spinors (See e.g. Ref. \cite{pauli}).
Given the properties of the matrix $C$ described in Eq. (\ref{uub}) and the form of the generators $S^{\rho\sigma}$ of the spinor representation of the proper orthochronous Lorentz group in Eq. (\ref{gene}) we have that the transpose $S^{\rho\sigma T}$ satisfies
\begin{eqnarray}
S^{\rho\sigma T}=\frac{1}{4}[\gamma^{\sigma T},\gamma^{\rho T}]=-\frac{1}{4}C[\gamma^{\rho},\gamma^{\sigma}]C=-CS^{\rho\sigma}C.
\end{eqnarray}
Thus it holds for any finite transformation $S(\Lambda)$ that $S(\Lambda)^TC=CS(\Lambda)^{-1}$.
This allows us to construct a Lorentz invariant bilinear form as  
\begin{eqnarray}
\psi^TC\varphi,
\end{eqnarray}
where $\psi$ and $\varphi$ are Dirac spinors.
This bilinear form transforms as $\psi^TS(\Lambda)^TCS(\Lambda)\varphi=\psi^TCS(\Lambda)^{-1}S(\Lambda)\varphi=\psi^TC\varphi$ for any spinor representation $S(\Lambda)$ of a proper orthochronous Lorentz transformation. Moreover $\psi^TC\varphi$ is invariant under parity inversion P.
This follows since $\gamma^0=(\gamma^{0})^{T}=(\gamma^{0})^{-1}$ and $\gamma^0C=C\gamma^0$ and thus $\psi^TS(\textrm{P})^TCS(\textrm{P})\varphi=\psi^T\gamma^0C\gamma^0\varphi=\psi^TC\varphi$.
However, $\psi^TC\varphi$ is not invariant under the CPT transformation but changes sign.
This follows since $S(\textrm{CPT})=-i\gamma^{5}$ and $\gamma^5C=C\gamma^5$ and $\gamma^5=(\gamma^{5})^{T}=(\gamma^{5})^{-1}$ and thus $\psi^TS(\textrm{CPT})^TCS(\textrm{CPT})\varphi=-\psi^T\gamma^5C\gamma^5\varphi=-\psi^TC\varphi$.

Next we recall that the matrix $\gamma^5$ anti-commutes with all $\gamma^{\mu}$. Therefore we see that it commutes with any generator $S^{\rho\sigma}$
\begin{eqnarray}
[S^{\rho\sigma},\gamma^5]=\frac{1}{4}[\gamma^{\rho},\gamma^{\sigma}]\gamma^5-\frac{1}{4}\gamma^5[\gamma^{\rho},\gamma^{\sigma}]=0.
\end{eqnarray}
It follows that $\gamma^5$ commutes with the spinor representation of any proper orthochronous Lorentz transformation $S(\Lambda)\gamma^5=\gamma^5S(\Lambda)$.
Therefore, we can construct another Lorentz invariant bilinear form as 
\begin{eqnarray}
\psi^TC\gamma^5\varphi.
\end{eqnarray}
This bilinear form transforms as $\psi^TS(\Lambda)^TC\gamma^5S(\Lambda)\varphi=\psi^TCS(\Lambda)^{-1}\gamma^5S(\Lambda)\varphi=\psi^TC\gamma^5S(\Lambda)^{-1}S(\Lambda)\varphi=\psi^TC\gamma^5\varphi$.
Moreover, $\psi^TC\gamma^5\varphi$ is not invariant under parity inversion but changes sign.
This follows since $\gamma^0\gamma^5=-\gamma^5\gamma^0$ and thus $\psi^TS(\textrm{P})^TC\gamma^5S(\textrm{P})\varphi=\psi^T\gamma^0C\gamma^5\gamma^0\varphi=-\psi^TC\gamma^5\varphi$.
Furthermore, $\psi^TC\gamma^5\varphi$ is not invariant under the CPT transformation but changes sign. 
This follows since $S(\textrm{CPT})=-i\gamma^{5}$ and thus $\psi^TS(\textrm{CPT})^TC\gamma^5S(\textrm{CPT})\varphi=-\psi^T\gamma^5C\gamma^5\gamma^5\varphi=-\psi^TC\gamma^5\varphi$. 

The two bilinear forms $\psi^TC\varphi$ and $\psi^TC\gamma^5\varphi$ are both skew-symmetric, i.e., $\psi^TC\varphi=-\varphi^TC\psi$ and $\psi^TC\gamma^5\varphi=-\varphi^TC\gamma^5\psi$ due to the anti-symmetry of $C$ and $C\gamma^5$ respectively. Thus $\psi^TC\psi=0$ and $\psi^TC\gamma^5\psi=0$. Moreover, the two bilinear forms are both non-degenerate. 
If $\xi^TC\gamma^5\chi=0$ for all $\chi$ it follows that $\xi=0$, and if $\xi^TC\gamma^5\chi=0$ for all $\xi$ it follows that $\chi=0$. Similarly, if $\xi^TC\chi=0$ for all $\chi$ it follows that $\xi=0$, and if $\xi^TC\chi=0$ for all $\xi$ it follows that $\chi=0$.

Note that because the U(1) phases of the spinor representations of P, T and C have to be chosen the U(1) phases acquired by the two bilinear forms $\psi^TC\varphi$ and $\psi^TC\gamma^5\varphi$ under the $S(\textrm{P})$ transformation or under the $S(\textrm{CPT})$ transformation are determined by these choices. However the difference by a factor of $-1$ between the phase acquired by $\psi^TC\varphi$ under the $S(\textrm{P})$ transformation and the phase acquired by $\psi^TC\gamma^5\varphi$ under the $S(\textrm{P})$ transformation, is independent of these choices.

\section{Behaviour of the bilinear forms under unitary spinor evolution generated by Dirac-like Hamiltonians}\label{ham}
Here we consider a subspace defined by a fixed particle momentum $\bold{k}$, i.e., a subspace spanned by the four basis elements $\phi_je^{i\bold{k}\cdot\bold{x}}$ with the same $\bold{k}$. Furthermore, as was done in Ref. \cite{spinorent}, we consider an evolution that acts unitarily on such a subspace and is generated by a Hamiltonian operator $H$. For the subspace to be invariant under the evolution it is required that $(\phi_je^{i\bold{k}\cdot\bold{x}},H\phi_le^{i\bold{k'}\cdot\bold{x}})\propto\delta_{\bold{k},\bold{k'}}$. To have such unitary action on the subspace we consider evolution generated by Hamiltonians that do not depend on the spatial coordinate $\bold{x}$.

Given the restriction to evolutions generated by Hamiltonians without spatial dependence we can consider the behaviour of the bilinear forms $\psi^TC\varphi$ and $\psi^TC\gamma^5\varphi$ for such evolutions.
It can be shown \cite{spinorent} that the bilinear form $\psi^TC\varphi$ is invariant, up to a U(1) phase, under evolutions generated by any Hamiltonians on the form $H^{2,3}(t)+H^{0}(t)$ where

\begin{eqnarray}
H^{2,3}(t)=\gamma^\mu\gamma^\nu\phi_{\mu\nu}(t)+\gamma^\mu\gamma^\nu\gamma^\rho\kappa_{\mu\nu\rho}(t),
\end{eqnarray}
and $H^0(t)=f(t)I$. This follows from the relation $CH^{2,3}(t)=-(H^{2,3}(t))^TC$. Similarly it can be shown that the bilinear form $\psi^TC\gamma^5\varphi$ is invariant, up to a U(1) phase, under evolutions generated by any Hamiltonians on the form $H^{1,2}(t)+H^{0}(t)$ where

\begin{eqnarray}
H^{1,2}(t)=\gamma^\mu\eta_{\mu}(t)+\gamma^\mu\gamma^\nu\lambda_{\mu\nu}(t).
\end{eqnarray}
This follows from the relation $C\gamma^5H^{1,2}(t)=-(H^{1,2}(t))^TC\gamma^5$.
See \ref{hamm} for a derivation of how the two bilinear forms behave under these kinds of evolutions.

The Dirac Hamiltonian contains a first degree term in the gamma matrices, the mass term $m\gamma^0$, a second degree term, the generalized canonical momentum term $\sum_{\mu=1,2,3}\gamma^{0}\gamma^{\mu}(i\partial_{\mu}-qA_{\mu}(t))$, as well as a zeroth degree term, the coupling to the scalar potential $qA_0(t)I$. Thus the Dirac Hamiltonian for a massive particle is of the type $H^{1,2}(t)+H^{0}(t)$ while the Dirac Hamiltonian for a massless particle is on both the form  $H^{1,2}(t)+H^{0}(t)$ and on the form $H^{2,3}(t)+H^{0}(t)$.

Apart from the terms in the standard Dirac Hamiltonian we can consider some additional or alternative Hamiltonian terms from different physical models. One example is a coupling to a scalar potential such as a Yukawa scalar boson $g\gamma^0\phi$ \cite{yukawa}, but such terms behave analogously to a mass term. A coupling to a pseudo-scalar potential such as a Yukawa pseudo-scalar boson $gi\gamma^0\gamma^5\phi$ on the other hand is third degree in the gamma matrices. Another kind of additional term is a coupling to a pseudo-vector potential $\sum_{\mu=1,2,3}\gamma^{0}\gamma^5\gamma^{\mu}(A_{\mu}(t))$ which is second degree in gamma matrices (See e.g. Ref. \cite{thaller}). A tensor term $i\gamma^{0}\sum_{\mu\nu=1,2,3}[\gamma^\mu,\gamma^\nu]F_{\mu\nu}(t)$ is third degree in gamma matrices. An example is an anomalous magnetic moment described by a Pauli-coupling $i\gamma^{0}\sum_{\mu\nu=1,2,3}[\gamma^\mu,\gamma^\nu](\partial_\mu A_\nu(t)-\partial_\nu A_\mu(t))$ to the electromagnetic tensor (See e.g. Ref. \cite{thaller} or \cite{das}).
A pseudo-tensor term $i\gamma^{0}\gamma^5\sum_{\mu\nu=1,2,3}[\gamma^\mu,\gamma^\nu]F_{\mu\nu}(t)$ is first degree in gamma matrices (See e.g. Ref. \cite{thaller}). 
Finally we can consider a chiral coupling of electroweak type to a vector boson, e.g. $\sum_\mu g\gamma^0\gamma^\mu(I\pm\gamma^5)Z_{\mu}$ \cite{weinberg}, which has terms of degree 2 and 4 in the gamma matrices.

The bilinear forms ${\psi^T}C{\varphi}$ and ${\psi^T}C\gamma^5{\varphi}$ are not invariant under evolution generated by Hamiltonians that have both a mass term or a coupling to a scalar potential and also a coupling to a pseudo-scalar or a tensor term. Neither are they invariant under evolution generated by Hamiltonians that contain both a tensor and pseudo-tensor term. 
Moreover, the bilinear forms ${\psi^T}C{\varphi}$ and ${\psi^T}C\gamma^5{\varphi}$ are not invariant under evolution generated by
Hamiltonians with chiral coupling of electroweak type to a vector boson.

The two bilinear forms ${\psi^T}C{\varphi}$ and ${\psi^T}C\gamma^5{\varphi}$ can be considered also in the context of Dirac or Weyl particles in solid state and molecular systems.
The Hamiltonian in Eq. (\ref{2d}) for Dirac particles in the 2D Dirac semimetal graphene  \cite{kotov,vass}, and 
the Hamiltonian in Eq. (\ref{3d}) for Dirac particles in 3D Dirac semimetals \cite{bohm} have only terms that are zeroth and second degree in the gamma matrices.
Therefore, for both these cases ${\psi^T}C\gamma^5{\varphi}$ and ${\psi^T}C{\varphi}$ are invariant up to a U(1) phase.

Terms that can be added to the Hamiltonians in this context are the Semenoff mass term $M_S\gamma^0\gamma^3$ \cite{semenoff} and the Haldane mass term $M_{H}\gamma^5\gamma^0\gamma^3$ \cite{haldane}. Both the Semenoff and Haldane mass terms are second degree in the gamma matrices, and thus ${\psi^T}C{\varphi}$ and ${\psi^T}C\gamma^5{\varphi}$ are still invariant, up to a U(1) phase, with these additions. 

\section{The invariance groups of the bilinear forms}\label{hinn}

Here we consider the groups that preserve the bilinear form ${\psi^T}C\gamma^5{\varphi}$, up to a U(1) phase, the groups that preserve the bilinear form ${\psi^T}C{\varphi}$, up to a U(1) phase, and the groups that preserve both the bilinear forms up to a U(1) phase. The relations between these groups and the evolutions generated by Dirac-like Hamiltonians and the spinor representation of the proper orthochronous Lorentz group is also described. These groups and their relations to the evolutions generated by Dirac-like Hamiltonians and the spinor representation of the proper orthochronous Lorentz group were previously described in Ref. \cite{spinorent}.

As described in Section \ref{ham} the non-degenerate skew-symmetric bilinear form
${\psi^T}C\gamma^5{\varphi}$
is invariant, up to a U(1) phase, under any evolution generated by Dirac Hamiltonians. In greater generality it is invariant, up to a U(1) phase, under any evolution generated by Hamiltonians of the type $H^{1,2}(t)+H^0(t)$. These latter evolutions form a Lie group of real dimension 11 generated by the exponentials of the real Lie algebra spanned by the 11 skew-Hermitian matrices $i\gamma^0$, $\gamma^1$, $\gamma^2$, $\gamma^3$, $i\gamma^0 \gamma^1$, $i\gamma^0 \gamma^2$, $i\gamma^0 \gamma^3$, $\gamma^1 \gamma^2$, $\gamma^1 \gamma^3$, $\gamma^2 \gamma^3$ and $iI$. This group $G^{C\gamma^5}_U$ is isomorphic to $\mathrm{Sp}(2)\times\mathrm{ U}(1)$ where $\mathrm{Sp}(2)$ is the compact symplectic group of $4\times 4$ matrices (See e.g. Ref. \cite{hall} Ch. 1.2.8.). This follows since a symplectic group is defined as the set of linear transformations on a vector space over the complex numbers $\mathbb{C}$ that preserve a given non-degenerate skew-symmetric bilinear form. All such symplectic groups on the same vector space are isomorphic and their isomorphism class is defined as {\it the} symplectic group $\mathrm{Sp}(n,\mathbb{C})$, where $n$ is the dimension of the vector space. Moreover, the compact symplectic group $\mathrm{Sp}(n/2)$ is defined as the intersection of $\mathrm{Sp}(n,\mathbb{C})$ with $\mathrm{SU}(n)$, i.e., $\mathrm{Sp}(n/2)=\mathrm{Sp}(n,\mathbb{C})\cap \mathrm{SU}(n)$.
As described in Ref. \cite{spinorent} the set of unitary transformations that can be implemented by non-zero mass strongly  continuous Dirac Hamiltonians is dense in the group $G^{C\gamma^5}_U$. Thus there is no continuous function that is invariant, up to a U(1) phase, under all evolutions generated by non-zero mass Dirac Hamiltonians that is not invariant, up to a U(1) phase, under $G^{C\gamma^5}_U$.

The invariance of ${\psi^T}C\gamma^5{\varphi}$, up to a U(1) phase, extends to a Lie group of real dimension 21 that contains $G^{C\gamma^5}_U$ as a subgroup. This larger group is generated by the exponentials of the real Lie algebra spanned by the 11 skew-Hermitian matrices $i\gamma^0$, $\gamma^1$, $\gamma^2$, $\gamma^3$, $i\gamma^0 \gamma^1$, $i\gamma^0 \gamma^2$, $i\gamma^0 \gamma^3$, $\gamma^1 \gamma^2$, $\gamma^1 \gamma^3$, $\gamma^2 \gamma^3$ and $iI$ and the 10 Hermitian matrices $\gamma^0$, $i\gamma^1$, $i\gamma^2$, $i\gamma^3$, $\gamma^0 \gamma^1$, $\gamma^0 \gamma^2$, $\gamma^0 \gamma^3$, $i\gamma^1 \gamma^2$, $i\gamma^1 \gamma^3$, and $i\gamma^2 \gamma^3$. This group $G^{C\gamma^5}$ is isomorphic to the group $\mathrm{Sp}(4,\mathbb{C})\times \mathrm{U}(1)$ where $\mathrm{Sp}(4,\mathbb{C})$ is the symplectic group of $4\times 4$ matrices (See e.g. Ref. \cite{hall} Ch. 1.2.4.).  The group $G^{C\gamma^5}$ contains the spinor representation of the proper orthochronous Lorentz group. As described in Ref. \cite{spinorent} the group $G^{C\gamma^5}$ is the smallest Lie group that contains both $G^{C\gamma^5}_U$ and the spinor representation of the proper orthochronous Lorentz group as subgroups and any continuous function that is invariant, up to a U(1) phase, under both $G^{C\gamma^5}_U$ and the spinor representation of the proper orthochronous Lorentz group is invariant, up to a U(1) phase, under $G^{C\gamma^5}$.

Likewise, as described in Section \ref{ham} the non-degenerate skew-symmetric bilinear form
${\psi^T}C{\varphi}$
is invariant, up to a U(1) phase, under any evolution generated by zero-mass Dirac Hamiltonians. In greater generality it is invariant, up to a U(1) phase, under any evolution generated by Hamiltonians of the type $H^{2,3}(t)+H^0(t)$. These latter evolutions form a Lie group of real dimension 11 generated by the exponentials of the real Lie algebra spanned by the 11 skew-Hermitian matrices $\gamma^5\gamma^0$, $i\gamma^5\gamma^1$, $i\gamma^5\gamma^2$, $i\gamma^5\gamma^3$, $i\gamma^0 \gamma^1$, $i\gamma^0 \gamma^2$, $i\gamma^0 \gamma^3$, $\gamma^1 \gamma^2$, $\gamma^1 \gamma^3$, $\gamma^2 \gamma^3$ and $iI$. This group $G^C_U$ is isomorphic to $\mathrm{Sp}(2)\times\mathrm{ U}(1)$. 
As described in Ref. \cite{spinorent} the set of unitary transformations that can be implemented by strongly  continuous Dirac Hamiltonians with zero mass and a coupling to a Yukawa pseudoscalar boson is dense in the group $G^{C}_U$. Thus there is no continuous function that is invariant, up to a U(1) phase, under all evolutions generated by Dirac Hamiltonians with zero mass and a coupling to a Yukawa pseudoscalar boson that is not invariant, up to a U(1) phase, under $G^{C}_U$.

The invariance of ${\psi^T}C{\varphi}$, up to a U(1) phase, extends to a group of real dimension 21 that contains $G^{C}_U$ as a subgroup. This larger group is generated by the exponentials of the real Lie algebra spanned by the 11 skew-Hermitian matrices $\gamma^5\gamma^0$, $i\gamma^5\gamma^1$,$i\gamma^5\gamma^2$, $i\gamma^5\gamma^3$, $i\gamma^0 \gamma^1$, $i\gamma^0 \gamma^2$, $i\gamma^0 \gamma^3$, $\gamma^1 \gamma^2$, $\gamma^1 \gamma^3$, $\gamma^2 \gamma^3$ and $iI$ and the 10 Hermitian matrices $i\gamma^5\gamma^0$, $\gamma^5\gamma^1$, $\gamma^5\gamma^2$, $\gamma^5\gamma^3$, $\gamma^0 \gamma^1$, $\gamma^0 \gamma^2$, $\gamma^0 \gamma^3$, $i\gamma^1 \gamma^2$, $i\gamma^1 \gamma^3$, and $i\gamma^2 \gamma^3$. This group $G^C$ is isomorphic to the group $\mathrm{Sp}(4,\mathbb{C})\times \mathrm{U}(1)$, and contains the spinor representation of the proper orthochronous Lorentz group. As described in Ref. \cite{spinorent} the group $G^{C}$ is the smallest Lie group that contains both $G^{C}_U$ and the spinor representation of the proper orthochronous Lorentz group as subgroups and any continuous function that is invariant, up to a U(1) phase, under both $G^{C}_U$ and the spinor representation of the proper orthochronous Lorentz group is invariant, up to a U(1) phase, under $G^{C}$.

While the Lie groups $G^{C\gamma^5}_U$ and $G^C_U$ are isomorphic they are not identical and their intersection $G^{C\gamma^5}_U\cap G^C_U$ is a Lie group of real dimension 7 that contains evolutions generated by Hamiltonians with only second and zeroth degree terms in the gamma matrices. The spinor representations of the spatial rotations are also included in this group but not the spinor representations of the Lorentz boosts. As described in Ref. \cite{spinorent} the set of unitary transformations that can be implemented by zero mass Dirac Hamiltonians is dense in the group $G^{C\gamma^5}_U\cap G^C_U$. Thus there is no continuous function that is invariant, up to a U(1) phase, under all evolutions generated by zero mass Dirac Hamiltonians that is not invariant, up to a U(1) phase, under $G^{C\gamma^5}_U\cap G^C_U$.

Similarly, the Lie groups $G^{C\gamma^5}$ and $G^C$ are isomorphic but not identical and their intersection  $G^{C\gamma^5}\cap G^C$ is a Lie group of real dimension 13 that also contains the evolutions generated by Hamiltonians with only second and zeroth degree terms in the gamma matrices and the spinor representations of the spatial rotations. Beyond this it contains non-unitary elements generated by the Hermitian second degree terms in the gamma matrices, which includes the spinor representations of the Lorentz boosts. As described in Ref. \cite{spinorent} the group $G^{C\gamma^5}\cap G^C$ is the smallest Lie group that contains both $G^{C\gamma^5}_U\cap G^C_U$ and the spinor representation of the proper orthochronous Lorentz group as subgroups and any continuous function that is invariant, up to a U(1) phase, under both $G^{C\gamma^5}_U\cap G^C_U$ and the spinor representation of the proper orthochronous Lorentz group is invariant, up to a U(1) phase, under $G^{C\gamma^5}\cap G^C$.

Besides the invariance of ${\psi^T}C\gamma^5{\varphi}$, up to a U(1) phase, under the connected Lie group $G^{C\gamma^5}$ and the invariance of ${\psi^T}C{\varphi}$, up to a U(1) phase, under the connected Lie group $G^C$ both the bilinear forms are invariant, up to a sign, under the 32 element finite group generated by the gamma matrices $\gamma^0,\gamma^1,\gamma^2,\gamma^3$, the so called {\it Dirac group} (See e.g. Ref. \cite{lomont}). As can be seen from Eq. (\ref{anti}) ${\psi^T}C{\varphi}$ is invariant under action by $\gamma^0$ but changes sign under action by $\gamma^1,\gamma^2$ and $\gamma^3$.  Similarly we see from Eq. (\ref{anti}) that ${\psi^T}C\gamma^5{\varphi}$ is invariant under action by $\gamma^1,\gamma^2$ and $\gamma^3$ but changes sign under action by $\gamma^0$. 

Note that if we choose a different physically equivalent set of matrices $\gamma'^0,\gamma'^1,\gamma'^2,\gamma'^3$ related to the gamma matrices $\gamma^0,\gamma^1,\gamma^2,\gamma^3$ by a similarity transformation $\gamma'^\mu=S\gamma^\mu S^{-1}$ where $S\in \mathrm{GL}(4,\mathbb{C})$ we can construct  two skew-symmetric Lorentz invariant bilinear forms as ${\psi^T}(S^{-1})^TCS^{-1}{\varphi}$ and ${\psi^T}(S^{-1})^TC\gamma^5S^{-1}{\varphi}$. 
The linear groups preserving these bilinear forms, up to a U(1) phase, are $SG^CS^{-1}$ and $SG^{C\gamma^5}S^{-1}$ respectively. Moreover, the set of evolutions generated by nonzero mass Dirac Hamiltonians is a dense subset of $SG_U^{C\gamma^5}S^{-1}$ and the set of evolutions generated by zero mass Dirac Hamiltonians with a coupling to a Yukawa pseudoscalar boson is a dense subset of $SG_U^{C}S^{-1}$. 

\section{Describing spinor entanglement properties using polynomial invariants}\label{spen}

A general framework for describing entanglement properties of a composite system of non-relativistic particles was developed and described in References \cite{popescu} and \cite{carteret}. This general framework has been adapted to the case of Dirac spinors in Ref. \cite{spinorent} and it was described how Lorentz invariant homogeneous polynomials that are zero for product states and invariant, up to a U(1) phase, under physically allowed local unitary evolutions can be used to partially characterize inequivalent forms of spinor entanglement. Here we outline this conceptual framework previously described in Ref. \cite{spinorent} and describe the properties of polynomial invariants. We then discuss such Lorentz invariant homogeneous polynomials for the specific cases of physically allowed local unitary evolutions described in Section \ref{ham} and Section \ref{hinn}. Further, we consider the restrictions of these Lorentz invariant homogeneous polynomials to the case of zero momenta free particles in an energy eigenstate and to the case of Weyl particles. Finally we comment on the relation between the Lorentz invariant homogeneous polynomials and so called entanglement measures.

In general a system of multiple particles can be entangled in a variety of qualitatively different ways. The entanglement present in the system is typically considered to be preserved by changes of local reference frames and by local unitary evolution of the individual particles.
Therefore if two entangled states of the system are such that each of them can be transformed into the other deterministically through changes of local reference frames and local unitary evolution we may consider them to be equivalently entangled. Moreover, if two states are identical up to multiplication by a constant $c\in \mathbb{C}-\{0\}$ we may consider them to be equivalently entangled. Therefore, within the conceptual framework described in Ref. \cite{spinorent} two entangled states are considered equivalently entangled if they can be transformed into each other by changes of local reference frames, local unitary evolutions and multiplication by a constant, and inequivalently entangled otherwise.
Following the terminology of Ref. \cite{spinorent} we refer to the changes of local reference frames and the local unitary evolutions jointly as the {\it local reversible operations}.

In the general case one can identify a number of different properties of the entanglement and use these to distinguish between the inequivalent types of entanglement. By the above argument any such property describing the entanglement must be unchanged by all local reversible operations \cite{popescu,carteret,ben}. Moreover, no such entanglement property should be found in a product state, i.e., a state that can be created using only local resources.

For a system of multiple Dirac spinors we may following Ref. \cite{spinorent} identify the local reversible operations acting on the spinors as the 
set of changes of local reference frames, i.e., the local spinor representations of the proper orthochronous Lorentz transformations and the set of local unitary spinor evolutions generated by the set of physically allowed Dirac Hamiltonians. Given this we have for a system of multiple Dirac particles with definite momenta three conditions that define a spinor entanglement property for pure states of such particles.

\begin{enumerate}
  \item  [(1)] Non-existence for any product state.
  \item[(2)]Invariance under changes of local inertial reference frames, i.e., local Lorentz invariance.
  \item [(3)] Invariance under local evolutions generated by physically allowed Dirac Hamiltonians that act unitarily on any fixed-momenta subspace.
  
\end{enumerate}

Next we can consider the the question of how to describe equivalence classes of states with the same entanglement properties.
A state $\psi_{ABC\dots}$ of a multipartite system belongs to the set $\mathcal{O}_{\psi_{ABC\dots}}$ that is made up of all states that can be obtained from $\psi_{ABC\dots}$ by local reversible operations. Following the terminology of Ref. \cite{spinorent} we call such a set $\mathcal{O}_{\psi_{ABC\dots}}$, an {\it orbit} of the local reversible operations. Any two different orbits are by definition disjoint and the Hilbert space of the system can be completely decomposed into the set of all such orbits. If the orbit $\mathcal{O}_{\psi_{ABC\dots}}$ is identical to the orbit $\mathcal{O}_{\phi_{ABC\dots}}$ up to element wise multiplication by a nonzero constant $c$, i.e., a map $m_c:\psi_{ABC\dots}\to c\psi_{ABC\dots}$ for some $c\in \mathbb{C}-\{0\}$, we consider them equivalent. We denote by $\tilde{\mathcal{O}}_{\psi_{ABC\dots}}$ the equivalence class of all orbits that are related to $\mathcal{O}_{\psi_{ABC\dots}}$ by the maps $m_{c}$ for all $c\in \mathbb{C}-\{0\}$.
Two states that belong to different equivalence classes are different with regard to some physical property that cannot be changed by any local reversible operation.
Therefore two entangled states $\psi_{ABC\dots}$ and $\phi_{ABC\dots}$ for which $\tilde{\mathcal{O}}_{\psi_{ABC\dots}}\neq\tilde{\mathcal{O}}_{\phi_{ABC\dots}}$ are by definition inequivalently entangled. The description of different kinds of entanglement in terms of inequivalence under local reversible operations has been considered for various systems of non-relativistic spin-$\frac{1}{2}$ particles (See e.g. References \cite{popescu,carteret,toni,higuchi,sud,tarrach,toumazet,ekert,grassl,kempe,lind2,car2}).

A method to characterize the different inequivalent types of entanglement in a system of particles is to construct parameters that distinguishes between different equivalence classes of entangled states. If we want such parameters to only distinguish between different inequivalent forms of entanglement and not between any other physical properties, we must require that these parameters are functions that take a constant value on any given equivalence class. Thus the parameters have to be invariant under any local reversible operation and invariant under multiplication of the state by any constant $c\in \mathbb{C}-\{0\}$.
A way to construct such parameters is to find a set of functions $f_i$ on the Hilbert space of the system that are invariant under local reversible operations with determinant 1 and that are homogeneous, i.e., satisfy $f_i(c\psi_{ABC\dots})=c^{k(i)}f_i(\psi_{ABC\dots})$ for $c\in \mathbb{C}-\{0\}$ where $k(i)$ is the degree of homogeneity of $f_i$. 
Given two such homogeneous functions $f_i\neq0$ and $f_j\neq0$ the ratio $f_i^{k(j)}/f_j^{k(i)}$ has degree of homogeneity zero. This ratio is invariant under any local reversible operation as well as invariant under multiplication of a state by any constant $c\in \mathbb{C}-\{0\}$. Functions of this kind can therefore be used to parametrize the set of equivalence classes $\tilde{\mathcal{O}}_{\psi_{ABC\dots}}$.
If we also require that the homogeneous functions $f_i$ take the value zero for all product states they are witnesses of entanglement, i.e., any nonzero value of such a function implies that the state is entangled. This kind of functions are referred to here and in Ref. \cite{spinorent} as {\it entanglement invariants}.

One way to construct entanglement invariants is as homogeneous polynomials in the state coefficients.
The characterization of entanglement using such polynomial invariants has been considered previously for several different systems of non-relativistic spin-$\frac{1}{2}$ particles (See e.g. References \cite{coffman,sud,verstraete2,luque,toumazet,wootters,wootters2,grassl,verstraete,kempe,luq4,miyake1}).
If for two entangled states there exists some ratio of two polynomial entanglement invariants with degree of homogeneity zero that takes different values for the two states it follows that the two states belong to different equivalence classes and thus are inequivalently entangled. However, if all such ratios of polynomial entanglement invariants takes the same values for two states, this does not imply that the two states belong to the same equivalence class.
Thus a given set of polynomial invariants may be able to distinguish between some but not all inequivalent types of entanglement. We say that such a set of polynomial invariants only provides a {\it partial characterization} of the entanglement properties of the system.
In the general case it is impossible to construct a set of polynomial invariants that can distinguish between all equivalence classes due to the possibility of an equivalence class being in the closure of another equivalence class. Since a ratio of two polynomials with degree of homogeneity zero is a continuous function and constant on an equivalence class it cannot distinguish between the class and its closure. Thus polynomial entanglement invariants in general provide only a partial characterization of entanglement properties.

We can consider different cases of physically allowed local unitary evolutions for a system of Dirac spinors. This means we have to consider different cases of local reversible operations. In some cases the polynomial invariants are the same for two different cases of local reversible operations and in some cases they are different. 
In particular the polynomial entanglement invariants constructed for different cases of local reversible operations can have either the same or alternatively different local invariance groups. 
Here we identify four principal physically motivated cases of local invariance group acting on a single particle.

For the case where the local unitary evolution acting on a given particle is generated by Dirac Hamiltonians with a mass term and more generally by Hamiltonians of the form $H^{1,2}(t)+H^{0}(t)$ the group of local unitary evolutions acting on the particle is dense in $G_U^{C\gamma^5}$. Therefore, in this case any polynomial entanglement invariant is invariant, up to a U(1) phase, under $G_U^{C\gamma^5}$ since it is a continuous function. Moreover, any such polynomial entanglement invariant is invariant, up to a U(1) phase, also under $G^{C\gamma^5}$ acting on the particle since it is also Lorentz invariant and $G^{C\gamma^5}$ is the smallest Lie group that contains $G_U^{C\gamma^5}$ and the spinor representation of the proper orthochronous Lorentz group as subgroups, as described in Section \ref{hinn}.

For the case where the local unitary evolution acting on a given particle is generated by Dirac Hamiltonians with zero mass and a coupling to a Yukawa pseudoscalar boson and more generally by Hamiltonians of the form $H^{2,3}(t)+H^{0}(t)$ the group of local unitary evolutions acting on the particle is dense in $G^{C}_U$. Therefore, in this case any polynomial entanglement invariant is invariant, up to a U(1) phase, under $G_U^{C}$ since it is a continuous function. Moreover, any such polynomial entanglement invariant is invariant, up to a U(1) phase, also under $G^{C}$ acting on the particle since it is also Lorentz invariant.

For the case where the local unitary evolution acting on a given particle is generated by Dirac Hamiltonians with zero mass and a no additional couplings the group of local unitary evolutions acting on the particle is dense in $G^{C\gamma^5}_U\cap G^C_U$. Therefore, in this case any polynomial entanglement invariant is invariant, up to a U(1) phase, under $G^{C\gamma^5}_U\cap G^C_U$ since it is a continuous function. Moreover, any such polynomial entanglement invariant is invariant, up to a U(1) phase, also under $G^{C\gamma^5}\cap G^C$ acting on the particle since it is a Lorentz invariant.

Finally we consider case where the local unitary evolution on a given particle is generated by Dirac Hamiltonians
with both a mass term and a coupling to a Yukawa pseudoscalar boson where either the mass or the coupling to the Yukawa pseudoscalar boson is a free parameter.
For this case the group of local unitary evolutions acting on the particle is dense in the smallest Lie group that contains both $G^{C\gamma^5}_U$ and $G^C_U$ as subgroups. This group is U(4). Therefore, in this case any polynomial entanglement invariant is invariant, up to a U(1) phase, under U(4) since it is a continuous function. Moreover, any such polynomial entanglement invariant is invariant, up to a U(1) phase, also under $\mathrm{U}(1)\times\mathrm{SL}(4,\mathbb{C})$ acting on the particle since it is a Lorentz invariant.

From the above four cases we can see that a set of polynomial entanglement invariants for a system of $n$ Dirac particles is invariant, up to a U(1) phase, under a Lie group $G_1\otimes G_2\otimes\dots\otimes G_n$ where each $G_k$ is one out of $G^C$, $G^{C\gamma^5}$, $G^{C\gamma^5}\cap G^C$ and $\mathrm{U}(1)\times\mathrm{SL}(4,\mathbb{C})$. The local groups $G_k$ depend on the physically allowed local evolution.
 
\subsection{Polynomial entanglement invariants in the case of zero momenta free particles in an energy eigenstate and for the case of Weyl particles}
We here consider the restriction of a polynomial entanglement invariant to a subspace.
In particular we consider some physically motivated subspaces.
One such case is a system of Dirac particles with zero momenta that are not coupled to any four-potentials and that are in an eigenstate of the local Dirac Hamiltonians. The state of such a system  belongs to a subspace invariant under local projections by $P_+=\frac{1}{2}(I+\gamma^0)$ or alternatively $P_-=\frac{1}{2}(I-\gamma^0)$. A system of particles of this kind is often identified with a system of non-relativistic free spin-$\frac{1}{2}$ particles or antiparticles. 
Another such case is a system of Weyl particles, i.e., zero mass Dirac particles with definite chiralities. The state of a system of Weyl particles belongs to a subspace invariant under local projections by $P_R=\frac{1}{2}(I+\gamma^5)$ or alternatively $P_L=\frac{1}{2}(I-\gamma^5)$.

Any polynomial entanglement invariant restricted to a subspace reduces to a polynomial that is invariant under the subgroup of the local reversible operations with determinant one that preserve the subspace. Moreover, since
any polynomial entanglement invariant is invariant, up to a U(1) phase, under the local action of either $G^C$, $G^{C\gamma^5}$, $G^{C\gamma^5}\cap G^C$ or $\mathrm{U}(1)\times\mathrm{SL}(4,\mathbb{C}) $ the restriction of the invariant to a subspace is invariant, up to a U(1) phase, under the local action of a subgroup of one of these groups that preserve the subspace. For any of the local groups $G^C$, $G^{C\gamma^5}$, $G^{C\gamma^5}\cap G^C$ or $\mathrm{U}(1)\times\mathrm{SL}(4,\mathbb{C}) $ and any subspace invariant under $P_+$, $P_-$, $P_L$ or $P_R$ the subgroup that preserves the subspace is isomorphic to $\mathrm{U}(1)\times\mathrm{SL}(2,\mathbb{C})$.
Therefore, any polynomial entanglement invariant reduces on a subspace defined by local projectors $P_+$, $P_-$, $P_L$ or $P_R$ to a polynomial that is invariant under local action by a group isomorphic to $\mathrm{SL}(2,\mathbb{C})$ or reduces to zero.

For a system of two Dirac spinors the subspaces defined by local invariance under $P_R$ or  $P_L$ and the subspaces defined by local invariance under $P_+$ or $P_-$ were considered in \cite{spinorent}.
The polynomial Lorentz invariants constructed in \cite{spinorent} were found to reduce on these subspaces either to zero or to a local $\mathrm{SL}(2,\mathbb{C})$ invariant polynomial, the Wootters concurrence \cite{wootters,wootters2}. The Wootters concurrence is a degree 2 polynomial in the state coefficients that was constructed to describe the spin entanglement of two non-relativistic spin-$\frac{1}{2}$ particles and other two-level systems. It generates the entire algebra of local $\mathrm{SL}(2,\mathbb{C})$ invariants for two non-relativistic spin-$\frac{1}{2}$ particles.

For higher numbers of particles a number of local $\mathrm{SL}(2,\mathbb{C})$ invariant polynomials have been constructed. For three non-relativistic spin-$\frac{1}{2}$ particles a polynomial of degree 4 that is invariant under local $\mathrm{SL}(2,\mathbb{C})$ is the Coffman-Kundu-Wootters 3-tangle \cite{coffman}. This polynomial generates the entire algebra of local $\mathrm{SL}(2,\mathbb{C})$ invariants for three non-relativistic spin-$\frac{1}{2}$ particles.
For a system of four non-relativistic spin-$\frac{1}{2}$ particles four polynomials of degrees 2, 4, 4, and 6 that are invariant under local $\mathrm{SL}(2,\mathbb{C})$ have been found by Luque and Thibon \cite{luque}. These four polynomials generate the entire algebra of local $\mathrm{SL}(2,\mathbb{C})$ invariants for four non-relativistic spin-$\frac{1}{2}$ particles.
Moreover, degree 2 polynomials invariant under local $\mathrm{SL}(2,\mathbb{C})$ for arbitrary even numbers of spin-$\frac{1}{2}$ particles have been described by Wong and Christensen \cite{wong}.

\subsection{Entanglement measures}
Here we briefly comment on a commonly used tool for describing entanglement, the so called {\it entanglement measures} \cite{entmes}. These are functions that are zero for product states and satisfy the condition of non-increase on average under local operations assisted by classical communication (LOCC)\cite{vidal}. This condition is called entanglement monotonicity and entanglement properties satisfying this condition can be quantified by the entanglement measures.  The entanglement measures are related to entanglement invariants in that the condition of entanglement monotonicity implies invariance under local reversible operations.  
However since classical communication between spacelike separated laboratories is not possible a different physical scenario where timelike or null paths connect the laboratories are required for LOCC. In this kind of scenario operations made in one laboratory are not prevented from influencing other laboratories along timelike or null curves by the causal structure of the spacetime. Therefore, other assumptions regarding causal influences between laboratories must be made for the notion of local operations to be meaningful. Given such assumptions one then needs to characterize the set of LOCC to identify entanglement properties that satisfy the condition of entanglement monotonicity. Considering this kind of physical scenario and characterizing the set of LOCC goes beyond the scope of this work and is left as an open problem.

\section{Constructing candidate entanglement invariants through tensor contractions}\label{ent}

We here describe a method for constructing polynomial entanglement invariants for a system of multiple Dirac spinors.
To do this we utilize the properties of the bilinear forms described in Sections \ref{invariants}, \ref{ham} and \ref{hinn} to construct polynomial entanglement invariants of the kind described in Section \ref{spen}. 
In particular the two bilinear forms $\psi^T(x)C\gamma^5\varphi(x)$ and $\psi^T(x)C\varphi(x)$ are both pointwise Lorentz invariant. Furthermore, if $\psi$ and $\varphi$ belong to a subspace defined by fixed particle momenta and spanned by spinorial degrees of freedom both $\psi^T(x)C\gamma^5\varphi(x)$ and $\psi^T(x)C\varphi(x)$ are invariant, up to a U(1) phase, under any evolution that acts unitarily on the subspace and is generated by Dirac Hamiltonians and zero-mass Dirac Hamiltonians, respectively. Finally, both $\psi^T(x)C\psi(x)$ and $\psi^T(x)C\gamma^5\psi(x)$ are identically zero since they are skew-symmetric.

For the case where the physically allowed local unitary evolutions of the spinors are described by Dirac Hamiltonians with zero-mass or arbitrary mass, or more generally where the group of physically allowed local transformations belong to $G^C$ or to $G^{C\gamma^5}$, the bilinear forms $\psi^TC\varphi$ and $\psi^TC\gamma^5\varphi$, respectively, have the desired local invariance properties of an entanglement invariant as described in Section \ref{spen}. Below we show how polynomials with these invariance properties can be constructed for any number of Dirac spinors in a way that generalizes the method used for the case of two Dirac spinors in Ref. \cite{spinorent}.

The physical scenario as described in Section \ref{dir} is one where a number of spacelike separated observers each has their own laboratory holding a Dirac or Weyl particle. We give the first five observers the names Alice, Bob, Charlie, David, and Erin, respectively, for convenience. Then we let the particles be in a joint state and assume that operations by one observer on the shared system can be made jointly with the other observers operations, i.e., assume that operations made by different observers commute. Further we make the assumption that a tensor product structure can be used to describe the shared system and use the tensor products of local basis elements $\phi_{j_A}e^{i\bold{k_A}\cdot\bold{x_A}}\otimes \phi_{j_B}e^{i\bold{k_B}\cdot\bold{x_B}}\otimes\phi_{j_C}e^{i\bold{k_C}\cdot\bold{x_C}}\otimes\dots$ as a basis. 
Then the state can be expanded in this basis as
\begin{eqnarray}
\psi_{ABC\dots}(t)=\sum_{\substack{j_A,j_B,j_C,\dots\\\bold{k_A},\bold{k_B},\bold{k_C},\dots} }\psi_{\substack{j_A,j_B,j_C,\dots\\\bold{k_A},\bold{k_B},\bold{k_C},\dots}}(t)\phi_{j_A}e^{i\bold{k_A}\cdot\bold{x_A}}\otimes \phi_{j_B}e^{i\bold{k_B}\cdot\bold{x_B}}\otimes \dots,
\end{eqnarray}
where $\psi_{j_A,j_B,j_C,\dots\bold{k_A},\bold{k_B},\bold{k_C},\dots}(t)$ are complex numbers. 

Next we consider a state that belongs to a subspace where the momenta $\bold{k_A},\bold{k_B},\bold{k_C},\dots$ are fixed, i.e., a subspace spanned by only the spinorial degrees of freedom. Then we can suppress the indices $\bold{k_A},\bold{k_B},\bold{k_C},\dots$ in the description of the state and let $\psi_{jkl\dots}\equiv \psi_{j_A,k_B,l_C,\dots\bold{k_A},\bold{k_B},\bold{k_C},\dots}$.
We can arrange these coefficients $\psi_{jkl\dots}$ as a tensor, i.e., a multi-dimensional array, by letting the spinor basis indices $jkl\dots$ be the tensor component indices. Then the state coefficients $\psi_{jk}$ of a bipartite state form a two-dimensional tensor, i.e., a matrix, where $j$ and $k$ are the column and row indices respectively (See Ref. \cite{spinorent}). Likewise for a tripartite state the state coefficients $\psi_{jkl}$ form a three-dimensional tensor where $j$, $k$ and $l$ are the indices of the three different dimensions, respectively. See Fig. \ref{tretens} for an illustration. In this way an $n$-partite state corresponds to an $n$-dimensional tensor. Let us denote this tensor $\Psi^{ABC\dots}$, and its components $\Psi^{ABC\dots}_{jkl\dots}\equiv\psi_{jkl\dots}$.

\begin{figure}
\centering
\includegraphics[scale=1]{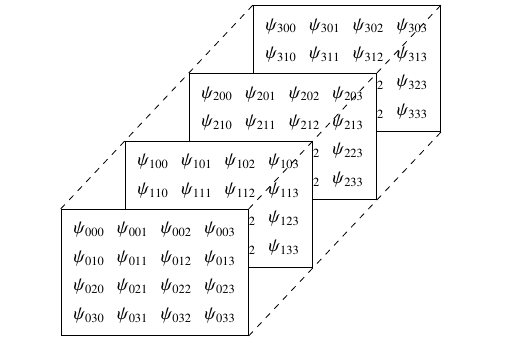}
\caption{Graphical representation of the three-dimensional tensor corresponding to the state of three Dirac spinors. \label{tretens}}
\end{figure}

Transformations $S^A$ on Alice's part of the system act on the first index of  $\Psi^{ABC\dots}$. Likewise, transformations $S^B$ on Bob's part of the system act on the second index, transformations $S^C$ on Charlies's part of the system act on the third index, and transformations on the $n$th observer's part act on the $n$th index

\begin{eqnarray}
\Psi^{ABC\dots}_{jkl\dots}\to\sum_{mno\dots} S^A_{jm}S^{B}_{kn}S^{C}_{lo}\dots\Psi^{ABC\dots}_{mno\dots} .
\end{eqnarray}

If we consider two copies of $\Psi^{ABC\dots}$ together with the matrix $C$ we can construct products of tensor components as $\Psi^{ABC\dots}_{jkl\dots}C_{jm}\Psi^{ABC\dots}_{mno\dots}$ and then take the sum over $j$ and $m$. This sum is invariant under the spinor representation of the proper orthochronous Lorentz group acting on Alices particle. This follows from the properties of the bilinear form $\psi^T C \varphi$ discussed in section \ref{invariants} since $\sum_{jm}\Psi^{ABC\dots}_{jkl\dots}C_{jm}\Psi^{ABC\dots}_{mno\dots}$ is a bilinear form for every fixed set of indices $kl\dots$ and $no\dots$.
In particular, for a transformation $S(\Lambda)$ in Alice's lab we have

\begin{eqnarray}
\sum_{jm}\Psi^{ABC\dots}_{jkl\dots}C_{jm}\Psi^{ABC\dots}_{mno\dots}\to &&\sum_{jmpq}\Psi^{ABC\dots}_{jkl\dots}S^T(\Lambda)_{jp}C_{pq}S(\Lambda)_{qm}\Psi^{ABC\dots}_{mno\dots}\nonumber\\
=&&\sum_{jm}\Psi^{ABC\dots}_{jkl\dots}C_{jm}\Psi^{ABC\dots}_{mno\dots}.
\end{eqnarray}
We refer to the construction $\sum_{jm}\Psi^{ABC\dots}_{jkl\dots}C_{jm}\Psi^{ABC\dots}_{mno\dots}$ as a {\it tensor sandwich contraction} where $C$ is being sandwiched between two copies of $\Psi^{ABC\dots}$.
Similarly we can make sandwich contractions over a pair of indices corresponding to any other observers particle with the matrix $C$ being sandwiched.
Furthermore, we can make sandwich contractions also with $C\gamma^5$ sandwiched between the copies of $\Psi^{ABC\dots}$.

Now consider an even number $n$ of copies of $\Psi^{ABC\dots}$. Next, for a given observer split the set of indices corresponding to that observer's particle into $n/2$ pairs. Then sandwich contract every such pair of indices with either $C$ or $C\gamma^5$ sandwiched inbetween the two copies of $\Psi^{ABC\dots}$, as described above. Then repeat this procedure for every other observer. The result is a scalar that is invariant under the spinor representation of the proper orthochronous Lorentz group in any lab. In particular it is a polynomial of degree $n$ in the state coefficients $\psi_{jkl\dots}$. Note that this procedure is very similar to Cayley's $\Omega$ process for constructing polynomial invariants under the special linear group \cite{omega}.

Due to the antisymmetry of $C$ and $C\gamma^5$ the polynomials constructed in this way are identically zero if $\Psi^{ABC\dots}$ can be factored as $\Psi^{A}\otimes\Psi^{BC\dots}$, i.e., if the state is a product state over the partition $A|BC\dots$. The same holds if any other observer's particle is in a product state with the rest of the particles.
On the other hand, a polynomial constructed in this way is not necessarily zero if the state is a product state over a partition that splits the particles into sets where each set has two or more particles. In this case the polynomial factorizes into polynomials defined on the different sets of particles in the partitioning. 

All polynomials constructed in this way are invariant, up to a U(1) phase, under local unitary evolutions generated by zero-mass Dirac Hamiltonians, as follows from the discussion in Section \ref{ham}. This also holds for zero-mass Dirac Hamiltonians with additional terms that are second degree in gamma matrices such as a Semenoff mass term \cite{semenoff} or Haldane mass term \cite{haldane}. More generally all such polynomial are invariant, up to a U(1) phase, under the group $G^{C}\cap G^{C\gamma^5}$ acting locally on any spinor.

If all indices corresponding to a given observer's particle are sandwich contracted with the matrix $C$ sandwiched for each pair of indices, the resulting polynomial is invariant, up to a U(1) phase, also for evolutions of that spinor generated by zero-mass Dirac Hamiltonians with additional terms that are second degree in gamma matrices or third degree in gamma matrices such as a pseudo-scalar Yukawa term or a Pauli coupling. More generally the polynomial is invariant, up to a U(1) phase, under the group $G^C$ acting on the spinor.

If all indices corresponding to a given observer's particle are sandwich contracted with the matrix $C\gamma^5$ sandwiched for each pair of indices, the resulting polynomial is invariant, up to a U(1) phase, for evolutions of that spinor generated by arbitrary-mass Dirac Hamiltonians. More generally the polynomial is invariant, up to a U(1) phase, under the group $G^{C\gamma^5}$ acting on the spinor.

If the indices corresponding to a given observer's particle are sandwich contracted with zero or an even positive number of the matrix $C\gamma^5$ sandwiched, the resulting polynomial is invariant under P in the lab of the given observer. If the number of sandwiched $C\gamma^5$ is odd the polynomial changes sign under P.

The polynomials constructed in this way satisfy the requirements described in Section \ref{spen} and are therefore considered potential candidates for describing entanglement properties of a system of multiple Dirac spinors. Note that the absolute values are invariant under the local action of the Dirac group, i.e., the finite group generated by the gamma matrices $\gamma^0,\gamma^1,\gamma^2,\gamma^3$, and in particular invariant under the local P, CT and CPT transformations.
Note also that by using $C$ and $C\gamma^5$ to construct the invariants we have in essence used the T and CP transformations. Using such inversion transformations follows the same general idea of using "state inversion" transformations \cite{wootters2,uhlmann,rungta} as the construction of the Wootters concurrence for the case of two non-relativistic spin-$\frac{1}{2}$ particles.

The computational difficulty in constructing invariant polynomials through sandwich contractions rises sharply with the number of spinors and the number of copies of $\Psi^{ABC\dots}$ involved in the contractions, i.e., the degree of the polynomials. 
However for the lowest polynomial degrees such invariants can be obtained with relatively small difficulty, in particular for three and four spinors. Furthermore, degree 2 polynomials can be obtained straightforwardly for all even numbers of spinors.

\section{The case of three spinors}\label{three}

For three Dirac spinors the state coefficients can be arranged as a $4\times 4\times 4$ tensor $\Psi^{ABC}$.
Transformations $S^A,S^B,S^C$ acting locally on Alice's, Bob's and Charlie's particle respectively are described as
\begin{eqnarray}
\Psi^{ABC}_{ijk}\to\sum_{lmn} S^A_{il}S^B_{jm}S^C_{kn}\Psi^{ABC}_{lmn}.
\end{eqnarray}

Lorentz invariants of degree 2 can be constructed as tensor sandwich contractions of the form $\sum_{ijklmn}X_{li}X_{mj}X_{nk}\Psi^{ABC}_{ijk}\Psi^{ABC}_{lmn}$ involving two copies of $\Psi^{ABC}$ where either $X=C$ or $X=C\gamma^5$ for each instance. However all the 8 degree 2 polynomials that can be constructed in this way are identically zero due to the antisymmetry of $C$ and $C\gamma^5$. 

The next lowest degree is 4.
From the above follows that invariants of degree 4 that factorize as a product of two degree 2 polynomials are identically zero.
Thus it remains to consider the tensor sandwich contractions involving four copies of $\Psi^{ABC}$ that do not factorize into two degree 2 polynomials. There are four different such ways to pair up the tensor indices of the four copies.  In writing these sandwich contractions we leave out the summation sign in the following, with the understanding that repeated indices are summed over. We also suppress the superscript $ABC$ of $\Psi$. The four different ways to contract the indices are

\begin{eqnarray}\label{tens}
I_a&=& X_{ij}X_{mk}X_{nl}\Psi_{jkl}\Psi_{qmn}X_{qr}X_{ps}X_{ut}\Psi_{rst}\Psi_{ipu},\nonumber\\
I_b&=& X_{ij}X_{mk}X_{nl}\Psi_{jkl}\Psi_{ipn}X_{qr}X_{ps}X_{ut}\Psi_{rst}\Psi_{qmu},\nonumber\\
I_c&=& X_{ij}X_{mk}X_{nl}\Psi_{jkl}\Psi_{imu}X_{qr}X_{ps}X_{ut}\Psi_{rst}\Psi_{qpn},\nonumber\\
I_d&=&X_{ij}X_{mk}X_{nl}\Psi_{jkl}\Psi_{ipt}X_{qr}X_{ps}X_{ut}\Psi_{rsn}\Psi_{qmu}.
\end{eqnarray}

\begin{figure}
\centering
\begin{subfigure}{.3\linewidth}
  \centering
  \includegraphics[scale=1]{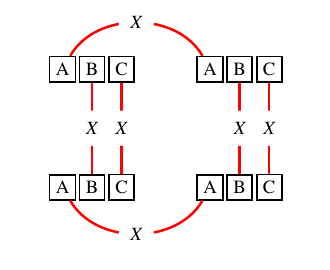}
  \caption{$I_a$}
\end{subfigure}%
\hspace{2cm}%
\begin{subfigure}{.3\linewidth}
  \centering
  \includegraphics[scale=1]{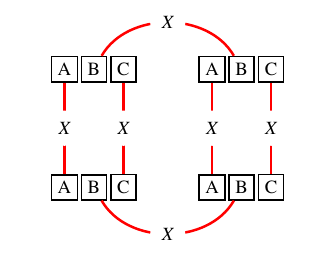}
   \caption{$I_b$}
\end{subfigure}\\
\begin{subfigure}{4cm}
  \centering
  \includegraphics[scale=1]{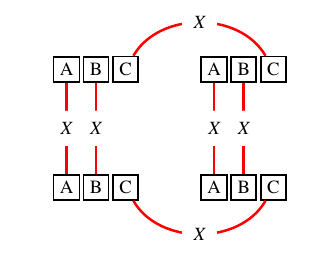}
   \caption{$I_c$}
\end{subfigure}
\hspace{2cm}%
\begin{subfigure}{4cm}
  \centering
  \includegraphics[scale=1]{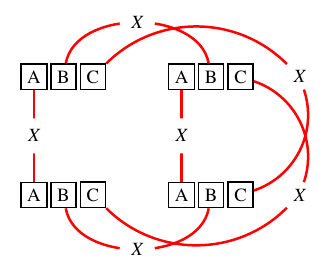}
   \caption{$I_d$}
\end{subfigure}
\caption{Graph representations of the four tensor sandwich contractions $I_a$, $I_b$, $I_c$, and $I_d$. Each of the four copies of $\Psi^{ABC}$ is represented by three boxes corresponding to the three tensor indices. The tensor sandwich contractions are represented by red lines connecting the contracted indices broken by $X$ representing the sandwiched tensor.}
\label{ris}
\end{figure}

The contraction $I_a$ is invariant with respect to a permutation of laboratories B and C. Similarly, the contraction $I_b$ is invariant with respect to a permutation of laboratories A and C and the contraction $I_c$ is invariant with respect to a permutation of laboratories A and B. The final contraction $I_d$ is invariant with respect to any permutation of the laboratories. 
Moreover, the different tensor contractions can be represented as graphs which provides an additional way to understand them.
See Fig. \ref{ris} for a graph representation of $I_a$, $I_b$, $I_c$, and $I_d$.

For each way to pair up the tensor indices we consider $(2^6-2^3)/2+2^3=36$ ways to choose the $X$s as either $C$ or $ C\gamma^5$. There is $2^3$ such choices where for each lab both $X$s corresponding to the lab are identical and there is $2^6-2^3$ choices that are not of this kind. For the latter case however each such choice for $I_a$, $I_b$, $I_c$, and $I_d$ is equivalent to at least one other choice. This can be understood in terms of graph isomorphisms of the graph representations of $I_a$, $I_b$, $I_c$, and $I_d$ (See fig. \ref{ris}), and reduces the number of such choices that need to be considered to $(2^6-2^3)/2$.
Thus we can consider a total of $36\times 4=144$ ways to construct polynomials by this method. These 144 different tensor sandwich contractions were calculated and their linear independence tested. Due to the large number of terms in the polynomials they are not all given in a fully written out form, but a few selected examples are given fully written out in \ref{polllu}.
The resulting polynomials are not all linearly independent and not all unique but one can select a set of 67 linearly independent polynomials. In the following we use the abbreviated notation $C^5\equiv C\gamma^5$ in giving the formal expressions for the polynomials. 
The polynomials can be divided into subsets based on their properties under parity inversion P in the different labs. This is useful since a polynomial cannot be linearly dependent on any polynomial with different behaviour under parity inversion P.

\subsection{Polynomials invariant under P in all labs}
There are 32 polynomials that are tensor sandwich contractions of the form in Eq. (\ref{tens}) and also invariant under P in Alice's, Bob's and in Charlie's lab. These Lorentz invariants are not all linearly independent, but a set of 23 linearly independent polynomials can be selected. Thus the polynomials span a 23 dimensional space. The 32 polynomials are

\begin{eqnarray}
I_{2a}&=& C^5_{ij}C^5_{mk}C^5_{nl}\Psi_{jkl}\Psi_{pmn}C^5_{pr}C^5_{qs}C^5_{ut}\Psi_{rst}\Psi_{iqu},\nonumber\\
I_{2b}&=& C^5_{ij}C^5_{mk}C^5_{nl}\Psi_{jkl}\Psi_{ipn}C^5_{qr}C^5_{ps}C^5_{ut}\Psi_{rst}\Psi_{qmu},\nonumber\\
I_{2c}&=& C^5_{ij}C^5_{mk}C^5_{nl}\Psi_{jkl}\Psi_{imp}C^5_{qr}C^5_{us}C^5_{pt}\Psi_{rst}\Psi_{qun},\nonumber\\
I_{2d}&=& C^5_{ij}C^5_{mk}C^5_{nl}\Psi_{jkl}\Psi_{ipt}C^5_{qr}C^5_{ps}C^5_{ut}\Psi_{rsn}\Psi_{qmu},
\end{eqnarray}
\begin{eqnarray}
I_{3a}&=&C_{ij}C_{mk}C_{nl}\Psi_{jkl}\Psi_{qmn}C_{qr}C_{ps}C_{ut}\Psi_{rst}\Psi_{ipu},\nonumber\\
I_{3b}&=&C_{ij}C_{mk}C_{nl}\Psi_{jkl}\Psi_{ipn}C_{qr}C_{ps}C_{ut}\Psi_{rst}\Psi_{qmu},\nonumber\\
I_{3c}&=&C_{ij}C_{mk}C_{nl}\Psi_{jkl}\Psi_{imu}C_{qr}C_{ps}C_{ut}\Psi_{rst}\Psi_{qpn},\nonumber\\
I_{3d}&=&C_{ij}C_{mk}C_{nl}\Psi_{jkl}\Psi_{ipt}C_{qr}C_{ps}C_{ut}\Psi_{rsn}\Psi_{qmu},
\end{eqnarray}
\begin{eqnarray}
I_{4a}&=&C^5_{ij}C_{mk}C^5_{nl}\Psi_{jkl}\Psi_{pmn}C^5_{pr}C_{qs}C^5_{ut}\Psi_{rst}\Psi_{iqu},\nonumber\\
I_{4b}&=&C^5_{ij}C_{mk}C^5_{nl}\Psi_{jkl}\Psi_{ipn}C^5_{qr}C_{ps}C^5_{ut}\Psi_{rst}\Psi_{qmu},\nonumber\\
I_{4c}&=&C^5_{ij}C_{mk}C^5_{nl}\Psi_{jkl}\Psi_{imp}C^5_{qr}C_{us}C^5_{pt}\Psi_{rst}\Psi_{qun},\nonumber\\
I_{4d}&=&C^5_{ij}C_{mk}C^5_{nl}\Psi_{jkl}\Psi_{ipt}C^5_{qr}C_{ps}C^5_{ut}\Psi_{rsn}\Psi_{qmu},
\end{eqnarray}
\begin{eqnarray}
I_{5a}&=&C^5_{ij}C^5_{mk}C_{nl}\Psi_{jkl}\Psi_{pmn}C^5_{pr}C^5_{qs}C_{ut}\Psi_{rst}\Psi_{iqu},\nonumber\\
I_{5b}&=&C^5_{ij}C^5_{mk}C_{nl}\Psi_{jkl}\Psi_{ipn}C^5_{qr}C^5_{ps}C_{ut}\Psi_{rst}\Psi_{qmu},\nonumber\\
I_{5c}&=&C^5_{ij}C^5_{mk}C_{nl}\Psi_{jkl}\Psi_{imp}C^5_{qr}C^5_{us}C_{pt}\Psi_{rst}\Psi_{qun},\nonumber\\
I_{5d}&=&C^5_{ij}C^5_{mk}C_{nl}\Psi_{jkl}\Psi_{ipt}C^5_{qr}C^5_{ps}C_{ut}\Psi_{rsn}\Psi_{qmu},
\end{eqnarray}
\begin{eqnarray}
I_{6a}&=&C_{ij}C_{mk}C^5_{nl}\Psi_{jkl}\Psi_{pmn}C_{pr}C_{qs}C^5_{ut}\Psi_{rst}\Psi_{iqu},\nonumber\\
I_{6b}&=&C_{ij}C_{mk}C^5_{nl}\Psi_{jkl}\Psi_{ipn}C_{qr}C_{ps}C^5_{ut}\Psi_{rst}\Psi_{qmu},\nonumber\\
I_{6c}&=&C_{ij}C_{mk}C^5_{nl}\Psi_{jkl}\Psi_{imp}C_{qr}C_{us}C^5_{pt}\Psi_{rst}\Psi_{qun},\nonumber\\
I_{6d}&=&C_{ij}C_{mk}C^5_{nl}\Psi_{jkl}\Psi_{ipt}C_{qr}C_{ps}C^5_{ut}\Psi_{rsn}\Psi_{qmu},
\end{eqnarray}
\begin{eqnarray}
I_{7a}&=&C_{ij}C^5_{mk}C^5_{nl}\Psi_{jkl}\Psi_{pmn}C_{pr}C^5_{qs}C^5_{ut}\Psi_{rst}\Psi_{iqu},\nonumber\\
I_{7b}&=&C_{ij}C^5_{mk}C^5_{nl}\Psi_{jkl}\Psi_{ipn}C_{qr}C^5_{ps}C^5_{ut}\Psi_{rst}\Psi_{qmu},\nonumber\\
I_{7c}&=&C_{ij}C^5_{mk}C^5_{nl}\Psi_{jkl}\Psi_{imp}C_{qr}C^5_{us}C^5_{pt}\Psi_{rst}\Psi_{qun},\nonumber\\
I_{7d}&=&C_{ij}C^5_{mk}C^5_{nl}\Psi_{jkl}\Psi_{ipt}C_{qr}C^5_{ps}C^5_{ut}\Psi_{rsn}\Psi_{qmu},
\end{eqnarray}
\begin{eqnarray}
I_{8a}&=&C_{ij}C^5_{mk}C_{nl}\Psi_{jkl}\Psi_{pmn}C_{pr}C^5_{qs}C_{ut}\Psi_{rst}\Psi_{iqu},\nonumber\\
I_{8b}&=&C_{ij}C^5_{mk}C_{nl}\Psi_{jkl}\Psi_{ipn}C_{qr}C^5_{ps}C_{ut}\Psi_{rst}\Psi_{qmu},\nonumber\\
I_{8c}&=&C_{ij}C^5_{mk}C_{nl}\Psi_{jkl}\Psi_{imp}C_{qr}C^5_{us}C_{pt}\Psi_{rst}\Psi_{qun},\nonumber\\
I_{8d}&=&C_{ij}C^5_{mk}C_{nl}\Psi_{jkl}\Psi_{ipt}C_{qr}C^5_{ps}C_{ut}\Psi_{rsn}\Psi_{qmu},
\end{eqnarray}
\begin{eqnarray}
I_{9a}&=&C^5_{ij}C_{mk}C_{nl}\Psi_{jkl}\Psi_{pmn}C^5_{pr}C_{qs}C_{ut}\Psi_{rst}\Psi_{iqu},\nonumber\\
I_{9b}&=&C^5_{ij}C_{mk}C_{nl}\Psi_{jkl}\Psi_{ipn}C^5_{qr}C_{ps}C_{ut}\Psi_{rst}\Psi_{qmu},\nonumber\\
I_{9c}&=&C^5_{ij}C_{mk}C_{nl}\Psi_{jkl}\Psi_{imp}C^5_{qr}C_{us}C_{pt}\Psi_{rst}\Psi_{qun},\nonumber\\
I_{9d}&=&C^5_{ij}C_{mk}C_{nl}\Psi_{jkl}\Psi_{ipt}C^5_{qr}C_{ps}C_{ut}\Psi_{rsn}\Psi_{qmu}.
\end{eqnarray}
The linear dependence of the polynomials can be described by the 9 equations
\begin{eqnarray}\label{grunt}
0&=& I_{5b} - I_{5c} + I_{5d} + I_{8b} - I_{8c} +I_{8d}, \nonumber\\
0&=& I_{3b} - I_{3c}+ I_{3d}  + I_{9b} -  I_{9c}  +I_{9d},\nonumber\\
0&=& I_{5a}- I_{5c} - I_{5d} + I_{9a} - I_{9c}-I_{9d}, 
\end{eqnarray}
\begin{eqnarray}
0&=&I_{8d} + I_{8a} - I_{8b} + 2 I_{7a} - I_{7b} - I_{7c} + I_{6a} + I_{6b} - 
 2 I_{6c} + I_{4b} - I_{4c} + I_{4d},
\nonumber\\
0&=&I_{2a} + I_{2b} - 2 I_{2c} + I_{4a} + I_{4b} - 2 I_{4c} + I_{6a} + I_{6b} - 
 2 I_{6c} + I_{7a} + I_{7b} - 2 I_{7c},
\nonumber\\
0&=&2I_{2a}- I_{2b} -  I_{2c} +2 I_{4a} - I_{4b} -  I_{4c}  + 
 2 I_{5a} - I_{5b} - I_{5c} + 
 2 I_{9a} - I_{9b} - I_{9c},
\nonumber\\
0&=&I_{2a}-2 I_{2b} +  I_{2c}  + I_{5a} - 2 I_{5b} +  I_{5c} + 
  I_{7a} -2I_{7b} + I_{7c} + I_{8a} - 2 I_{8b} + I_{8c},
\nonumber\\
0&=&I_{3a} + I_{3b} - 2 I_{3c} + I_{5a} + I_{5b} - 2 I_{5c} + I_{8a} + I_{8b} - 
 2 I_{8c} + I_{9a} + I_{9b}- 2 I_{9c},
\nonumber\\
0&=&2 I_{3a} - I_{3b} - I_{3c} + 2 I_{6a} - I_{6b} - I_{6c} + 
 2 I_{7a} - I_{7b} - I_{7c} + 2 I_{8a} - I_{8b} - I_{8c}.
\end{eqnarray}

\subsection{Polynomials only invariant under P in Alice's and Charlie's labs}
There is 16 polynomials that are tensor sandwich contractions of the form in Eq. (\ref{tens}) and also invariant under P in Alice's and Charlie's lab, but not in Bob's lab. These Lorentz invariants are not all linearly independent, but a set of 8 linearly independent polynomials can be selected. Thus the polynomials span an 8 dimensional space. The 16 polynomials are

\begin{eqnarray}
I_{10a}&=& C^5_{ij}C^5_{mk}C^5_{nl}\Psi_{jkl}\Psi_{pmn}C^5_{pr}C_{qs}C^5_{ut}\Psi_{rst}\Psi_{iqu},\nonumber\\
I_{10b}&=&C^5_{ij}C^5_{mk}C^5_{nl}\Psi_{jkl}\Psi_{ipn}C^5_{qr}C_{ps}C^5_{ut}\Psi_{rst}\Psi_{qmu},\nonumber\\
I_{10c}&=&C^5_{ij}C^5_{mk}C^5_{nl}\Psi_{jkl}\Psi_{imp}C^5_{qr}C_{us}C^5_{pt}\Psi_{rst}\Psi_{qun},\nonumber\\
I_{10d}&=&C^5_{ij}C^5_{mk}C^5_{nl}\Psi_{jkl}\Psi_{ipt}C^5_{qr}C_{ps}C^5_{ut}\Psi_{rsn}\Psi_{qmu},
\end{eqnarray}
\begin{eqnarray}
 I_{16a}&=&C_{ij}C_{mk}C_{nl}\Psi_{jkl}\Psi_{pmn}C_{pr}C^5_{qs}C_{ut}\Psi_{rst}\Psi_{iqu},\nonumber\\
 I_{16b}&=&C_{ij}C_{mk}C_{nl}\Psi_{jkl}\Psi_{ipn}C_{qr}C^5_{ps}C_{ut}\Psi_{rst}\Psi_{qmu},\nonumber\\
I_{16c}&=&C_{ij}C_{mk}C_{nl}\Psi_{jkl}\Psi_{imp}C_{qr}C^5_{us}C_{pt}\Psi_{rst}\Psi_{qun},\nonumber\\
I_{16d}&=&C_{ij}C_{mk}C_{nl}\Psi_{jkl}\Psi_{ipt}C_{qr}C^5_{ps}C_{ut}\Psi_{rsn}\Psi_{qmu},
\end{eqnarray}
\begin{eqnarray}
 I_{17a}&=&C_{ij}C^5_{mk}C^5_{nl}\Psi_{jkl}\Psi_{pmn}C_{pr}C_{qs}C^5_{ut}\Psi_{rst}\Psi_{iqu},\nonumber\\
I_{17b}&=&C_{ij}C^5_{mk}C^5_{nl}\Psi_{jkl}\Psi_{ipn}C_{qr}C_{ps}C^5_{ut}\Psi_{rst}\Psi_{qmu},\nonumber\\
I_{17c}&=&C_{ij}C^5_{mk}C^5_{nl}\Psi_{jkl}\Psi_{imp}C_{qr}C_{us}C^5_{pt}\Psi_{rst}\Psi_{qun},\nonumber\\
I_{17d}&=&C_{ij}C^5_{mk}C^5_{nl}\Psi_{jkl}\Psi_{ipt}C_{qr}C_{ps}C^5_{ut}\Psi_{rsn}\Psi_{qmu},
\end{eqnarray}
\begin{eqnarray}
 I_{18a}&=&C^5_{ij}C^5_{mk}C_{nl}\Psi_{jkl}\Psi_{pmn}C^5_{pr}C_{qs}C_{ut}\Psi_{rst}\Psi_{iqu},\nonumber\\
I_{18b}&=&C^5_{ij}C^5_{mk}C_{nl}\Psi_{jkl}\Psi_{ipn}C^5_{qr}C_{ps}C_{ut}\Psi_{rst}\Psi_{qmu},\nonumber\\
I_{18c}&=&C^5_{ij}C^5_{mk}C_{nl}\Psi_{jkl}\Psi_{imp}C^5_{qr}C_{us}C_{pt}\Psi_{rst}\Psi_{qun},\nonumber\\
I_{18d}&=&C^5_{ij}C^5_{mk}C_{nl}\Psi_{jkl}\Psi_{ipt}C^5_{qr}C_{ps}C_{ut}\Psi_{rsn}\Psi_{qmu}.
\end{eqnarray}
The linear dependence of the polynomials can be described by the 8 equations
 \begin{eqnarray}
I_{16a} - I_{16c} - I_{16d}&=&0,\nonumber\\
I_{10a} - I_{10c} - I_{10d}&=&0,\nonumber\\
I_{18a} - I_{18c} - I_{18d}&=&0,\nonumber\\
I_{17a} - I_{17c} - I_{17d}&=&0,
 \end{eqnarray}
 \begin{eqnarray}\label{gnu}
2 I_{10a} - I_{10b} - I_{10c} + 
 2 I_{18a} - I_{18b} - I_{18c}&=&0,\nonumber\\
I_{10a} + I_{10b} - 2 I_{10c} + I_{17a} + I_{17b}- 2I_{17c}&=&0,\nonumber\\
I_{16a} + I_{16b} - 2 I_{16c} + I_{18a} + I_{18b}- 2I_{18c}&=&0,\nonumber\\
2 I_{16a} - I_{16b} - I_{16c} + 2 I_{17a} - I_{17b} - I_{17c}&=&0.
\end{eqnarray}

\subsection{Polynomials only invariant under P in Alice's and Bob's labs}
There is 16 polynomials that are tensor sandwich contractions of the form in Eq. (\ref{tens}) and also invariant under P in Alice's and Bob's lab, but not in Charlie's lab. These Lorentz invariants are not all linearly independent, but a set of 8 linearly independent polynomials can be selected. Thus the polynomials span an 8 dimensional space. The 16 polynomials are

\begin{eqnarray}
I_{19a}&=&C_{ij}C_{mk}C_{nl}\Psi_{jkl}\Psi_{pmn}C_{pr}C_{qs}C^5_{ut}\Psi_{rst}\Psi_{iqu},\nonumber\\
I_{19b}&=&C_{ij}C_{mk}C_{nl}\Psi_{jkl}\Psi_{ipn}C_{qr}C_{ps}C^5_{ut}\Psi_{rst}\Psi_{qmu},\nonumber\\
I_{19c}&=&C_{ij}C_{mk}C_{nl}\Psi_{jkl}\Psi_{imp}C_{qr}C_{us}C^5_{pt}\Psi_{rst}\Psi_{qun},\nonumber\\
I_{19d}&=&C_{ij}C_{mk}C_{nl}\Psi_{jkl}\Psi_{ipt}C_{qr}C_{ps}C^5_{ut}\Psi_{rsn}\Psi_{qmu},
\end{eqnarray}
\begin{eqnarray}
I_{20a}&=&C^5_{ij}C^5_{mk}C^5_{nl}\Psi_{jkl}\Psi_{pmn}C^5_{pr}C^5_{qs}C_{ut}\Psi_{rst}\Psi_{iqu},\nonumber\\
I_{20b}&=&C^5_{ij}C^5_{mk}C^5_{nl}\Psi_{jkl}\Psi_{ipn}C^5_{qr}C^5_{ps}C_{ut}\Psi_{rst}\Psi_{qmu},\nonumber\\
I_{20c}&=&C^5_{ij}C^5_{mk}C^5_{nl}\Psi_{jkl}\Psi_{imp}C^5_{qr}C^5_{us}C_{pt}\Psi_{rst}\Psi_{qun},\nonumber\\
I_{20d}&=&C^5_{ij}C^5_{mk}C^5_{nl}\Psi_{jkl}\Psi_{ipt}C^5_{qr}C^5_{ps}C_{ut}\Psi_{rsn}\Psi_{qmu},
\end{eqnarray}
\begin{eqnarray}
 I_{21a}&=&C^5_{ij}C_{mk}C_{nl}\Psi_{jkl}\Psi_{pmn}C^5_{pr}C_{qs}C^5_{ut}\Psi_{rst}\Psi_{iqu},\nonumber\\
I_{21b}&=&C^5_{ij}C_{mk}C_{nl}\Psi_{jkl}\Psi_{ipn}C^5_{qr}C_{ps}C^5_{ut}\Psi_{rst}\Psi_{qmu},\nonumber\\
 I_{21c}&=&C^5_{ij}C_{mk}C_{nl}\Psi_{jkl}\Psi_{imp}C^5_{qr}C_{us}C^5_{pt}\Psi_{rst}\Psi_{qun},\nonumber\\
I_{21d}&=&C^5_{ij}C_{mk}C_{nl}\Psi_{jkl}\Psi_{ipt}C^5_{qr}C_{ps}C^5_{ut}\Psi_{rsn}\Psi_{qmu},
\end{eqnarray}
\begin{eqnarray}
I_{22a}&=&C_{ij}C^5_{mk}C^5_{nl}\Psi_{jkl}\Psi_{pmn}C_{pr}C^5_{qs}C_{ut}\Psi_{rst}\Psi_{iqu},\nonumber\\
I_{22b}&=&C_{ij}C^5_{mk}C^5_{nl}\Psi_{jkl}\Psi_{ipn}C_{qr}C^5_{ps}C_{ut}\Psi_{rst}\Psi_{qmu},\nonumber\\
I_{22c}&=&C_{ij}C^5_{mk}C^5_{nl}\Psi_{jkl}\Psi_{imp}C_{qr}C^5_{us}C_{pt}\Psi_{rst}\Psi_{qun},\nonumber\\
I_{22d}&=&C_{ij}C^5_{mk}C^5_{nl}\Psi_{jkl}\Psi_{ipt}C_{qr}C^5_{ps}C_{ut}\Psi_{rsn}\Psi_{qmu}.
\end{eqnarray}
The linear dependence of the polynomials can be described by the 8 equations
 \begin{eqnarray}
I_{20a} - I_{20b} + I_{20d}&=&0,\nonumber\\
I_{19a} - I_{19b} + I_{19d}&=&0,\nonumber\\
I_{21a} - I_{21b} + I_{21d}&=&0,\nonumber\\
I_{22a} - I_{22b} +I_{22d}&=&0,
 \end{eqnarray}
 \begin{eqnarray}\label{anteloop}
 2I_{20a}- I_{20b}- I_{20c} +2I_{21a}- I_{21b}- I_{21c}&=&0,\nonumber\\
 I_{19a}- 2I_{19b}+I_{19c} +I_{21a}-2 I_{21b}+ I_{21c}&=&0,\nonumber\\
 I_{20a}-2 I_{20b}+ I_{20c} +I_{22a}-2 I_{22b}+ I_{22c}&=&0,\nonumber\\
 2I_{19a}- I_{19b}- I_{19c} +2I_{22a}- I_{22b}- I_{22c}&=&0.
 \end{eqnarray}

\subsection{Polynomials only invariant under P in Bob's and Charlie's labs }
There are 16 polynomials that are tensor sandwich contractions of the form in Eq. (\ref{tens}) and also invariant under P in Bob's and Charlie's lab, but not in Alice's lab. These Lorentz invariants are not all linearly independent, but a set of 8 linearly independent polynomials can be selected. Thus the polynomials span an 8 dimensional space. The 16 polynomials are

\begin{eqnarray}
I_{23a}&=& C_{ij}C_{mk}C_{nl}\Psi_{jkl}\Psi_{pmn}C^5_{pr}C_{qs}C_{ut}\Psi_{rst}\Psi_{iqu},\nonumber\\
I_{23b}&=&C_{ij}C_{mk}C_{nl}\Psi_{jkl}\Psi_{ipn}C^5_{qr}C_{ps}C_{ut}\Psi_{rst}\Psi_{qmu},\nonumber\\
I_{23c}&=&C_{ij}C_{mk}C_{nl}\Psi_{jkl}\Psi_{imp}C^5_{qr}C_{us}C_{pt}\Psi_{rst}\Psi_{qun},\nonumber\\
I_{23d}&=&C_{ij}C_{mk}C_{nl}\Psi_{jkl}\Psi_{ipt}C^5_{qr}C_{ps}C_{ut}\Psi_{rsn}\Psi_{qmu},
\end{eqnarray}
\begin{eqnarray}
 I_{24a}&=&C^5_{ij}C^5_{mk}C^5_{nl}\Psi_{jkl}\Psi_{pmn}C_{pr}C^5_{qs}C^5_{ut}\Psi_{rst}\Psi_{iqu},\nonumber\\
I_{24b}&=& C^5_{ij}C^5_{mk}C^5_{nl}\Psi_{jkl}\Psi_{ipn}C_{qr}C^5_{ps}C^5_{ut}\Psi_{rst}\Psi_{qmu},\nonumber\\
 I_{24c}&=&C^5_{ij}C^5_{mk}C^5_{nl}\Psi_{jkl}\Psi_{imp}C_{qr}C^5_{us}C^5_{pt}\Psi_{rst}\Psi_{qun},\nonumber\\
I_{24d}&=&C^5_{ij}C^5_{mk}C^5_{nl}\Psi_{jkl}\Psi_{ipt}C_{qr}C^5_{ps}C^5_{ut}\Psi_{rsn}\Psi_{qmu},
\end{eqnarray}
\begin{eqnarray}
 I_{25a}&=&C^5_{ij}C_{mk}C^5_{nl}\Psi_{jkl}\Psi_{pmn}C_{pr}C_{qs}C^5_{ut}\Psi_{rst}\Psi_{iqu},\nonumber\\
I_{25b}&=&C^5_{ij}C_{mk}C^5_{nl}\Psi_{jkl}\Psi_{ipn}C_{qr}C_{ps}C^5_{ut}\Psi_{rst}\Psi_{qmu},\nonumber\\
I_{25c}&=&C^5_{ij}C_{mk}C^5_{nl}\Psi_{jkl}\Psi_{imp}C_{qr}C_{us}C^5_{pt}\Psi_{rst}\Psi_{qun},\nonumber\\
I_{25d}&=&C^5_{ij}C_{mk}C^5_{nl}\Psi_{jkl}\Psi_{ipt}C_{qr}C_{ps}C^5_{ut}\Psi_{rsn}\Psi_{qmu},
\end{eqnarray}
\begin{eqnarray}
 I_{26a}&=&C^5_{ij}C^5_{mk}C_{nl}\Psi_{jkl}\Psi_{pmn}C_{pr}C^5_{qs}C_{ut}\Psi_{rst}\Psi_{iqu},\nonumber\\
I_{26b}&=&C^5_{ij}C^5_{mk}C_{nl}\Psi_{jkl}\Psi_{ipn}C_{qr}C^5_{ps}C_{ut}\Psi_{rst}\Psi_{qmu},\nonumber\\
I_{26c}&=&C^5_{ij}C^5_{mk}C_{nl}\Psi_{jkl}\Psi_{imp}C_{qr}C^5_{us}C_{pt}\Psi_{rst}\Psi_{qun},\nonumber\\
I_{26d}&=&C^5_{ij}C^5_{mk}C_{nl}\Psi_{jkl}\Psi_{ipt}C_{qr}C^5_{ps}C_{ut}\Psi_{rsn}\Psi_{qmu}.
 \end{eqnarray}
The linear dependence of the polynomials can be described by the 8 equations
 \begin{eqnarray}
I_{23b} - I_{23c} + I_{23d}&=&0,\nonumber\\
I_{24b} - I_{24c} + I_{24d}&=&0,\nonumber\\
I_{26b} - I_{26c} + I_{26d}&=&0,\nonumber\\
I_{25b} - I_{25c} + I_{25d}&=&0,
 \end{eqnarray}
 \begin{eqnarray}\label{zebra}
I_{23a} -2 I_{23b} + I_{23c} + 
  I_{25a} -2 I_{25b} + I_{25c}&=&0,\nonumber\\
I_{24a} -2 I_{24b} + I_{24c} + 
  I_{26a} -2 I_{26b} + I_{26c}&=&0,\nonumber\\
I_{24a} + I_{24b} - 2 I_{24c} + I_{25a} + I_{25b}- 2I_{25c}&=&0,\nonumber\\
I_{23a} + I_{23b} - 2 I_{23c} + I_{26a} + I_{26b}- 2I_{26c}&=&0.
\end{eqnarray}

\subsection{Polynomials only invariant under P in Bob's lab }
There is 16 different ways to construct polynomials as tensor sandwich contractions of the form in Eq. (\ref{tens}) that are invariant under P in Bob's lab, but not in Alice's and Charlie's labs. The resulting Lorentz invariants are not all linearly independent and not all unique, but a set of 5 linearly independent polynomials can be selected. Thus the polynomials span a 5 dimensional space. The 16 polynomials are

\begin{eqnarray}
I_{27a}&=& C_{ij}C_{mk}C_{nl}\Psi_{jkl}\Psi_{pmn}C^5_{pr}C_{qs}C^5_{ut}\Psi_{rst}\Psi_{iqu},\nonumber\\
I_{27b}&=&C_{ij}C_{mk}C_{nl}\Psi_{jkl}\Psi_{ipn}C^5_{qr}C_{ps}C^5_{ut}\Psi_{rst}\Psi_{qmu},\nonumber\\
I_{27c}&=& C_{ij}C_{mk}C_{nl}\Psi_{jkl}\Psi_{imp}C^5_{qr}C_{us}C^5_{pt}\Psi_{rst}\Psi_{qun},\nonumber\\
I_{27d}&=&C_{ij}C_{mk}C_{nl}\Psi_{jkl}\Psi_{ipt}C^5_{qr}C_{ps}C^5_{ut}\Psi_{rsn}\Psi_{qmu},
\end{eqnarray}
\begin{eqnarray}
I_{28a}&=& C^5_{ij}C^5_{mk}C^5_{nl}\Psi_{jkl}\Psi_{pmn}C_{pr}C^5_{qs}C_{ut}\Psi_{rst}\Psi_{iqu},\nonumber\\
I_{28b}&=&  C^5_{ij}C^5_{mk}C^5_{nl}\Psi_{jkl}\Psi_{ipn}C_{qr}C^5_{ps}C_{ut}\Psi_{rst}\Psi_{qmu},\nonumber\\
I_{28c}&=& C^5_{ij}C^5_{mk}C^5_{nl}\Psi_{jkl}\Psi_{imp}C_{qr}C^5_{us}C_{pt}\Psi_{rst}\Psi_{qun},\nonumber\\
I_{28d}&=&C^5_{ij}C^5_{mk}C^5_{nl}\Psi_{jkl}\Psi_{ipt}C_{qr}C^5_{ps}C_{ut}\Psi_{rsn}\Psi_{qmu},
\end{eqnarray}
\begin{eqnarray}
I_{29a}&=& C^5_{ij}C_{mk}C_{nl}\Psi_{jkl}\Psi_{pmn}C_{pr}C_{qs}C^5_{ut}\Psi_{rst}\Psi_{iqu},\nonumber\\
I_{29b}&=& C^5_{ij}C_{mk}C_{nl}\Psi_{jkl}\Psi_{ipn}C_{qr}C_{ps}C^5_{ut}\Psi_{rst}\Psi_{qmu},\nonumber\\
I_{29c}&=&  C^5_{ij}C_{mk}C_{nl}\Psi_{jkl}\Psi_{imp}C_{qr}C_{us}C^5_{pt}\Psi_{rst}\Psi_{qun},\nonumber\\
I_{29d}&=&C^5_{ij}C_{mk}C_{nl}\Psi_{jkl}\Psi_{ipt}C_{qr}C_{ps}C^5_{ut}\Psi_{rsn}\Psi_{qmu},
\end{eqnarray}
\begin{eqnarray}
I_{30a}&=& C^5_{ij}C^5_{mk}C_{nl}\Psi_{jkl}\Psi_{pmn}C_{pr}C^5_{qs}C^5_{ut}\Psi_{rst}\Psi_{iqu},\nonumber\\
I_{30b}&=&C^5_{ij}C^5_{mk}C_{nl}\Psi_{jkl}\Psi_{ipn}C_{qr}C^5_{ps}C^5_{ut}\Psi_{rst}\Psi_{qmu},\nonumber\\
I_{30c}&=& C^5_{ij}C^5_{mk}C_{nl}\Psi_{jkl}\Psi_{imp}C_{qr}C^5_{us}C^5_{pt}\Psi_{rst}\Psi_{qun},\nonumber\\
I_{30d}&=&C^5_{ij}C^5_{mk}C_{nl}\Psi_{jkl}\Psi_{ipt}C_{qr}C^5_{ps}C^5_{ut}\Psi_{rsn}\Psi_{qmu}.
\end{eqnarray}
The linear dependence of the polynomials can be described by the 11 equations

\begin{eqnarray}\label{hippo}
I_{27a} &=&I_{29a},\nonumber\\
I_{28a} &=& I_{30a},\nonumber\\
I_{27c}&=&I_{29c},\nonumber\\
I_{28c} &=&I_{30c},\nonumber\\
I_{27a}-I_{27b}&=& I_{29b} - I_{29c},\nonumber\\
 I_{28a}- I_{28b}&=& I_{30b} - I_{30c},\nonumber\\
 \frac{1}{2}( I_{27c}- I_{27a})=I_{27d} =- I_{28d}&=&- I_{29d}=I_{30d}=\frac{1}{2}(I_{28a}-I_{28c} ).
\end{eqnarray}
In particular we see that all the polynomials on the form $I_d$ are equal up to a sign. Note that the equalities $I_{27a} =I_{29a}$, $I_{28a} = I_{30a}$, $I_{27c} = I_{29c}$, and $I_{28c} = I_{30c}$ can be understood as graph isomorphisms of the graphs corresponding to the tensor contractions (See Fig. \ref{ris}).

\subsection{Polynomials only invariant under P in Charlie's lab }
There is 16 different ways to construct polynomials as tensor sandwich contractions of the form in Eq. (\ref{tens}) that are invariant under P in Charlie's lab, but not in Alice's and Bob's labs. The resulting Lorentz invariants are not all linearly independent and not all unique, but a set of 5 linearly independent polynomials can be selected. Thus the polynomials span a 5 dimensional space. The 16 polynomials are

\begin{eqnarray}
 I_{31a}&=&C_{ij}C_{mk}C_{nl}\Psi_{jkl}\Psi_{pmn}C^5_{pr}C^5_{qs}C_{ut}\Psi_{rst}\Psi_{iqu},\nonumber\\
I_{31b}&=&C_{ij}C_{mk}C_{nl}\Psi_{jkl}\Psi_{ipn}C^5_{qr}C^5_{ps}C_{ut}\Psi_{rst}\Psi_{qmu},\nonumber\\
I_{31c}&=&C_{ij}C_{mk}C_{nl}\Psi_{jkl}\Psi_{imp}C^5_{qr}C^5_{us}C_{pt}\Psi_{rst}\Psi_{qun},\nonumber\\
I_{31d}&=&C_{ij}C_{mk}C_{nl}\Psi_{jkl}\Psi_{ipt}C^5_{qr}C^5_{ps}C_{ut}\Psi_{rsn}\Psi_{qmu},
\end{eqnarray}
\begin{eqnarray}
 I_{32a}&=&C^5_{ij}C^5_{mk}C^5_{nl}\Psi_{jkl}\Psi_{pmn}C_{pr}C_{qs}C^5_{ut}\Psi_{rst}\Psi_{iqu},\nonumber\\
I_{32b}&=&C^5_{ij}C^5_{mk}C^5_{nl}\Psi_{jkl}\Psi_{ipn}C_{qr}C_{ps}C^5_{ut}\Psi_{rst}\Psi_{qmu},\nonumber\\
I_{32c}&=&C^5_{ij}C^5_{mk}C^5_{nl}\Psi_{jkl}\Psi_{imp}C_{qr}C_{us}C^5_{pt}\Psi_{rst}\Psi_{qun},\nonumber\\
I_{32d}&=&C^5_{ij}C^5_{mk}C^5_{nl}\Psi_{jkl}\Psi_{ipt}C_{qr}C_{ps}C^5_{ut}\Psi_{rsn}\Psi_{qmu},
\end{eqnarray}
\begin{eqnarray}
 I_{33a}&=&C_{ij}C^5_{mk}C^5_{nl}\Psi_{jkl}\Psi_{pmn}C^5_{pr}C_{qs}C^5_{ut}\Psi_{rst}\Psi_{iqu},\nonumber\\
 I_{33b}&=&C_{ij}C^5_{mk}C^5_{nl}\Psi_{jkl}\Psi_{ipn}C^5_{qr}C_{ps}C^5_{ut}\Psi_{rst}\Psi_{qmu},\nonumber\\
I_{33c}&= &C_{ij}C^5_{mk}C^5_{nl}\Psi_{jkl}\Psi_{imp}C^5_{qr}C_{us}C^5_{pt}\Psi_{rst}\Psi_{qun},\nonumber\\
I_{33d}&=&C_{ij}C^5_{mk}C^5_{nl}\Psi_{jkl}\Psi_{ipt}C^5_{qr}C_{ps}C^5_{ut}\Psi_{rsn}\Psi_{qmu},
\end{eqnarray}
\begin{eqnarray}
I_{34a}&=&C^5_{ij}C_{mk}C_{nl}\Psi_{jkl}\Psi_{pmn}C_{pr}C^5_{qs}C_{ut}\Psi_{rst}\Psi_{iqu},\nonumber\\
I_{34b}&=&C^5_{ij}C_{mk}C_{nl}\Psi_{jkl}\Psi_{ipn}C_{qr}C^5_{ps}C_{ut}\Psi_{rst}\Psi_{qmu},\nonumber\\
I_{34c}&=&C^5_{ij}C_{mk}C_{nl}\Psi_{jkl}\Psi_{imp}C_{qr}C^5_{us}C_{pt}\Psi_{rst}\Psi_{qun},\nonumber\\
I_{34d}&=&C^5_{ij}C_{mk}C_{nl}\Psi_{jkl}\Psi_{ipt}C_{qr}C^5_{ps}C_{ut}\Psi_{rsn}\Psi_{qmu}.
\end{eqnarray}
The linear dependence of the polynomials can be described by the 11 equations
\begin{eqnarray}\label{rhino}
I_{31a} &=& I_{34a},\nonumber\\
I_{32a} &=& I_{33a},\nonumber\\
I_{31b} &=& I_{34b},\nonumber\\
I_{32b} &=& I_{33b},\nonumber\\
I_{32c} - I_{32a}&=&  I_{33b} - I_{33c},\nonumber\\
I_{31c} - I_{31a}&=&  I_{34b} - I_{34c},\nonumber\\
\frac{1}{2}(I_{31a} - I_{31b})=I_{31d}=-I_{32d}&=&-I_{33d}=I_{34d}=\frac{1}{2}( I_{32b}- I_{32a} ).
\end{eqnarray}
In particular we see that all the polynomials on the form $I_d$ are equal up to a sign. 
Note that the equalities $I_{31a} =I_{34a}$, $I_{32a} = I_{33a}$, $I_{31b} = I_{34b}$, and $I_{32b} = I_{33b}$ can be understood as graph isomorphisms of the graphs corresponding to the tensor contractions (See Fig. \ref{ris}).

\subsection{Polynomials only invariant under P in Alice's lab }
There is 16 different ways to construct polynomials as tensor sandwich contractions of the form in Eq. (\ref{tens}) that are invariant under P in Alice's lab, but not in Bob's and Charlie's labs. The resulting Lorentz invariants are not all linearly independent and not all unique, but a set of 5 linearly independent polynomials can be selected. Thus the polynomials span a 5 dimensional space. The 16 polynomials are

\begin{eqnarray}
I_{35a}&=& C_{ij}C_{mk}C_{nl}\Psi_{jkl}\Psi_{pmn}C_{pr}C^5_{qs}C^5_{ut}\Psi_{rst}\Psi_{iqu},\nonumber\\
I_{35b}&=&C_{ij}C_{mk}C_{nl}\Psi_{jkl}\Psi_{ipn}C_{qr}C^5_{ps}C^5_{ut}\Psi_{rst}\Psi_{qmu},\nonumber\\
I_{35c}&=& C_{ij}C_{mk}C_{nl}\Psi_{jkl}\Psi_{imp}C_{qr}C^5_{us}C^5_{pt}\Psi_{rst}\Psi_{qun},\nonumber\\
I_{35d}&=&C_{ij}C_{mk}C_{nl}\Psi_{jkl}\Psi_{ipt}C_{qr}C^5_{ps}C^5_{ut}\Psi_{rsn}\Psi_{qmu},
\end{eqnarray}
\begin{eqnarray}
I_{36a}&=&  C^5_{ij}C^5_{mk}C^5_{nl}\Psi_{jkl}\Psi_{pmn}C^5_{pr}C_{qs}C_{ut}\Psi_{rst}\Psi_{iqu},\nonumber\\
I_{36b}&=& C^5_{ij}C^5_{mk}C^5_{nl}\Psi_{jkl}\Psi_{ipn}C^5_{qr}C_{ps}C_{ut}\Psi_{rst}\Psi_{qmu},\nonumber\\
I_{36c}&=&  C^5_{ij}C^5_{mk}C^5_{nl}\Psi_{jkl}\Psi_{imp}C^5_{qr}C_{us}C_{pt}\Psi_{rst}\Psi_{qun},\nonumber\\
I_{36d}&=&C^5_{ij}C^5_{mk}C^5_{nl}\Psi_{jkl}\Psi_{ipt}C^5_{qr}C_{ps}C_{ut}\Psi_{rsn}\Psi_{qmu},
\end{eqnarray}
\begin{eqnarray}
 I_{37a}&=&C^5_{ij}C^5_{mk}C_{nl}\Psi_{jkl}\Psi_{pmn}C^5_{pr}C_{qs}C^5_{ut}\Psi_{rst}\Psi_{iqu},\nonumber\\
I_{37b}&=&  C^5_{ij}C^5_{mk}C_{nl}\Psi_{jkl}\Psi_{ipn}C^5_{qr}C_{ps}C^5_{ut}\Psi_{rst}\Psi_{qmu},\nonumber\\
 I_{37c}&=&C^5_{ij}C^5_{mk}C_{nl}\Psi_{jkl}\Psi_{imp}C^5_{qr}C_{us}C^5_{pt}\Psi_{rst}\Psi_{qun},\nonumber\\
I_{37d}&=&C^5_{ij}C^5_{mk}C_{nl}\Psi_{jkl}\Psi_{ipt}C^5_{qr}C_{ps}C^5_{ut}\Psi_{rsn}\Psi_{qmu},
\end{eqnarray}
\begin{eqnarray}
I_{38a}&=& C_{ij}C^5_{mk}C_{nl}\Psi_{jkl}\Psi_{pmn}C_{pr}C_{qs}C^5_{ut}\Psi_{rst}\Psi_{iqu},\nonumber\\
I_{38b}&=&  C_{ij}C^5_{mk}C_{nl}\Psi_{jkl}\Psi_{ipn}C_{qr}C_{ps}C^5_{ut}\Psi_{rst}\Psi_{qmu},\nonumber\\
I_{38c}&=& C_{ij}C^5_{mk}C_{nl}\Psi_{jkl}\Psi_{imp}C_{qr}C_{us}C^5_{pt}\Psi_{rst}\Psi_{qun},\nonumber\\
I_{38d}&=&C_{ij}C^5_{mk}C_{nl}\Psi_{jkl}\Psi_{ipt}C_{qr}C_{ps}C^5_{ut}\Psi_{rsn}\Psi_{qmu}.
\end{eqnarray}
The linear dependence of the polynomials can be described by the 11 equations

\begin{eqnarray}\label{giraffe}
I_{35b}& =& I_{38b},\nonumber\\
I_{36b} &=& I_{37b},\nonumber\\
I_{35c} &=& I_{38c},\nonumber\\
I_{36c} &=& I_{37c},\nonumber\\
I_{35a} - I_{35c}&=& I_{38b}-I_{38a},\nonumber \\
I_{36a}- I_{36c}&= & I_{37b}-I_{37a}, \nonumber\\
\frac{1}{2}(I_{35b} - I_{35c})=I_{35d} =- I_{36d}&=&I_{37d}=- I_{38d}=\frac{1}{2}(  I_{37c} -I_{37b}).
\end{eqnarray}
In particular we see that all the polynomials on the form $I_d$ are equal up to a sign. 
Note that the equalities $I_{35b} =I_{38b}$, $I_{36b} = I_{37b}$, $I_{35c} = I_{38c}$, and $I_{36c} = I_{37c}$ can be understood as graph isomorphisms of the graphs corresponding to the tensor contractions (See Fig. \ref{ris}).

\subsection{Polynomials not invariant under P in any lab }
There is 16 different ways to construct polynomials as tensor sandwich contractions of the form in Eq. (\ref{tens}) that are not invariant under P in any lab. The resulting Lorentz invariants are not all linearly independent and not all unique, but a set of 5 linearly independent polynomials can be selected. Thus the polynomials span a 5 dimensional space. The 16 polynomials are

\begin{eqnarray}
I_{11a}&=&C^5_{ij}C^5_{mk}C^5_{nl}\Psi_{jkl}\Psi_{pmn}C_{pr}C_{qs}C_{ut}\Psi_{rst}\Psi_{iqu},\nonumber\\
I_{11b}&=&  C^5_{ij}C^5_{mk}C^5_{nl}\Psi_{jkl}\Psi_{ipn}C_{qr}C_{ps}C_{ut}\Psi_{rst}\Psi_{qmu},\nonumber\\
I_{11c}&=&  C^5_{ij}C^5_{mk}C^5_{nl}\Psi_{jkl}\Psi_{imp}C_{qr}C_{us}C_{pt}\Psi_{rst}\Psi_{qun},\nonumber\\
I_{11d}&=&C^5_{ij}C^5_{mk}C^5_{nl}\Psi_{jkl}\Psi_{ipt}C_{qr}C_{ps}C_{ut}\Psi_{rsn}\Psi_{qmu},
\end{eqnarray}
\begin{eqnarray}
I_{12a}&=&   C^5_{ij}C_{mk}C^5_{nl}\Psi_{jkl}\Psi_{pmn}C_{pr}C^5_{qs}C_{ut}\Psi_{rst}\Psi_{iqu},\nonumber\\
I_{12b}&=&  C^5_{ij}C_{mk}C^5_{nl}\Psi_{jkl}\Psi_{ipn}C_{qr}C^5_{ps}C_{ut}\Psi_{rst}\Psi_{qmu},\nonumber\\
I_{12c}&=&   C^5_{ij}C_{mk}C^5_{nl}\Psi_{jkl}\Psi_{imp}C_{qr}C^5_{us}C_{pt}\Psi_{rst}\Psi_{qun},\nonumber\\
I_{12d}&=&C^5_{ij}C_{mk}C^5_{nl}\Psi_{jkl}\Psi_{ipt}C_{qr}C^5_{ps}C_{ut}\Psi_{rsn}\Psi_{qmu},
\end{eqnarray}
\begin{eqnarray}
I_{14a}&=&  C^5_{ij}C^5_{mk}C_{nl}\Psi_{jkl}\Psi_{pmn}C_{pr}C_{qs}C^5_{ut}\Psi_{rst}\Psi_{iqu},\nonumber\\
I_{14b}&=&  C^5_{ij}C^5_{mk}C_{nl}\Psi_{jkl}\Psi_{ipn}C_{qr}C_{ps}C^5_{ut}\Psi_{rst}\Psi_{qmu},\nonumber\\
I_{14c}&=&  C^5_{ij}C^5_{mk}C_{nl}\Psi_{jkl}\Psi_{imp}C_{qr}C_{us}C^5_{pt}\Psi_{rst}\Psi_{qun},\nonumber\\
I_{14d}&=&C^5_{ij}C^5_{mk}C_{nl}\Psi_{jkl}\Psi_{ipt}C_{qr}C_{ps}C^5_{ut}\Psi_{rsn}\Psi_{qmu},
\end{eqnarray}
\begin{eqnarray}
I_{15a}&=&  C^5_{ij}C_{mk}C_{nl}\Psi_{jkl}\Psi_{pmn}C_{pr}C^5_{qs}C^5_{ut}\Psi_{rst}\Psi_{iqu},\nonumber\\
I_{15b}&=&   C^5_{ij}C_{mk}C_{nl}\Psi_{jkl}\Psi_{ipn}C_{qr}C^5_{ps}C^5_{ut}\Psi_{rst}\Psi_{qmu},\nonumber\\
I_{15c}&=&  C^5_{ij}C_{mk}C_{nl}\Psi_{jkl}\Psi_{imp}C_{qr}C^5_{us}C^5_{pt}\Psi_{rst}\Psi_{qun},\nonumber\\
I_{15d}&=&C^5_{ij}C_{mk}C_{nl}\Psi_{jkl}\Psi_{ipt}C_{qr}C^5_{ps}C^5_{ut}\Psi_{rsn}\Psi_{qmu}.
\end{eqnarray}
The linear dependence of the polynomials can be described by the 11 equations
\begin{eqnarray}
I_{14a}&=&I_{12a},\nonumber\\
I_{11a}&=&I_{15a},\nonumber\\
I_{12b}&=&I_{11b},\nonumber\\
I_{14b}&=&I_{15b},\nonumber\\
I_{14c}&=&I_{11c},\nonumber\\
I_{12c}&=&I_{15c},\nonumber\\
I_{11d}=-I_{12d} &=&-I_{14d}=I_{15d}, \nonumber\\
I_{11a}+I_{12a}=I_{11c}&+&I_{12c}=I_{11b}+I_{14b}.
\end{eqnarray}
In particular we see that all the polynomials on the form $I_d$ are equal up to a sign. 
Note that the equalities $I_{14a} =I_{12a}$, $I_{11a} = I_{15a}$, $I_{12b} = I_{11b}$, $I_{14b} = I_{15b}$, $I_{14c} = I_{11c}$, $I_{12c} = I_{15c}$, $I_{12d} = I_{14d}$, and $I_{11d} = I_{15d}$   can be understood as graph isomorphisms of the graphs corresponding to the tensor contractions (See Fig. \ref{ris}).

\subsection{Polynomials invariant under $\mathrm{U}(1)\times\mathrm{SL}(4,\mathbb{C}) $ up to a U(1) phase in one of the labs}
While the polynomials constructed as tensor sandwich contractions are designed to be invariant, up to a U(1) phase, under either $G^{C\gamma^5}$, $G^{C}$ or $G^{C}\cap G^{C\gamma^5}$ in any given lab, some linear combinations of the polynomials given above can be seen to be invariant, up to a U(1) phase, under both $G^{C}$ and $G^{C\gamma^5}$ in at least one lab.
For example the polynomial $I_{5b} - I_{5c} + I_{5d}$ is by construction invariant, up to a U(1) phase, under $G^{C\gamma^5}$ in Alice's lab and the polynomial $ I_{8b} - I_{8c} +I_{8d}$ is by construction invariant, up to a U(1) phase, under $G^{C}$ in Alice's lab, but the linear dependence relation $0= I_{5b} - I_{5c} + I_{5d} + I_{8b} - I_{8c} +I_{8d}$ given in Eq. (\ref{grunt}) implies that both the polynomial $I_{5b} - I_{5c} + I_{5d}$ and the polynomial $ I_{8b} - I_{8c} +I_{8d}$ are invariant, up to a U(1) phase, under both $G^{C}$ and $G^{C\gamma^5}$ in Alice's lab. Similar relations hold for many other linear combinations of the polynomials and other labs. Since the polynomials are continuous functions it follows that
any polynomial that is invariant, up to a U(1) phase, under both $G^{C}$ and $G^{C\gamma^5}$ is invariant, up to a U(1) phase, under the smallest Lie group that contains $G^{C}$ and $G^{C\gamma^5}$ as subgroups (See e.g. Ref. \cite{spinorent} Theorem 2 or  Ref. \cite{wall} Ch. 1.3.3.). This Lie group is $\mathrm{U}(1)\times\mathrm{SL}(4,\mathbb{C})$.
For the polynomials invariant under P in all labs we have the following twelve linear dependence relations

\begin{eqnarray}
I_{5b} - I_{5c} + I_{5d} + I_{8b} - I_{8c} +I_{8d}&=&0, \nonumber\\
I_{3b} - I_{3c}+ I_{3d}  + I_{9b} -  I_{9c}  +I_{9d}&=&0,\nonumber\\
I_{5a}- I_{5c} - I_{5d} + I_{9a} - I_{9c}-I_{9d}&=&0,\nonumber\\
I_{7a} - I_{7c} - I_{7d} - I_{6d} - I_{6c} + I_{6a}&=&0,\nonumber\\
I_{6b} - I_{6c} + I_{6d} + I_{4d} - I_{4c} + I_{4b}&=&0,\nonumber\\
I_{4a} - I_{4b} + I_{4d} + I_{9d} - I_{9b} + I_{9a}&=&0,\nonumber\\
I_{7a} - I_{7b} + I_{7d} + I_{8d} - I_{8b} + I_{8a}&=&0,\nonumber\\
I_{3a} - I_{3b} + I_{3d} + I_{6d} - I_{6b} + I_{6a}&=&0,\nonumber\\
I_{3a} - I_{3c} - I_{3d} - I_{8d} - I_{8c} + I_{8a}&=&0,\nonumber\\
I_{2a} - I_{2b} + I_{2d} + I_{5d} - I_{5b} + I_{5a}&=&0,\nonumber\\
I_{2a} - I_{2c} - I_{2d} - I_{4d} - I_{4c} + I_{4a}&=&0,\nonumber\\
I_{2b} - I_{2c} + I_{2d} + I_{7d} - I_{7c} + I_{7b}&=&0.
\end{eqnarray}
From these relations follows that the polynomials $I_{5b} - I_{5c} + I_{5d}$, $I_{8b} - I_{8c} +I_{8d}$, $I_{3b} - I_{3c}+ I_{3d}$, $I_{9b} -  I_{9c}  +I_{9d}$, $I_{6b} - I_{6c} + I_{6d}$, $I_{4d} - I_{4c} + I_{4b}$,  $I_{2b} - I_{2c} + I_{2d}$, and $ I_{7d} - I_{7c} + I_{7b}$ are invariant under $\mathrm{U}(1)\times\mathrm{SL}(4,\mathbb{C})$ in Alice's lab.
Moreover, the polynomials $I_{5a}- I_{5c} - I_{5d}$, $ I_{9a} - I_{9c}-I_{9d}$, $I_{7a} - I_{7c} - I_{7d}$, $- I_{6d} - I_{6c} + I_{6a}$,
$I_{3a} - I_{3c} - I_{3d}$, $- I_{8d} - I_{8c} + I_{8a}$, 
$I_{2a} - I_{2c} - I_{2d}$, and $- I_{4d} - I_{4c} + I_{4a}$ are invariant under $\mathrm{U}(1)\times\mathrm{SL}(4,\mathbb{C})$ in Bob's lab. 
Finally, the polynomials $I_{4a} - I_{4b} + I_{4d}$, $I_{9d} - I_{9b} + I_{9a}$, $I_{7a} - I_{7b} + I_{7d}$, $ I_{8d} - I_{8b} + I_{8a}$, $I_{3a} - I_{3b} + I_{3d}$, $I_{6d} - I_{6b} + I_{6a}$, $I_{2a} - I_{2b} + I_{2d}$, and $I_{5d} - I_{5b} + I_{5a}$
are invariant under $\mathrm{U}(1)\times\mathrm{SL}(4,\mathbb{C})$ in Charlie's lab.

For the polynomials only invariant under P in Alice's and Charlie's labs we can see from Eq. (\ref{gnu}) that the polynomials $I_{10a} + I_{10b} - 2 I_{10c}$, $ I_{17a} + I_{17b}- 2I_{17c}$, $I_{16a} + I_{16b} - 2 I_{16c}$, and $I_{18a} + I_{18b}- 2I_{18c}$
are invariant under $\mathrm{U}(1)\times\mathrm{SL}(4,\mathbb{C})$ in Alice's lab.
Moreover, the polynomials $2 I_{10a} - I_{10b} - I_{10c}$, $2 I_{18a} - I_{18b} - I_{18c}$, $2 I_{16a} - I_{16b} - I_{16c}$, and $ 2 I_{17a} - I_{17b} - I_{17c}$
are invariant under $\mathrm{U}(1)\times\mathrm{SL}(4,\mathbb{C})$ in Charlie's lab.

For the polynomials only invariant under P in Alice's and Bob's labs we can see from Eq. (\ref{anteloop}) that the polynomials $I_{19a}- 2I_{19b}+I_{19c}$, $ 2I_{21a}- I_{21b}- I_{21c}$, $I_{20a}-2 I_{20b}+ I_{20c} $, and $I_{22a}-2 I_{22b}+ I_{22c}$
are invariant under $\mathrm{U}(1)\times\mathrm{SL}(4,\mathbb{C})$ in Alice's lab.
Moreover, the polynomials $2I_{20a}- I_{20b}- I_{20c}$, $2I_{21a}- I_{21b}- I_{21c}$, $2I_{19a}- I_{19b}- I_{19c}$, and $2I_{22a}- I_{22b}- I_{22c}$
are invariant under $\mathrm{U}(1)\times\mathrm{SL}(4,\mathbb{C})$ in Bob's lab.

For the polynomials only invariant under P in Bob's and Charlie's labs we can see from Eq. (\ref{zebra}) that the polynomials $I_{24a} + I_{24b} - 2 I_{24c}$, $ I_{25a} + I_{25b}- 2I_{25c}$, $I_{23a} + I_{23b} - 2 I_{23c}$, and $ I_{26a} + I_{26b}- 2I_{26c}$
are invariant under $\mathrm{U}(1)\times\mathrm{SL}(4,\mathbb{C})$ in Bob's lab.
Moreover, the polynomials $I_{23a} -2 I_{23b} + I_{23c}$, $I_{25a} -2 I_{25b} + I_{25c}$, $I_{24a} -2 I_{24b} + I_{24c}$, and $ I_{26a} -2 I_{26b} + I_{26c}$
are invariant under $\mathrm{U}(1)\times\mathrm{SL}(4,\mathbb{C})$ in Charlie's lab.

For the polynomials only invariant under P in Bob's lab we can see from Eq. (\ref{hippo}) that the polynomial $I_{27d}$
is invariant under $\mathrm{U}(1)\times\mathrm{SL}(4,\mathbb{C})$ in Bob's lab.
Similarly, for the polynomials only invariant under P in Charlie's lab we can see from Eq. (\ref{rhino}) that the polynomial $I_{31d}$
is invariant under $\mathrm{U}(1)\times\mathrm{SL}(4,\mathbb{C})$ in Charlie's lab.
Finally, for the polynomials only invariant under P in Alice's lab we can see from Eq. (\ref{giraffe}) that the polynomial $I_{35d}$
is invariant under $\mathrm{U}(1)\times\mathrm{SL}(4,\mathbb{C})$ in Alice's lab.

\subsection{Weyl particles}
We can consider the case where Alice's, Bob's and Charlie's particles are Weyl particles, i.e., have a definite chirality. Then the shared state of the spinors is invariant under some combination of projections $P_L^A$ or $P_R^A$ by Alice, $P_L^B$ or $P_R^B$ by Bob and, $P_L^C$ or $P_R^C$ by Charlie and the Lorentz invariants on the form $I_a$, $I_b$ and $I_c$ reduce, up to a sign, to $128\tau$ where

\begin{eqnarray}\label{tau}
\tau=&& (\psi_{011}\psi_{ 100} -\psi_{ 010}\psi_{ 101} -\psi_{ 001}\psi_{ 110} +\psi_{ 000}\psi_{ 111})^2 \nonumber\\
&&- 4 (\psi_{ 001} \psi_{ 010} -\psi_{ 000 }\psi_{ 011}) (\psi_{ 101}\psi_{ 110} -\psi_{ 100}\psi_{ 111}).
\end{eqnarray}
The polynomial $\tau$ is the Coffman-Kundu-Wootters 3-tangle \cite{coffman}, which is equal to the Cayley hyperdeterminant \cite{cayley} of the sub-tensor of $\Psi^{ABC}$ defined by the elements with only indices equal to $0$ or $1$.

The reduction of the polynomials to a form where all state coefficients have only two spinor basis indices $0$ and $1$ is due to the symmetry $\psi_{jkl}=(-1)^{|LA|}\psi_{(j-2) kl}=(-1)^{|LB|}\psi_{j (k-2)l}=(-1)^{|LC|}\psi_{j k(l-2)}$ of the shared state, where $|LA|=1$ if the state is invariant under $P_L^A$ and otherwise zero, $|LB|=1$ if the state is invariant under $P_L^B$ and otherwise zero, $|LC|=1$ if the state is invariant under $P_L^C$ and otherwise zero and $j$,$k$, and $l$ are defined modulo 4. Thus, for the case of Weyl particles the polynomials on the form $I_{a}$, $I_{b}$, and $I_{c}$ become essentially equivalent to the Coffman-Kundu-Wootters 3-tangle.
 
The invariants on the form $I_{d}$ do not reduce to the 3-tangle when Alice's, Bob's and Charlie's particles are Weyl particles, but instead reduce to zero.
However, we can consider also a scenario where only two of the particles are Weyl particles. For example let Alice's and Bob's particles be Weyl particles i.e., let the shared state be invariant under some combination of projections $P_L^A$ or $P_R^A$ by Alice and $P_L^B$ or $P_R^B$ by Bob but no condition is imposed on Charlie's particle. Then the polynomials on the form $I_{d}$ that are invariant under P in Charlie's lab , i.e., $I_{2d},I_{3d},I_{4d},I_{5d},I_{6d},I_{7d},I_{8d},I_{9d},I_{10d}$ $,I_{16d},I_{17d},I_{18d},I_{23d},I_{24d},
I_{25d},I_{26d},I_{31d},I_{32d},I_{33d},$ and $ I_{34d}$ reduce to $64V$, up to a sign, where

\begin{eqnarray}
V=  
 &&(\psi_{ 003} \psi_{012}-\psi_{ 002 }\psi_{ 013})(\psi_{ 101}\psi_{ 110}-\psi_{ 100}\psi_{ 111})\nonumber\\ 
  &&+(\psi_{ 003}\psi_{ 011}-\psi_{ 001}\psi_{ 013})(\psi_{ 100}\psi_{ 112} -\psi_{ 102}\psi_{ 110})\nonumber \\ 
  &&+(\psi_{ 003}\psi_{ 010}-\psi_{ 000 }\psi_{ 013})(\psi_{ 102}\psi_{ 111} -\psi_{ 101}\psi_{ 112})\nonumber\\
   &&+ (\psi_{ 000 }\psi_{ 011}-\psi_{ 000 }\psi_{  010})(\psi_{ 102}\psi_{ 113}-\psi_{ 103}\psi_{ 112 })\nonumber\\
 &&+(\psi_{ 000 }\psi_{ 012}-\psi_{ 002 }\psi_{ 010})(\psi_{ 103}\psi_{ 111}-\psi_{ 101}\psi_{ 113})\nonumber\\
   &&+(\psi_{ 001}\psi_{ 012}-\psi_{ 002 }\psi_{ 011})(\psi_{ 100}\psi_{ 113 }   -\psi_{ 103}\psi_{ 110}),  
\end{eqnarray} 
while the remaining invariants on the form $I_{d}$ that are not invariant under P in Charlie's lab reduce to zero. The polynomial $V$, called the $2\times 2\times 4$ tangle, was described in Refs. \cite{verstraete,moor} and is invariant under $\mathrm{SL}(4,\mathbb{C})$ acting on Charlie's spinor. The cases where Bob's and Charlie's or Alice's and Charlie's particles are Weyl particles are completely analogous. For these cases the invariants on the form $I_{d}$ that are invariant under P in the lab that holds the non-chiral particle reduce to multiples of polynomials obtained from $V$ by permuting the state coefficient indices accordingly.

\subsection{Eigenspaces of the local Dirac Hamiltonians in the case of zero momenta and zero four-potentials}

We can consider the case of zero momentum and zero four-potential, sometimes described as the non-relativistic limit of a free particle. If Alice's, Bob's and Charlie's particles are all in this limit and also in an eigenstate of the local Dirac Hamiltonians the shared state is invariant under some combination of projections $P_+^A$ or $P_-^A$ by Alice, $P_+^B$ or $P_-^B$ by Bob, and $P_+^C$ or $P_-^C$ by Charlie. In this case only a $2\times 2\times 2$ subtensor of $\Psi^{ABC}$ is nonzero.
As a consequence the polynomials $I_{3a}$, $I_{3b}$ and $I_{3c}$ reduce, up to a sign and a relabelling of the indices, to $2\tau$ where $\tau$ is the Coffman-Kundu-Wootters 3-tangle \cite{coffman} given in Eq. (\ref{tau}). All other polynomials of degree 4 constructed above are zero for this case. In particular any polynomial obtained through tensor sandwich contractions where one or more contractions involve $C\gamma^5$ is zero in this case since $P_+C\gamma^5P_+=P_-C\gamma^5P_-=0$.

\subsection{Examples of tripartite spinor entangled states}
Here we consider a few examples of tripartite entangled states to illustrate how inequivalent forms of spinor entanglement are distinguished by the polynomials.
We can consider analogues of the tripartite entangled Greenberger-Horne-Zeilinger (GHZ)\cite{ghz} state for non-relativistic spin-$\frac{1}{2}$ particles.
One such state is
\begin{eqnarray}
\frac{1}{\sqrt{2}}({\phi_0^A}\otimes{\phi_0^B}\otimes{\phi_0^C}+{\phi_1^A}\otimes{\phi_1^B}\otimes{\phi_1^C}).
\end{eqnarray}
For this state $|I_{3a}|=|I_{3b}|=|I_{3c}|=1/2$, but all the other degree 4 polynomials are identically zero. Similarly we can construct a state
\begin{eqnarray}
\frac{1}{\sqrt{2}}({\phi_2^A}\otimes{\phi_2^B}\otimes{\phi_2^C}+{\phi_1^A}\otimes{\phi_1^B}\otimes{\phi_1^C}).
\end{eqnarray}
for which $|I_{2a}|=|I_{2b}|=|I_{2c}|=1/2$, but all the other degree 4 polynomials are identically zero. Furthermore, one can construct GHZ-like states that involve more than two basis spinors. For example
\begin{eqnarray}
\frac{1}{\sqrt{2}}({\phi_2^A}\otimes{\phi_0^B}\otimes{\phi_2^C}+{\phi_1^A}\otimes{\phi_1^B}\otimes{\phi_1^C}).
\end{eqnarray}
for which $|I_{4a}|=|I_{4b}|=|I_{4c}|=1/2$, but all the other degree 4 polynomials are identically zero, and
\begin{eqnarray}
\frac{1}{\sqrt{2}}({\phi_2^A}\otimes{\phi_0^B}\otimes{\phi_0^C}+{\phi_1^A}\otimes{\phi_3^B}\otimes{\phi_1^C}).
\end{eqnarray}
for which $|I_{5a}|=|I_{5b}|=|I_{5c}|=1/2$, but all the other degree 4 polynomials are identically zero. As for the above four examples, for any GHZ like state of this kind only one triplet of polynomials of the types $I_a,I_b$ and $I_c$ is non-zero, and these polynomials are invariant under P in all labs.

A generalization of the GHZ state with three terms is
\begin{eqnarray}
\frac{1}{{\sqrt{3}}}({\phi_0^A}\otimes{\phi_0^B}\otimes{\phi_0^C}+{\phi_1^A}\otimes{\phi_1^B}\otimes{\phi_1^C}+{\phi_2^A}\otimes{\phi_2^B}\otimes{\phi_2^C}).
\end{eqnarray}
For this state $|I_{2a}|=|I_{2b}|=|I_{2c}|=|I_{3a}|=|I_{3b}|=|I_{3c}|=2/9$ and $|I_{11a}|=|I_{11b}|=|I_{11c}|=|I_{15a}|=|I_{12b}|=|I_{14c}|=1/9$ while all other degree 4 polynomials are identically zero. 
Another generalization of the GHZ state with four terms is
\begin{eqnarray}
\frac{1}{{2}}({\phi_0^A}\otimes{\phi_0^B}\otimes{\phi_0^C}+{\phi_1^A}\otimes{\phi_1^B}\otimes{\phi_1^C}+{\phi_2^A}\otimes{\phi_2^B}\otimes{\phi_2^C}+{\phi_3^A}\otimes{\phi_3^B}\otimes{\phi_3^C}).
\end{eqnarray}
For this state $|I_{2a}|=|I_{2b}|=|I_{2c}|=|I_{3a}|=|I_{3b}|=|I_{3c}|=|I_{4b}|=|I_{5c}|=|I_{6c}|=|I_{7a}|=|I_{8b}|=|I_{9a}|=1/4$ while all other degree 4 polynomials are identically zero. 

An example of a state that is not on the GHZ form is
\begin{eqnarray}
\frac{1}{{2}}({\phi_0^A}\otimes{\phi_0^B}\otimes{\phi_0^C}+{\phi_1^A}\otimes{\phi_0^B}\otimes{\phi_1^C}+{\phi_1^A}\otimes{\phi_1^B}\otimes{\phi_0^C}+{\phi_2^A}\otimes{\phi_1^B}\otimes{\phi_1^C}).
\end{eqnarray}
For this state $|I_{23a}|=|I_{23b}|=|I_{23c}|=1/4$ but all other degree 4 polynomials are identically zero. 
This state was considered in Refs. \cite{miyake1,moor} although not in the context of Dirac spinors.
Another state not on the GHZ form is
\begin{eqnarray}
\frac{1}{{2}}({\phi_0^A}\otimes{\phi_0^B}\otimes{\phi_0^C}+{\phi_0^A}\otimes{\phi_1^B}\otimes{\phi_1^C}+{\phi_1^A}\otimes{\phi_0^B}\otimes{\phi_2^C}+{\phi_1^A}\otimes{\phi_1^B}\otimes{\phi_3^C}).
\end{eqnarray}
for which $|I_{3a}|=|I_{6a}|=1/2$ and $|I_{3c}|=|I_{3d}|=|I_{6c}|=|I_{6d}|=1/4$ but all other degree 4 polynomials are identically zero. This state was considered in Refs. \cite{miyake,moor} although not in the context of Dirac spinors.

We can see that the eight examples above belong to eight different entanglement classes that can be discriminated by the polynomial invariants.

\section{The case of four spinors}\label{four}
For four Dirac spinors the state coefficients can be arranged as a $4\times 4\times 4\times 4$ tensor $\Psi^{ABCD}$.
Transformations $S^A,S^B,S^C$, and $S^D$ acting locally on Alice's, Bob's, Charlie's and David's particle, respectively, are described as
\begin{eqnarray}
\Psi^{ABCD}_{ijkl}\to\sum_{mnop}S^A_{im}S^B_{jn}S^C_{ko}S^D_{lp}\Psi^{ABCD}_{mnop}.
\end{eqnarray}
Lorentz invariants of degree 2 can be constructed as tensor contractions of the form 
\begin{eqnarray}
\sum_{jklmnpqr}X_{nj}X_{pk}X_{ql}X_{rm}\Psi^{ABCD}_{jklm}\Psi^{ABCD}_{npqr}, 
\end{eqnarray}
involving two copies of $\Psi^{ABCD}$ where either $X=C$ or $X=C\gamma^5$ for each instance. Unlike the case of three particles the 16 degree 2 Lorentz invariants constructed this way are non-zero.
These 16 different polynomials were calculated and their linear independence tested. 
Due to the large number of terms in the polynomials they are not all given in a fully written out form, but a few selected examples are given fully written out in \ref{polllu}.
In writing these contractions we leave out the summation sign in the following, with the understanding that repeated indices are summed over. We also suppress the superscript $ABCD$ of $\Psi$ and use the abbreviated notation $C^5\equiv C\gamma^5$.

The 16 polynomials of degree 2 are linearly independent and given by

\begin{eqnarray}\label{ghj}
H_a&=&C_{nj}C_{pk}C_{ql}C_{rm}\Psi_{jklm}\Psi_{npqr},\nonumber\\
H_b&=&C^5_{nj}C^5_{pk}C^5_{ql}C^5_{rm}\Psi_{jklm}\Psi_{npqr},\nonumber\\
H_c&=&C_{nj}C_{pk}C_{ql}C^5_{rm}\Psi_{jklm}\Psi_{npqr},\nonumber\\
H_d&=&C_{nj}C_{pk}C^5_{ql}C_{rm}\Psi_{jklm}\Psi_{npqr},\nonumber\\
H_e&=&C^5_{nj}C_{pk}C_{ql}C_{rm}\Psi_{jklm}\Psi_{npqr},\nonumber\\
H_f&=&C_{nj}C^5_{pk}C_{ql}C_{rm}\Psi_{jklm}\Psi_{npqr},\nonumber\\
H_g&=&C_{nj}C_{pk}C^5_{ql}C^5_{rm}\Psi_{jklm}\Psi_{npqr},\nonumber\\
H_h&=&C^5_{nj}C_{pk}C_{ql}C^5_{rm}\Psi_{jklm}\Psi_{npqr},\nonumber\\
H_i&=&C_{nj}C^5_{pk}C_{ql}C^5_{rm}\Psi_{jklm}\Psi_{npqr},\nonumber\\
H_j&=&C^5_{nj}C_{pk}C^5_{ql}C_{rm}\Psi_{jklm}\Psi_{npqr},\nonumber\\
H_k&=&C_{nj}C^5_{pk}C^5_{ql}C_{rm}\Psi_{jklm}\Psi_{npqr},\nonumber\\
H_l&=&C^5_{nj}C^5_{pk}C_{ql}C_{rm}\Psi_{jklm}\Psi_{npqr},\nonumber\\
H_m&=&C^5_{nj}C^5_{pk}C^5_{ql}C_{rm}\Psi_{jklm}\Psi_{npqr},\nonumber\\
H_n&=&C^5_{nj}C^5_{pk}C_{ql}C^5_{rm}\Psi_{jklm}\Psi_{npqr},\nonumber\\
H_o&=&C_{nj}C^5_{pk}C^5_{ql}C^5_{rm}\Psi_{jklm}\Psi_{npqr},\nonumber\\
H_p&=&C^5_{nj}C_{pk}C^5_{ql}C^5_{rm}\Psi_{jklm}\Psi_{npqr}.
\end{eqnarray}

The next lowest degree is 4.
There are 13 ways to pair the tensor indices of four copies of $\Psi^{ABCD}$ to create polynomials that do not factorize into two degree 2 polynomials

\begin{eqnarray}\label{cvn}
W_a&=&X_{nj}X_{pk}X_{ql}X_{rm}\Psi_{jklm}\Psi_{uvqr} X_{uw}X_{vx}X_{sy}X_{tz}\Psi_{npyz}\Psi_{wxst},\nonumber\\
W_b&=&X_{nj}X_{pk}X_{ql}X_{rm}\Psi_{jklm}\Psi_{uvqr} X_{uw}X_{vx}X_{sy}X_{tz}\Psi_{nxyz}\Psi_{wpst},\nonumber\\
W_c&=&X_{nj}X_{pk}X_{ql}X_{rm}\Psi_{jklm}\Psi_{nvqz} X_{uw}X_{vx}X_{sy}X_{tz}\Psi_{upsr}\Psi_{wxyt },\nonumber\\
W_d&=&X_{nj}X_{pk}X_{ql}X_{rm}\Psi_{jklm}\Psi_{nvqz} X_{uw}X_{vx}X_{sy}X_{tz}\Psi_{uxsr}\Psi_{wpyt },\nonumber\\
W_e&=&X_{nj}X_{pk}X_{ql}X_{rm}\Psi_{jklm}\Psi_{npst} X_{uw}X_{vx}X_{sy}X_{tz}\Psi_{uvqz}\Psi_{wxyr},\nonumber\\
W_f&=&X_{nj}X_{pk}X_{ql}X_{rm}\Psi_{jklm}\Psi_{upqt} X_{uw}X_{vx}X_{sy}X_{tz}\Psi_{nvsr}\Psi_{wxyz},\nonumber\\
W_g&=&X_{nj}X_{pk}X_{ql}X_{rm}\Psi_{jklm}\Psi_{upqt} X_{uw}X_{vx}X_{sy}X_{tz}\Psi_{nvsz}\Psi_{wxyr},\nonumber\\
W_h&=&X_{nj}X_{pk}X_{ql}X_{rm}\Psi_{jklm}\Psi_{upsr} X_{uw}X_{vx}X_{sy}X_{tz}\Psi_{nvyt}\Psi_{wxqz},\nonumber\\
W_i&=&X_{nj}X_{pk}X_{ql}X_{rm}\Psi_{jklm}\Psi_{nvsr} X_{uw}X_{vx}X_{sy}X_{tz}\Psi_{upyt}\Psi_{wxqz},\nonumber\\
W_j&=&X_{nj}X_{pk}X_{ql}X_{rm}\Psi_{jklm}\Psi_{nvqr} X_{uw}X_{vx}X_{sy}X_{tz}\Psi_{uxyt}\Psi_{wpsz},\nonumber\\
W_k&=&X_{nj}X_{pk}X_{ql}X_{rm}\Psi_{jklm}\Psi_{upqr} X_{uw}X_{vx}X_{sy}X_{tz}\Psi_{nvst}\Psi_{wxyz},\nonumber\\
W_l&=&X_{nj}X_{pk}X_{ql}X_{rm}\Psi_{jklm}\Psi_{npsr} X_{uw}X_{vx}X_{sy}X_{tz}\Psi_{uvyz}\Psi_{wxqt},\nonumber\\
W_m&=&X_{nj}X_{pk}X_{ql}X_{rm}\Psi_{jklm}\Psi_{npqt} X_{uw}X_{vx}X_{sy}X_{tz}\Psi_{uvsr}\Psi_{wxyz}.
\end{eqnarray} 

We can categorize the tensor contractions in Eq. (\ref{cvn}) in terms of their permutation symmetries. The three contractions
$W_a$, $W_c$ and $W_f$ are each invariant with respect to permutations of two disjoint pairs of laboratories. The contraction $W_a$ is invariant with respect to permutations of AB and permutations of CD. The contraction $W_c$ is invariant with respect to permutations of AC and permutations of BD. The contraction $W_f$ is invariant with respect to permutations of AD and permutations of BC.
Six of the contractions $W_b$, $W_d$, $W_e$, $W_g$, $W_h$, and $W_i$ are each invariant with respect to permutations of a single pair of laboratories. The contraction $W_b$ is invariant with respect to permutations of CD while $W_d$ is invariant with respect to permutations of AC, $W_e$ with respect to  AB, $W_g$ with respect to BC, $W_h$ with respect to BD, and $W_i$ with respect to AD. The final four contractions $W_j$, $W_k$, $W_l$, and $W_m$ are each invariant with respect to permutations of a triple of laboratories.
The contraction $W_j$ is invariant with respect to permutations of ACD, while $W_k$ is invariant with respect to permutations of BCD, $W_l$ with respect to ABD, and $W_m$ with respect to ABC. 
Moreover, the different tensor sandwich contractions can be represented as graphs.
See \ref{graphs} for the graph representations of the tensor contractions in Eq. (\ref{cvn}).

For each pairing of tensor indices we can consider $(2^{8}-2^4)/2+2^4=136$ ways to choose the $X$s as either $C$ or $C^5$. There is $2^4$ choices where for each lab both $X$s corresponding to the lab are identical and there is $2^8-2^4$ choices that are not of this kind. For the latter case however each such choice for is equivalent to at least one other choice. This can be understood in terms of graph isomorphisms of the graph representations of the tensor contraction in Eq. (\ref{cvn}), and reduces the number of such choices that need to be considered to $(2^8-2^4)/2$. 
Thus we can consider a total of $136\times 13=1768$ ways to construct polynomials by this method. 

A complete list of polynomials will not be given here. Instead only a selection of 26 polynomials were computed and tested for linear dependence. Thirteen Lorentz invariant polynomials were constructed by choosing all $X$s in Eq. (\ref{cvn}) as $C$ and 13 polynomials were constructed by choosing all $X$s in Eq. (\ref{cvn}) as $C^5$.
These polynomials are invariant under P in all labs. When evaluating the possible linear dependencies of these polynomials on products of the degree 2 polynomials in Eq. (\ref{ghj}) we note that the square of any of these degree 2 polynomials is invariant under P in all labs while any product of two different degree 2 polynomials is not. Therefore only the squares need to be considered.
Evaluating the possible linear dependencies showed that the 26 degree 4 polynomials together with the 16 squares of the degree 2 polynomials form a 41 dimensional polynomial space. Thus there is a single equation describing the linear dependence. 
Due to the large number of terms in the polynomials they are not all given in a fully written out form, but two selected examples are given fully written out in \ref{polllu}.
The 26 polynomials of degree four are

\begin{eqnarray}\label{ffr}
T_a&=&C_{nj}C_{pk}C_{ql}C_{rm}\Psi_{jklm}\Psi_{uvqr} C_{uw}C_{vx}C_{sy}C_{tz}\Psi_{npyz}\Psi_{wxst},\nonumber\\
T_b&=&C_{nj}C_{pk}C_{ql}C_{rm}\Psi_{jklm}\Psi_{uvqr} C_{uw}C_{vx}C_{sy}C_{tz}\Psi_{nxyz}\Psi_{wpst},\nonumber\\
T_c&=&C_{nj}C_{pk}C_{ql}C_{rm}\Psi_{jklm}\Psi_{nvqz} C_{uw}C_{vx}C_{sy}C_{tz}\Psi_{upsr}\Psi_{wxyt },\nonumber\\
T_d&=&C_{nj}C_{pk}C_{ql}C_{rm}\Psi_{jklm}\Psi_{nvqz} C_{uw}C_{vx}C_{sy}C_{tz}\Psi_{uxsr}\Psi_{wpyt },\nonumber\\
T_e&=&C_{nj}C_{pk}C_{ql}C_{rm}\Psi_{jklm}\Psi_{npst} C_{uw}C_{vx}C_{sy}C_{tz}\Psi_{uvqz}\Psi_{wxyr},\nonumber\\
T_f&=&C_{nj}C_{pk}C_{ql}C_{rm}\Psi_{jklm}\Psi_{upqt} C_{uw}C_{vx}C_{sy}C_{tz}\Psi_{nvsr}\Psi_{wxyz},\nonumber\\
T_g&=&C_{nj}C_{pk}C_{ql}C_{rm}\Psi_{jklm}\Psi_{upqt} C_{uw}C_{vx}C_{sy}C_{tz}\Psi_{nvsz}\Psi_{wxyr},\nonumber\\
T_h&=&C_{nj}C_{pk}C_{ql}C_{rm}\Psi_{jklm}\Psi_{upsr} C_{uw}C_{vx}C_{sy}C_{tz}\Psi_{nvyt}\Psi_{wxqz},\nonumber\\
T_i&=&C_{nj}C_{pk}C_{ql}C_{rm}\Psi_{jklm}\Psi_{nvsr} C_{uw}C_{vx}C_{sy}C_{tz}\Psi_{upyt}\Psi_{wxqz},\nonumber\\
T_j&=&C_{nj}C_{pk}C_{ql}C_{rm}\Psi_{jklm}\Psi_{nvqr} C_{uw}C_{vx}C_{sy}C_{tz}\Psi_{uxyt}\Psi_{wpsz},\nonumber\\
T_k&=&C_{nj}C_{pk}C_{ql}C_{rm}\Psi_{jklm}\Psi_{upqr} C_{uw}C_{vx}C_{sy}C_{tz}\Psi_{nvst}\Psi_{wxyz},\nonumber\\
T_l&=&C_{nj}C_{pk}C_{ql}C_{rm}\Psi_{jklm}\Psi_{npsr} C_{uw}C_{vx}C_{sy}C_{tz}\Psi_{uvyz}\Psi_{wxqt},\nonumber\\
T_m&=&C_{nj}C_{pk}C_{ql}C_{rm}\Psi_{jklm}\Psi_{npqt} C_{uw}C_{vx}C_{sy}C_{tz}\Psi_{uvsr}\Psi_{wxyz},
\end{eqnarray} 
and
\begin{eqnarray}\label{ffr2}
Y_a&=&C_{nj}^5C_{pk}^5C_{ql}^5C_{rm}^5\Psi_{jklm}\Psi_{uvqr} C_{uw}^5C_{vx}^5C_{sy}^5C_{tz}^5\Psi_{npyz}\Psi_{wxst},\nonumber\\
Y_b&=&C_{nj}^5C_{pk}^5C_{ql}^5C_{rm}^5\Psi_{jklm}\Psi_{uvqr} C_{uw}^5C_{vx}^5C_{sy}^5C_{tz}^5\Psi_{nxyz}\Psi_{wpst},\nonumber\\
Y_c&=&C_{nj}^5C_{pk}^5C_{ql}^5C_{rm}^5\Psi_{jklm}\Psi_{nvqz} C_{uw}^5C_{vx}^5C_{sy}^5C_{tz}^5\Psi_{upsr}\Psi_{wxyt },\nonumber\\
Y_d&=&C_{nj}^5C_{pk}^5C_{ql}^5C_{rm}^5\Psi_{jklm}\Psi_{nvqz} C_{uw}^5C_{vx}^5C_{sy}^5C_{tz}^5\Psi_{uxsr}\Psi_{wpyt },\nonumber\\
Y_e&=&C_{nj}^5C_{pk}^5C_{ql}^5C_{rm}^5\Psi_{jklm}\Psi_{npst} C_{uw}^5C_{vx}^5C_{sy}^5C_{tz}^5\Psi_{uvqz}\Psi_{wxyr},\nonumber\\
Y_f&=&C_{nj}^5C_{pk}^5C_{ql}^5C_{rm}^5\Psi_{jklm}\Psi_{upqt} C_{uw}^5C_{vx}^5C_{sy}^5C_{tz}^5\Psi_{nvsr}\Psi_{wxyz},\nonumber\\
Y_g&=&C_{nj}^5C_{pk}^5C_{ql}^5C_{rm}^5\Psi_{jklm}\Psi_{upqt} C_{uw}^5C_{vx}^5C_{sy}^5C_{tz}^5\Psi_{nvsz}\Psi_{wxyr},\nonumber\\
Y_h&=&C_{nj}^5C_{pk}^5C_{ql}^5C_{rm}^5\Psi_{jklm}\Psi_{upsr} C_{uw}^5C_{vx}^5C_{sy}^5C_{tz}^5\Psi_{nvyt}\Psi_{wxqz},\nonumber\\
Y_i&=&C_{nj}^5C_{pk}^5C_{ql}^5C_{rm}^5\Psi_{jklm}\Psi_{nvsr} C_{uw}^5C_{vx}^5C_{sy}^5C_{tz}^5\Psi_{upyt}\Psi_{wxqz},\nonumber\\
Y_j&=&C_{nj}^5C_{pk}^5C_{ql}^5C_{rm}^5\Psi_{jklm}\Psi_{nvqr} C_{uw}^5C_{vx}^5C_{sy}^5C_{tz}^5\Psi_{uxyt}\Psi_{wpsz},\nonumber\\
Y_k&=&C_{nj}^5C_{pk}^5C_{ql}^5C_{rm}^5\Psi_{jklm}\Psi_{upqr} C_{uw}^5C_{vx}^5C_{sy}^5C_{tz}^5\Psi_{nvst}\Psi_{wxyz},\nonumber\\
Y_l&=&C_{nj}^5C_{pk}^5C_{ql}^5C_{rm}^5\Psi_{jklm}\Psi_{npsr} C_{uw}^5C_{vx}^5C_{sy}^5C_{tz}^5\Psi_{uvyz}\Psi_{wxqt},\nonumber\\
Y_m&=&C_{nj}^5C_{pk}^5C_{ql}^5C_{rm}^5\Psi_{jklm}\Psi_{npqt} C_{uw}^5C_{vx}^5C_{sy}^5C_{tz}^5\Psi_{uvsr}\Psi_{wxyz}.
\end{eqnarray} 
The equation describing the linear dependence is
\begin{align}
H_b^2 - H_a^2=2(& T_a -  T_b - T_c-  T_d -  T_e+  T_f  -  T_g 
  -  T_h -  T_i -  T_j -  T_k -  T_l - 
  T_m \nonumber\\&  -  Y_a +  Y_b  +  Y_c+  Y_d +  Y_e  -  Y_f +  Y_g +
  Y_h + Y_i +  Y_j + Y_k +  Y_l +  Y_m).
\end{align}
Note that the polynomial $H_a^2+2( T_a -  T_b - T_c-  T_d -  T_e+  T_f  -  T_g 
  -  T_h -  T_i  -  T_j -  T_k -  T_l - 
  T_m )$ is invariant, up to a U(1) phase, under $G^{C}$ in all labs and the polynomial $H_b^2+2(  Y_a -  Y_b  -  Y_c-  Y_d -  Y_e  +  Y_f -  Y_g -
  Y_h - Y_i -  Y_j - Y_k -  Y_l -  Y_m)$ is invariant, up to a U(1) phase, under $G^{C\gamma^5}$ in all labs. Therefore this equation implies that both these two polynomials are invariant, up to a U(1) phase, in all labs under the smallest Lie group that contains both $G^{C}$ and $G^{C\gamma^5}$ as subgroups (See Ref. \cite{spinorent} Theorem 2). The smallest Lie group that contains both $G^{C}$ and $G^{C\gamma^5}$ as subgroups is $\mathrm{U}(1)\times\mathrm{SL}(4,\mathbb{C})$. 

\subsection{Weyl particles}
If Alice's, Bob's, Charlie's and David's particles are Weyl particles the shared state is invariant under some combination of projections $P_L^A$ or $P_R^A$ by Alice, $P_L^B$ or $P_R^B$ by Bob, $P_L^C$ or $P_R^C$ by Charlie, and $P_L^D$ or $P_R^D$ by David. 
Thus the shared state of the spinors has the symmetry $\psi_{jklm}=(-1)^{|LA|}\psi_{(j-2) klm}=(-1)^{|LB|}\psi_{j (k-2)lm}=(-1)^{|LC|}\psi_{j k(l-2)m}=(-1)^{|LD|}\psi_{ jkl(m-2)}$, where $|LA|=1$ if the state is invariant under $P_L^A$ and otherwise zero, $|LB|=1$ if the state is invariant under $P_L^B$ and otherwise zero, $|LC|=1$ if the state is invariant under $P_L^C$ and otherwise zero, $|LD|=1$ if the state is invariant under $P_L^D$ and otherwise zero, and $j$,$k$,$l$, and $m$ are defined modulo 4.

In this case all the Lorentz invariants of degree 2 reduce to $16H$, up to a sign, where

\begin{eqnarray}\label{hpo}
H= &&\psi_{0000}\psi_{1111}-\psi_{0111}\psi_{1000} +\psi_{0110}\psi_{1001} +\psi_{0101}\psi_{1010} \nonumber\\
&&-\psi_{0100}\psi_{1011} + \psi_{ 0011}\psi_{1100} -\psi_{0010}\psi_{1101} -\psi_{0001}\psi_{1110},
\end{eqnarray}
 is the degree 2 polynomial local $\mathrm{SL}(2,\mathbb{C})$ invariant for the case of four non-relativistic spin-$\frac{1}{2}$ particles described by Luque and Thibon in Ref. \cite{luque}, by Verstraete, Dehaene and De Moor in Ref. \cite{moor} and by Wong and Christensen in Ref. \cite{wong}. 
Reference \cite{luque} also describe two degree 4 invariants $L$ and $M$ for the same scenario.
These are given by

\begin{eqnarray}\label{lpo}
L=&&\psi_{ 0011}\psi_{ 0110}\psi_{ 1001}\psi_{ 1100} -\psi_{ 0010}\psi_{ 0111}\psi_{1001}\psi_{ 1100}\nonumber\\&& - 
\psi_{ 0011}\psi_{ 0101}\psi_{ 1010}\psi_{ 1100 }+\psi_{ 0001}\psi_{ 0111}\psi_{ 1010}\psi_{ 1100 }\nonumber\\&& + 
\psi_{ 0010}\psi_{ 0101}\psi_{ 1011}\psi_{ 1100} -\psi_{ 0001}\psi_{ 0110}\psi_{ 1011}\psi_{ 1100 }\nonumber\\&&- 
\psi_{ 0011}\psi_{ 0110}\psi_{ 1000}\psi_{ 1101} +\psi_{ 0010}\psi_{ 0111}\psi_{ 1000}\psi_{ 1101}\nonumber\\&& + 
\psi_{ 0011}\psi_{ 0100}\psi_{ 1010}\psi_{ 1101} -\psi_{ 0000}\psi_{ 0111}\psi_{ 1010}\psi_{ 1101}\nonumber\\&& - 
\psi_{ 0010}\psi_{ 0100}\psi_{ 1011}\psi_{ 1101} +\psi_{ 0000}\psi_{ 0110}\psi_{ 1011}\psi_{ 1101}\nonumber\\&& + 
\psi_{ 0011}\psi_{ 0101}\psi_{ 1000}\psi_{ 1110} -\psi_{ 0001}\psi_{ 0111}\psi_{ 1000}\psi_{ 1110}\nonumber\\&& - 
\psi_{ 0011}\psi_{ 0100}\psi_{ 1001}\psi_{ 1110} +\psi_{ 0000}\psi_{ 0111}\psi_{ 1001}\psi_{ 1110}\nonumber\\&& + 
\psi_{ 0001}\psi_{ 0100}\psi_{ 1011}\psi_{ 1110} -\psi_{ 0000}\psi_{ 0101}\psi_{ 1011}\psi_{ 1110}\nonumber\\&& - 
\psi_{ 0010}\psi_{ 0101}\psi_{ 1000}\psi_{ 1111} +\psi_{ 0001}\psi_{ 0110}\psi_{ 1000}\psi_{ 1111}\nonumber\\&& + 
\psi_{ 0010}\psi_{ 0100}\psi_{ 1001}\psi_{ 1111} -\psi_{ 0000}\psi_{ 0110}\psi_{ 1001}\psi_{ 1111}\nonumber\\&& - 
\psi_{ 0001}\psi_{ 0100}\psi_{ 1010}\psi_{ 1111} +\psi_{ 0000}\psi_{ 0101}\psi_{ 1010}\psi_{ 1111},
\end{eqnarray}
and 
 \begin{eqnarray}\label{mpo}
M =&& -\psi_{ 0101}\psi_{ 0110}\psi_{ 1001}\psi_{ 1010} +\psi_{ 0100}\psi_{ 0111}\psi_{ 1001}\psi_{ 1010}\nonumber\\&& + 
 \psi_{ 0101}\psi_{ 0110}\psi_{ 1000}\psi_{ 1011} -\psi_{ 0100}\psi_{ 0111}\psi_{ 1000}\psi_{ 1011}\nonumber\\&& + 
 \psi_{ 0011}\psi_{ 0101}\psi_{ 1010}\psi_{ 1100} -\psi_{ 0001}\psi_{ 0111}\psi_{ 1010}\psi_{ 1100}\nonumber\\&& - 
\psi_{  0010}\psi_{ 0101}\psi_{ 1011}\psi_{ 1100} +\psi_{ 0000}\psi_{ 0111}\psi_{ 1011}\psi_{ 1100}\nonumber\\&& - 
 \psi_{ 0011}\psi_{ 0100}\psi_{ 1010}\psi_{ 1101} +\psi_{ 0001}\psi_{ 0110}\psi_{ 1010}\psi_{ 1101}\nonumber\\&& + 
\psi_{  0010}\psi_{ 0100}\psi_{ 1011}\psi_{ 1101} -\psi_{ 0000}\psi_{ 0110}\psi_{ 1011}\psi_{ 1101}\nonumber\\&& - 
\psi_{  0011}\psi_{ 0101}\psi_{ 1000}\psi_{ 1110} +\psi_{ 0001}\psi_{ 0111}\psi_{ 1000}\psi_{ 1110}\nonumber\\&& + 
\psi_{  0010}\psi_{ 0101}\psi_{ 1001}\psi_{ 1110} -\psi_{ 0000}\psi_{ 0111}\psi_{ 1001}\psi_{ 1110 }\nonumber\\&&- 
\psi_{  0001}\psi_{ 0010}\psi_{ 1101}\psi_{ 1110} +\psi_{ 0000}\psi_{ 0011}\psi_{ 1101}\psi_{ 1110} \nonumber\\&&+ 
\psi_{  0011}\psi_{ 0100}\psi_{ 1000}\psi_{ 1111} -\psi_{ 0001}\psi_{ 0110}\psi_{ 1000}\psi_{ 1111 }\nonumber\\&&- 
\psi_{  0010}\psi_{ 0100}\psi_{ 1001}\psi_{ 1111} +\psi_{ 0000}\psi_{ 0110}\psi_{ 1001}\psi_{ 1111 }\nonumber\\&&+ 
\psi_{  0001}\psi_{ 0010}\psi_{ 1100}\psi_{ 1111} -\psi_{ 0000}\psi_{ 0011}\psi_{ 1100}\psi_{ 1111}. 
 \end{eqnarray}
The polynomials $L,M,$ and $H^2$ are linearly independent.

Much like the degree 2 invariants all reduce to a multiple of $H$ for the case of Weyl particles, the degree 4 invariants reduce to linear combinations of $H^2,L,$ and $M$.
The polynomials $T_a$ and $Y_a$ reduce to $512(H^2  - 2 L - 4 M)$, the polynomials $T_c$ and $Y_c$ reduce to $512( -H^2 - 4 L - 2  M)$ while $T_f$ and $Y_f$ reduce to $512 (H^2-2L +2M)$.
The polynomials $T_b,T_e,Y_b$ and $Y_e$ reduce to $1024 (-L - 2 M)$, the polynomials
$T_d,T_h,Y_d$ and $Y_h$ reduce to $1024 (2L +  M)$, while $T_g,T_i,Y_g$ and $Y_i$ reduce to $1024 (-L +M)$. The polynomials $T_j,T_k,T_l,T_m,Y_j,Y_k,Y_l$, and $Y_m$ reduce to $512H^2$.

\subsection{Eigenspaces of the local Dirac Hamiltonians in the case of zero momenta and zero four-potentials}

We can consider the case of zero momentum and zero four-potential, sometimes described as the non-relativistic limit of a free particle. If Alice's, Bob's, Charlie's and David's particles are all in this limit and also in an eigenstate of the local Dirac Hamiltonians the shared state is invariant under some combination of projections $P_+^A$ or $P_-^A$ by Alice, $P_+^B$ or $P_-^B$ by Bob, $P_+^C$ or $P_-^C$ by Charlie, and $P_+^D$ or $P_-^D$ by David. In this case only a $2\times 2\times 2 \times 2$ subtensor of $\Psi^{ABCD}$ is nonzero.

For eigenspaces of the local Dirac Hamiltonians in this limit the polynomial $H_a$ reduces up to a relabelling of the indices to $2H$ where $H$ is the polynomial given in Eq. (\ref{hpo}). All the other degree 2 polynomials given in Eq. (\ref{ghj}) reduce to zero. This follows since $P_+C\gamma^5P_+=P_-C\gamma^5P_-=0$.
The degree 4 polynomials given in Eq. (\ref{ffr}) reduce up to a relabelling of the indices to linear combinations of $H^2$, $L$ and $M$ where $L$ is given in Eq. (\ref{lpo}) and $M$ is given in Eq. (\ref{mpo}).
The polynomial $T_a$ reduces up to a relabelling of the indices to $2 H^2- 4 L - 8 M $. 
The polynomial $T_c$ reduces up to a relabelling of the indices to $- 2 H^2 - 8 L - 4 M$.
The polynomial $T_f$ reduces up to a relabelling of the indices to $2 H^2 -4 L+ 4 M  $.
The polynomials $T_b$ and $T_e$ reduce up to a relabelling of the indices to $ - 4 L- 8 M$.
The polynomials $T_d$ and $T_h$ reduce up to a relabelling of the indices to $8 L+ 4 M $.
The polynomials $T_g$ and $T_i$ reduce up to a relabelling of the indices to $-4 L+ 4 M $.
Finally, the polynomials $T_j$, $T_k$, $T_l$ and $T_m$ reduce up to a relabelling of the indices to $ 2 H^2$.
Since we have that $P_+C\gamma^5P_+=P_-C\gamma^5P_-=0$ all polynomials of degree 4 given in Eq. (\ref{ffr2}) reduce to zero. Moreover, this property of $C\gamma^5$ implies that the polynomials in Eq. (\ref{ffr}) are the only polynomials on the form given in  Eq. (\ref{cvn}) that do not reduce to zero in this case.

\subsection{Examples of fourpartite spinor entangled states}

Here we consider a few examples of fourpartite entangled states to illustrate how inequivalent forms of spinor entanglement are distinguished by the polynomials.
We can consider analogues of the four-partite entangled Greenberger-Horne-Zeilinger (GHZ)\cite{ghz} state for non-relativistic spin-$\frac{1}{2}$ particles.
One such state is
\begin{eqnarray}
\frac{1}{\sqrt{2}}({\phi_0^A}\otimes{\phi_0^B}\otimes{\phi_0^C}\otimes{\phi_0^D}+{\phi_1^A}\otimes{\phi_1^B}\otimes{\phi_1^C}\otimes{\phi_1^D}).
\end{eqnarray}
For this state $|H_{a}|=1$, but all the other degree 2 polynomials are identically zero and $|T_a|=|T_c|=|T_f|=|T_j|=|T_k|=|T_l|=|T_m|=1/2$, while the other degree 4 polynomials in Eq. (\ref{ffr}) and in  Eq. (\ref{ffr2}) are zero.

Similarly we can construct a state
\begin{eqnarray}
\frac{1}{\sqrt{2}}({\phi_2^A}\otimes{\phi_2^B}\otimes{\phi_2^C}\otimes{\phi_2^D}+{\phi_1^A}\otimes{\phi_1^B}\otimes{\phi_1^C}\otimes{\phi_1^D}).
\end{eqnarray}
for which $|H_{b}|=1$, but all the other degree 2 polynomials are identically zero and $|Y_a|=|Y_c|=|Y_f|=|Y_j|=|Y_k|=|Y_l|=|Y_m|=1/2$, while the other degree 4 polynomials in Eq. (\ref{ffr}) and in  Eq. (\ref{ffr2}) are zero.

 Furthermore, one can construct analogues of the GHZ state that involve more than two basis spinors. One example is
\begin{eqnarray}
\frac{1}{\sqrt{2}}({\phi_0^A}\otimes{\phi_1^B}\otimes{\phi_1^C}\otimes{\phi_3^D}+{\phi_1^A}\otimes{\phi_0^B}\otimes{\phi_0^C}\otimes{\phi_0^D}).
\end{eqnarray}
for which $|H_{c}|=1$, but all the other degree 2 polynomials are identically zero and all the degree 4 polynomials in Eq. (\ref{ffr}) and in  Eq. (\ref{ffr2}) are zero. Another example is
\begin{eqnarray}
\frac{1}{\sqrt{2}}({\phi_0^A}\otimes{\phi_1^B}\otimes{\phi_3^C}\otimes{\phi_3^D}+{\phi_1^A}\otimes{\phi_0^B}\otimes{\phi_0^C}\otimes{\phi_2^D}).
\end{eqnarray}
for which $|H_{d}|=1$, but all the other degree 2 polynomials are identically zero and all the degree 4 polynomials in Eq. (\ref{ffr}) and in  Eq. (\ref{ffr2}) are zero. As is the case for the four examples given here, for any GHZ like state of this kind only one of the polynomials of degree 2 is nonzero.
 
An example of a generalized GHZ state with three terms is

\begin{eqnarray}
\frac{1}{\sqrt{3}}({\phi_0^A}\otimes{\phi_0^B}\otimes{\phi_0^C}\otimes{\phi_0^D}+{\phi_1^A}\otimes{\phi_1^B}\otimes{\phi_1^C}\otimes{\phi_1^D}+{\phi_2^A}\otimes{\phi_2^B}\otimes{\phi_2^C}\otimes{\phi_2^D}).
\end{eqnarray}
for which $|H_{a}|=|H_{b}|=2/3$, but all the other degree 2 polynomials are identically zero, and $|T_a|=|T_c|=|T_f|=|T_j|=|T_k|=|T_l|=|T_m|=|Y_a|=|Y_c|=|Y_f|=|Y_j|=|Y_k|=|Y_l|=|Y_m|=2/9$ while the other degree 4 polynomials in Eq. (\ref{ffr}) and in  Eq. (\ref{ffr2}) are zero. A generalized GHZ state with four terms is

\begin{align}
\frac{1}{{2}}(&{\phi_0^A}\otimes{\phi_0^B}\otimes{\phi_0^C}\otimes{\phi_0^D}+{\phi_1^A}\otimes{\phi_1^B}\otimes{\phi_1^C}\otimes{\phi_1^D}\nonumber\\&+{\phi_2^A}\otimes{\phi_2^B}\otimes{\phi_2^C}\otimes{\phi_2^D}+{\phi_3^A}\otimes{\phi_3^B}\otimes{\phi_3^C}\otimes{\phi_3^D}).
\end{align}
for which $|H_{a}|=|H_{b}|=1$, but all the other degree 2 polynomials are identically zero, and $|T_a|=|T_c|=|T_f|=|T_j|=|T_k|=|T_l|=|T_m|=|Y_a|=|Y_c|=|Y_f|=|Y_j|=|Y_k|=|Y_l|=|Y_m|=1/4$ while the other degree 4 polynomials in Eq. (\ref{ffr}) and in  Eq. (\ref{ffr2}) are zero.

An example of a state not on the GHZ form is the analogue of the so called cluster state \cite{briegel}
\begin{align}
\frac{1}{{2}}(&{\phi_0^A}\otimes{\phi_0^B}\otimes{\phi_0^C}\otimes{\phi_0^D}-{\phi_1^A}\otimes{\phi_1^B}\otimes{\phi_1^C}\otimes{\phi_1^D}\nonumber\\&+{\phi_0^A}\otimes{\phi_0^B}\otimes{\phi_1^C}\otimes{\phi_1^D}+{\phi_1^A}\otimes{\phi_1^B}\otimes{\phi_0^C}\otimes{\phi_0^D}).
\end{align}
For this state all the degree 2 polynomials are zero and $|T_a|=|T_b|=|T_e|=1/2$, and $|T_c|=|T_d|=|T_f|=|T_g|=|T_h|=|T_i|=1/4$ while the other degree 4 polynomials in Eq. (\ref{ffr}) and in  Eq. (\ref{ffr2}) are zero.

Another analogue of the cluster state is
\begin{align}
\frac{1}{{2}}(&{\phi_0^A}\otimes{\phi_0^B}\otimes{\phi_0^C}\otimes{\phi_0^D}-{\phi_3^A}\otimes{\phi_3^B}\otimes{\phi_3^C}\otimes{\phi_3^D}\nonumber\\&+{\phi_0^A}\otimes{\phi_0^B}\otimes{\phi_3^C}\otimes{\phi_3^D}+{\phi_3^A}\otimes{\phi_3^B}\otimes{\phi_0^C}\otimes{\phi_0^D}).
\end{align}
For this state all the degree 2 polynomials are zero and $|Y_a|=|Y_b|=|Y_e|=1/2$, and $|Y_c|=|Y_d|=|Y_f|=|Y_g|=|Y_h|=|Y_i|=1/4$ while the other degree 4 polynomials in Eq. (\ref{ffr}) and in  Eq. (\ref{ffr2}) are zero.

A further example of a state for which all degree 2 polynomials are zero is the state
\begin{align}
\frac{1}{{2}}(&{\phi_0^A}\otimes{\phi_2^B}\otimes{\phi_2^C}\otimes{\phi_2^D}+{\phi_2^A}\otimes{\phi_1^B}\otimes{\phi_3^C}\otimes{\phi_3^D}\nonumber\\&+{\phi_1^A}\otimes{\phi_3^B}\otimes{\phi_1^C}\otimes{\phi_1^D}+{\phi_3^A}\otimes{\phi_0^B}\otimes{\phi_0^C}\otimes{\phi_0^D}).
\end{align}
Moreover, for this state $|T_a|=|Y_b|=1/4$ while the other degree 4 polynomials in Eq. (\ref{ffr}) and in  Eq. (\ref{ffr2})  are zero. 
Likewise, all degree 2 polynomials are zero for the state
\begin{align}
\frac{1}{{2}}(&{\phi_1^A}\otimes{\phi_1^B}\otimes{\phi_2^C}\otimes{\phi_2^D}+{\phi_2^A}\otimes{\phi_3^B}\otimes{\phi_3^C}\otimes{\phi_3^D}\nonumber\\&+{\phi_0^A}\otimes{\phi_2^B}\otimes{\phi_1^C}\otimes{\phi_1^D}+{\phi_3^A}\otimes{\phi_0^B}\otimes{\phi_0^C}\otimes{\phi_0^D}),
\end{align}
and $|T_b|=|Y_k|=1/4$ while the other degree 4 polynomials in Eq. (\ref{ffr}) and in  Eq. (\ref{ffr2}) are zero. 

We can see that the ten examples above belong to ten different entanglement classes that can be discriminated by the polynomial invariants.

\subsection{Examples of bipartite entangled four spinor states}
The polynomials constructed as tensor sandwich contractions for the case of four spinors can be non-zero also for states that are not fourpartite entangled but only bipartite entangled.
An example of a state of four Dirac spinors that is a product state over the partitioning AB|CD but bipartite entangled on AB and CD is
\begin{align}
\frac{1}{{2}}(&{\phi_0^A}\otimes{\phi_0^B}\otimes{\phi_0^C}\otimes{\phi_0^D}+{\phi_1^A}\otimes{\phi_1^B}\otimes{\phi_1^C}\otimes{\phi_1^D}\nonumber\\&+{\phi_0^A}\otimes{\phi_0^B}\otimes{\phi_1^C}\otimes{\phi_1^D}+{\phi_1^A}\otimes{\phi_1^B}\otimes{\phi_0^C}\otimes{\phi_0^D}).
\end{align}
For this state $|H_a|=1$ while the other degree 2 polynomials are zero and $|T_a|=1$, $|T_b|=|T_e|=|T_j|=|T_k|=|T_l|=|T_m|=1/2$, and $|T_c|=|T_d|=|T_f|=|T_g|=|T_h|=|T_i|=1/4$ while the other degree 4 polynomials in Eq. (\ref{ffr}) and  Eq. (\ref{ffr2}) are zero.
Similarly, the state
\begin{align}
\frac{1}{{2}}(&{\phi_0^A}\otimes{\phi_2^B}\otimes{\phi_2^C}\otimes{\phi_0^D}+{\phi_1^A}\otimes{\phi_3^B}\otimes{\phi_1^C}\otimes{\phi_1^D}\nonumber\\&+{\phi_0^A}\otimes{\phi_3^B}\otimes{\phi_1^C}\otimes{\phi_0^D}+{\phi_1^A}\otimes{\phi_2^B}\otimes{\phi_2^C}\otimes{\phi_1^D}).
\end{align}
is a product state over the partitioning AD|BC but bipartite entangled on AD and BC.
For this state $|H_d|=1$ while the other degree 2 polynomials are zero and all degree 4 polynomials in Eq. (\ref{ffr}) and  Eq. (\ref{ffr2}) are zero.

\section{The case of five spinors}\label{five}
For five Dirac spinors the state coefficients can be arranged as a $4\times 4\times 4\times 4\times 4$ tensor $\Psi^{ABCDE}$.
Transformations $S^A,S^B,S^C,S^D$, and $S^E$ acting locally on Alice's, Bob's, Charlie's, David's, and Erin's particle, respectively, are described as
\begin{eqnarray}
\Psi^{ABCDE}_{ijklm}\to\sum_{nopqr}S^A_{in}S^B_{jo}S^C_{kp}S^D_{lq}S^E_{mr}\Psi^{ABCDE}_{nopqr}.
\end{eqnarray}
Lorentz invariants of degree 2 can be constructed as tensor sandwich contractions involving two copies of $\Psi^{ABCDE}$, however all such polynomials are identically zero due to the antisymmetry of $C$ and $C\gamma^5$. 

The next lowest degree is 4.
From the above follows that invariants of degree 4 that factorize as a product of two degree 2 polynomials are identically zero.
Thus it remains to consider the tensor sandwich contractions involving four copies of $\Psi^{ABCDE}$ that do not factorize into two degree 2 polynomials. There are 40 inequivalent such ways to pair up the tensor indices of the four copies. These are given in \ref{wides}.
For each pairing of tensor indices we can consider $(2^{10}-2^5)/2+2^5=528$ ways to choose the $X$s as either $C$ or $C\gamma^5$. This gives a total of 21120 ways to construct polynomials by this method. Due to the computational difficulty in constructing the polynomials and testing their linear independence no explicit examples are given here.

\section{Degree 2 polynomials for more than five spinors}\label{getto}

For $N$ Dirac spinors the state coefficients can be arranged as a $4\times 4\times \dots\times 4$ tensor $\Psi^{ABCD\dots N}$.
Transformations $S^A,S^B,S^C,\dots S^N$ acting locally on the particles, are described as
\begin{eqnarray}
\Psi^{ABCD\dots N}_{ijk\dots l}\to\sum_{mno\dots p}S^A_{im}S^B_{jn}S^C_{ko}\dots S^N_{lp}\Psi^{ABCD\dots N}_{mno\dots p}.
\end{eqnarray}
Lorentz invariants of degree 2 can be constructed as tensor sandwich contractions of the form 
\begin{eqnarray}\label{hjl}
\sum_{jkl\dots mnpq \dots r}X_{nj}X_{pk}X_{ql}\dots X_{rm}\Psi^{ABC\dots N}_{jkl\dots m}\Psi^{ABC\dots N}_{npq\dots r},
\end{eqnarray}
involving two copies of $\Psi^{ABC\dots N}$ where either $X=C$ or $X=C\gamma^5$ for each instance. For odd $N$ these polynomials are zero due to the antisymmetry of $C$ and $C\gamma^5$. For every even $N$ on the other hand there is $2^N$ linearly independent polynomials on the form in Eq. (\ref{hjl}) corresponding to the different choices of $X$ as either $X=C$ or $X=C\gamma^5$.

If the particles of all observers are Weyl particles the shared state is invariant under some combination of projections, $P_L$ or $P_R$ for each observer. In this case the polynomials on the form in Eq. (\ref{hjl}) for even $N$ all reduce, up to a sign, to $2^N\tau_{1\dots N}$ where $\tau_{1\dots N}=\sum_{jkl\dots mnpq \dots r\in\{0,1\}}\epsilon_{nj}\epsilon_{pk}\epsilon_{ql}\dots \epsilon_{rm}\psi^{ABC\dots N}_{jkl\dots m}\psi^{ABC\dots N}_{npq\dots r}$ where the Levi-Civita antisymmetric symbol $\epsilon_{jk}$ is defined by $1=\epsilon_{10}=-\epsilon_{01},$ and $\epsilon_{00}=\epsilon_{11}=0$. The polynomial $\tau_{1\dots N}$ is the $N$-tangle introduced by Wong and Christensen in Ref. \cite{wong}.

In the case of zero momenta and zero four-potentials and with the state of the particles belonging to an eigenspace of the local Dirac Hamiltonians the polynomial where each $X=C$ reduces up to a relabelling of the indices to $\tau_{1\dots N}$ while the $2^N-1$ other degree 2 polynomials are zero.

The lowest nonzero case not already described in Section \ref{four} or in Ref. \cite{spinorent} is six Dirac spinors. For this case there is 64 linearly independent polynomials on the form in Eq. (\ref{hjl}) corresponding to the different choices of $X$ as either $X=C$ or $X=C\gamma^5$. For Weyl particles all these polynomials reduce, up to a sign, to $64\tau_{1\dots 6}$ where $\tau_{1\dots 6}$ is given by

\begin{eqnarray}
\tau_{1\dots 6}= &&\psi_{000000}\psi_{111111}-\psi_{011111}\psi_{100000} -\psi_{101111}\psi_{010000}\nonumber\\&&-\psi_{110111}\psi_{001000}-\psi_{111011}\psi_{000100}-\psi_{111101}\psi_{000010}\nonumber\\&&-\psi_{111110}\psi_{000001}+\psi_{001111}\psi_{110000}
+\psi_{010111}\psi_{101000}\nonumber\\&&+\psi_{011011}\psi_{100100}+\psi_{011101}\psi_{100010}+\psi_{011110}\psi_{100001}\nonumber\\&&+\psi_{100111}\psi_{011000}+\psi_{101011}\psi_{010100}+\psi_{101101}\psi_{010010}\nonumber\\&&+\psi_{101110}\psi_{010001}+\psi_{110011}\psi_{001100}
+\psi_{110101}\psi_{001010}\nonumber\\&&+\psi_{110110}\psi_{001001}+\psi_{111001}\psi_{000110}+\psi_{111010}\psi_{000101}\nonumber\\&&+\psi_{111100}\psi_{000011}-\psi_{111000}\psi_{000111}-\psi_{110100}\psi_{001011}\nonumber\\&&-\psi_{110010}\psi_{001101}-\psi_{110001}\psi_{001110}-\psi_{101100}\psi_{010011}\nonumber\\&&-\psi_{101010}\psi_{010101}-\psi_{101001}\psi_{010110}-\psi_{100110}\psi_{011001}\nonumber\\&&-\psi_{100101}\psi_{011010}-\psi_{100011}\psi_{011100}.
\end{eqnarray}
The polynomial $\tau_{1\dots 6}$ is the $6$-tangle of Wong and Christensen \cite{wong}.

In the case of zero momenta and zero four-potentials and with the state of the particles belonging to an eigenspace of the local Dirac Hamiltonians the polynomial where each $X=C$ reduces up to a relabelling of the indices to $\tau_{1\dots 6}$ while the 63 other degree 2 polynomials are zero.

\section{General properties of polynomial invariants of connected complex reductive Lie groups}\label{lli}
For a system of $n$ Dirac spinors any polynomial constructed by tensor sandwich contractions is invariant, up to a U(1) phase, under some Lie group $G_1\otimes G_2\otimes\dots\otimes G_n$ where each group $G_k$ is one out of $G^C$, $G^{C\gamma^5}$, $G^{C\gamma^5}\cap G^C$ and $\mathrm{U}(1)\times\mathrm{SL}(4,\mathbb{C})$. Here we consider some general properties of these groups and the polynomials invariant, up to a U(1) phase, under these groups.

For each of the groups $G^C$, $G^{C\gamma^5}$, $G^{C\gamma^5}\cap G^C$ and $\mathrm{U}(1)\times\mathrm{SL}(4,\mathbb{C})$ we can consider its determinant one subgroup. The determinant one subgroup of $\mathrm{U}(1)\times\mathrm{SL}(4,\mathbb{C})$ is clearly 
$\mathrm{SL}(4,\mathbb{C})$. Denote the determinant one subgroup of $G^C$ by $\mathrm{S}G^C$ and denote the determinant one subgroup of $G^{C\gamma^5}$ by $\mathrm{S}G^{C\gamma^5}$. 
The group $\mathrm{S}G^C$ consists of all linear transformations that preserve the skew-symmetric bilinear form $\psi^TC\varphi$. The group $\mathrm{S}G^{C\gamma^5}$ consists of all linear transformations that preserve the skew-symmetric bilinear form $\psi^TC\gamma^5\varphi$.
Both $\mathrm{S}G^C$ and $\mathrm{S}G^{C\gamma^5}$ are isomorphic to $\mathrm{Sp}(4,\mathbb{C})$.
The group $\mathrm{S}G^{C\gamma^5}\cap \mathrm{S}G^C$ consists of all linear transformations that preserve both of the skew-symmetric bilinear forms $\psi^TC\varphi$ and $\psi^TC\gamma^5\varphi$ and is isomorphic to $\mathrm{SL}(2,\mathbb{C})\times\mathrm{SL}(2,\mathbb{C})$.

Next consider $G_U^C$ and $G_U^{C\gamma^5}$ and denote their determinant one subgroups by $\mathrm{S}G_U^C$ and $\mathrm{S}G_U^{C\gamma^5}$ respectively. The groups $\mathrm{S}G_U^C$ and $\mathrm{S}G_U^{C\gamma^5}$ are the maximal compact subgroups of $\mathrm{S}G^C$ and $\mathrm{S}G^{C\gamma^5}$, respectively, i.e., $\mathrm{S}G_U^C=\mathrm{S}G^C\cap \mathrm{U}(n)$ and $\mathrm{S}G_U^{C\gamma^5}=\mathrm{S}G^{C\gamma^5}\cap \mathrm{U}(n)$. Both $\mathrm{S}G_U^C$ and $\mathrm{S}G_U^{C\gamma^5}$ are isomorphic to $\mathrm{Sp}(2)$. The group $\mathrm{S}G_U^{C\gamma^5}\cap \mathrm{S}G_U^C$ is the maximal compact subgroup of $\mathrm{S}G^{C\gamma^5}\cap \mathrm{S}G^C$ and it is isomorphic to $\mathrm{SU}(2)\times\mathrm{SU}(2)$.

Let $\mathrm{Lie}(G)$ be the Lie algebra of a Lie group $G\subset\mathrm{GL}(n,\mathbb{C})$ and let $\mathrm{Lie}(K)$ be the Lie algebra of its maximal compact subgroup $K\equiv G\cap \mathrm{U}(n)$. We say that
$G$ is the complexification of $K$ if $\mathrm{Lie}(G)=\mathrm{Lie}(K)+i\mathrm{Lie}(K)$. The group $\mathrm{S}G^C$ is the complexification of $\mathrm{S}G_U^C$, the group $\mathrm{S}G^{C\gamma^5}$ is the complexification of $\mathrm{S}G_U^{C\gamma^5}$, the group $\mathrm{S}G^{C\gamma^5}\cap \mathrm{S}G^C$ is the complexification of $\mathrm{S}G_U^{C\gamma^5}\cap \mathrm{S}G_U^C$ and $\mathrm{SL}(4,\mathbb{C})$ is the complexification of $\mathrm{SU}(4)$.

A Lie subgroup of $\mathrm{GL}(n,\mathbb{C})$ that is the complexification of a compact Lie group is here called a {\it complex reductive group} following Onishchik and Vinberg (See Ref. \cite{vinberg} Ch. 5 {\S}2. 5$^{\circ}$). Note that this property is invariant under conjugation of the group by $S\in \mathrm{GL}(n,\mathbb{C})$, i.e., if a group $G\subset \mathrm{GL}(n,\mathbb{C})$ is complex reductive so is the group $SGS^{-1}$ defined by the elements $SgS^{-1}$ for $g\in G$. Thus the groups that preserve the bilinear forms ${\psi^T}(S^{-1})^TCS^{-1}{\varphi}$ and ${\psi^T}(S^{-1})^TC\gamma^5S^{-1}{\varphi}$ are complex reductive regardless of the choice of matrices $S\gamma^0S^{-1},S\gamma^1S^{-1},S\gamma^2S^{-1},S\gamma^3S^{-1}$ satisfying Eq. (\ref{anti}). However, since every compact Lie group can be represented as a subgroup of a unitary group (See e.g. Ref. \cite{nolan} Ch. 2.1.2), we can consider the case of compact subgroups of $\mathrm{U}(n)$ and their complexifications without loss of generality.

Any polynomial constructed by tensor sandwich contractions is invariant under some group $G_1\otimes G_2\otimes\dots\otimes G_n$ acting on $n$ Dirac spinors where each group $G_k$ is one out of the groups  $\mathrm{S}G^C$, $\mathrm{S}G^{C\gamma^5}$, $\mathrm{S}G^{C\gamma^5}\cap \mathrm{S}G^C$ and $\mathrm{SL}(4,\mathbb{C})$. The tensor product of two connected complex reductive Lie groups is a connected complex reductive Lie group and thus $G_1\otimes G_2\otimes\dots\otimes G_n$ is a connected complex reductive Lie group. The algebra of polynomial invariants of a connected complex reductive Lie group have a number of useful general properties. We describe some of these properties in the following.

Let $G$ be a connected Lie subgroup of $\mathrm{GL}(n,\mathbb{C})$ that is the complexification of its maximal compact subgroup.
Then the algebra of polynomial invariants under $G$ and the algebra of polynomial invariants under its maximal compact subgroup are the same. To see this we consider the following theorem based on the so called Unitarian Trick \cite{hurw,weeyl}

\begin{theorem}\label{uii}
Let $G$ be a connected Lie subgroup of $\mathrm{GL}(n,\mathbb{C})$ and let $K\equiv G\cap \mathrm{U}(n)$. Assume that $G$ is the complexification of $K$. 
Let $\mathcal{O}(\mathbb{C}^n)$ be the algebra over $\mathbb{C}$ of polynomials on $\mathbb{C}^n$, let $\mathcal{O}(\mathbb{C}^n)^G\subset \mathcal{O}(\mathbb{C}^n)$ be the subalgebra of polynomials that are invariant under $G$, i.e., $\mathcal{O}(\mathbb{C}^n)^G\equiv\{f\in\mathcal{O}(\mathbb{C}^n)|f(gx)=f(x) \phantom{t}\textrm{for all}\phantom{t} g\in G, x\in \mathbb{C}^n\}$, and let $\mathcal{O}(\mathbb{C}^n)^K\subset \mathcal{O}(\mathbb{C}^n)$ be the subalgebra of polynomials that are invariant under $K$, i.e., $\mathcal{O}(\mathbb{C}^n)^K\equiv\{f\in\mathcal{O}(\mathbb{C}^n)|f(gx)=f(x) \phantom{t}\textrm{for all}\phantom{t} g\in K, x\in \mathbb{C}^n\}$. Then $\mathcal{O}(\mathbb{C}^n)^G=\mathcal{O}(\mathbb{C}^n)^K$.
\end{theorem}
\begin{proof}
Let $F\in\mathcal{O}(\mathbb{C}^n)$ and $v\in \mathbb{C}^n$. For any $n\times n$ matrix $M$ consider the function on $\mathbb{R}\times \mathbb{C}^n$ defined by $F(e^{Mt}v)$. Then $\frac{d}{dt}F(e^{Mt}v)|_{t=0}=Mv\cdot\nabla F(v)$.
Since $F(e^{M(t+a)}v)=F(e^{Mt}e^{Ma}v)$ we see that $\frac{d}{dt}F(e^{Mt}v)|_{t=0}=0$ for all $v\in \mathbb{C}^n$  if and only if $\frac{d}{dt}F(e^{Mt}v)|_{t=a}=0$ for all $a\in \mathbb{R}$ and all $v\in \mathbb{C}^n$. This is equivalent to the condition $F(e^{Mt}v)=F(v)$ for all $t\in \mathbb{R}$ and all $v\in \mathbb{C}^n$.

Now consider the Lie algebra $\mathrm{Lie}(G)$. Since $G$ is connected every element $g\in G$ can be written as $g=e^{M_1}e^{M_2}\dots e^{M_l}$ for $M_k\in \mathrm{Lie}(G)$. Therefore the condition $\frac{d}{dt}F(e^{M_kt}v)|_{t=0}=M_kv\cdot\nabla F(v)=0$ for all $M_k\in \mathrm{Lie}(G)$ is equivalent to $F$ being invariant under $G$.
Next consider the Lie algebra $\mathrm{Lie}(K)$. Since $K$ is connected every element $g\in K$ can be written as $g=e^{M_1}e^{M_2}\dots e^{M_l}$ for $M_k\in \mathrm{Lie}(K)$. Therefore the condition $\frac{d}{dt}F(e^{M_kt}v)|_{t=0}=M_kv\cdot\nabla F(v)=0$ for all $M_k\in \mathrm{Lie}(K)$ is equivalent to $F$ being invariant under $K$.

If $\frac{d}{dt}F(e^{M_kt}v)|_{t=0}=M_kv\cdot\nabla F(v)=0$ for some $M_k\in \mathrm{Lie}(K)$ we see that $\frac{d}{dt}F(e^{iM_kt}v)|_{t=0}=iM_kv\cdot\nabla F(v)=0$. 
Since $\mathrm{Lie}(G)=\mathrm{Lie}(K)+i\mathrm{Lie}(K)$ it follows that if $F$ is invariant under $K$ it is also invariant under $G$.
Moreover, since $K$ is a subgroup of $G$ we have that if $F$ is invariant under $G$ it is invariant under $K$.
Thus $\mathcal{O}(\mathbb{C}^n)^G=\mathcal{O}(\mathbb{C}^n)^K$.
See also e.g. Ref. \cite{dolgachev} for the special case of $\mathrm{SL}(2, \mathbb{C})$ and $\mathrm{SU}(2)$.
\end{proof}

Another property of a connected complex reductive Lie group is that the algebra of polynomial invariants of the group is finitely generated.
This is the content of the so called Hilbert's Finiteness Theorem \cite{hilbert,weeyl,nagata}
\begin{theorem}\label{fin}
Let $G$ be a connected complex reductive Lie subgroup of $\mathrm{GL}(n,\mathbb{C})$. Let $\mathcal{O}(\mathbb{C}^n)$ be the algebra over $\mathbb{C}$ of polynomials on $\mathbb{C}^n$ and let $\mathcal{O}(\mathbb{C}^n)^G\subset \mathcal{O}(\mathbb{C}^n)$ be the subalgebra of polynomials that are invariant under $G$, i.e., $\mathcal{O}(\mathbb{C}^n)^G\equiv\{f\in\mathcal{O}(\mathbb{C}^n)|f(gx)=f(x) \phantom{t}\textrm{for all}\phantom{t} g\in G, x\in \mathbb{C}^n\}$. Then $\mathcal{O}(\mathbb{C}^n)^G$ is finitely generated over $\mathbb{C}$.
\end{theorem}
\begin{proof}
See e.g. Ref. \cite{nolan} Ch. 3.1.4, Ref. \cite{wall} Ch. 5.1.1 or Ref. \cite{nagata}.
\end{proof}
This theorem implies that for any number $n$ of spacelike separated Dirac spinors and any group $G_1\otimes G_2\otimes\dots\otimes G_n$ where each $G_k$ is one out of $\mathrm{S}G^C$, $\mathrm{S}G^{C\gamma^5}$, $\mathrm{S}G^{C\gamma^5}\cap \mathrm{S}G^C$ and $\mathrm{SL}(4,\mathbb{C})$, a finite number of polynomials invariant under this group generate the algebra of all such invariants. Moreover, by Theorem \ref{uii} the same finitely generated algebra is the invariants of the maximal compact subgroup of $G_1\otimes G_2\otimes\dots\otimes G_n$. Note, however that this does not imply that such a finite number of polynomials can necessarily be constructed as tensor sandwich contractions.

If two orbits of a connected complex reductive Lie group are such that their closures do not overlap there exists a polynomial invariant of the group that distinguishes between the orbit closures.
\begin{theorem}
Let $G$ be a connected complex reductive Lie subgroup of $\mathrm{GL}(n,\mathbb{C})$. Let $\mathcal{O}(\mathbb{C}^n)$ be the algebra over $\mathbb{C}$ of polynomials on $\mathbb{C}^n$ and let $\mathcal{O}(\mathbb{C}^n)^G\subset \mathcal{O}(\mathbb{C}^n)$ be the subalgebra of polynomials that are invariant under $G$, i.e., $\mathcal{O}(\mathbb{C}^n)^G\equiv\{f\in\mathcal{O}(\mathbb{C}^n)|f(gx)=f(x) \phantom{t}\textrm{for all}\phantom{t} g\in G, x\in \mathbb{C}^n\}$. Let $Gx$ and $Gy$ be $G$-orbits for $x,y\in \mathbb{C}^n$ such that their closures $\overline{Gx}$ and $\overline{Gy}$ satisfy $\overline{Gx}\cap \overline{Gy}=\emptyset$. Then there exist $f\in\mathcal{O}(\mathbb{C}^n)^G$ such that $f_{|\overline{Gx}}=1$ and $f_{|\overline{Gy}}=0$.
\end{theorem}
\begin{proof}
See e.g. Ref. \cite{nolan} Ch. 3.1.4.
\end{proof}

Consider a set $\Sigma$ of states on which all polynomial invariants take constant values and that is not a subset of any other set on which all polynomial invariants take constant values.
A property of the polynomial invariants of a connected complex reductive Lie group is that 
any such subset $\Sigma$ of states contains a unique closed orbit.
In particular this implies that every open orbit has one and only one closed orbit in its closure.

\begin{theorem}\label{tsw}
Let $G$ be a connected complex reductive Lie subgroup of $\mathrm{GL}(n,\mathbb{C})$. Let $\mathcal{O}(\mathbb{C}^n)$ be the algebra over $\mathbb{C}$ of polynomials on $\mathbb{C}^n$ and let $\mathcal{O}(\mathbb{C}^n)^G\subset \mathcal{O}(\mathbb{C}^n)$ be the subalgebra of polynomials that are invariant under $G$, i.e., $\mathcal{O}(\mathbb{C}^n)^G\equiv\{f\in\mathcal{O}(\mathbb{C}^n)|f(gx)=f(x) \phantom{t}\textrm{for all}\phantom{t} g\in G, x\in \mathbb{C}^n\}$. 
Let $S_x=\{y\in \mathbb{C}^n| f(y)=f(x)\phantom{t}\textrm{for all}\phantom{t} f\in \mathcal{O}(\mathbb{C}^n)^G\}$. 
Then each set $S_x$ contains a unique closed orbit of $G$. If $Gw$ is the unique closed orbit in $S_x$ and $v\in S_x$ then $Gw\subset \overline {Gv}$.
\end{theorem}
\begin{proof}
See e.g. Ref. \cite{nolan} Ch. 3.4.1.
\end{proof}

Theorem \ref{tsw} together with Theorem \ref{fin} implies that there exists a finite number $m$ of polynomials such that the set of $m$-tuples of values of these polynomials is in one-to-one correspondence with the set of closed orbits. Thus the polynomials provide coordinates for the set of closed orbits. In particular, for any group $G_1\otimes G_2\otimes\dots\otimes G_n$ where each $G_k$ is one out of $\mathrm{S}G^C$, $\mathrm{S}G^{C\gamma^5}$, $\mathrm{S}G^{C\gamma^5}\cap \mathrm{S}G^C$ and $\mathrm{SL}(4,\mathbb{C})$ a finite number of polynomial invariants distinguishes all closed orbits of $G_1\otimes G_2\otimes\dots\otimes G_n$. An open orbit on the other hand cannot be distinguished by such polynomials from the closed orbit in its closure. However it is an open question if a set of polynomials with these properties can be constructed as tensor sandwich contractions.

\section{Discussion and Conclusions}\label{diss}
In this work we have considered the problem of describing the spinor entanglement of a system of multiple Dirac particles with definite momenta held by spacelike separated observers. The general approach followed is the same as was considered in Ref. \cite{spinorent} for the case of two Dirac particles. We reviewed some properties of the Dirac equation, the spinor representation of the Lorentz group and the charge conjugation transformation, as well as some properties of Lorentz invariant skew-symmetric bilinear forms.
The assumption was made that the local curvature of spacetime can be neglected and each particle described as belonging to a Minkowski space. Beyond this, we assumed that the physical scenario is such that it is warranted to treat particle momentum eigenmodes as having a finite spatial extent. 
Moreover, we assumed that the tensor products of the individual particle momentum eigenmodes can be used as a basis for the multi-particle states.

Given these assumptions and utilizing the properties of the skew-symmetric bilinear forms we described a method to construct polynomials in the state coefficients of the system of spacelike separated Dirac particles. This method is a generalization of the method used in Ref. \cite{spinorent} for the case of two Dirac particles.
The polynomials constructed by this method are invariant under the spinor representations of the local proper orthochronous Lorentz groups. Moreover, each such Lorentz invariant polynomial is identically zero for states where any one of the spinors is in a product state with the other spinors. 

For the case of three and four Dirac spinors polynomials of degree 2 and 4 were constructed and their linear independence tested. For three spinors no non-zero degree 2 polynomials can be constructed, but a set of 67 linearly independent degree 4 polynomials were given. For four spinors 16 linearly independent degree 2 polynomials were derived and a further 26 polynomials of degree 4 were given. A larger number of degree 4 polynomials exist but was not derived due to the complexity of constructing a large number of such polynomials and testing their linear independence. For five spinors no nonzero degree 2 polynomial can be constructed. The degree 4 polynomials that can be constructed for five spinors using the method in this work were described but not computed. For any even number $N$ of Dirac spinors $2^N$ linearly independent polynomials of degree 2 can be constructed and these polynomials were described.

For the case of Weyl particles, i.e., particles with definite chirality, the Lorentz invariant polynomials for three Dirac spinors reduce to either a multiple of the Coffman-Kundu-Wootters 3-tangle \cite{coffman} or alternatively are identically zero. If only two of the particles are Weyl particles some of the polynomials reduce to the $2\times 2\times 4$ tangle described in  Refs. \cite{verstraete,moor}. For the case of four Weyl particles the Lorentz invariant polynomials derived here reduce to linear combinations of the polynomials found by Luque and Thibon \cite{luque}. For the case of an even number $N$ of Weyl particles the degree 2 polynomials described here reduce to multiples of the $N$-tangle introduced by Wong and Christensen \cite{wong}.

In the case of zero particle momenta and zero four-potentials we considered the eigenspaces of the local Dirac Hamiltonians. These eigenspaces are often identified with non-relativistic free spin-$\frac{1}{2}$ particles and antiparticles. On these eigenspaces the polynomials constructed in this work all reduce to linear combinations of the previously known polynomials constructed for non-relativistic spin-$\frac{1}{2}$ particles or are identically zero. For three particles the polynomials either reduce to a multiple of the Coffman-Kundu-Wootters 3-tangle \cite{coffman} or are zero. For four particles the polynomials reduce to linear combinations of the polynomials found by Luque and Thibon \cite{luque} or are zero. For the case of an even number $N$ of particles the degree 2 polynomials constructed here reduce to the $N$-tangle of Wong and Christensen \cite{wong} or are zero. Since a system of Dirac particles can always be described in their respective rest frames, the previously known polynomials constructed for non-relativistic spin-$\frac{1}{2}$ particles can always be used for the case of free particles in eigenstates of the local Dirac Hamiltonians.

We considered evolutions generated by local Hamiltonians that act unitarily on the subspaces defined by fixed momenta, i.e., the subspaces spanned by the spinorial degrees of freedom. All polynomials derived using the method given here are invariant, up to a U(1) phase, under such local unitary evolution generated by zero-mass Dirac Hamiltonians and zero-mass Dirac Hamiltonians with additional terms that are second degree in Dirac gamma matrices such as a Semenoff mass term \cite{semenoff} or a Haldane mass term \cite{haldane}.
Some polynomials are by construction invariant, up to a U(1) phase, also under local unitary evolution generated by arbitrary-mass Dirac Hamiltonians and some are invariant, up to a U(1) phase, under local unitary evolution generated by zero-mass Dirac Hamiltonians with additional terms that are third degree in  Dirac gamma matrices such as a Pauli coupling or a Yukawa pseudo-scalar coupling.

Because of these properties the polynomials constructed by the method in this work are considered potential candidates for describing the entanglement of the spinor degrees of freedom in a system of multiple Dirac particles with either zero or arbitrary mass, or zero-mass with an additional coupling such as a Yukawa pseudo-scalar. Polynomials of this kind can be used to partially characterize qualitatively inequivalent forms of spinor entanglement as described in Ref. \cite{spinorent}. Such a characterization is not complete since only the disjoint orbit-closures of the invariance groups of the polynomials can be distinguished by the polynomials. Therefore there exist inequivalently spinor entangled states that cannot be distinguished by any set of polynomials of this kind. Moreover, as described in Ref. \cite{spinorent}
for any set of polynomials that are invariant under the spinor representation of the proper orthochronous Lorentz group there exist spinor entangled states for which all the polynomials are zero.

The general properties of the algebras of Lorentz invariant polynomials that are invariant, up to a U(1) phase, under evolution generated locally by arbitrary-mass Dirac Hamiltonians, by zero-mass Dirac Hamiltonians, or by zero-mass Dirac Hamiltonians with a Yukawa pseudo-scalar coupling were discussed. Any such algebra is generated by a finite number of polynomials but if such a set of generators can be constructed by the method described in this work is not clear.

In Ref. \cite{spinorent} it was descried how the absolute values of Lorentz invariant polynomials for two Dirac spinors can be extended to functions on the set of states that are incoherent mixtures, i.e., mixed states, through convex roof extensions \cite{lima,wakker,uhlmannn}. Convex roof extensions can be made also for the case of the polynomials constructed for multiple Dirac spinors in this work. Therefore such convex roof extensions can provide a partial characterization of the different types of multi-spinor entanglement of incoherent mixtures. These convex roof extensions are by definition identically zero for all incoherent mixtures of product states, i.e., for all separable states.

The polynomial Lorentz invariants in this work were constructed to describe multi-spinor entanglement for the case of Dirac particles with definite momenta but it is an open question whether similar constructions can be made for the case without definite momenta.
Another open question is the description of entanglement in a scenario that allows for communication between the labs holding the particles. 
In such a scenario one can consider multi-spinor entanglement properties that can be quantified by entanglement measures \cite{entmes}, i.e., multi-spinor entanglement properties that satisfy a condition of non-increase on average under any local operations assisted by classical communication \cite{vidal}.
More broadly one may consider if there are other conceptualizations of spinor entanglement properties that give a description that is complementary to the one given here.

\section*{Declaration of interests}
The authors declare that they have no known competing financial interests or personal relationships that could have appeared to influence the work reported in this paper.

\section*{Acknowledgment}
The author thanks the anonymous referees for constructive comments that improved the work.
The author also thanks the anonymous referees for constructive comments on an earlier version of this work that motivated several major improvements.
Support
from the European Research Council Consolidator Grant QITBOX (Grant Agreement No. 617337), the Spanish MINECO
(Project FOQUS FIS2013-46768-P, Severo Ochoa grant SEV-
2015-0522), Fundaci\'o Privada Cellex, the Generalitat de Catalunya
(SGR 875) and the John Templeton Foundation is acknowledged.

\appendix

\section{The use of real and rational numbers to quantify properties in experiments}\label{opp}
Any measurement that quantifies a property is an operational procedure that eventually concludes and results in an output number that is registered by an experimenter. In any experiment only finitely many such procedures can be performed and any memory storing the output numbers has finite capacity.

Any real number is by definition the limit of a Cauchy sequence of rational numbers \cite{cantor}. Such a sequence terminates only if the limit is a rational number itself. Thus if the limit is an irrational number the sequence does not terminate.
Consequently, the representation of an irrational number in a base-$n$ positional numeral system is a non-terminating and non-repeating sequence of digits for any $n$.

It follows that no memory with finite capacity can store an irrational number, and thus the output numbers from any given experiment is by necessity a finite set of rational numbers. In any finite set of rational numbers the elements are multiples of their greatest common divisor $q$.
Therefore, in any given experiment one cannot distinguish between the continuous spectrum of real numbers $\mathbb{R}$ and a discrete spectrum $nq$, where $n\in \mathbb{Z}$. Moreover, the experiment cannot distinguish a set $S\subset\mathbb{R}$ containing $mq$ for some $m\in \mathbb{Z}$ from $mq$ itself if $S$ does not contain any other multiple of $q$.

In particular, for any set of measured momentum vectors on the form $\bold{k}=(k_1,k_2,k_3)$ a rectangular box with periodic boundary conditions, and whose sides align with the directions corresponding to the components of the $\bold{k}$, can be found such that all the measured momenta are simultaneously allowed by the box dimensions. Alternatively, for any set of measured momenta each of the measured $\bold{k}$ can be matched to a Schwartz function with compact support in momentum space that contains $\bold{k}$ in such a way that no two of the Schwartz functions have overlapping support and no two of the measured momenta are in the support of the same Schwartz function.

\section{The bilinear forms ${\psi^T}C{\varphi}$ and ${\psi^T}C\gamma^5{\varphi}$ undergoing unitary spinor evolution generated by Dirac-like Hamiltonians}\label{hamm}
Here we show how to derive the properties of the bilinear forms ${\psi^T}C{\varphi}$ and ${\psi^T}C\gamma^5{\varphi}$ described in Section \ref{ham}.
We consider a subspace defined by a fixed particle momentum $\bold{k}$, i.e., a subspace spanned by the four basis elements $\phi_je^{i\bold{k}\cdot\bold{x}}$ with the same $\bold{k}$. Furthermore we assume that the evolution acts unitarily on such a subspace and is generated by a Hamiltonian operator $H$. For the subspace to be invariant under the evolution it is required that $(\phi_je^{i\bold{k}\cdot\bold{x}},H\phi_le^{i\bold{k'}\cdot\bold{x}})\propto\delta_{\bold{k},\bold{k'}}$. Therefore, to have such unitary action on the subspace we consider evolution generated by Hamiltonians that do not depend on the spatial coordinate $\bold{x}$.

We again consider the inner product on a subspace of this kind
\begin{eqnarray}
({\psi(t)},{\varphi(t)})_{\bold{k}}=\psi^{\dagger}(t)\varphi(t),
\end{eqnarray}
and assume that the time dependent Hamiltonian $H(s)$ is bounded and strongly continuous, i.e., for all ${\psi}$ and $s$ it holds that $\lim_{t\to s}||H(t){\psi}-H(s){\psi}||=0$ where $||\cdot||$ is the norm induced by the inner product. Then the evolution operator can be expressed as an ordered exponential as described in the following theorem.

\begin{theorem}
Let $t\in\mathbb{R}\to H(t)$ be a strongly continuous map into the bounded Hermitian operators on a Hilbert space $\mathcal{H}$. Then there exists an evolution operator $U(t,s)$ such that for all ${\psi}\in\mathcal{H}$ we have that
${\psi(t)}=U(t,s){\psi(s)}$ and $\partial_{t}U(t,s)=-iH(t)U(t,s)$. This evolution operator can be expressed as an ordered exponential
\begin{eqnarray*}
U(t,r)=&&\mathcal{T}_{\leftarrow}\{e^{-i\int_{r}^tH(s)ds}\}\nonumber\\
\equiv&&\sum_{n=0}^\infty(-i)^n\int_{r}^t\int_{r}^{s_n}\int_{r}^{s_{n-1}}\int_{r}^{s_{2}}H(s_n)\dots H(s_{1})ds_1\dots ds_{n-2} ds_{n-1} ds_{n},\nonumber\\
\end{eqnarray*}
and satisfies $U(t,t)=I$ and $U(r,s)U(s,t)=U(r,t)$.
\end{theorem}
\begin{proof}
See e.g. Ref. \cite{reed}.
\end{proof} 
For an evolution operator on an ordered exponential form we can consider its conjugate transpose, complex conjugate, and transpose in the given basis, 
\begin{eqnarray}
U(t,r)&&=\mathcal{T}_{\leftarrow}\{e^{-i\int_{r}^tH(s)ds}\},\nonumber\\
U(t,r)^\dagger &&=\mathcal{T}_{\rightarrow}\{e^{i\int_{r}^tH(s)ds}\},\nonumber\\
U(t,r)^*&&=\mathcal{T}_{\leftarrow}\{e^{i\int_{r}^t H^T(s)ds}\},\nonumber\\
U(t,r)^T&&=\mathcal{T}_{\rightarrow}\{e^{-i\int_{r}^t H^T(s)ds}\}.
\end{eqnarray}

Now assume that $X$ is a time independent matrix such that for all $s$ it holds that $XH(s)=-H(s)^TX$. Then it follows that $X\mathcal{T}_{\leftarrow}\{e^{-i\int_{0}^tH(s)ds}\}=\mathcal{T}_{\leftarrow}\{e^{i\int_{0}^tH(s)^Tds}\}X$, and the bilinear form $\psi^{T}X\varphi$ is invariant under the evolution generated by $H(s)$
\begin{eqnarray}
{\psi^T} U(t,0)^T  XU(t,0){\varphi}={\psi^T}\mathcal{T}_{\rightarrow}\{e^{-i\int_{0}^t H^T(s)ds}\}
\mathcal{T}_{\leftarrow}\{e^{i\int_{0}^t H^T(s)ds}\}X{\varphi}={\psi^T}X{\varphi}.
\end{eqnarray}

Next we consider the possibilities $X=C$ and $X=C\gamma^5$.
For products of different numbers of distinct gamma matrices we have for $C$ that

\begin{eqnarray}\label{c}
(\gamma^\mu)^TC&=&C\gamma^\mu\nonumber\\
(\gamma^\mu\gamma^\nu)^TC&=&-C\gamma^\mu\gamma^\nu\nonumber\\
(\gamma^\mu\gamma^\nu\gamma^\rho)^TC&=&-C\gamma^\mu\gamma^\nu\gamma^\rho\nonumber\\
(\gamma^\mu\gamma^\nu\gamma^\rho\gamma^\sigma)^TC&=&C\gamma^\mu\gamma^\nu\gamma^\rho\gamma^\sigma,
\end{eqnarray}
and for $C\gamma^5$ that
\begin{eqnarray}\label{c5}
(\gamma^\mu)^TC\gamma^5&=&-C\gamma^5\gamma^\mu\nonumber\\
(\gamma^\mu\gamma^\nu)^TC\gamma^5&=&-C\gamma^5\gamma^\mu\gamma^\nu\nonumber\\
(\gamma^\mu\gamma^\nu\gamma^\rho)^TC\gamma^5&=&C\gamma^5\gamma^\mu\gamma^\nu\gamma^\rho\nonumber\\
(\gamma^\mu\gamma^\nu\gamma^\rho\gamma^\sigma)^TC\gamma^5&=&C\gamma^5\gamma^\mu\gamma^\nu\gamma^\rho\gamma^\sigma.
\end{eqnarray}
From these equations we can see that a Hamiltonian term that is second or third degree in the gamma matrices, i.e., a Hamiltonian term of the form
\begin{eqnarray}
H^{2,3}(t)=\gamma^\mu\gamma^\nu\phi_{\mu\nu}(t)+\gamma^\mu\gamma^\nu\gamma^\rho\kappa_{\mu\nu\rho}(t),
\end{eqnarray}
satisfies $CH^{2,3}(t)=-(H^{2,3}(t))^TC$. Similarly, a Hamiltonian term that is first or second degree in the gamma matrices, i.e., a Hamiltonian term of the form
\begin{eqnarray}
H^{1,2}(t)=\gamma^\mu\eta_{\mu}(t)+\gamma^\mu\gamma^\nu\lambda_{\mu\nu}(t),
\end{eqnarray}
satisfies $C\gamma^5H^{1,2}(t)=-(H^{1,2}(t))^TC\gamma^5$.
Any Hamiltonian term $H^0(t)=f(t)I$ that is zeroth degree in the gamma matrices, i.e., proportional to the identity matrix, is clearly its own transpose and commutes with both $C\gamma^5$ and $C$.

The Dirac Hamiltonian contains a first degree term in the gamma matrices, the mass term $m\gamma^0$, and a second degree term, the generalized canonical momentum term $\sum_{\mu=1,2,3}\gamma^{0}\gamma^{\mu}(i\partial_{\mu}-qA_{\mu}(t))$.
Beyond this, it has a zeroth degree term, the coupling to the scalar potential $qA_0(t)I$. However, any zeroth degree term $f(t)I$ can be removed from the Hamiltonian by a change of variables. If we define a new spinor as
$\psi'= e^{-i\theta(t)}\psi$ the new Hamiltonian $H'$ defined by $i\partial_t\psi'(t)=H'\psi'(t)$ is 
$H'=H+\gamma^0\sum_\mu\gamma^\mu\partial_\mu \theta(t)$. This change of variables amounts to a change of local U(1) gauge (See e.g. Ref. \cite{griffiths}). For the choice $\theta(t)=-\int_{t_0}^tf(s)ds$ we see that since $f(t)-\partial_t \int_{t_0}^tf(s)ds=0$ the term proportional to the identity in $H'$ is identically zero. Except for the zeroth degree term the new Hamiltonian $H'$ in general contains terms with the same degrees in gamma matrices as $H$. It cannot acquire terms with degrees different from those of the terms in $H$. We can thus remove the zeroth degree term $qA_0(t)I$ from the Dirac Hamiltonian by choosing $\theta(t)=-\int_{t_0}^tqA_0(s)ds$. For this choice let $U'(t,t_0)$ be the evolution generated by $H'$. Then we can see that $\psi(t)=e^{-iq\int_{t_0}^tA_0(s)ds}\psi'(t)=e^{-iq\int_{t_0}^tA_0(s)ds}U'(t,t_0)\psi'(t_0)=e^{-iq\int_{t_0}^tA_0(s)ds}U'(t,t_0)\psi(t_0)$.

Therefore, as described in Ref. \cite{spinorent}, for an evolution $U_D(t,0)$ generated by Dirac Hamiltonians it holds for the bilinear form ${\psi^T}C\gamma^5{\varphi}$ that

\begin{eqnarray}\label{ghhh}
&&\psi^{T}U_D(t,0)^TC\gamma^5U_D(t,0)\varphi\nonumber\\
&&=e^{-2iq\int_{0}^tA_0(s)ds}{\psi^T}{U'}_D(t,0)^T C\gamma^5 U'_D(t,0){\varphi}\nonumber\\
&&=e^{-2iq\int_{0}^tA_0(s)ds}{\psi^T}C\gamma^5{\varphi}.
\end{eqnarray}

Similarly, for an evolution $U_W(t,0)$ generated by zero-mass Dirac Hamiltonians it holds for the bilinear form ${\psi^T}C{\varphi}$ that 

\begin{eqnarray}\label{ujjjk}
&&\psi^{T}U_W(t,0)^TU_W(t,0)\varphi \nonumber\\
&&=e^{-2iq\int_{0}^tA_0(s)ds}{\psi^T}{U'}_W(t,0)^TCU'_W(t,0){\varphi}\nonumber\\
&&=e^{-2iq\int_{0}^tA_0(s)ds}{\psi^T}C{\varphi}.
\end{eqnarray}

Thus the bilinear form $\psi^TC\varphi$ is invariant, up to a U(1) phase, under evolutions generated by any Hamiltonians on the form $H^{2,3}(t)+H^{0}(t)$ and the bilinear form $\psi^TC\gamma^5\varphi$ is invariant, up to a U(1) phase, under evolutions generated by any Hamiltonians on the form $H^{1,2}(t)+H^{0}(t)$.

\section{Graph representations of the tensor sandwich contractions that yield degree 4 polynomials for four Dirac spinors}\label{graphs}

Here we give the graph representations of the tensor sandwich contractions in Eq. (\ref{cvn}). The contractions $W_a$, $W_c$, and $W_f$ that are each invariant with respect to permutations of two disjoint pairs of laboratories are given in Fig. \ref{ris3}. The contractions  $W_b$, $W_d$, $W_e$, $W_g$, $W_h$, and $W_i$ that are each invariant with respect to permutations of a single pair of laboratories are given in Fig. \ref{ris2}. Finally, the contractions $W_j$, $W_k$, $W_l$, and $W_m$ that are each invariant with respect to permutations of a triple of laboratories are given in Fig. \ref{ris4}.

\begin{figure}
\centering
\begin{subfigure}{.3\linewidth}
  \centering
  \includegraphics[scale=1]{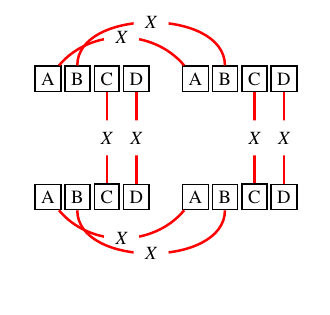}
  \caption{$W_a$}
\end{subfigure}%
\hspace{2cm}%
\begin{subfigure}{.3\linewidth}
  \centering
  \includegraphics[scale=1]{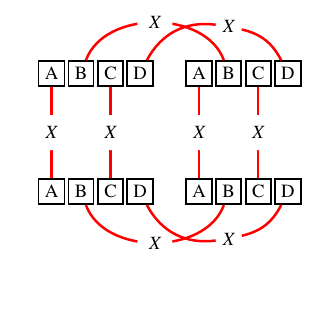}
   \caption{$W_c$}
\end{subfigure}\\
\begin{subfigure}{4cm}
  \centering
  \includegraphics[scale=1]{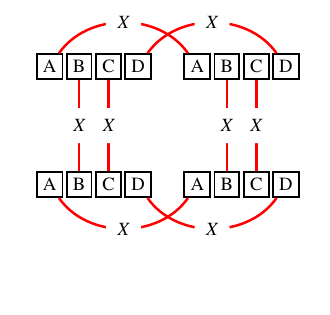}
   \caption{$W_f$}
\end{subfigure}
\caption{Graph representations of the tensor sandwich contractions $W_a$, $W_c$, and $W_f$. Each of the four copies of $\Psi^{ABCD}$ is represented by four boxes corresponding to the four tensor indices. The tensor sandwich contractions are represented by red lines connecting the contracted indices broken by $X$ representing the sandwiched tensor.}
\label{ris3}
\end{figure}

\begin{figure}
\centering
\begin{subfigure}{.3\linewidth}
  \centering
  \includegraphics[scale=1]{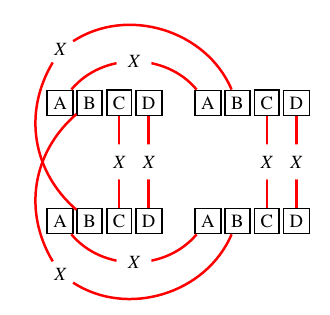}
  \caption{$W_b$}
\end{subfigure}%
\hspace{2cm}%
\begin{subfigure}{.3\linewidth}
  \centering
  \includegraphics[scale=1]{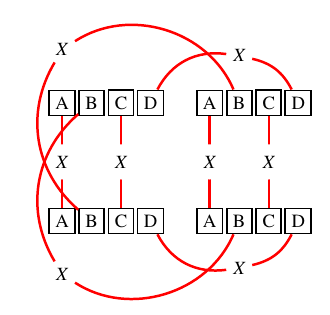}
   \caption{$W_d$}
\end{subfigure}\\
\begin{subfigure}{4cm}
  \centering
  \includegraphics[scale=1]{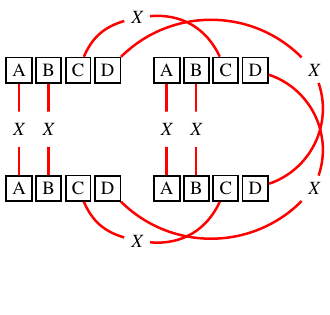}
   \caption{$W_e$}
\end{subfigure}
\hspace{2cm}%
\begin{subfigure}{4cm}
  \centering
  \includegraphics[scale=1]{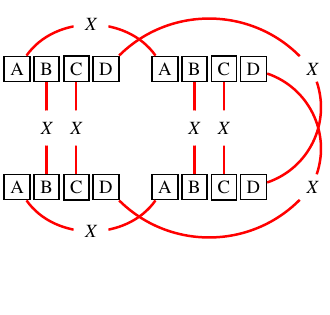}
   \caption{$W_g$}
\end{subfigure}\\
\begin{subfigure}{4cm}
  \centering
  \includegraphics[scale=1]{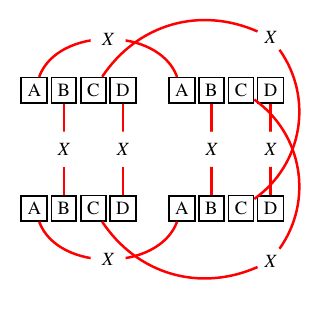}
   \caption{$W_h$}
\end{subfigure}
\hspace{2cm}%
\begin{subfigure}{4cm}
  \centering
  \includegraphics[scale=1]{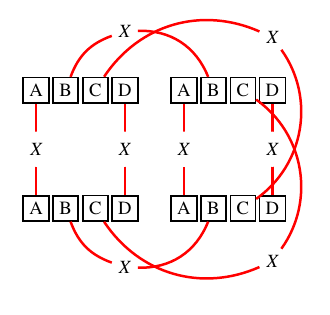}
   \caption{$W_i$}
\end{subfigure}
\caption{Graph representations of the tensor sandwich contractions $W_b$, $W_d$, $W_e$, $W_g$, $W_h$, and $W_i$. Each of the four copies of $\Psi^{ABCD}$ is represented by four boxes corresponding to the four tensor indices. The tensor sandwich contractions are represented by red lines connecting the contracted indices broken by $X$ representing the sandwiched tensor.}
\label{ris2}
\end{figure}

\begin{figure}
\centering
\begin{subfigure}{.3\linewidth}
  \centering
  \includegraphics[scale=1]{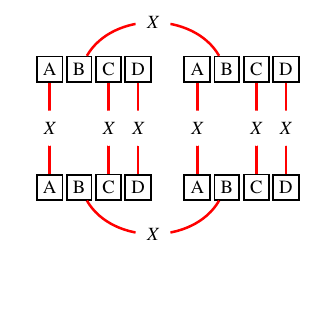}
  \caption{$W_j$}
\end{subfigure}%
\hspace{2cm}%
\begin{subfigure}{.3\linewidth}
  \centering
  \includegraphics[scale=1]{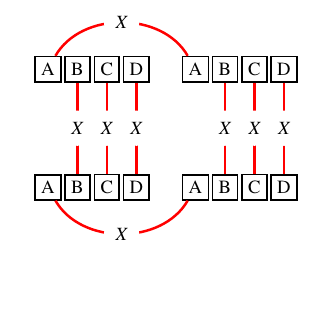}
   \caption{$W_k$}
\end{subfigure}\\
\begin{subfigure}{4cm}
  \centering
  \includegraphics[scale=1]{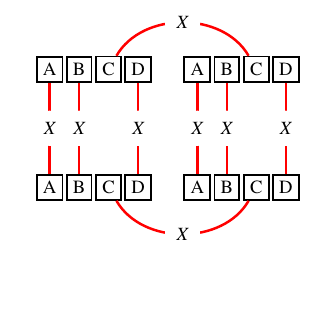}
   \caption{$W_l$}
\end{subfigure}
\hspace{2cm}%
\begin{subfigure}{.3\linewidth}
  \centering
  \includegraphics[scale=1]{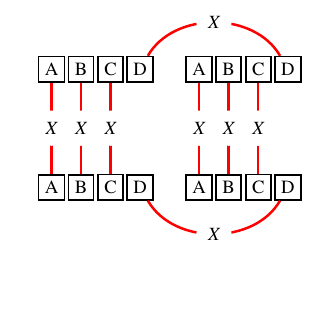}
   \caption{$W_m$}
\end{subfigure}
\caption{Graph representations of the tensor sandwich contractions $W_j$, $W_k$, $W_l$, and $W_m$.
 Each of the four copies of $\Psi^{ABCD}$ is represented by four boxes corresponding to the four tensor indices. The tensor sandwich contractions are represented by red lines connecting the contracted indices broken by $X$ representing the sandwiched tensor.}
\label{ris4}
\end{figure}

\clearpage
\section{Degree 4 polynomials for five spinors}\label{wides}

Here we give the 40 inequivalent ways to pair up the tensor indices of four copies of the $4\times 4\times 4\times 4\times 4$ tensor $\Psi^{ABCDE}$ that do not factorize into two degree 2 polynomials. In writing these sandwich contractions we leave out the summation sign with the understanding that repeated indices are summed over. We also suppress the superscript $ABCDE$ of $\Psi$. The 40 inequivalent ways to contract the indices are

\begin{eqnarray}\label{cvn2}
F_1=&&X_{gl} X_{hm}X_{in}X_{jo}X_{ku}\Psi_{ghijk}\Psi_{lmnop} X_{qv} X_{rw}X_{sx}X_{ty}X_{pz}\Psi_{qrstu}\Psi_{vwxyz},\nonumber\\
F_2=&&X_{gl} X_{hm}X_{in}X_{jt}X_{kp}\Psi_{ghijk}\Psi_{lmnop} X_{qv} X_{rw}X_{sx}X_{oy}X_{uz}\Psi_{qrstu}\Psi_{vwxyz},\nonumber\\
F_3=&&X_{gl} X_{hm}X_{is}X_{jo}X_{kp}\Psi_{ghijk}\Psi_{lmnop} X_{qv} X_{rw}X_{nx}X_{ty}X_{uz}\Psi_{qrstu}\Psi_{vwxyz},\nonumber\\
F_4=&&X_{gl} X_{hr}X_{in}X_{jo}X_{kp}\Psi_{ghijk}\Psi_{lmnop} X_{qv} X_{mw}X_{sx}X_{ty}X_{uz}\Psi_{qrstu}\Psi_{vwxyz},\nonumber\\
F_5=&&X_{gq} X_{hm}X_{in}X_{jo}X_{kp}\Psi_{ghijk}\Psi_{lmnop} X_{lv} X_{rw}X_{sx}X_{ty}X_{uz}\Psi_{qrstu}\Psi_{vwxyz},\nonumber\\
F_6=&&X_{gl} X_{hm}X_{in}X_{jt}X_{ku}\Psi_{ghijk}\Psi_{lmnop} X_{qv} X_{rw}X_{sx}X_{oy}X_{pz}\Psi_{qrstu}\Psi_{vwxyz},\nonumber\\
F_7=&&X_{gl} X_{hm}X_{in}X_{jt}X_{kz}\Psi_{ghijk}\Psi_{lmnop} X_{qv} X_{rw}X_{sx}X_{oy}X_{pu}\Psi_{qrstu}\Psi_{vwxyz},\nonumber\\
F_8=&&X_{gl} X_{hm}X_{is}X_{jt}X_{kp}\Psi_{ghijk}\Psi_{lmnop} X_{qv} X_{rw}X_{nx}X_{oy}X_{uz}\Psi_{qrstu}\Psi_{vwxyz},\nonumber\\
F_9=&&X_{gl} X_{hm}X_{is}X_{jy}X_{kp}\Psi_{ghijk}\Psi_{lmnop} X_{qv} X_{rw}X_{nx}X_{ot}X_{uz}\Psi_{qrstu}\Psi_{vwxyz},\nonumber\\
F_{10}=&&X_{gl} X_{hr}X_{is}X_{jo}X_{kp}\Psi_{ghijk}\Psi_{lmnop} X_{qv} X_{mw}X_{nx}X_{ty}X_{uz}\Psi_{qrstu}\Psi_{vwxyz},\nonumber\\
F_{11}=&&X_{gl} X_{hr}X_{ix}X_{jo}X_{kp}\Psi_{ghijk}\Psi_{lmnop} X_{qv} X_{mw}X_{ns}X_{ty}X_{uz}\Psi_{qrstu}\Psi_{vwxyz},\nonumber\\
F_{12}=&&X_{gq} X_{hr}X_{in}X_{jo}X_{kp}\Psi_{ghijk}\Psi_{lmnop} X_{lv} X_{mw}X_{sx}X_{ty}X_{uz}\Psi_{qrstu}\Psi_{vwxyz},\nonumber\\
F_{13}=&&X_{gq} X_{hw}X_{in}X_{jo}X_{kp}\Psi_{ghijk}\Psi_{lmnop} X_{lv} X_{mr}X_{sx}X_{ty}X_{uz}\Psi_{qrstu}\Psi_{vwxyz},\nonumber\\
F_{14}=&&X_{gq} X_{hm}X_{in}X_{jo}X_{ku}\Psi_{ghijk}\Psi_{lmnop} X_{lv} X_{rw}X_{sx}X_{ty}X_{pz}\Psi_{qrstu}\Psi_{vwxyz},\nonumber\\
F_{15}=&&X_{gq} X_{hm}X_{in}X_{jo}X_{kz}\Psi_{ghijk}\Psi_{lmnop} X_{lv} X_{rw}X_{sx}X_{ty}X_{pu}\Psi_{qrstu}\Psi_{vwxyz},\nonumber\\
F_{16}=&&X_{gq} X_{hm}X_{in}X_{jt}X_{kp}\Psi_{ghijk}\Psi_{lmnop} X_{lv} X_{rw}X_{sx}X_{oy}X_{uz}\Psi_{qrstu}\Psi_{vwxyz},\nonumber\\
F_{17}=&&X_{gq} X_{hm}X_{in}X_{jy}X_{kp}\Psi_{ghijk}\Psi_{lmnop} X_{lv} X_{rw}X_{sx}X_{ot}X_{uz}\Psi_{qrstu}\Psi_{vwxyz},\nonumber\\
F_{18}=&&X_{gl} X_{hm}X_{is}X_{jo}X_{ku}\Psi_{ghijk}\Psi_{lmnop} X_{qv} X_{rw}X_{nx}X_{ty}X_{pz}\Psi_{qrstu}\Psi_{vwxyz},\nonumber\\
F_{19}=&&X_{gl} X_{hm}X_{is}X_{jo}X_{kz}\Psi_{ghijk}\Psi_{lmnop} X_{qv} X_{rw}X_{nx}X_{ty}X_{pu}\Psi_{qrstu}\Psi_{vwxyz},\nonumber\\
F_{20}=&&X_{gl} X_{hr}X_{in}X_{jt}X_{kp}\Psi_{ghijk}\Psi_{lmnop} X_{qv} X_{mw}X_{sx}X_{oy}X_{uz}\Psi_{qrstu}\Psi_{vwxyz},\nonumber\\
F_{21}=&&X_{gl} X_{hr}X_{in}X_{jy}X_{kp}\Psi_{ghijk}\Psi_{lmnop} X_{qv} X_{mw}X_{sx}X_{ot}X_{uz}\Psi_{qrstu}\Psi_{vwxyz},\nonumber\\
F_{22}=&&X_{gq} X_{hm}X_{is}X_{jo}X_{kp}\Psi_{ghijk}\Psi_{lmnop} X_{lv} X_{rw}X_{ix}X_{ty}X_{uz}\Psi_{qrstu}\Psi_{vwxyz},\nonumber\\
F_{23}=&&X_{gq} X_{hm}X_{ix}X_{jo}X_{kp}\Psi_{ghijk}\Psi_{lmnop} X_{lv} X_{rw}X_{is}X_{ty}X_{uz}\Psi_{qrstu}\Psi_{vwxyz},\nonumber\\
F_{24}=&&X_{gl} X_{hr}X_{in}X_{jo}X_{ku}\Psi_{ghijk}\Psi_{lmnop} X_{qv} X_{mw}X_{sx}X_{ty}X_{pz}\Psi_{qrstu}\Psi_{vwxyz},\nonumber\\
F_{25}=&&X_{gl} X_{hr}X_{in}X_{jo}X_{kz}\Psi_{ghijk}\Psi_{lmnop} X_{qv} X_{mw}X_{sx}X_{ty}X_{pu}\Psi_{qrstu}\Psi_{vwxyz},\nonumber\\
F_{26}=&&X_{gq} X_{hm}X_{in}X_{jy}X_{kz}\Psi_{ghijk}\Psi_{lmnop} X_{lv} X_{rw}X_{sx}X_{to}X_{up}\Psi_{qrstu}\Psi_{vwxyz},\nonumber\\
F_{27}=&&X_{gq} X_{hm}X_{ix}X_{jo}X_{kz}\Psi_{ghijk}\Psi_{lmnop} X_{lv} X_{rw}X_{sn}X_{ty}X_{up}\Psi_{qrstu}\Psi_{vwxyz},\nonumber\\
F_{28}=&&X_{gq} X_{hw}X_{in}X_{jo}X_{kz}\Psi_{ghijk}\Psi_{lmnop} X_{lv} X_{rm}X_{sx}X_{ty}X_{up}\Psi_{qrstu}\Psi_{vwxyz},\nonumber\\
F_{29}=&&X_{gl} X_{hr}X_{in}X_{jy}X_{kz}\Psi_{ghijk}\Psi_{lmnop} X_{qv} X_{mw}X_{sx}X_{to}X_{up}\Psi_{qrstu}\Psi_{vwxyz},\nonumber\\
F_{30}=&&X_{gl} X_{hr}X_{ix}X_{jo}X_{kz}\Psi_{ghijk}\Psi_{lmnop} X_{qv} X_{mw}X_{sn}X_{ty}X_{up}\Psi_{qrstu}\Psi_{vwxyz},\nonumber\\
F_{31}=&&X_{gl} X_{hr}X_{ix}X_{jy}X_{kp}\Psi_{ghijk}\Psi_{lmnop} X_{qv} X_{mw}X_{sn}X_{to}X_{uz}\Psi_{qrstu}\Psi_{vwxyz},\nonumber\\
F_{32}=&&X_{gl} X_{hm}X_{is}X_{jy}X_{kz}\Psi_{ghijk}\Psi_{lmnop} X_{qv} X_{rw}X_{nx}X_{ty}X_{up}\Psi_{qrstu}\Psi_{vwxyz},\nonumber\\
F_{33}=&&X_{gl} X_{hw}X_{is}X_{jo}X_{kz}\Psi_{ghijk}\Psi_{lmnop} X_{qv} X_{rm}X_{nx}X_{ty}X_{up}\Psi_{qrstu}\Psi_{vwxyz},\nonumber\\
F_{34}=&&X_{gl} X_{hw}X_{is}X_{jy}X_{kp}\Psi_{ghijk}\Psi_{lmnop} X_{qv} X_{rm}X_{nx}X_{to}X_{uz}\Psi_{qrstu}\Psi_{vwxyz},\nonumber\\
F_{35}=&&X_{gl} X_{hm}X_{ix}X_{jt}X_{kz}\Psi_{ghijk}\Psi_{lmnop} X_{qv} X_{rw}X_{sn}X_{oy}X_{up}\Psi_{qrstu}\Psi_{vwxyz},\nonumber\\
F_{36}=&&X_{gl} X_{hw}X_{in}X_{jt}X_{kz}\Psi_{ghijk}\Psi_{lmnop} X_{qv} X_{rm}X_{sx}X_{oy}X_{up}\Psi_{qrstu}\Psi_{vwxyz},\nonumber\\
F_{37}=&&X_{gl} X_{hw}X_{ix}X_{jt}X_{kp}\Psi_{ghijk}\Psi_{lmnop} X_{qv} X_{rm}X_{sn}X_{oy}X_{uz}\Psi_{qrstu}\Psi_{vwxyz},\nonumber\\
F_{38}=&&X_{gl} X_{hm}X_{ix}X_{jy}X_{ku}\Psi_{ghijk}\Psi_{lmnop} X_{qv} X_{rw}X_{sn}X_{to}X_{pz}\Psi_{qrstu}\Psi_{vwxyz},\nonumber\\
F_{39}=&&X_{gl} X_{hw}X_{in}X_{jy}X_{ku}\Psi_{ghijk}\Psi_{lmnop} X_{qv} X_{rm}X_{sx}X_{to}X_{pz}\Psi_{qrstu}\Psi_{vwxyz},\nonumber\\
F_{40}=&&X_{gl} X_{hw}X_{ix}X_{jo}X_{ku}\Psi_{ghijk}\Psi_{lmnop} X_{qv} X_{rm}X_{sn}X_{ty}X_{pz}\Psi_{qrstu}\Psi_{vwxyz}.\nonumber\\
\end{eqnarray}
For each pairing of tensor indices we can consider $(2^{10}-2^5)/2+2^5=528$ ways to choose the $X$s as either $C$ or $C\gamma^5$. This gives a total of 21120 tensor sandwich contractions that yield degree 4 polynomials for five spinors.

\section{Selected explicit polynomials}\label{polllu}
Here we give a few examples of explicitly written out polynomials. For the case of three Dirac spinors the polynomials $I_{2a},I_{2b},I_{2c}$ are given as examples of polynomials of the types $I_a,I_b$ and $I_c$ that are invariant under $G^{C\gamma^5}$, up to a U(1) phase, in all labs and invariant under P in all labs. The polynomials $I_{3a},I_{3b},I_{3c}$ are given as examples of polynomials of the types $I_a,I_b$ and $I_c$ that are invariant under $G^{C}$, up to a U(1) phase, in all labs and invariant under P in all labs. The polynomial $I_{3d}$ is given as an example of the polynomials of type $I_d$. The polynomial $I_{23a}$ is given as an example of a polynomial that is P invariant only in two of the labs, Bob's and Charlie's. The polynomial $I_{35a}$ is given as an example of a polynomial that is P invariant only in one lab, Alice's. Finally the polynomial $I_{11a}$ is given as an example of a polynomial that is not P invariant in any lab.
For the case of four Dirac spinors the polynomials $H_a,H_b,H_c$, and $H_d$ are given as examples of the degree 2 polynomials, and the polynomials $T_l$ and $Y_l$ are given as examples of the degree 4 polynomials.

We first give the polynomials $I_{2a},I_{2b},I_{2c}$ that are invariant under $G^{C\gamma^5}$, up to a U(1) phase, in all labs and invariant under P in all labs. Explicitly written out they are
\small
\begin{align}
   I_{2a} = -2 (&\psi_{133}\psi_{200} - \psi_{132}\psi_{201} + \psi_{131}\psi_{202} - \psi_{130}\psi_{203} - \psi_{123}\psi_{210} +
       \psi_{122}\psi_{211} - \psi_{121}\psi_{212} + \psi_{120}\psi_{213}\nonumber\\ &+ \psi_{113}\psi_{220} - \psi_{112}\psi_{221} + 
      \psi_{111}\psi_{222} - \psi_{110}\psi_{223} - \psi_{103}\psi_{230} + \psi_{102}\psi_{231} - \psi_{101}\psi_{232} + 
      \psi_{100}\psi_{233})^2\nonumber\\- 
   2 (&\psi_{033}\psi_{300} - \psi_{032}\psi_{301} + \psi_{031}\psi_{302} - \psi_{030}\psi_{303} - \psi_{023}\psi_{310} + 
      \psi_{022}\psi_{311} - \psi_{021}\psi_{312} + \psi_{020}\psi_{313}\nonumber\\ &+ \psi_{013}\psi_{320} - \psi_{012}\psi_{321} + 
      \psi_{011}\psi_{322} - \psi_{010}\psi_{323} - \psi_{003}\psi_{330} + \psi_{002}\psi_{331} - \psi_{001}\psi_{332} + 
      \psi_{000}\psi_{333})^2\nonumber\\ + 
   8 (&\psi_{113}\psi_{120} - \psi_{112}\psi_{121} + \psi_{111}\psi_{122} - \psi_{110}\psi_{123} - \psi_{103}\psi_{130} + 
      \psi_{102}\psi_{131} - \psi_{101}\psi_{132} + \psi_{100}\psi_{133})\nonumber\\\times  (&\psi_{213}\psi_{220} - \psi_{212}\psi_{221} + 
      \psi_{211}\psi_{222} - \psi_{210}\psi_{223} - \psi_{203}\psi_{230} + \psi_{202}\psi_{231} - \psi_{201}\psi_{232} + 
      \psi_{200}\psi_{233})\nonumber\\ + 
   8 (&\psi_{013}\psi_{020} - \psi_{012}\psi_{021} + \psi_{011}\psi_{022} - \psi_{010}\psi_{023} - \psi_{003}\psi_{030} + 
      \psi_{002}\psi_{031} - \psi_{001}\psi_{032} + \psi_{000}\psi_{033}) \nonumber\\\times (&\psi_{313}\psi_{320} - \psi_{312}\psi_{321} + 
      \psi_{311}\psi_{322} - \psi_{310}\psi_{323} - \psi_{303}\psi_{330} + \psi_{302}\psi_{331} - \psi_{301}\psi_{332} + 
      \psi_{300}\psi_{333})\nonumber\\  + 
   4 (&\psi_{033}\psi_{200} - \psi_{032}\psi_{201} + \psi_{031}\psi_{202} - \psi_{030}\psi_{203} - \psi_{023}\psi_{210} + 
      \psi_{022}\psi_{211} - \psi_{021}\psi_{212} + \psi_{020}\psi_{213}\nonumber\\ &+ \psi_{013}\psi_{220} - \psi_{012}\psi_{221} + 
      \psi_{011}\psi_{222} - \psi_{010}\psi_{223} - \psi_{003}\psi_{230} + \psi_{002}\psi_{231} - \psi_{001}\psi_{232} + 
      \psi_{000}\psi_{233})\nonumber\\\times  (&\psi_{133}\psi_{300} - \psi_{132}\psi_{301} + \psi_{131}\psi_{302} - \psi_{130}\psi_{303} - 
      \psi_{123}\psi_{310} + \psi_{122}\psi_{311} - \psi_{121}\psi_{312} + \psi_{120}\psi_{313}\nonumber\\ &+ \psi_{113}\psi_{320} - 
      \psi_{112}\psi_{321} + \psi_{111}\psi_{322} - \psi_{110}\psi_{323} - \psi_{103}\psi_{330} + \psi_{102}\psi_{331} - 
      \psi_{101}\psi_{332} + \psi_{100}\psi_{333})\nonumber\\ - 
   4 (&\psi_{033}\psi_{100} - \psi_{032}\psi_{101} + \psi_{031}\psi_{102} - \psi_{030}\psi_{103} - \psi_{023}\psi_{110} + 
      \psi_{022}\psi_{111} - \psi_{021}\psi_{112} + \psi_{020}\psi_{113}\nonumber\\ &+ \psi_{013}\psi_{120} - \psi_{012}\psi_{121} + 
      \psi_{011}\psi_{122} - \psi_{010}\psi_{123} - \psi_{003}\psi_{130} + \psi_{002}\psi_{131} - \psi_{001}\psi_{132} + 
      \psi_{000}\psi_{133})\nonumber\\\times  (&\psi_{233}\psi_{300} - \psi_{232}\psi_{301} + \psi_{231}\psi_{302} - \psi_{230}\psi_{303} - 
      \psi_{223}\psi_{310} + \psi_{222}\psi_{311} - \psi_{221}\psi_{312} + \psi_{220}\psi_{313}\nonumber\\ &+ \psi_{213}\psi_{320} - 
      \psi_{212}\psi_{321} + \psi_{211}\psi_{322} - \psi_{210}\psi_{323} - \psi_{203}\psi_{330} + \psi_{202}\psi_{331} - 
      \psi_{201}\psi_{332} + \psi_{200}\psi_{333}),
\end{align}
\begin{align}
 I_{2b} = -2 (&\psi_{123}\psi_{210}- \psi_{122}\psi_{211}+ \psi_{121}\psi_{212}- \psi_{120}\psi_{213}+ \psi_{113}\psi_{220}-
      \psi_{112}\psi_{221}+ \psi_{111}\psi_{222}- \psi_{110}\psi_{223}\nonumber\\&- \psi_{023}\psi_{310}+ \psi_{022}\psi_{311}- 
     \psi_{021}\psi_{312}+ \psi_{020}\psi_{313}- \psi_{013}\psi_{320}+ \psi_{012}\psi_{321}- \psi_{011}\psi_{322}+ 
     \psi_{010}\psi_{323})^2\nonumber\\ - 
  2 (&\psi_{133}\psi_{200}- \psi_{132}\psi_{201}+ \psi_{131}\psi_{202}- \psi_{130}\psi_{203}+ \psi_{103}\psi_{230}- 
     \psi_{102}\psi_{231}+ \psi_{101}\psi_{232}- \psi_{100}\psi_{233}\nonumber\\&- \psi_{033}\psi_{300}+ \psi_{032}\psi_{301}- 
     \psi_{031}\psi_{302}+ \psi_{030}\psi_{303}- \psi_{003}\psi_{330}+ \psi_{002}\psi_{331}- \psi_{001}\psi_{332}+ 
     \psi_{000}\psi_{333})^2\nonumber\\+ 
  8 (&\psi_{113}\psi_{210}- \psi_{112}\psi_{211}+ \psi_{111}\psi_{212}- \psi_{110}\psi_{213}- \psi_{013}\psi_{310}+ 
     \psi_{012}\psi_{311}- \psi_{011}\psi_{312}+ \psi_{010}\psi_{313})\nonumber\\\times  (&\psi_{123}\psi_{220}- \psi_{122}\psi_{221}+ 
     \psi_{121}\psi_{222}- \psi_{120}\psi_{223}- \psi_{023}\psi_{320}+ \psi_{022}\psi_{321}- \psi_{021}\psi_{322}+ 
     \psi_{020}\psi_{323})\nonumber\\ + 
  8 (&\psi_{103}\psi_{200}- \psi_{102}\psi_{201}+ \psi_{101}\psi_{202}- \psi_{100}\psi_{203}- \psi_{003}\psi_{300}+ 
     \psi_{002}\psi_{301}- \psi_{001}\psi_{302}+ \psi_{000}\psi_{303})\nonumber\\ \times (&\psi_{133}\psi_{230}- \psi_{132}\psi_{231}+ 
     \psi_{131}\psi_{232}- \psi_{130}\psi_{233}- \psi_{033}\psi_{330}+ \psi_{032}\psi_{331}- \psi_{031}\psi_{332}+ 
     \psi_{030}\psi_{333})\nonumber\\+ 
  4 (&\psi_{123}\psi_{200}- \psi_{122}\psi_{201}+ \psi_{121}\psi_{202}- \psi_{120}\psi_{203}+ \psi_{103}\psi_{220}- 
     \psi_{102}\psi_{221}+ \psi_{101}\psi_{222}- \psi_{100}\psi_{223}\nonumber\\&- \psi_{023}\psi_{300}+ \psi_{022}\psi_{301}- 
     \psi_{021}\psi_{302}+ \psi_{020}\psi_{303}- \psi_{003}\psi_{320}+ \psi_{002}\psi_{321}- \psi_{001}\psi_{322}+ 
     \psi_{000}\psi_{323})\nonumber\\\times (&\psi_{133}\psi_{210}- \psi_{132}\psi_{211}+ \psi_{131}\psi_{212}- \psi_{130}\psi_{213}+ 
     \psi_{113}\psi_{230}- \psi_{112}\psi_{231}+ \psi_{111}\psi_{232}- \psi_{110}\psi_{233}\nonumber\\&- \psi_{033}\psi_{310}+ 
     \psi_{032}\psi_{311}- \psi_{031}\psi_{312}+ \psi_{030}\psi_{313}- \psi_{013}\psi_{330}+ \psi_{012}\psi_{331}- 
     \psi_{011}\psi_{332}+ \psi_{010}\psi_{333})\nonumber\\ - 
  4 (&\psi_{113}\psi_{200}- \psi_{112}\psi_{201}+ \psi_{111}\psi_{202}- \psi_{110}\psi_{203}+ \psi_{103}\psi_{210}- 
     \psi_{102}\psi_{211}+ \psi_{101}\psi_{212}- \psi_{100}\psi_{213}\nonumber\\&- \psi_{013}\psi_{300}+ \psi_{012}\psi_{301}- 
     \psi_{011}\psi_{302}+ \psi_{010}\psi_{303}- \psi_{003}\psi_{310}+ \psi_{002}\psi_{311}- \psi_{001}\psi_{312}+ 
     \psi_{000}\psi_{313}) \nonumber\\\times (&\psi_{133}\psi_{220}- \psi_{132}\psi_{221}+ \psi_{131}\psi_{222}- \psi_{130}\psi_{223}+ 
     \psi_{123}\psi_{230}- \psi_{122}\psi_{231}+ \psi_{121}\psi_{232}- \psi_{120}\psi_{233}\nonumber\\&- \psi_{033}\psi_{320}+ 
     \psi_{032}\psi_{321}- \psi_{031}\psi_{322}+ \psi_{030}\psi_{323}- \psi_{023}\psi_{330}+ \psi_{022}\psi_{331}- 
     \psi_{021}\psi_{332}+ \psi_{020}\psi_{333}),   
\end{align}
\normalsize
and
\small
\begin{align}
 I_{2c} = -2 (&\psi_{132}\psi_{201}+ \psi_{131}\psi_{202}- \psi_{122}\psi_{211}- \psi_{121}\psi_{212}+ \psi_{112}\psi_{221}+
      \psi_{111}\psi_{222}- \psi_{102}\psi_{231}- \psi_{101}\psi_{232}\nonumber\\&- \psi_{032}\psi_{301}- \psi_{031}\psi_{302}+ 
     \psi_{022}\psi_{311}+ \psi_{021}\psi_{312}- \psi_{012}\psi_{321}- \psi_{011}\psi_{322}+ \psi_{002}\psi_{331}+ 
     \psi_{001}\psi_{332})^2\nonumber\\ - 
  2 (&\psi_{133}\psi_{200}+ \psi_{130}\psi_{203}- \psi_{123}\psi_{210}- \psi_{120}\psi_{213}+ \psi_{113}\psi_{220}+ 
     \psi_{110}\psi_{223}- \psi_{103}\psi_{230}- \psi_{100}\psi_{233}\nonumber\\&- \psi_{033}\psi_{300}- \psi_{030}\psi_{303}+ 
     \psi_{023}\psi_{310}+ \psi_{020}\psi_{313}- \psi_{013}\psi_{320}- \psi_{010}\psi_{323}+ \psi_{003}\psi_{330}+ 
     \psi_{000}\psi_{333})^2\nonumber\\+ 
  8 (&\psi_{131}\psi_{201}- \psi_{121}\psi_{211}+ \psi_{111}\psi_{221}- \psi_{101}\psi_{231}- \psi_{031}\psi_{301}+ 
     \psi_{021}\psi_{311}- \psi_{011}\psi_{321}+ \psi_{001}\psi_{331})\nonumber\\\times  (&\psi_{132}\psi_{202}- \psi_{122}\psi_{212}+ 
     \psi_{112}\psi_{222}- \psi_{102}\psi_{232}- \psi_{032}\psi_{302}+ \psi_{022}\psi_{312}- \psi_{012}\psi_{322}+ 
     \psi_{002}\psi_{332})\nonumber\\ + 
  8 (&\psi_{130}\psi_{200}- \psi_{120}\psi_{210}+ \psi_{110}\psi_{220}- \psi_{100}\psi_{230}- \psi_{030}\psi_{300}+ 
     \psi_{020}\psi_{310}- \psi_{010}\psi_{320}+ \psi_{000}\psi_{330})\nonumber\\\times (& \psi_{133}\psi_{203}- \psi_{123}\psi_{213}+ 
     \psi_{113}\psi_{223}- \psi_{103}\psi_{233}- \psi_{033}\psi_{303}+ \psi_{023}\psi_{313}- \psi_{013}\psi_{323}+ 
     \psi_{003}\psi_{333})\nonumber\\+ 
  4 (&\psi_{132}\psi_{200}+ \psi_{130}\psi_{202}- \psi_{122}\psi_{210}- \psi_{120}\psi_{212}+ \psi_{112}\psi_{220}+ 
     \psi_{110}\psi_{222}- \psi_{102}\psi_{230}- \psi_{100}\psi_{232}\nonumber\\&- \psi_{032}\psi_{300}- \psi_{030}\psi_{302}+ 
     \psi_{022}\psi_{310}+ \psi_{020}\psi_{312}- \psi_{012}\psi_{320}- \psi_{010}\psi_{322}+ \psi_{002}\psi_{330}+ 
     \psi_{000}\psi_{332})\nonumber\\\times  (&\psi_{133}\psi_{201}+ \psi_{131}\psi_{203}- \psi_{123}\psi_{211}- \psi_{121}\psi_{213}+ 
     \psi_{113}\psi_{221}+ \psi_{111}\psi_{223}- \psi_{103}\psi_{231}- \psi_{101}\psi_{233}\nonumber\\&- \psi_{033}\psi_{301}- 
     \psi_{031}\psi_{303}+ \psi_{023}\psi_{311}+ \psi_{021}\psi_{313}- \psi_{013}\psi_{321}- \psi_{011}\psi_{323}+ 
     \psi_{003}\psi_{331}+ \psi_{001}\psi_{333})\nonumber\\ - 
  4 (&\psi_{131}\psi_{200}+ \psi_{130}\psi_{201}- \psi_{121}\psi_{210}- \psi_{120}\psi_{211}+ \psi_{111}\psi_{220}+ 
     \psi_{110}\psi_{221}- \psi_{101}\psi_{230}- \psi_{100}\psi_{231}\nonumber\\&- \psi_{031}\psi_{300}- \psi_{030}\psi_{301}+ 
     \psi_{021}\psi_{310}+ \psi_{020}\psi_{311}- \psi_{011}\psi_{320}- \psi_{010}\psi_{321}+ \psi_{001}\psi_{330}+ 
     \psi_{000}\psi_{331}) \nonumber\\\times (&\psi_{133}\psi_{202}+ \psi_{132}\psi_{203}- \psi_{123}\psi_{212}- \psi_{122}\psi_{213}+ 
     \psi_{113}\psi_{222}+ \psi_{112}\psi_{223}- \psi_{103}\psi_{232}- \psi_{102}\psi_{233}\nonumber\\&- \psi_{033}\psi_{302}- 
     \psi_{032}\psi_{303}+ \psi_{023}\psi_{312}+ \psi_{022}\psi_{313}- \psi_{013}\psi_{322}- \psi_{012}\psi_{323}+ 
     \psi_{003}\psi_{332}+ \psi_{002}\psi_{333}). 
\end{align}
\normalsize
The polynomials $I_{3a},I_{3b},I_{3c}$ are invariant under $G^C$, up to a U(1) phase, in all labs and invariant under P in all labs. Explicitly written out they are given by
\small
   \begin{align}
I_{3a} = -2 (&\psi_{011}\psi_{ 100} -\psi_{ 010 }\psi_{ 101} +\psi_{ 013}\psi_{ 102} -\psi_{ 012}\psi_{ 103 }-\psi_{ 001}   \psi_{ 110 }+
     \psi_{ 000}\psi_{ 111} -\psi_{ 003}\psi_{ 112} +\psi_{ 002}\psi_{ 113}\nonumber\\ &+\psi_{ 031}\psi_{ 120} -\psi_{ 030}\psi_{ 121} + 
   \psi_{  033}\psi_{ 122} -\psi_{ 032}\psi_{ 123 }-\psi_{ 021}\psi_{ 130} +\psi_{ 020}\psi_{ 131} -\psi_{ 023}\psi_{ 132} + 
    \psi_{ 022}\psi_{ 133})^2\nonumber\\
  - 
  2 (&\psi_{211} \psi_{300} - \psi_{210} \psi_{301} + \psi_{213} \psi_{302} - \psi_{212} \psi_{303} - \psi_{201} \psi_{310} + 
     \psi_{200} \psi_{311} - \psi_{203} \psi_{312} + \psi_{202} \psi_{313}\nonumber\nonumber\\ &+ \psi_{231} \psi_{320} - \psi_{230} \psi_{321} + 
     \psi_{233} \psi_{322} - \psi_{232 }\psi_{323} - \psi_{221} \psi_{330} + \psi_{220} \psi_{331} - \psi_{223} \psi_{332} + 
     \psi_{222} \psi_{333})^2\nonumber\\  
     + 
  8 (&\psi_{ 001}\psi_{ 010} -\psi_{ 000}\psi_{ 011} +\psi_{ 003}\psi_{ 012} -\psi_{ 002}\psi_{ 013} +\psi_{ 021}\psi_{ 030 }- 
    \psi_{ 020}\psi_{ 031} +\psi_{ 023}\psi_{ 032} -\psi_{ 022}\psi_{ 033})\nonumber\\\times   (&\psi_{ 101}\psi_{ 110} -\psi_{ 100} \psi_{ 111} + 
    \psi_{ 103}\psi_{ 112} -\psi_{ 102}\psi_{ 113} +\psi_{ 121}\psi_{ 130 }-\psi_{ 120 }\psi_{ 131} +\psi_{ 123}\psi_{ 132} - 
    \psi_{ 122}\psi_{ 133})\nonumber\\
      + 
  8 (&\psi_{201} \psi_{210} - \psi_{200} \psi_{211} + \psi_{203} \psi_{212} - \psi_{202} \psi_{213} + \psi_{221 }\psi_{230} - 
     \psi_{220} \psi_{231} + \psi_{223 }\psi_{232 }- \psi_{222} \psi_{233})\nonumber\\\times   (&\psi_{301} \psi_{310} - \psi_{300} \psi_{311} + 
     \psi_{303} \psi_{312} - \psi_{302} \psi_{313 }+ \psi_{321} \psi_{330} - \psi_{320} \psi_{331} + \psi_{323} \psi_{332} - 
     \psi_{322} \psi_{333})\nonumber\\
     - 
  4 (&\psi_{ 111}\psi_{ 200} -\psi_{ 110}\psi_{ 201} + \psi_{113} \psi_{202} - \psi_{112} \psi_{203} - \psi_{101} \psi_{210} + 
     \psi_{100} \psi_{211} -\psi_{ 103}\psi_{ 212} +\psi_{ 102}\psi_{ 213} \nonumber\\&+\psi_{ 131}\psi_{ 220} - \psi_{130} \psi_{221} + 
     \psi_{133} \psi_{222} -\psi_{ 132}\psi_{ 223} -\psi_{ 121}\psi_{ 230} + \psi_{120} \psi_{231 }- \psi_{123} \psi_{232} + 
     \psi_{122} \psi_{233})\nonumber\\\times  (&\psi_{011} \psi_{300} - \psi_{010} \psi_{301} + \psi_{013} \psi_{302} - \psi_{012} \psi_{303} - 
     \psi_{001} \psi_{310} + \psi_{000} \psi_{311} - \psi_{003} \psi_{312} + \psi_{002} \psi_{313}\nonumber\\ &+ \psi_{031} \psi_{320} - 
     \psi_{030} \psi_{321} + \psi_{033} \psi_{322} - \psi_{032} \psi_{323} - \psi_{021} \psi_{330} + \psi_{020} \psi_{331} - 
     \psi_{023} \psi_{332} + \psi_{022} \psi_{333})\nonumber\\ + 
  4 (&\psi_{011 }\psi_{200} - \psi_{010} \psi_{201} + \psi_{013} \psi_{202} - \psi_{012} \psi_{203} - \psi_{001} \psi_{210} + 
     \psi_{000} \psi_{211} - \psi_{003} \psi_{212} + \psi_{002} \psi_{213}\nonumber\\& + \psi_{031} \psi_{220} - \psi_{030} \psi_{221} + 
     \psi_{033} \psi_{222} - \psi_{032} \psi_{223} - \psi_{021} \psi_{230} + \psi_{020 }\psi_{231 }- \psi_{023} \psi_{232} + 
     \psi_{022} \psi_{233})\nonumber\\\times   (&\psi_{111} \psi_{300} - \psi_{110 }\psi_{301} + \psi_{113} \psi_{302} - \psi_{112} \psi_{303} - 
     \psi_{101} \psi_{310} + \psi_{100} \psi_{311} - \psi_{103} \psi_{312} + \psi_{102} \psi_{313}\nonumber\\ &+ \psi_{131} \psi_{320} - 
     \psi_{130} \psi_{321} + \psi_{133} \psi_{322} - \psi_{132} \psi_{323} - \psi_{121} \psi_{330 }+ \psi_{120 }\psi_{331} - 
     \psi_{123} \psi_{332} + \psi_{122} \psi_{333}),
\end{align}
\begin{align}
 I_{3b} = -2 (&\psi_{011}\psi_{100}- \psi_{010}\psi_{101}+ \psi_{013}\psi_{102}- \psi_{012}\psi_{103}+ \psi_{001}\psi_{110}-
       \psi_{000}\psi_{111}+ \psi_{003}\psi_{112}- \psi_{002}\psi_{113}\nonumber\\&+ \psi_{211}\psi_{300}- \psi_{210}\psi_{301}+ 
      \psi_{213}\psi_{302}- \psi_{212}\psi_{303}+ \psi_{201}\psi_{310}- \psi_{200}\psi_{311}+ \psi_{203}\psi_{312}- 
      \psi_{202}\psi_{313})^2\nonumber\\ - 
   2 (&\psi_{031}\psi_{120}- \psi_{030}\psi_{121}+ \psi_{033}\psi_{122}- \psi_{032}\psi_{123}+ \psi_{021}\psi_{130}- 
      \psi_{020}\psi_{131}+ \psi_{023}\psi_{132}- \psi_{022}\psi_{133}\nonumber\\ &+\psi_{231}\psi_{320}- \psi_{230}\psi_{321}+ 
      \psi_{233}\psi_{322}- \psi_{232}\psi_{323}+ \psi_{221}\psi_{330}- \psi_{220}\psi_{331}+ \psi_{223}\psi_{332}- 
      \psi_{222}\psi_{333})^2\nonumber\\ + 
   8 (&\psi_{021}\psi_{120}- \psi_{020}\psi_{121}+ \psi_{023}\psi_{122}- \psi_{022}\psi_{123}+ \psi_{221}\psi_{320}- 
      \psi_{220}\psi_{321}+ \psi_{223}\psi_{322}- \psi_{222}\psi_{323})\nonumber\\\times (&\psi_{031}\psi_{130}- \psi_{030}\psi_{131}+ 
      \psi_{033}\psi_{132}- \psi_{032}\psi_{133}+ \psi_{231}\psi_{330}- \psi_{230}\psi_{331}+ \psi_{233}\psi_{332}- 
      \psi_{232}\psi_{333})\nonumber\\ + 
   8 (&\psi_{001}\psi_{100}- \psi_{000}\psi_{101}+ \psi_{003}\psi_{102}- \psi_{002}\psi_{103}+ \psi_{201}\psi_{300}- 
      \psi_{200}\psi_{301}+ \psi_{203}\psi_{302}- \psi_{202}\psi_{303})\nonumber\\\times (&\psi_{011}\psi_{110}- \psi_{010}\psi_{111}+ 
      \psi_{013}\psi_{112}- \psi_{012}\psi_{113}+ \psi_{211}\psi_{310}- \psi_{210}\psi_{311}+ \psi_{213}\psi_{312}- 
      \psi_{212}\psi_{313})\nonumber\\ - 
   4 (&\psi_{021}\psi_{110}- \psi_{020}\psi_{111}+ \psi_{023}\psi_{112}- \psi_{022}\psi_{113}+ \psi_{011}\psi_{120}- 
      \psi_{010}\psi_{121}+ \psi_{013}\psi_{122}- \psi_{012}\psi_{123}\nonumber\\& + \psi_{221}\psi_{310}- \psi_{220}\psi_{311}+ 
      \psi_{223}\psi_{312}- \psi_{222}\psi_{313}+ \psi_{211}\psi_{320}- \psi_{210}\psi_{321}+ \psi_{213}\psi_{322}- 
      \psi_{212}\psi_{323})\nonumber\\\times  (&\psi_{031}\psi_{100}- \psi_{030}\psi_{101}+ \psi_{033}\psi_{102}- \psi_{032}\psi_{103}+ 
      \psi_{001}\psi_{130}- \psi_{000}\psi_{131}+ \psi_{003}\psi_{132}- \psi_{002}\psi_{133}\nonumber\\ &+ \psi_{231}\psi_{300}- 
      \psi_{230}\psi_{301}+ \psi_{233}\psi_{302}- \psi_{232}\psi_{303}+ \psi_{201}\psi_{330}- \psi_{200}\psi_{331}+ 
      \psi_{203}\psi_{332}- \psi_{202}\psi_{333})\nonumber\\ + 
   4 (&\psi_{021}\psi_{100}- \psi_{020}\psi_{101}+ \psi_{023}\psi_{102}- \psi_{022}\psi_{103}+ \psi_{001}\psi_{120}- 
      \psi_{000}\psi_{121}+ \psi_{003}\psi_{122}- \psi_{002}\psi_{123}\nonumber\\&+ \psi_{221}\psi_{300}- \psi_{220}\psi_{301}+ 
      \psi_{223}\psi_{302}- \psi_{222}\psi_{303}+ \psi_{201}\psi_{320}- \psi_{200}\psi_{321}+ \psi_{203}\psi_{322}- 
      \psi_{202}\psi_{323})\nonumber\\\times (&\psi_{031}\psi_{110}- \psi_{030}\psi_{111}+ \psi_{033}\psi_{112}- \psi_{032}\psi_{113}+ 
      \psi_{011}\psi_{130}- \psi_{010}\psi_{131}+ \psi_{013}\psi_{132}- \psi_{012}\psi_{133}\nonumber\\ &+ \psi_{231}\psi_{310}- 
      \psi_{230}\psi_{311}+ \psi_{233}\psi_{312}- \psi_{232}\psi_{313}+ \psi_{211}\psi_{330}- \psi_{210}\psi_{331}+ 
      \psi_{213}\psi_{332}- \psi_{212}\psi_{333}),     
\end{align}
\normalsize
and
\small
\begin{align}
  I_{3c} = -2 (&\psi_{011}\psi_{100} + \psi_{010}\psi_{101} - \psi_{001}\psi_{110} - \psi_{000}\psi_{111} + \psi_{031}\psi_{120} +
       \psi_{030}\psi_{121} - \psi_{021}\psi_{130} - \psi_{020}\psi_{131}\nonumber\\ &+ \psi_{211}\psi_{300} + \psi_{210}\psi_{301} - 
      \psi_{201}\psi_{310} - \psi_{200}\psi_{311} + \psi_{231}\psi_{320} + \psi_{230}\psi_{321} - \psi_{221}\psi_{330} - 
      \psi_{220}\psi_{331})^2\nonumber\\  - 
   2 (&\psi_{013}\psi_{102} + \psi_{012}\psi_{103} - \psi_{003}\psi_{112} - \psi_{002}\psi_{113} + \psi_{033}\psi_{122} + 
      \psi_{032}\psi_{123} - \psi_{023}\psi_{132} - \psi_{022}\psi_{133}\nonumber\\& + \psi_{213}\psi_{302} + \psi_{212}\psi_{303} - 
      \psi_{203}\psi_{312} - \psi_{202}\psi_{313} + \psi_{233}\psi_{322} + \psi_{232}\psi_{323} - \psi_{223}\psi_{332} - 
      \psi_{222}\psi_{333})^2\nonumber\\ + 
   8 (&\psi_{012}\psi_{102} - \psi_{002}\psi_{112} + \psi_{032}\psi_{122} - \psi_{022}\psi_{132} + \psi_{212}\psi_{302} - 
      \psi_{202}\psi_{312} + \psi_{232}\psi_{322} - \psi_{222}\psi_{332})\nonumber\\\times (&\psi_{013}\psi_{103} - \psi_{003}\psi_{113} + 
      \psi_{033}\psi_{123} - \psi_{023}\psi_{133} + \psi_{213}\psi_{303} - \psi_{203}\psi_{313} + \psi_{233}\psi_{323} - 
      \psi_{223}\psi_{333})\nonumber\\ + 
   8 (&\psi_{010}\psi_{100} - \psi_{000}\psi_{110} + \psi_{030}\psi_{120} - \psi_{020}\psi_{130} + \psi_{210}\psi_{300} - 
      \psi_{200}\psi_{310} + \psi_{230}\psi_{320} - \psi_{220}\psi_{330})\nonumber\\\times (&\psi_{011}\psi_{101} - \psi_{001}\psi_{111} + 
      \psi_{031}\psi_{121} - \psi_{021}\psi_{131} + \psi_{211}\psi_{301} - \psi_{201}\psi_{311} + \psi_{231}\psi_{321} - 
      \psi_{221}\psi_{331})\nonumber\\ - 
   4 (&\psi_{012}\psi_{101} + \psi_{011}\psi_{102} - \psi_{002}\psi_{111} - \psi_{001}\psi_{112} + \psi_{032}\psi_{121} + 
      \psi_{031}\psi_{122} - \psi_{022}\psi_{131} - \psi_{021}\psi_{132}\nonumber\\& + \psi_{212}\psi_{301} + \psi_{211}\psi_{302} - 
      \psi_{202}\psi_{311} - \psi_{201}\psi_{312} + \psi_{232}\psi_{321} + \psi_{231}\psi_{322} - \psi_{222}\psi_{331} - 
      \psi_{221}\psi_{332})\nonumber\\\times (&\psi_{013}\psi_{100} + \psi_{010}\psi_{103} - \psi_{003}\psi_{110} - \psi_{000}\psi_{113} + 
      \psi_{033}\psi_{120} + \psi_{030}\psi_{123} - \psi_{023}\psi_{130} - \psi_{020}\psi_{133}\nonumber\\ &+ \psi_{213}\psi_{300} + 
      \psi_{210}\psi_{303} - \psi_{203}\psi_{310} - \psi_{200}\psi_{313} + \psi_{233}\psi_{320} + \psi_{230}\psi_{323} - 
      \psi_{223}\psi_{330} - \psi_{220}\psi_{333})\nonumber\\ + 
   4 (&\psi_{012}\psi_{100} + \psi_{010}\psi_{102} - \psi_{002}\psi_{110} - \psi_{000}\psi_{112} + \psi_{032}\psi_{120} + 
      \psi_{030}\psi_{122} - \psi_{022}\psi_{130} - \psi_{020}\psi_{132}\nonumber\\ &+ \psi_{212}\psi_{300} + \psi_{210}\psi_{302} - 
      \psi_{202}\psi_{310} - \psi_{200}\psi_{312} + \psi_{232}\psi_{320} + \psi_{230}\psi_{322} - \psi_{222}\psi_{330} - 
      \psi_{220}\psi_{332})\nonumber\\\times (&\psi_{013}\psi_{101} + \psi_{011}\psi_{103} - \psi_{003}\psi_{111} - \psi_{001}\psi_{113} + 
      \psi_{033}\psi_{121} + \psi_{031}\psi_{123} - \psi_{023}\psi_{131} - \psi_{021}\psi_{133}\nonumber\\ &+ \psi_{213}\psi_{301} + 
      \psi_{211}\psi_{303} - \psi_{203}\psi_{311} - \psi_{201}\psi_{313} + \psi_{233}\psi_{321} + \psi_{231}\psi_{323} - 
      \psi_{223}\psi_{331} - \psi_{221}\psi_{333}).    
\end{align}
\normalsize
An example of a polynomial on the form $I_d$ is $I_{3d}$. It is given by $I_{3d}=4(Z_1+Z_2)$ where

\small
\begin{align}
Z_1=&(\psi_{001}\psi_{013}+ \psi_{021}\psi_{033}- \psi_{023}\psi_{031}- \psi_{003}\psi_{011}) (\psi_{102}\psi_{110}- 
      \psi_{100}\psi_{112}+ \psi_{122}\psi_{130}- \psi_{120}\psi_{132})\nonumber\\&
   + (\psi_{022}\psi_{031}+ \psi_{002}\psi_{011}- \psi_{001}\psi_{012}- \psi_{021}\psi_{032}) (\psi_{103}\psi_{110}- 
      \psi_{100}\psi_{113}+ \psi_{123}\psi_{130}- \psi_{120}\psi_{133})\nonumber\\&
   + (\psi_{003}\psi_{012}+ \psi_{020}\psi_{031}- \psi_{021}\psi_{030}- \psi_{002}\psi_{013}) (\psi_{101}\psi_{110}- 
      \psi_{100}\psi_{111}- \psi_{123}\psi_{132}+ \psi_{122}\psi_{133})\nonumber\\&
   + (\psi_{002}\psi_{033}+ \psi_{000}\psi_{031}- \psi_{003}\psi_{032}- \psi_{001}\psi_{030}) (\psi_{111}\psi_{120}- 
      \psi_{110}\psi_{121}+ \psi_{113}\psi_{122}- \psi_{112}\psi_{123})\nonumber\\&
   + (\psi_{000}\psi_{012}+ \psi_{020}\psi_{032}- \psi_{022}\psi_{030}- \psi_{002}\psi_{010}) (\psi_{103}\psi_{111}- 
      \psi_{101}\psi_{113}+ \psi_{123}\psi_{131}- \psi_{121}\psi_{133})\nonumber\\&
   + (\psi_{023}\psi_{030}+ \psi_{003}\psi_{010}- \psi_{020}\psi_{033}- \psi_{000}\psi_{013}) (\psi_{102}\psi_{111}- 
      \psi_{101}\psi_{112}+ \psi_{122}\psi_{131}- \psi_{121}\psi_{132})\nonumber\\&
   + (\psi_{021}\psi_{010}+ \psi_{023}\psi_{012}- \psi_{020}\psi_{011}- \psi_{022}\psi_{013}) (\psi_{101}\psi_{130}- 
      \psi_{100}\psi_{131}+ \psi_{103}\psi_{132}- \psi_{102}\psi_{133})\nonumber\\&
   + (\psi_{001}\psi_{010}+ \psi_{022}\psi_{033}- \psi_{000}\psi_{011}- \psi_{023}\psi_{032}) (\psi_{103}\psi_{112}- 
      \psi_{102}\psi_{113}- \psi_{121}\psi_{130}+ \psi_{120}\psi_{131})\nonumber\\&
   + (\psi_{011}\psi_{030}+ \psi_{013}\psi_{032}- \psi_{010}\psi_{031}- \psi_{012}\psi_{033}) (\psi_{101}\psi_{120}- 
      \psi_{100}\psi_{121}+ \psi_{103}\psi_{122}- \psi_{102}\psi_{123})\nonumber\\&
   + (\psi_{001}\psi_{020}+ \psi_{003}\psi_{022}- \psi_{000}\psi_{021}- \psi_{002}\psi_{023}) (\psi_{111}\psi_{130}- 
      \psi_{110}\psi_{131}+ \psi_{113}\psi_{132}- \psi_{112}\psi_{133})\nonumber\\&
   + (\psi_{121}\psi_{200}- \psi_{120}\psi_{201}+ \psi_{123}\psi_{202}- \psi_{122}\psi_{203}) (\psi_{330}\psi_{011}+ 
      \psi_{332}\psi_{013}- \psi_{333}\psi_{012}- \psi_{010}\psi_{331})\nonumber\\&
   + (\psi_{131}\psi_{200}- \psi_{130}\psi_{201}+ \psi_{133}\psi_{202}- \psi_{132}\psi_{203}) ( 
     \psi_{321}\psi_{010}+ \psi_{323}\psi_{012}- \psi_{320}\psi_{011}- \psi_{322}\psi_{013})\nonumber\\&
   + (\psi_{101}\psi_{200}- \psi_{100}\psi_{201}+ \psi_{103}\psi_{202}- \psi_{102}\psi_{203}) (\psi_{310}\psi_{011}+ 
      \psi_{312}\psi_{013}- \psi_{313}\psi_{012}- \psi_{311}\psi_{010})\nonumber\\&
   + (\psi_{111}\psi_{210}- \psi_{110}\psi_{211}+ \psi_{113}\psi_{212}- \psi_{112}\psi_{213}) (\psi_{302}\psi_{003}+ 
      \psi_{300}\psi_{001}- \psi_{303}\psi_{002}- \psi_{301}\psi_{000})\nonumber\\&
   + (\psi_{111}\psi_{201}- \psi_{101}\psi_{211}+ \psi_{131}\psi_{221}- \psi_{121}\psi_{231}) (\psi_{310}\psi_{000}+ 
      \psi_{330}\psi_{020}- \psi_{320}\psi_{030}- \psi_{300}\psi_{010})\nonumber\\&
   + (\psi_{111}\psi_{202}- \psi_{101}\psi_{212}+ \psi_{131}\psi_{222}- \psi_{121}\psi_{232}) (\psi_{300}\psi_{013}+ 
      \psi_{320}\psi_{033}- \psi_{330}\psi_{023}- \psi_{310}\psi_{003})\nonumber\\&
   + (\psi_{110}\psi_{202}- \psi_{100}\psi_{212}+ \psi_{130}\psi_{222}- \psi_{120}\psi_{232}) (\psi_{311}\psi_{003}+ 
      \psi_{331}\psi_{023}- \psi_{321}\psi_{033}- \psi_{301}\psi_{013})\nonumber\\&
   + (\psi_{121}\psi_{210}- \psi_{120}\psi_{211}+ \psi_{123}\psi_{212}- \psi_{122}\psi_{213}) (\psi_{331}\psi_{000}+ 
      \psi_{333}\psi_{002}- \psi_{330}\psi_{001}- \psi_{332}\psi_{003})\nonumber\\&
   + (\psi_{131}\psi_{210}- \psi_{130}\psi_{211}+ \psi_{133}\psi_{212}- \psi_{132}\psi_{213}) (\psi_{322}\psi_{003}+ 
      \psi_{320}\psi_{001}- \psi_{323}\psi_{002}- \psi_{321}\psi_{000})\nonumber\\&
   + (\psi_{112}\psi_{202}- \psi_{102}\psi_{212}+ \psi_{132}\psi_{222}- \psi_{122}\psi_{232}) (\psi_{333}\psi_{023}+ 
      \psi_{313}\psi_{003}- \psi_{303}\psi_{013}- \psi_{323}\psi_{033})\nonumber\\&
   + (\psi_{110}\psi_{200}- \psi_{100}\psi_{210}+ \psi_{130}\psi_{220}- \psi_{120}\psi_{230}) (\psi_{311}\psi_{001}+ 
      \psi_{331}\psi_{021}- \psi_{321}\psi_{031}- \psi_{301}\psi_{011})\nonumber\\&
   + (\psi_{113}\psi_{201}- \psi_{103}\psi_{211}+ \psi_{133}\psi_{221}- \psi_{123}\psi_{231}) (\psi_{332}\psi_{020}+ 
      \psi_{312}\psi_{000}- \psi_{302}\psi_{010}- \psi_{322}\psi_{030})\nonumber\\&
   + (\psi_{112}\psi_{201}- \psi_{102}\psi_{211}+ \psi_{132}\psi_{221}- \psi_{122}\psi_{231}) (\psi_{323}\psi_{030}+ 
      \psi_{303}\psi_{010}- \psi_{333}\psi_{020}- \psi_{313}\psi_{000})\nonumber\\&
   + (\psi_{113}\psi_{200}- \psi_{103}\psi_{210}+ \psi_{133}\psi_{220}- \psi_{123}\psi_{230}) (\psi_{322}\psi_{031}+ 
      \psi_{302}\psi_{011}- \psi_{312}\psi_{001}- \psi_{332}\psi_{021})\nonumber\\&
   + (\psi_{112}\psi_{200}- \psi_{102}\psi_{210}+ \psi_{132}\psi_{220}- \psi_{122}\psi_{230}) (\psi_{333}\psi_{021}+ 
      \psi_{313}\psi_{001}- \psi_{303}\psi_{011}- \psi_{323}\psi_{031})\nonumber\\&
   + (\psi_{111}\psi_{230}- \psi_{110}\psi_{231}+ \psi_{113}\psi_{232}- \psi_{112}\psi_{233}) ( 
     \psi_{300}\psi_{021}+ \psi_{302}\psi_{023}- \psi_{303}\psi_{022}- \psi_{301}\psi_{020})\nonumber\\&
   + (\psi_{101}\psi_{230}- \psi_{100}\psi_{231}+ \psi_{103}\psi_{232}- \psi_{102}\psi_{233}) (\psi_{313}\psi_{022}+ 
      \psi_{311}\psi_{020}- \psi_{312}\psi_{023}- \psi_{310}\psi_{021})\nonumber\\&
   + (\psi_{131}\psi_{230}- \psi_{130}\psi_{231}+ \psi_{133}\psi_{232}- \psi_{132}\psi_{233}) (\psi_{320}\psi_{021}+ 
      \psi_{322}\psi_{023}- \psi_{321}\psi_{020}- \psi_{323}\psi_{022})\nonumber\\&
   + (\psi_{111}\psi_{220}- \psi_{110}\psi_{221}+ \psi_{113}\psi_{222}- \psi_{112}\psi_{223}) ( 
     \psi_{303}\psi_{032}+ \psi_{301}\psi_{030}- \psi_{300}\psi_{031}- \psi_{302}\psi_{033})\nonumber\\&
    + (\psi_{101}\psi_{220}- \psi_{100}\psi_{221}+ \psi_{103}\psi_{222}- \psi_{102}\psi_{223}) (\psi_{310}\psi_{031}+ 
      \psi_{312}\psi_{033}- \psi_{313}\psi_{032}- \psi_{311}\psi_{030})\nonumber\\&
   + (\psi_{111}\psi_{203}- \psi_{101}\psi_{213}+ \psi_{131}\psi_{223}- \psi_{121}\psi_{233}) (\psi_{310}\psi_{002}+ 
      \psi_{330}\psi_{022}- \psi_{320}\psi_{032}- \psi_{300}\psi_{012})\nonumber\\&
   + (\psi_{110}\psi_{203}- \psi_{100}\psi_{213}+ \psi_{130}\psi_{223}- \psi_{120}\psi_{233}) (\psi_{321}\psi_{032}+ 
      \psi_{301}\psi_{012}- \psi_{311}\psi_{002}- \psi_{331}\psi_{022})\nonumber\\&
   + (\psi_{113}\psi_{203}- \psi_{103}\psi_{213}+ \psi_{133}\psi_{223}- \psi_{123}\psi_{233}) (\psi_{312}\psi_{002}+ 
      \psi_{332}\psi_{022}- \psi_{322}\psi_{032}- \psi_{302}\psi_{012})\nonumber\\&
   + (\psi_{121}\psi_{220}- \psi_{120}\psi_{221}+ \psi_{123}\psi_{222}- \psi_{122}\psi_{223}) (\psi_{332}\psi_{033}+ 
      \psi_{330}\psi_{031}- \psi_{333}\psi_{032}- \psi_{331}\psi_{030})\nonumber\\&
   + (\psi_{203}\psi_{212}+ \psi_{220}\psi_{231}- \psi_{221}\psi_{230}- \psi_{202}\psi_{213}) (\psi_{301}\psi_{310}- 
      \psi_{300}\psi_{311}- \psi_{323}\psi_{332}+ \psi_{322}\psi_{333})\nonumber\\&
   + (\psi_{201}\psi_{213}+ \psi_{221}\psi_{233}- \psi_{203}\psi_{211}- \psi_{223}\psi_{231}) (\psi_{302}\psi_{310}- 
      \psi_{300}\psi_{312}+ \psi_{322}\psi_{330}- \psi_{320}\psi_{332})\nonumber\\&
   + (\psi_{211}\psi_{230}- \psi_{210}\psi_{231}+ \psi_{213}\psi_{232}- \psi_{212}\psi_{233}) (\psi_{301}\psi_{320}- 
      \psi_{300}\psi_{321}+ \psi_{303}\psi_{322}- \psi_{302}\psi_{323})\nonumber\\&
   + (\psi_{223}\psi_{230}+ \psi_{203}\psi_{210}- \psi_{220}\psi_{233}- \psi_{200}\psi_{213}) (\psi_{302}\psi_{311}- 
      \psi_{301}\psi_{312}+ \psi_{322}\psi_{331}- \psi_{321}\psi_{332})\nonumber\\&
   + (\psi_{222}\psi_{231}+ \psi_{202}\psi_{211}- \psi_{221}\psi_{232}- \psi_{201}\psi_{212}) (\psi_{303}\psi_{310}- 
      \psi_{300}\psi_{313}+ \psi_{323}\psi_{330}- \psi_{320}\psi_{333})\nonumber\\&
   + (\psi_{200}\psi_{212}+ \psi_{220}\psi_{232}- \psi_{202}\psi_{210}- \psi_{222}\psi_{230}) (\psi_{303}\psi_{311}- 
      \psi_{301}\psi_{313}+ \psi_{323}\psi_{331}- \psi_{321}\psi_{333})\nonumber\\&
   + (\psi_{222}\psi_{233}+ \psi_{201}\psi_{210}- \psi_{223}\psi_{232}- \psi_{200}\psi_{211}) (\psi_{303}\psi_{312}- 
      \psi_{302}\psi_{313}- \psi_{321}\psi_{330}+ \psi_{320}\psi_{331})\nonumber\\&
   + (\psi_{200}\psi_{231}+ \psi_{202}\psi_{233}- \psi_{201}\psi_{230}- \psi_{203}\psi_{232}) (\psi_{311}\psi_{320}- 
      \psi_{310}\psi_{321}+ \psi_{313}\psi_{322}- \psi_{312}\psi_{323})\nonumber\\&
   + (\psi_{223}\psi_{212}+ \psi_{221}\psi_{210}- \psi_{220}\psi_{211}- \psi_{222}\psi_{213}) ( 
     \psi_{301}\psi_{330}- \psi_{300}\psi_{331}+ \psi_{303}\psi_{332}- \psi_{302}\psi_{333})\nonumber\\&
   + (\psi_{203}\psi_{222}+ \psi_{220}\psi_{201}- \psi_{221}\psi_{200}- \psi_{202}\psi_{223}) ( 
     \psi_{311}\psi_{330}- \psi_{310}\psi_{331}+ \psi_{313}\psi_{332}- \psi_{312}\psi_{333}),
     \end{align}
     \normalsize
     and
     \small
     \begin{align} 
   Z_2=&\psi_{032}\psi_{323}(\psi_{112}\psi_{203}- \psi_{102}\psi_{213}+ \psi_{131}\psi_{220}- \psi_{130}\psi_{221}+ 
      \psi_{133}\psi_{222}- \psi_{122}\psi_{233})\nonumber\\&
   + \psi_{022}\psi_{333}( \psi_{102}\psi_{213}-\psi_{112}\psi_{203}- \psi_{132}\psi_{223}+ \psi_{121}\psi_{230}- 
      \psi_{120}\psi_{231}+ \psi_{123}\psi_{232})\nonumber\\&
   + \psi_{010} \psi_{301}(\psi_{111}\psi_{200}+ \psi_{113}\psi_{202}- \psi_{112}\psi_{203}- \psi_{100}\psi_{211}+ 
      \psi_{130}\psi_{221}- \psi_{120}\psi_{231})\nonumber\\&
   + \psi_{011}\psi_{300}(\psi_{110}\psi_{201}- \psi_{113}\psi_{202}+ \psi_{112}\psi_{203}- \psi_{101}\psi_{210}+ 
      \psi_{131}\psi_{220}- \psi_{121}\psi_{230})\nonumber\\&
   + \psi_{012}\psi_{303}(\psi_{111}\psi_{200}- \psi_{110}\psi_{201}+ \psi_{113}\psi_{202}- \psi_{102}\psi_{213}+ 
      \psi_{132}\psi_{223}- \psi_{122}\psi_{233})\nonumber\\&
   + \psi_{013}\psi_{302}(\psi_{110}\psi_{201}-\psi_{111}\psi_{200}+ \psi_{112}\psi_{203}- \psi_{103}\psi_{212}+ 
      \psi_{133}\psi_{222}- \psi_{123}\psi_{232})\nonumber\\&
   + \psi_{001}\psi_{310}(\psi_{100}\psi_{211}-\psi_{111}\psi_{200}- \psi_{103}\psi_{212}+ \psi_{102}\psi_{213}- 
      \psi_{131}\psi_{220}+ \psi_{121}\psi_{230})\nonumber\\&
   + \psi_{031}\psi_{320}(\psi_{111}\psi_{200}- \psi_{101}\psi_{210}+ \psi_{130}\psi_{221}- \psi_{133}\psi_{222}+ 
      \psi_{132}\psi_{223}- \psi_{121}\psi_{230})\nonumber\\&
   + \psi_{000}\psi_{311}(\psi_{101}\psi_{210}-\psi_{110}\psi_{201}+ \psi_{103}\psi_{212}- \psi_{102}\psi_{213}- 
      \psi_{130}\psi_{221}+ \psi_{120}\psi_{231}) \nonumber\\&
   + \psi_{021}\psi_{330}( \psi_{101}\psi_{210}-\psi_{111}\psi_{200}- \psi_{131}\psi_{220}+ \psi_{120}\psi_{231}- 
      \psi_{123}\psi_{232}+ \psi_{122}\psi_{233})\nonumber\\&
   + \psi_{020}\psi_{331}( \psi_{100}\psi_{211}-\psi_{110}\psi_{201}- \psi_{130}\psi_{221}+ \psi_{121}\psi_{230}+ 
      \psi_{123}\psi_{232}- \psi_{122}\psi_{233})\nonumber\\&
   + \psi_{030}\psi_{321}(\psi_{110}\psi_{201}- \psi_{100}\psi_{211}+ \psi_{131}\psi_{220}+ \psi_{133}\psi_{222}- 
      \psi_{132}\psi_{223}- \psi_{120}\psi_{231})\nonumber\\&
   + \psi_{003}\psi_{312}(\psi_{100}\psi_{211}-\psi_{113}\psi_{202}- \psi_{101}\psi_{210}+ \psi_{102}\psi_{213}- 
      \psi_{133}\psi_{222}+ \psi_{123}\psi_{232})\nonumber\\&
   + \psi_{033}\psi_{322}(\psi_{113}\psi_{202}- \psi_{103}\psi_{212}- \psi_{131}\psi_{220}+ \psi_{130}\psi_{221}+ 
      \psi_{132}\psi_{223}- \psi_{123}\psi_{232}) \nonumber\\&
   + \psi_{002}\psi_{313}(\psi_{101}\psi_{210}-\psi_{112}\psi_{203}- \psi_{100}\psi_{211}+ \psi_{103}\psi_{212}- 
      \psi_{132}\psi_{223}+ \psi_{122}\psi_{233})\nonumber\\& 
   + \psi_{023}\psi_{332}(\psi_{103}\psi_{212}-\psi_{113}\psi_{202}- \psi_{133}\psi_{222}- \psi_{121}\psi_{230}+ 
      \psi_{120}\psi_{231}+ \psi_{122}\psi_{233}). 
\end{align}
\normalsize
An example of a polynomial that is invariant under P only in Bob's and Charlie's labs is $I_{23a}$. Written out it is 
\small   
\begin{align}
   I_{23a} = -4 (&\psi_{101}\psi_{110} - \psi_{100}\psi_{111}+ \psi_{103}\psi_{112}- \psi_{102}\psi_{113}+ 
      \psi_{121}\psi_{130} - \psi_{120}\psi_{131}+ \psi_{123}\psi_{132}- \psi_{122}\psi_{133}\nonumber\\&+\psi_{301}\psi_{310} - \psi_{300}\psi_{311}+ \psi_{303}\psi_{312}- \psi_{302}\psi_{313}+ 
      \psi_{321}\psi_{330} - \psi_{320}\psi_{331}+ \psi_{323}\psi_{332}- \psi_{322}\psi_{333})\nonumber\\\times  (&\psi_{011}\psi_{200} - 
      \psi_{010}\psi_{201}+ \psi_{013}\psi_{202}- \psi_{012}\psi_{203}- \psi_{001}\psi_{210}+ \psi_{000}\psi_{211}- 
      \psi_{003}\psi_{212}+ \psi_{002}\psi_{213}\nonumber\\&+ \psi_{031}\psi_{220} - \psi_{030}\psi_{221}+ \psi_{033}\psi_{222}- 
      \psi_{032}\psi_{223}- \psi_{021}\psi_{230}+ \psi_{020}\psi_{231}- \psi_{023}\psi_{232}+ \psi_{022}\psi_{233})\nonumber\\ - 
   4 (&\psi_{001}\psi_{010} - \psi_{000}\psi_{011}+ \psi_{003}\psi_{012}- \psi_{002}\psi_{013}+ \psi_{021}\psi_{030} - 
      \psi_{020}\psi_{031}+ \psi_{023}\psi_{032}- \psi_{022}\psi_{033}\nonumber\\&+\psi_{201}\psi_{210} - \psi_{200}\psi_{211}+ \psi_{203}\psi_{212}- \psi_{202}\psi_{213}+ \psi_{221}\psi_{230} - 
      \psi_{220}\psi_{231}+ \psi_{223}\psi_{232}- \psi_{222}\psi_{233})\nonumber\\\times  (&\psi_{111}\psi_{300} - \psi_{110}\psi_{301}+ 
      \psi_{113}\psi_{302}- \psi_{112}\psi_{303}- \psi_{101}\psi_{310}+ \psi_{100}\psi_{311}- \psi_{103}\psi_{312}+ 
      \psi_{102}\psi_{313}\nonumber\\&+ \psi_{131}\psi_{320} - \psi_{130}\psi_{321}+ \psi_{133}\psi_{322}- \psi_{132}\psi_{323}- 
      \psi_{121}\psi_{330}+ \psi_{120}\psi_{331}- \psi_{123}\psi_{332}+ \psi_{122}\psi_{333})\nonumber\\ - 
   2 (&\psi_{011}\psi_{100} - \psi_{010}\psi_{101}+ \psi_{013}\psi_{102}- \psi_{012}\psi_{103}- \psi_{001}\psi_{110}+ 
      \psi_{000}\psi_{111}- \psi_{003}\psi_{112}+ \psi_{002}\psi_{113}\nonumber\\&+ \psi_{031}\psi_{120} - \psi_{030}\psi_{121}+ 
      \psi_{033}\psi_{122}- \psi_{032}\psi_{123}- \psi_{021}\psi_{130}+ \psi_{020}\psi_{131}- \psi_{023}\psi_{132}+ 
      \psi_{022}\psi_{133}\nonumber\\&+\psi_{211}\psi_{300} - \psi_{210}\psi_{301}+ \psi_{213}\psi_{302}- \psi_{212}\psi_{303}- 
      \psi_{201}\psi_{310}+ \psi_{200}\psi_{311}- \psi_{203}\psi_{312}+ \psi_{202}\psi_{313}\nonumber\\&+ \psi_{231}\psi_{320} - 
      \psi_{230}\psi_{321}+ \psi_{233}\psi_{322}- \psi_{232}\psi_{323}- \psi_{221}\psi_{330}+ \psi_{220}\psi_{331}- 
      \psi_{223}\psi_{332}+ \psi_{222}\psi_{333}) \nonumber\\ \times (&\psi_{111}\psi_{200} - \psi_{110}\psi_{201}+ \psi_{113}\psi_{202}- \psi_{112}\psi_{203}- 
      \psi_{101}\psi_{210}+ \psi_{100}\psi_{211}- \psi_{103}\psi_{212}+ \psi_{102}\psi_{213}\nonumber\\&+ \psi_{131}\psi_{220} - 
      \psi_{130}\psi_{221}+ \psi_{133}\psi_{222}- \psi_{132}\psi_{223}- \psi_{121}\psi_{230}+ \psi_{120}\psi_{231}- 
      \psi_{123}\psi_{232}+ \psi_{122}\psi_{233}\nonumber\\&+\psi_{011}\psi_{300} - \psi_{010}\psi_{301}+ \psi_{013}\psi_{302}- \psi_{012}\psi_{303}- 
      \psi_{001}\psi_{310}+ \psi_{000}\psi_{311}- \psi_{003}\psi_{312}+ \psi_{002}\psi_{313}\nonumber\\&+ \psi_{031}\psi_{320} - 
      \psi_{030}\psi_{321}+ \psi_{033}\psi_{322}- \psi_{032}\psi_{323}- \psi_{021}\psi_{330}+ \psi_{020}\psi_{331}- 
      \psi_{023}\psi_{332}+ \psi_{022}\psi_{333}). 
      \end{align}
      \normalsize
  An example of a polynomial that is invariant under P only in Alice's lab is $I_{35a}$. Written out it is 
  \footnotesize    
      \begin{align}
         I_{35a} =  - 
   4 (&\psi_{001}\psi_{010}- \psi_{000}\psi_{011}+ \psi_{003}\psi_{012}- \psi_{002}\psi_{013}+ \psi_{021}\psi_{030}- 
      \psi_{020}\psi_{031}+ \psi_{023}\psi_{032}- \psi_{022}\psi_{033})\nonumber\\\times(&\psi_{113}\psi_{120}- \psi_{112}\psi_{121}+ 
      \psi_{111}\psi_{122}- \psi_{110}\psi_{123}- \psi_{103}\psi_{130}+ \psi_{102}\psi_{131}- \psi_{101}\psi_{132}+ 
      \psi_{100}\psi_{133})\nonumber\\ - 
   4 (&\psi_{013}\psi_{020}- \psi_{012}\psi_{021}+ \psi_{011}\psi_{022}- \psi_{010}\psi_{023}- \psi_{003}\psi_{030}+ 
      \psi_{002}\psi_{031}- \psi_{001}\psi_{032}+ \psi_{000}\psi_{033})\nonumber\\\times (&\psi_{101}\psi_{110}- \psi_{100}\psi_{111}+ 
      \psi_{103}\psi_{112}- \psi_{102}\psi_{113}+ \psi_{121}\psi_{130}- \psi_{120}\psi_{131}+ \psi_{123}\psi_{132}- 
      \psi_{122}\psi_{133})\nonumber\\ - 
   4 (&\psi_{201}\psi_{210}- \psi_{200}\psi_{211}+ \psi_{203}\psi_{212}- \psi_{202}\psi_{213}+ \psi_{221}\psi_{230}- 
      \psi_{220}\psi_{231}+ \psi_{223}\psi_{232}- \psi_{222}\psi_{233})\nonumber\\\times (&\psi_{313}\psi_{320}- \psi_{312}\psi_{321}+ 
      \psi_{311}\psi_{322}- \psi_{310}\psi_{323}- \psi_{303}\psi_{330}+ \psi_{302}\psi_{331}- \psi_{301}\psi_{332}+ 
      \psi_{300}\psi_{333})\nonumber\\ - 
   4 (&\psi_{213}\psi_{220}- \psi_{212}\psi_{221}+ \psi_{211}\psi_{222}- \psi_{210}\psi_{223}- \psi_{203}\psi_{230}+ 
      \psi_{202}\psi_{231}- \psi_{201}\psi_{232}+ \psi_{200}\psi_{233})\nonumber\\\times (&\psi_{301}\psi_{310}- \psi_{300}\psi_{311}+ 
      \psi_{303}\psi_{312}- \psi_{302}\psi_{313}+ \psi_{321}\psi_{330}- \psi_{320}\psi_{331}+ \psi_{323}\psi_{332}- 
      \psi_{322}\psi_{333})\nonumber\\-2 (&\psi_{033}\psi_{100}- \psi_{032}\psi_{101}+ \psi_{031}\psi_{102}- \psi_{030}\psi_{103}- 
      \psi_{023}\psi_{110}+ \psi_{022}\psi_{111}- \psi_{021}\psi_{112}+ \psi_{020}\psi_{113}\nonumber\\&+ \psi_{013}\psi_{120}- 
      \psi_{012}\psi_{121}+ \psi_{011}\psi_{122}- \psi_{010}\psi_{123}- \psi_{003}\psi_{130}+ \psi_{002}\psi_{131}- 
      \psi_{001}\psi_{132}+ \psi_{000}\psi_{133})\nonumber\\\times (&\psi_{011}\psi_{100}- \psi_{010}\psi_{101}+ \psi_{013}\psi_{102}- 
      \psi_{012}\psi_{103}- \psi_{001}\psi_{110}+ \psi_{000}\psi_{111}- \psi_{003}\psi_{112}+ \psi_{002}\psi_{113}\nonumber\\&+ 
      \psi_{031}\psi_{120}- \psi_{030}\psi_{121}+ \psi_{033}\psi_{122}- \psi_{032}\psi_{123}- \psi_{021}\psi_{130}+ 
      \psi_{020}\psi_{131}- \psi_{023}\psi_{132}+ \psi_{022}\psi_{133})\nonumber\\ - 
   2 (&\psi_{111}\psi_{200}- \psi_{110}\psi_{201}+ \psi_{113}\psi_{202}- \psi_{112}\psi_{203}- \psi_{101}\psi_{210}+ 
      \psi_{100}\psi_{211}- \psi_{103}\psi_{212}+ \psi_{102}\psi_{213}\nonumber\\&+ \psi_{131}\psi_{220}- \psi_{130}\psi_{221}+ 
      \psi_{133}\psi_{222}- \psi_{132}\psi_{223}- \psi_{121}\psi_{230}+ \psi_{120}\psi_{231}- \psi_{123}\psi_{232}+ 
      \psi_{122}\psi_{233})\nonumber\\\times (&\psi_{033}\psi_{300}- \psi_{032}\psi_{301}+ \psi_{031}\psi_{302}- \psi_{030}\psi_{303}- 
      \psi_{023}\psi_{310}+ \psi_{022}\psi_{311}- \psi_{021}\psi_{312}+ \psi_{020}\psi_{313}\nonumber\\&+ \psi_{013}\psi_{320}- 
      \psi_{012}\psi_{321}+ \psi_{011}\psi_{322}- \psi_{010}\psi_{323}- \psi_{003}\psi_{330}+ \psi_{002}\psi_{331}- 
      \psi_{001}\psi_{332}+ \psi_{000}\psi_{333})\nonumber\\ - 
   2 (&\psi_{133}\psi_{200}- \psi_{132}\psi_{201}+ \psi_{131}\psi_{202}- \psi_{130}\psi_{203}- \psi_{123}\psi_{210}+ 
      \psi_{122}\psi_{211}- \psi_{121}\psi_{212}+ \psi_{120}\psi_{213}\nonumber\\&+ \psi_{113}\psi_{220}- \psi_{112}\psi_{221}+ 
      \psi_{111}\psi_{222}- \psi_{110}\psi_{223}- \psi_{103}\psi_{230}+ \psi_{102}\psi_{231}- \psi_{101}\psi_{232}+ 
      \psi_{100}\psi_{233})\nonumber\\\times (&\psi_{011}\psi_{300}- \psi_{010}\psi_{301}+ \psi_{013}\psi_{302}- \psi_{012}\psi_{303}- 
      \psi_{001}\psi_{310}+ \psi_{000}\psi_{311}- \psi_{003}\psi_{312}+ \psi_{002}\psi_{313}\nonumber\\&+ \psi_{031}\psi_{320}- 
      \psi_{030}\psi_{321}+ \psi_{033}\psi_{322}- \psi_{032}\psi_{323}- \psi_{021}\psi_{330}+ \psi_{020}\psi_{331}- 
      \psi_{023}\psi_{332}+ \psi_{022}\psi_{333})\nonumber\\ + 
   2 (&\psi_{011}\psi_{200}- \psi_{010}\psi_{201}+ \psi_{013}\psi_{202}- \psi_{012}\psi_{203}- \psi_{001}\psi_{210}+ 
      \psi_{000}\psi_{211}- \psi_{003}\psi_{212}+ \psi_{002}\psi_{213}\nonumber\\&+ \psi_{031}\psi_{220}- \psi_{030}\psi_{221}+ 
      \psi_{033}\psi_{222}- \psi_{032}\psi_{223}- \psi_{021}\psi_{230}+ \psi_{020}\psi_{231}- \psi_{023}\psi_{232}+ 
      \psi_{022}\psi_{233})\nonumber\\\times (&\psi_{133}\psi_{300}- \psi_{132}\psi_{301}+ \psi_{131}\psi_{302}- \psi_{130}\psi_{303}- 
      \psi_{123}\psi_{310}+ \psi_{122}\psi_{311}- \psi_{121}\psi_{312}+ \psi_{120}\psi_{313}\nonumber\\&+ \psi_{113}\psi_{320}- 
      \psi_{112}\psi_{321}+ \psi_{111}\psi_{322}- \psi_{110}\psi_{323}- \psi_{103}\psi_{330}+ \psi_{102}\psi_{331}- 
      \psi_{101}\psi_{332}+ \psi_{100}\psi_{333})\nonumber\\ + 
   2 (&\psi_{033}\psi_{200}- \psi_{032}\psi_{201}+ \psi_{031}\psi_{202}- \psi_{030}\psi_{203}- \psi_{023}\psi_{210}+ 
      \psi_{022}\psi_{211}- \psi_{021}\psi_{212}+ \psi_{020}\psi_{213}\nonumber\\&+ \psi_{013}\psi_{220}- \psi_{012}\psi_{221}+ 
      \psi_{011}\psi_{222}- \psi_{010}\psi_{223}- \psi_{003}\psi_{230}+ \psi_{002}\psi_{231}- \psi_{001}\psi_{232}+ 
      \psi_{000}\psi_{233})\nonumber\\\times (&\psi_{111}\psi_{300}- \psi_{110}\psi_{301}+ \psi_{113}\psi_{302}- \psi_{112}\psi_{303}- 
      \psi_{101}\psi_{310}+ \psi_{100}\psi_{311}- \psi_{103}\psi_{312}+ \psi_{102}\psi_{313}\nonumber\\&+ \psi_{131}\psi_{320}- 
      \psi_{130}\psi_{321}+ \psi_{133}\psi_{322}- \psi_{132}\psi_{323}- \psi_{121}\psi_{330}+ \psi_{120}\psi_{331}- 
      \psi_{123}\psi_{332}+ \psi_{122}\psi_{333})\nonumber\\ - 
   2 (&\psi_{233}\psi_{300}- \psi_{232}\psi_{301}+ \psi_{231}\psi_{302}- \psi_{230}\psi_{303}- \psi_{223}\psi_{310}+ 
      \psi_{222}\psi_{311}- \psi_{221}\psi_{312}+ \psi_{220}\psi_{313}\nonumber\\&+ \psi_{213}\psi_{320}- \psi_{212}\psi_{321}+ 
      \psi_{211}\psi_{322}- \psi_{210}\psi_{323}- \psi_{203}\psi_{330}+ \psi_{202}\psi_{331}- \psi_{201}\psi_{332}+ 
      \psi_{200}\psi_{333})\nonumber\\\times (&\psi_{211}\psi_{300}- \psi_{210}\psi_{301}+ \psi_{213}\psi_{302}- \psi_{212}\psi_{303}- 
      \psi_{201}\psi_{310}+ \psi_{200}\psi_{311}- \psi_{203}\psi_{312}+ \psi_{202}\psi_{313}\nonumber\\&+ \psi_{231}\psi_{320}- 
      \psi_{230}\psi_{321}+ \psi_{233}\psi_{322}- \psi_{232}\psi_{323}- \psi_{221}\psi_{330}+ \psi_{220}\psi_{331}- 
      \psi_{223}\psi_{332}+ \psi_{222}\psi_{333}).
      \end{align}
      \normalsize
An example of a polynomial that is not invariant under P in any lab is $I_{11a}$. Written out it is  
\footnotesize
 \begin{align}
      I_{11a} = -2 (&\psi_{101}\psi_{110}- \psi_{100}\psi_{111}+ \psi_{103}\psi_{112}- \psi_{102}\psi_{113}+ 
       \psi_{121}\psi_{130}- \psi_{120}\psi_{131}+ \psi_{123}\psi_{132}- \psi_{122}\psi_{133}\nonumber\\& + \psi_{301}\psi_{310}- 
       \psi_{300}\psi_{311}+ \psi_{303}\psi_{312}- \psi_{302}\psi_{313}+ \psi_{321}\psi_{330}- \psi_{320}\psi_{331}+ 
       \psi_{323}\psi_{332}- \psi_{322}\psi_{333})\nonumber\\\times  (&\psi_{033}\psi_{200}- \psi_{032}\psi_{201}+ \psi_{031}\psi_{202}- 
     \psi_{030}\psi_{203}- \psi_{023}\psi_{210}+ \psi_{022}\psi_{211}- \psi_{021}\psi_{212}+ \psi_{020}\psi_{213}\nonumber\\&+ 
     \psi_{013}\psi_{220}- \psi_{012}\psi_{221}+ \psi_{011}\psi_{222}- \psi_{010}\psi_{223}- \psi_{003}\psi_{230}+ 
     \psi_{002}\psi_{231}- \psi_{001}\psi_{232}+ \psi_{000}\psi_{233})\nonumber\\ + 
  2 (&\psi_{113}\psi_{120}- \psi_{112}\psi_{121}+ \psi_{111}\psi_{122}- \psi_{110}\psi_{123}- \psi_{103}\psi_{130}+ 
       \psi_{102}\psi_{131}- \psi_{101}\psi_{132}+ \psi_{100}\psi_{133}\nonumber\\& +  \psi_{313}\psi_{320}- \psi_{312}\psi_{321}+ 
       \psi_{311}\psi_{322}- \psi_{310}\psi_{323}- \psi_{303}\psi_{330}+ \psi_{302}\psi_{331}- \psi_{301}\psi_{332}+ 
       \psi_{300}\psi_{333})\nonumber\\\times  (&\psi_{011}\psi_{200}- \psi_{010}\psi_{201}+ \psi_{013}\psi_{202}- \psi_{012}\psi_{203}- 
     \psi_{001}\psi_{210}+ \psi_{000}\psi_{211}- \psi_{003}\psi_{212}+ \psi_{002}\psi_{213}\nonumber\\&+ \psi_{031}\psi_{220}- 
     \psi_{030}\psi_{221}+ \psi_{033}\psi_{222}- \psi_{032}\psi_{223}- \psi_{021}\psi_{230}+ \psi_{020}\psi_{231}- 
     \psi_{023}\psi_{232}+ 
     \psi_{022}\psi_{233})\nonumber\\ - 
  2 (& \psi_{001}\psi_{010}- \psi_{000}\psi_{011}+ \psi_{003}\psi_{012}- \psi_{002}\psi_{013}+ \psi_{021}\psi_{030}- 
        \psi_{020}\psi_{031}+ \psi_{023}\psi_{032}- \psi_{022}\psi_{033}\nonumber\\& \psi_{201}\psi_{210}- \psi_{200}\psi_{211}+ 
       \psi_{203}\psi_{212}- \psi_{202}\psi_{213}+ \psi_{221}\psi_{230}- \psi_{220}\psi_{231}+ \psi_{223}\psi_{232}- 
       \psi_{222}\psi_{233} )\nonumber\\\times  (&\psi_{133}\psi_{300}- \psi_{132}\psi_{301}+ \psi_{131}\psi_{302}- \psi_{130}\psi_{303}- 
     \psi_{123}\psi_{310}+ \psi_{122}\psi_{311}- \psi_{121}\psi_{312}+ \psi_{120}\psi_{313}\nonumber\\&+ \psi_{113}\psi_{320}- 
     \psi_{112}\psi_{321}+ \psi_{111}\psi_{322}- \psi_{110}\psi_{323}- \psi_{103}\psi_{330}+ \psi_{102}\psi_{331}- 
     \psi_{101}\psi_{332}+ \psi_{100}\psi_{333})\nonumber\\ + 
  2 (& \psi_{013}\psi_{020}- \psi_{012}\psi_{021}+ \psi_{011}\psi_{022}- \psi_{010}\psi_{023}- \psi_{003}\psi_{030}+ 
       \psi_{002}\psi_{031}- \psi_{001}\psi_{032}+ \psi_{000}\psi_{033}\nonumber\\& + \psi_{213}\psi_{220}- \psi_{212}\psi_{221}+ 
       \psi_{211}\psi_{222}- \psi_{210}\psi_{223}- \psi_{203}\psi_{230}+ \psi_{202}\psi_{231}- \psi_{201}\psi_{232}+ 
       \psi_{200}\psi_{233} )\nonumber\\\times  (&\psi_{111}\psi_{300}- \psi_{110}\psi_{301}+ \psi_{113}\psi_{302}- \psi_{112}\psi_{303}- 
     \psi_{101}\psi_{310}+ \psi_{100}\psi_{311}- \psi_{103}\psi_{312}+ \psi_{102}\psi_{313}\nonumber\\&+ \psi_{131}\psi_{320}- 
     \psi_{130}\psi_{321}+ \psi_{133}\psi_{322}- \psi_{132}\psi_{323}- \psi_{121}\psi_{330}+ \psi_{120}\psi_{331}- 
     \psi_{123}\psi_{332}+ \psi_{122}\psi_{333})\nonumber\\ - (&\psi_{011}\psi_{100}- \psi_{010}\psi_{101}+ \psi_{013}\psi_{102}- \psi_{012}\psi_{103}- 
       \psi_{001}\psi_{110}+ \psi_{000}\psi_{111}- \psi_{003}\psi_{112}+ \psi_{002}\psi_{113}\nonumber\\&+ \psi_{031}\psi_{120}- 
       \psi_{030}\psi_{121}+ \psi_{033}\psi_{122}- \psi_{032}\psi_{123}- \psi_{021}\psi_{130}+ \psi_{020}\psi_{131}- 
       \psi_{023}\psi_{132}+ \psi_{022}\psi_{133}\nonumber\\& + \psi_{211}\psi_{300}- \psi_{210}\psi_{301}+ \psi_{213}\psi_{302}- 
       \psi_{212}\psi_{303}- \psi_{201}\psi_{310}+ \psi_{200}\psi_{311}- \psi_{203}\psi_{312}+ \psi_{202}\psi_{313}\nonumber\\&+ 
       \psi_{231}\psi_{320}- \psi_{230}\psi_{321}+ \psi_{233}\psi_{322}- \psi_{232}\psi_{323}- \psi_{221}\psi_{330}+ 
       \psi_{220}\psi_{331}- \psi_{223}\psi_{332}+ \psi_{222}\psi_{333}) \nonumber\\\times (&\psi_{133}\psi_{200}- \psi_{132}\psi_{201}+ 
       \psi_{131}\psi_{202}- \psi_{130}\psi_{203}- \psi_{123}\psi_{210}+ \psi_{122}\psi_{211}- \psi_{121}\psi_{212}+ 
       \psi_{120}\psi_{213}\nonumber\\&+ \psi_{113}\psi_{220}- \psi_{112}\psi_{221}+ \psi_{111}\psi_{222}- \psi_{110}\psi_{223}- 
       \psi_{103}\psi_{230}+ \psi_{102}\psi_{231}- \psi_{101}\psi_{232}+ \psi_{100}\psi_{233}\nonumber\\& + \psi_{033}\psi_{300}- 
       \psi_{032}\psi_{301}+ \psi_{031}\psi_{302}- \psi_{030}\psi_{303}- \psi_{023}\psi_{310}+ \psi_{022}\psi_{311}- 
       \psi_{021}\psi_{312}+ \psi_{020}\psi_{313}\nonumber\\&+ \psi_{013}\psi_{320}- \psi_{012}\psi_{321}+ \psi_{011}\psi_{322}- 
       \psi_{010}\psi_{323}- \psi_{003}\psi_{330}+ \psi_{002}\psi_{331}- \psi_{001}\psi_{332}+ 
       \psi_{000}\psi_{333})\nonumber\\ - (&\psi_{033}\psi_{100}- \psi_{032}\psi_{101}+ \psi_{031}\psi_{102}- \psi_{030}\psi_{103}-
        \psi_{023}\psi_{110}+ \psi_{022}\psi_{111}- \psi_{021}\psi_{112}+ \psi_{020}\psi_{113}\nonumber\\&+ \psi_{013}\psi_{120}- 
       \psi_{012}\psi_{121}+ \psi_{011}\psi_{122}- \psi_{010}\psi_{123}- \psi_{003}\psi_{130}+ \psi_{002}\psi_{131}- 
       \psi_{001}\psi_{132}+ \psi_{000}\psi_{133}\nonumber\\& + \psi_{233}\psi_{300}- \psi_{232}\psi_{301}+ \psi_{231}\psi_{302}- 
       \psi_{230}\psi_{303}- \psi_{223}\psi_{310}+ \psi_{222}\psi_{311}- \psi_{221}\psi_{312}+ \psi_{220}\psi_{313}\nonumber\\&+ 
       \psi_{213}\psi_{320}- \psi_{212}\psi_{321}+ \psi_{211}\psi_{322}- \psi_{210}\psi_{323}- \psi_{203}\psi_{330}+ 
       \psi_{202}\psi_{331}- \psi_{201}\psi_{332}+ \psi_{200}\psi_{333})\nonumber\\\times ( & \psi_{111}\psi_{200}- \psi_{110}\psi_{201}+ 
       \psi_{113}\psi_{202}- \psi_{112}\psi_{203}- \psi_{101}\psi_{210}+ \psi_{100}\psi_{211}- \psi_{103}\psi_{212}+ 
       \psi_{102}\psi_{213}\nonumber\\&+ \psi_{131}\psi_{220}- \psi_{130}\psi_{221}+ \psi_{133}\psi_{222}- \psi_{132}\psi_{223}- 
       \psi_{121}\psi_{230}+ \psi_{120}\psi_{231}- \psi_{123}\psi_{232}+ \psi_{122}\psi_{233}\nonumber\\& + \psi_{011}\psi_{300}- 
       \psi_{010}\psi_{301}+ \psi_{013}\psi_{302}- \psi_{012}\psi_{303}- \psi_{001}\psi_{310}+ \psi_{000}\psi_{311}- 
       \psi_{003}\psi_{312}+ \psi_{002}\psi_{313}\nonumber\\&+ \psi_{031}\psi_{320}- \psi_{030}\psi_{321}+ \psi_{033}\psi_{322}- 
       \psi_{032}\psi_{323}- \psi_{021}\psi_{330}+ \psi_{020}\psi_{331}- \psi_{023}\psi_{332}+ \psi_{022}\psi_{333}).
\end{align}   
   \normalsize

Examples of the degree 2 polynomials for four Dirac spinors are $H_a,H_b,H_c$ and $H_d$. They are given by
\footnotesize
\begin{align}
 H_a =2(&-\psi_{ 0111}\psi_{ 1000} + \psi_{ 0110}\psi_{ 1001} - \psi_{ 0113}\psi_{ 1002} + 
   \psi_{ 0112}\psi_{ 1003} + \psi_{ 0101}\psi_{ 1010} - \psi_{ 0100}\psi_{ 1011} + \psi_{ 0103}\psi_{ 1012} - 
   \psi_{ 0102}\psi_{ 1013}\nonumber\\& - \psi_{ 0131}\psi_{ 1020 }+ \psi_{ 0130}\psi_{ 1021} - \psi_{ 0133}\psi_{ 1022 }+ 
   \psi_{ 0132}\psi_{ 1023} +  \psi_{0121} \psi_{1030} - \psi_{0120} \psi_{1031} + \psi_{0123} \psi_{1032} - 
    \psi_{0122} \psi_{1033}\nonumber\\& +  \psi_{0011} \psi_{1100} -  \psi_{0010 }\psi_{1101} +  \psi_{0013} \psi_{1102} - 
    \psi_{0012} \psi_{1103} -  \psi_{0001} \psi_{1110 }+  \psi_{0000} \psi_{1111 }-  \psi_{0003} \psi_{1112} + 
   \psi_{ 0002}\psi_{ 1113}\nonumber\\& + \psi_{ 0031}\psi_{ 1120} - \psi_{ 0030 }\psi_{ 1121} + \psi_{ 0033}\psi_{ 1122} - 
   \psi_{ 0032}\psi_{ 1123 }- \psi_{ 0021}\psi_{ 1130} + \psi_{ 0020}\psi_{ 1131} - \psi_{ 0023}\psi_{ 1132} + 
   \psi_{ 0022}\psi_{ 1133}\nonumber\\& -  \psi_{0311} \psi_{1200} +  \psi_{0310 }\psi_{1201} -  \psi_{0313 }\psi_{1202} + 
    \psi_{0312} \psi_{1203} +  \psi_{0301} \psi_{1210} -  \psi_{0300} \psi_{1211} +  \psi_{0303} \psi_{1212} - 
    \psi_{0302} \psi_{1213}\nonumber\\& -  \psi_{0331 }\psi_{1220} +  \psi_{0330} \psi_{1221} -  \psi_{0333} \psi_{1222} + 
    \psi_{0332} \psi_{1223} +  \psi_{0321} \psi_{1230} -  \psi_{0320 }\psi_{1231 }+  \psi_{0323 }\psi_{1232 }- 
    \psi_{0322} \psi_{1233}\nonumber\\& +  \psi_{0211} \psi_{1300} -  \psi_{0210 }\psi_{1301} +  \psi_{0213 }\psi_{1302} - 
    \psi_{0212} \psi_{1303} -  \psi_{0201} \psi_{1310} +  \psi_{0200} \psi_{1311} - \psi_{0203 }\psi_{1312} + 
    \psi_{0202} \psi_{1313}\nonumber\\& +  \psi_{0231} \psi_{1320} -  \psi_{0230} \psi_{1321 }+  \psi_{0233} \psi_{1322} - 
    \psi_{0232} \psi_{1323} -  \psi_{0221} \psi_{1330 }+  \psi_{0220} \psi_{1331 }- \psi_{0223 }\psi_{1332} + 
    \psi_{0222} \psi_{1333}\nonumber\\& -  \psi_{2111} \psi_{3000} +  \psi_{2110} \psi_{3001 }-  \psi_{2113 }\psi_{3002} + 
    \psi_{2112} \psi_{3003} +  \psi_{2101} \psi_{3010} -  \psi_{2100} \psi_{3011} +  \psi_{2103} \psi_{3012} - 
    \psi_{2102} \psi_{3013}\nonumber\\& -  \psi_{2131} \psi_{3020} +  \psi_{2130 }\psi_{3021} -  \psi_{2133} \psi_{3022} + 
    \psi_{2132} \psi_{3023} +  \psi_{2121} \psi_{3030} -  \psi_{2120} \psi_{3031 }+  \psi_{2123} \psi_{3032} - 
    \psi_{2122} \psi_{3033}\nonumber\\& +  \psi_{2011} \psi_{3100} -  \psi_{2010 }\psi_{3101 }+  \psi_{2013 }\psi_{3102} - 
    \psi_{2012} \psi_{3103} -  \psi_{2001} \psi_{3110} +  \psi_{2000 }\psi_{3111} -  \psi_{2003} \psi_{3112} + 
    \psi_{2002} \psi_{3113}\nonumber\\& +  \psi_{2031} \psi_{3120} -  \psi_{2030} \psi_{3121} +  \psi_{2033 }\psi_{3122 }- 
    \psi_{2032} \psi_{3123} -  \psi_{2021} \psi_{3130 }+  \psi_{2020} \psi_{3131 }-  \psi_{2023} \psi_{3132 }+ 
    \psi_{2022} \psi_{3133}\nonumber\\& -  \psi_{2311} \psi_{3200} +  \psi_{2310} \psi_{3201} -  \psi_{2313} \psi_{3202} + 
    \psi_{2312} \psi_{3203} +  \psi_{2301} \psi_{3210 }-  \psi_{2300 }\psi_{3211} +  \psi_{2303} \psi_{3212} - 
    \psi_{2302} \psi_{3213}\nonumber\\& -  \psi_{2331} \psi_{3220} +  \psi_{2330} \psi_{3221} -  \psi_{2333} \psi_{3222} + 
    \psi_{2332} \psi_{3223} +  \psi_{2321} \psi_{3230 }-  \psi_{2320} \psi_{3231} +  \psi_{2323 }\psi_{3232} - 
    \psi_{2322} \psi_{3233}\nonumber\\& +  \psi_{2211} \psi_{3300} - \psi_{2210} \psi_{3301} +  \psi_{2213} \psi_{3302} - 
    \psi_{2212} \psi_{3303} -  \psi_{2201} \psi_{3310} +  \psi_{2200 }\psi_{3311 }-  \psi_{2203 }\psi_{3312} + 
    \psi_{2202} \psi_{3313}\nonumber\\& +  \psi_{2231} \psi_{3320} -  \psi_{2230} \psi_{3321} +  \psi_{2233} \psi_{3322} - 
    \psi_{2232} \psi_{3323} -  \psi_{2221} \psi_{3330} + \psi_{2220 }\psi_{3331} -  \psi_{2223 }\psi_{3332 }+ 
    \psi_{2222 }\psi_{3333}),
\end{align}

\begin{align}
 H_b =2( &\psi_{1333}\psi_{2000}- \psi_{1332}\psi_{2001}+ \psi_{1331}\psi_{2002}- 
   \psi_{1330}\psi_{2003}- \psi_{1323}\psi_{2010}+ \psi_{1322}\psi_{2011}- \psi_{1321}\psi_{2012}+ 
   \psi_{1320}\psi_{2013}\nonumber\\&+ \psi_{1313}\psi_{2020}- \psi_{1312}\psi_{2021}+ \psi_{1311}\psi_{2022}- 
   \psi_{1310}\psi_{2023}- \psi_{1303}\psi_{2030}+ \psi_{1302}\psi_{2031}- \psi_{1301}\psi_{2032}+ 
   \psi_{1300}\psi_{2033}\nonumber\\&- \psi_{1233}\psi_{2100}+ \psi_{1232}\psi_{2101}- \psi_{1231}\psi_{2102}+ 
   \psi_{1230}\psi_{2103}+ \psi_{1223}\psi_{2110}- \psi_{1222}\psi_{2111}+ \psi_{1221}\psi_{2112}- 
   \psi_{1220}\psi_{2113}\nonumber\\&- \psi_{1213}\psi_{2120}+ \psi_{1212}\psi_{2121}- \psi_{1211}\psi_{2122}+ 
   \psi_{1210}\psi_{2123}+ \psi_{1203}\psi_{2130}- \psi_{1202}\psi_{2131}+ \psi_{1201}\psi_{2132}- 
   \psi_{1200}\psi_{2133}\nonumber\\&+ \psi_{1133}\psi_{2200}- \psi_{1132}\psi_{2201}+ \psi_{1131}\psi_{2202}- 
   \psi_{1130}\psi_{2203}- \psi_{1123}\psi_{2210}+ \psi_{1122}\psi_{2211}- \psi_{1121}\psi_{2212}+ 
   \psi_{1120}\psi_{2213}\nonumber\\&+ \psi_{1113}\psi_{2220}- \psi_{1112}\psi_{2221}+ \psi_{1111}\psi_{2222}- 
   \psi_{1110}\psi_{2223}- \psi_{1103}\psi_{2230}+ \psi_{1102}\psi_{2231}- \psi_{1101}\psi_{2232}+ 
   \psi_{1100}\psi_{2233}\nonumber\\&- \psi_{1033}\psi_{2300}+ \psi_{1032}\psi_{2301}- \psi_{1031}\psi_{2302}+ 
   \psi_{1030}\psi_{2303}+ \psi_{1023}\psi_{2310}- \psi_{1022}\psi_{2311}+ \psi_{1021}\psi_{2312}- 
   \psi_{1020}\psi_{2313}\nonumber\\&- \psi_{1013}\psi_{2320}+ \psi_{1012}\psi_{2321}- \psi_{1011}\psi_{2322}+ 
   \psi_{1010}\psi_{2323}+ \psi_{1003}\psi_{2330}- \psi_{1002}\psi_{2331}+ \psi_{1001}\psi_{2332}- 
   \psi_{1000}\psi_{2333}\nonumber\\&- \psi_{0333}\psi_{3000}+ \psi_{0332}\psi_{3001}- \psi_{0331}\psi_{3002}+ 
   \psi_{0330}\psi_{3003}+ \psi_{0323}\psi_{3010}- \psi_{0322}\psi_{3011}+ \psi_{0321}\psi_{3012}- 
   \psi_{0320}\psi_{3013}\nonumber\\&- \psi_{0313}\psi_{3020}+ \psi_{0312}\psi_{3021}- \psi_{0311}\psi_{3022}+ 
   \psi_{0310}\psi_{3023}+ \psi_{0303}\psi_{3030}- \psi_{0302}\psi_{3031}+ \psi_{0301}\psi_{3032}- 
   \psi_{0300}\psi_{3033}\nonumber\\&+ \psi_{0233}\psi_{3100}- \psi_{0232}\psi_{3101}+ \psi_{0231}\psi_{3102}- 
   \psi_{0230}\psi_{3103}- \psi_{0223}\psi_{3110}+ \psi_{0222}\psi_{3111}- \psi_{0221}\psi_{3112}+ 
   \psi_{0220}\psi_{3113}\nonumber\\&+ \psi_{0213}\psi_{3120}- \psi_{0212}\psi_{3121}+ \psi_{0211}\psi_{3122}- 
   \psi_{0210}\psi_{3123}- \psi_{0203}\psi_{3130}+ \psi_{0202}\psi_{3131}- \psi_{0201}\psi_{3132}+ 
   \psi_{0200}\psi_{3133}\nonumber\\&- \psi_{0133}\psi_{3200}+ \psi_{0132}\psi_{3201}- \psi_{0131}\psi_{3202}+ 
   \psi_{0130}\psi_{3203}+ \psi_{0123}\psi_{3210}- \psi_{0122}\psi_{3211}+ \psi_{0121}\psi_{3212}- 
   \psi_{0120}\psi_{3213}\nonumber\\&- \psi_{0113}\psi_{3220}+ \psi_{0112}\psi_{3221}- \psi_{0111}\psi_{3222}+ 
   \psi_{0110}\psi_{3223}+ \psi_{0103}\psi_{3230}- \psi_{0102}\psi_{3231}+ \psi_{0101}\psi_{3232}- 
   \psi_{0100}\psi_{3233}\nonumber\\&+ \psi_{0033}\psi_{3300}- \psi_{0032}\psi_{3301}+ \psi_{0031}\psi_{3302}- 
   \psi_{0030}\psi_{3303}- \psi_{0023}\psi_{3310}+ \psi_{0022}\psi_{3311}- \psi_{0021}\psi_{3312}+ 
   \psi_{0020}\psi_{3313}\nonumber\\&+ \psi_{0013}\psi_{3320}- \psi_{0012}\psi_{3321}+ \psi_{0011}\psi_{3322}- 
   \psi_{0010}\psi_{3323}- \psi_{0003}\psi_{3330}+ \psi_{0002}\psi_{3331}- \psi_{0001}\psi_{3332}+ 
   \psi_{0000}\psi_{3333}),
\end{align}

\begin{align}
  H_c = 2(&-\psi_{0113}\psi_{1000}+ \psi_{0112}\psi_{1001}- \psi_{0111}\psi_{1002}+ 
   \psi_{0110}\psi_{1003}+ \psi_{0103}\psi_{1010}- \psi_{0102}\psi_{1011}+ \psi_{0101}\psi_{1012}- 
   \psi_{0100}\psi_{1013}\nonumber\\&- \psi_{0133}\psi_{1020}+ \psi_{0132}\psi_{1021}- \psi_{0131}\psi_{1022}+ 
   \psi_{0130}\psi_{1023}+ \psi_{0123}\psi_{1030}- \psi_{0122}\psi_{1031}+ \psi_{0121}\psi_{1032}- 
   \psi_{0120}\psi_{1033}\nonumber\\&+ \psi_{0013}\psi_{1100}- \psi_{0012}\psi_{1101}+ \psi_{0011}\psi_{1102}- 
   \psi_{0010}\psi_{1103}- \psi_{0003}\psi_{1110}+ \psi_{0002}\psi_{1111}- \psi_{0001}\psi_{1112}+ 
   \psi_{0000}\psi_{1113}\nonumber\\&+ \psi_{0033}\psi_{1120}- \psi_{0032}\psi_{1121}+ \psi_{0031}\psi_{1122}- 
   \psi_{0030}\psi_{1123}- \psi_{0023}\psi_{1130}+ \psi_{0022}\psi_{1131}- \psi_{0021}\psi_{1132}+ 
   \psi_{0020}\psi_{1133}\nonumber\\&- \psi_{0313}\psi_{1200}+ \psi_{0312}\psi_{1201}- \psi_{0311}\psi_{1202}+ 
   \psi_{0310}\psi_{1203}+ \psi_{0303}\psi_{1210}- \psi_{0302}\psi_{1211}+ \psi_{0301}\psi_{1212}- 
   \psi_{0300}\psi_{1213}\nonumber\\&- \psi_{0333}\psi_{1220}+ \psi_{0332}\psi_{1221}- \psi_{0331}\psi_{1222}+ 
   \psi_{0330}\psi_{1223}+ \psi_{0323}\psi_{1230}- \psi_{0322}\psi_{1231}+ \psi_{0321}\psi_{1232}- 
   \psi_{0320}\psi_{1233}\nonumber\\&+ \psi_{0213}\psi_{1300}- \psi_{0212}\psi_{1301}+ \psi_{0211}\psi_{1302}- 
   \psi_{0210}\psi_{1303}- \psi_{0203}\psi_{1310}+ \psi_{0202}\psi_{1311}- \psi_{0201}\psi_{1312}+ 
   \psi_{0200}\psi_{1313}\nonumber\\&+ \psi_{0233}\psi_{1320}- \psi_{0232}\psi_{1321}+ \psi_{0231}\psi_{1322}- 
   \psi_{0230}\psi_{1323}- \psi_{0223}\psi_{1330}+ \psi_{0222}\psi_{1331}- \psi_{0221}\psi_{1332}+ 
   \psi_{0220}\psi_{1333}\nonumber\\&- \psi_{2113}\psi_{3000}+ \psi_{2112}\psi_{3001}- \psi_{2111}\psi_{3002}+ 
   \psi_{2110}\psi_{3003}+ \psi_{2103}\psi_{3010}- \psi_{2102}\psi_{3011}+ \psi_{2101}\psi_{3012}- 
   \psi_{2100}\psi_{3013}\nonumber\\&- \psi_{2133}\psi_{3020}+ \psi_{2132}\psi_{3021}- \psi_{2131}\psi_{3022}+ 
   \psi_{2130}\psi_{3023}+ \psi_{2123}\psi_{3030}- \psi_{2122}\psi_{3031}+ \psi_{2121}\psi_{3032}- 
   \psi_{2120}\psi_{3033}\nonumber\\&+ \psi_{2013}\psi_{3100}- \psi_{2012}\psi_{3101}+ \psi_{2011}\psi_{3102}- 
   \psi_{2010}\psi_{3103}- \psi_{2003}\psi_{3110}+ \psi_{2002}\psi_{3111}- \psi_{2001}\psi_{3112}+ 
   \psi_{2000}\psi_{3113}\nonumber\\&+ \psi_{2033}\psi_{3120}- \psi_{2032}\psi_{3121}+ \psi_{2031}\psi_{3122}- 
   \psi_{2030}\psi_{3123}- \psi_{2023}\psi_{3130}+ \psi_{2022}\psi_{3131}- \psi_{2021}\psi_{3132}+ 
   \psi_{2020}\psi_{3133}\nonumber\\&- \psi_{2313}\psi_{3200}+ \psi_{2312}\psi_{3201}- \psi_{2311}\psi_{3202}+ 
   \psi_{2310}\psi_{3203}+ \psi_{2303}\psi_{3210}- \psi_{2302}\psi_{3211}+ \psi_{2301}\psi_{3212}- 
   \psi_{2300}\psi_{3213}\nonumber\\&- \psi_{2333}\psi_{3220}+ \psi_{2332}\psi_{3221}- \psi_{2331}\psi_{3222}+ 
   \psi_{2330}\psi_{3223}+ \psi_{2323}\psi_{3230}- \psi_{2322}\psi_{3231}+ \psi_{2321}\psi_{3232}- 
   \psi_{2320}\psi_{3233}\nonumber\\&+ \psi_{2213}\psi_{3300}- \psi_{2212}\psi_{3301}+ \psi_{2211}\psi_{3302}- 
   \psi_{2210}\psi_{3303}- \psi_{2203}\psi_{3310}+ \psi_{2202}\psi_{3311}- \psi_{2201}\psi_{3312}+ 
   \psi_{2200}\psi_{3313}\nonumber\\&+ \psi_{2233}\psi_{3320}- \psi_{2232}\psi_{3321}+ \psi_{2231}\psi_{3322}- 
   \psi_{2230}\psi_{3323}- \psi_{2223}\psi_{3330}+ \psi_{2222}\psi_{3331}- \psi_{2221}\psi_{3332}+ 
   \psi_{2220}\psi_{3333}),
\end{align}
\normalsize
and
\footnotesize
\begin{align}
   H_d = 2(&- \psi_{0131}\psi_{1000}+\psi_{0130}\psi_{1001}-\psi_{0133}\psi_{1002}+ 
  \psi_{0132}\psi_{1003}+\psi_{0121}\psi_{1010}-\psi_{0120}\psi_{1011}+\psi_{0123}\psi_{1012}- 
  \psi_{0122}\psi_{1013}\nonumber\\&-\psi_{0111}\psi_{1020}+\psi_{0110}\psi_{1021}-\psi_{0113}\psi_{1022}+ 
  \psi_{0112}\psi_{1023}+\psi_{0101}\psi_{1030}-\psi_{0100}\psi_{1031}+\psi_{0103}\psi_{1032}- 
  \psi_{0102}\psi_{1033}\nonumber\\&+\psi_{0031}\psi_{1100}-\psi_{0030}\psi_{1101}+\psi_{0033}\psi_{1102}- 
  \psi_{0032}\psi_{1103}-\psi_{0021}\psi_{1110}+\psi_{0020}\psi_{1111}-\psi_{0023}\psi_{1112}+ 
  \psi_{0022}\psi_{1113}\nonumber\\&+\psi_{0011}\psi_{1120}-\psi_{0010}\psi_{1121}+\psi_{0013}\psi_{1122}- 
  \psi_{0012}\psi_{1123}-\psi_{0001}\psi_{1130}+\psi_{0000}\psi_{1131}-\psi_{0003}\psi_{1132}+ 
  \psi_{0002}\psi_{1133}\nonumber\\&-\psi_{0331}\psi_{1200}+\psi_{0330}\psi_{1201}-\psi_{0333}\psi_{1202}+ 
  \psi_{0332}\psi_{1203}+\psi_{0321}\psi_{1210}-\psi_{0320}\psi_{1211}+\psi_{0323}\psi_{1212}- 
  \psi_{0322}\psi_{1213}\nonumber\\&-\psi_{0311}\psi_{1220}+\psi_{0310}\psi_{1221}-\psi_{0313}\psi_{1222}+ 
  \psi_{0312}\psi_{1223}+\psi_{0301}\psi_{1230}-\psi_{0300}\psi_{1231}+\psi_{0303}\psi_{1232}- 
  \psi_{0302}\psi_{1233}\nonumber\\&+\psi_{0231}\psi_{1300}-\psi_{0230}\psi_{1301}+\psi_{0233}\psi_{1302}- 
  \psi_{0232}\psi_{1303}-\psi_{0221}\psi_{1310}+\psi_{0220}\psi_{1311}-\psi_{0223}\psi_{1312}+ 
  \psi_{0222}\psi_{1313}\nonumber\\&+\psi_{0211}\psi_{1320}-\psi_{0210}\psi_{1321}+\psi_{0213}\psi_{1322}- 
  \psi_{0212}\psi_{1323}-\psi_{0201}\psi_{1330}+\psi_{0200}\psi_{1331}-\psi_{0203}\psi_{1332}+ 
  \psi_{0202}\psi_{1333}\nonumber\\&-\psi_{2131}\psi_{3000}+\psi_{2130}\psi_{3001}-\psi_{2133}\psi_{3002}+ 
  \psi_{2132}\psi_{3003}+\psi_{2121}\psi_{3010}-\psi_{2120}\psi_{3011}+\psi_{2123}\psi_{3012}- 
  \psi_{2122}\psi_{3013}\nonumber\\&-\psi_{2111}\psi_{3020}+\psi_{2110}\psi_{3021}-\psi_{2113}\psi_{3022}+ 
  \psi_{2112}\psi_{3023}+\psi_{2101}\psi_{3030}-\psi_{2100}\psi_{3031}+\psi_{2103}\psi_{3032}- 
  \psi_{2102}\psi_{3033}\nonumber\\&+\psi_{2031}\psi_{3100}-\psi_{2030}\psi_{3101}+\psi_{2033}\psi_{3102}- 
  \psi_{2032}\psi_{3103}-\psi_{2021}\psi_{3110}+\psi_{2020}\psi_{3111}-\psi_{2023}\psi_{3112}+ 
  \psi_{2022}\psi_{3113}\nonumber\\&+\psi_{2011}\psi_{3120}-\psi_{2010}\psi_{3121}+\psi_{2013}\psi_{3122}- 
  \psi_{2012}\psi_{3123}-\psi_{2001}\psi_{3130}+\psi_{2000}\psi_{3131}-\psi_{2003}\psi_{3132}+ 
  \psi_{2002}\psi_{3133}\nonumber\\&-\psi_{2331}\psi_{3200}+\psi_{2330}\psi_{3201}-\psi_{2333}\psi_{3202}+ 
  \psi_{2332}\psi_{3203}+\psi_{2321}\psi_{3210}-\psi_{2320}\psi_{3211}+\psi_{2323}\psi_{3212}- 
  \psi_{2322}\psi_{3213}\nonumber\\&-\psi_{2311}\psi_{3220}+\psi_{2310}\psi_{3221}-\psi_{2313}\psi_{3222}+ 
  \psi_{2312}\psi_{3223}+\psi_{2301}\psi_{3230}-\psi_{2300}\psi_{3231}+\psi_{2303}\psi_{3232}- 
  \psi_{2302}\psi_{3233}\nonumber\\&+\psi_{2231}\psi_{3300}-\psi_{2230}\psi_{3301}+\psi_{2233}\psi_{3302}- 
  \psi_{2232}\psi_{3303}-\psi_{2221}\psi_{3310}+\psi_{2220}\psi_{3311}-\psi_{2223}\psi_{3312}+ 
  \psi_{2222}\psi_{3313}\nonumber\\&+\psi_{2211}\psi_{3320}-\psi_{2210}\psi_{3321}+\psi_{2213}\psi_{3322}- 
  \psi_{2212}\psi_{3323}-\psi_{2201}\psi_{3330}+\psi_{2200}\psi_{3331}-\psi_{2203}\psi_{3332}+ 
  \psi_{2202}\psi_{3333}).
\end{align}
\normalsize
Examples of the degree four polynomials for four Dirac spinors are $T_l$ and $Y_l$. These are given by

\footnotesize
\begin{align}
 T_l= 2 ( &\psi_{0111}\psi_{1000}-  \psi_{0110}\psi_{1001}+  \psi_{0113}\psi_{1002}-  \psi_{0112}\psi_{1003}- 
     \psi_{0101}\psi_{1010}+  \psi_{0100}\psi_{1011}-  \psi_{0103}\psi_{1012}+  \psi_{0102}\psi_{1013}\nonumber\\&- 
     \psi_{0011}\psi_{1100}+  \psi_{0010}\psi_{1101}-  \psi_{0013}\psi_{1102}+  \psi_{0012}\psi_{1103}+ 
     \psi_{0001}\psi_{1110}-  \psi_{0000}\psi_{1111}+  \psi_{0003}\psi_{1112}-  \psi_{0002}\psi_{1113}\nonumber\\&+ 
     \psi_{0311}\psi_{1200}-  \psi_{0310}\psi_{1201}+  \psi_{0313}\psi_{1202}-  \psi_{0312}\psi_{1203}- 
     \psi_{0301}\psi_{1210}+  \psi_{0300}\psi_{1211}-  \psi_{0303}\psi_{1212}+  \psi_{0302}\psi_{1213}\nonumber\\&- 
     \psi_{0211}\psi_{1300}+  \psi_{0210}\psi_{1301}-  \psi_{0213}\psi_{1302}+  \psi_{0212}\psi_{1303}+ 
     \psi_{0201}\psi_{1310}-  \psi_{0200}\psi_{1311}+  \psi_{0203}\psi_{1312}-  \psi_{0202}\psi_{1313}\nonumber\\&+ 
     \psi_{2111}\psi_{3000}-  \psi_{2110}\psi_{3001}+  \psi_{2113}\psi_{3002}-  \psi_{2112}\psi_{3003}- 
     \psi_{2101}\psi_{3010}+  \psi_{2100}\psi_{3011}-  \psi_{2103}\psi_{3012}+  \psi_{2102}\psi_{3013}\nonumber\\&- 
     \psi_{2011}\psi_{3100}+  \psi_{2010}\psi_{3101}-  \psi_{2013}\psi_{3102}+  \psi_{2012}\psi_{3103}+ 
     \psi_{2001}\psi_{3110}-  \psi_{2000}\psi_{3111}+  \psi_{2003}\psi_{3112}-  \psi_{2002}\psi_{3113}\nonumber\\&+ 
     \psi_{2311}\psi_{3200}-  \psi_{2310}\psi_{3201}+  \psi_{2313}\psi_{3202}-  \psi_{2312}\psi_{3203}- 
     \psi_{2301}\psi_{3210}+  \psi_{2300}\psi_{3211}-  \psi_{2303}\psi_{3212}+  \psi_{2302}\psi_{3213}\nonumber\\&- 
     \psi_{2211}\psi_{3300}+  \psi_{2210}\psi_{3301}-  \psi_{2213}\psi_{3302}+  \psi_{2212}\psi_{3303}+ 
     \psi_{2201}\psi_{3310}-  \psi_{2200}\psi_{3311}+  \psi_{2203}\psi_{3312}- 
     \psi_{2202}\psi_{3313})^2\nonumber\\ + 
 2 ( &\psi_{0131}\psi_{1020}-  \psi_{0130}\psi_{1021}+  \psi_{0133}\psi_{1022}-  \psi_{0132}\psi_{1023}- 
     \psi_{0121}\psi_{1030}+  \psi_{0120}\psi_{1031}-  \psi_{0123}\psi_{1032}+  \psi_{0122}\psi_{1033}\nonumber\\&- 
     \psi_{0031}\psi_{1120}+  \psi_{0030}\psi_{1121}-  \psi_{0033}\psi_{1122}+  \psi_{0032}\psi_{1123}+ 
     \psi_{0021}\psi_{1130}-  \psi_{0020}\psi_{1131}+  \psi_{0023}\psi_{1132}-  \psi_{0022}\psi_{1133}\nonumber\\&+ 
     \psi_{0331}\psi_{1220}-  \psi_{0330}\psi_{1221}+  \psi_{0333}\psi_{1222}-  \psi_{0332}\psi_{1223}- 
     \psi_{0321}\psi_{1230}+  \psi_{0320}\psi_{1231}-  \psi_{0323}\psi_{1232}+  \psi_{0322}\psi_{1233}\nonumber\\&- 
     \psi_{0231}\psi_{1320}+  \psi_{0230}\psi_{1321}-  \psi_{0233}\psi_{1322}+  \psi_{0232}\psi_{1323}+ 
     \psi_{0221}\psi_{1330}-  \psi_{0220}\psi_{1331}+  \psi_{0223}\psi_{1332}-  \psi_{0222}\psi_{1333}\nonumber\\&+ 
     \psi_{2131}\psi_{3020}-  \psi_{2130}\psi_{3021}+  \psi_{2133}\psi_{3022}-  \psi_{2132}\psi_{3023}- 
     \psi_{2121}\psi_{3030}+  \psi_{2120}\psi_{3031}-  \psi_{2123}\psi_{3032}+  \psi_{2122}\psi_{3033}\nonumber\\&- 
     \psi_{2031}\psi_{3120}+  \psi_{2030}\psi_{3121}-  \psi_{2033}\psi_{3122}+  \psi_{2032}\psi_{3123}+ 
     \psi_{2021}\psi_{3130}-  \psi_{2020}\psi_{3131}+  \psi_{2023}\psi_{3132}-  \psi_{2022}\psi_{3133}\nonumber\\&+ 
     \psi_{2331}\psi_{3220}-  \psi_{2330}\psi_{3221}+  \psi_{2333}\psi_{3222}-  \psi_{2332}\psi_{3223}- 
     \psi_{2321}\psi_{3230}+  \psi_{2320}\psi_{3231}-  \psi_{2323}\psi_{3232}+  \psi_{2322}\psi_{3233}\nonumber\\&- 
     \psi_{2231}\psi_{3320}+  \psi_{2230}\psi_{3321}-  \psi_{2233}\psi_{3322}+  \psi_{2232}\psi_{3323}+ 
     \psi_{2221}\psi_{3330}-  \psi_{2220}\psi_{3331}+  \psi_{2223}\psi_{3332}-  \psi_{2222}\psi_{3333})^2\nonumber\\+ 
 4 (&- \psi_{0121}\psi_{1010}+  \psi_{0120}\psi_{1011}-  \psi_{0123}\psi_{1012}+  \psi_{0122}\psi_{1013}+ 
     \psi_{0111}\psi_{1020}-  \psi_{0110}\psi_{1021}+  \psi_{0113}\psi_{1022}-  \psi_{0112}\psi_{1023}\nonumber\\&+ 
     \psi_{0021}\psi_{1110}-  \psi_{0020}\psi_{1111}+  \psi_{0023}\psi_{1112}-  \psi_{0022}\psi_{1113}- 
     \psi_{0011}\psi_{1120}+  \psi_{0010}\psi_{1121}-  \psi_{0013}\psi_{1122}+  \psi_{0012}\psi_{1123}\nonumber\\&- 
     \psi_{0321}\psi_{1210}+  \psi_{0320}\psi_{1211}-  \psi_{0323}\psi_{1212}+  \psi_{0322}\psi_{1213}+ 
     \psi_{0311}\psi_{1220}-  \psi_{0310}\psi_{1221}+  \psi_{0313}\psi_{1222}-  \psi_{0312}\psi_{1223}\nonumber\\&+ 
     \psi_{0221}\psi_{1310}-  \psi_{0220}\psi_{1311}+  \psi_{0223}\psi_{1312}-  \psi_{0222}\psi_{1313}- 
     \psi_{0211}\psi_{1320}+  \psi_{0210}\psi_{1321}-  \psi_{0213}\psi_{1322}+  \psi_{0212}\psi_{1323}\nonumber\\&- 
     \psi_{2121}\psi_{3010}+  \psi_{2120}\psi_{3011}-  \psi_{2123}\psi_{3012}+  \psi_{2122}\psi_{3013}+ 
     \psi_{2111}\psi_{3020}-  \psi_{2110}\psi_{3021}+  \psi_{2113}\psi_{3022}-  \psi_{2112}\psi_{3023}\nonumber\\&+ 
     \psi_{2021}\psi_{3110}-  \psi_{2020}\psi_{3111}+  \psi_{2023}\psi_{3112}-  \psi_{2022}\psi_{3113}- 
     \psi_{2011}\psi_{3120}+  \psi_{2010}\psi_{3121}-  \psi_{2013}\psi_{3122}+  \psi_{2012}\psi_{3123}\nonumber\\&- 
     \psi_{2321}\psi_{3210}+  \psi_{2320}\psi_{3211}-  \psi_{2323}\psi_{3212}+  \psi_{2322}\psi_{3213}+ 
     \psi_{2311}\psi_{3220}-  \psi_{2310}\psi_{3221}+  \psi_{2313}\psi_{3222}-  \psi_{2312}\psi_{3223}\nonumber\\&+ 
     \psi_{2221}\psi_{3310}-  \psi_{2220}\psi_{3311}+  \psi_{2223}\psi_{3312}-  \psi_{2222}\psi_{3313}- 
     \psi_{2211}\psi_{3320}+  \psi_{2210}\psi_{3321}-  \psi_{2213}\psi_{3322}+ 
     \psi_{2212}\psi_{3323})\nonumber\\\times ( &\psi_{0131}\psi_{1000}-  \psi_{0130}\psi_{1001}+  \psi_{0133}\psi_{1002}- 
     \psi_{0132}\psi_{1003}-  \psi_{0101}\psi_{1030}+  \psi_{0100}\psi_{1031}-  \psi_{0103}\psi_{1032}+ 
     \psi_{0102}\psi_{1033}\nonumber\\&-  \psi_{0031}\psi_{1100}+  \psi_{0030}\psi_{1101}-  \psi_{0033}\psi_{1102}+ 
     \psi_{0032}\psi_{1103}+  \psi_{0001}\psi_{1130}-  \psi_{0000}\psi_{1131}+  \psi_{0003}\psi_{1132}- 
     \psi_{0002}\psi_{1133}\nonumber\\&+  \psi_{0331}\psi_{1200}-  \psi_{0330}\psi_{1201}+  \psi_{0333}\psi_{1202}- 
     \psi_{0332}\psi_{1203}-  \psi_{0301}\psi_{1230}+  \psi_{0300}\psi_{1231}-  \psi_{0303}\psi_{1232}+ 
     \psi_{0302}\psi_{1233}\nonumber\\&-  \psi_{0231}\psi_{1300}+  \psi_{0230}\psi_{1301}-  \psi_{0233}\psi_{1302}+ 
     \psi_{0232}\psi_{1303}+  \psi_{0201}\psi_{1330}-  \psi_{0200}\psi_{1331}+  \psi_{0203}\psi_{1332}- 
     \psi_{0202}\psi_{1333}\nonumber\\&+  \psi_{2131}\psi_{3000}-  \psi_{2130}\psi_{3001}+  \psi_{2133}\psi_{3002}- 
     \psi_{2132}\psi_{3003}-  \psi_{2101}\psi_{3030}+  \psi_{2100}\psi_{3031}-  \psi_{2103}\psi_{3032}+ 
     \psi_{2102}\psi_{3033}\nonumber\\&-  \psi_{2031}\psi_{3100}+  \psi_{2030}\psi_{3101}-  \psi_{2033}\psi_{3102}+ 
     \psi_{2032}\psi_{3103}+  \psi_{2001}\psi_{3130}-  \psi_{2000}\psi_{3131}+  \psi_{2003}\psi_{3132}- 
     \psi_{2002}\psi_{3133}\nonumber\\&+  \psi_{2331}\psi_{3200}-  \psi_{2330}\psi_{3201}+  \psi_{2333}\psi_{3202}- 
     \psi_{2332}\psi_{3203}-  \psi_{2301}\psi_{3230}+  \psi_{2300}\psi_{3231}-  \psi_{2303}\psi_{3232}+ 
     \psi_{2302}\psi_{3233}\nonumber\\&-  \psi_{2231}\psi_{3300}+  \psi_{2230}\psi_{3301}-  \psi_{2233}\psi_{3302}+ 
     \psi_{2232}\psi_{3303}+  \psi_{2201}\psi_{3330}-  \psi_{2200}\psi_{3331}+  \psi_{2203}\psi_{3332}- 
     \psi_{2202}\psi_{3333})\nonumber\\ + 
 4 ( &\psi_{0121}\psi_{1000}-  \psi_{0120}\psi_{1001}+  \psi_{0123}\psi_{1002}-  \psi_{0122}\psi_{1003}- 
     \psi_{0101}\psi_{1020}+  \psi_{0100}\psi_{1021}-  \psi_{0103}\psi_{1022}+  \psi_{0102}\psi_{1023}\nonumber\\&- 
     \psi_{0021}\psi_{1100}+  \psi_{0020}\psi_{1101}-  \psi_{0023}\psi_{1102}+  \psi_{0022}\psi_{1103}+ 
     \psi_{0001}\psi_{1120}-  \psi_{0000}\psi_{1121}+  \psi_{0003}\psi_{1122}-  \psi_{0002}\psi_{1123}\nonumber\\&+ 
     \psi_{0321}\psi_{1200}-  \psi_{0320}\psi_{1201}+  \psi_{0323}\psi_{1202}-  \psi_{0322}\psi_{1203}- 
     \psi_{0301}\psi_{1220}+  \psi_{0300}\psi_{1221}-  \psi_{0303}\psi_{1222}+  \psi_{0302}\psi_{1223}\nonumber\\&- 
     \psi_{0221}\psi_{1300}+  \psi_{0220}\psi_{1301}-  \psi_{0223}\psi_{1302}+  \psi_{0222}\psi_{1303}+ 
     \psi_{0201}\psi_{1320}-  \psi_{0200}\psi_{1321}+  \psi_{0203}\psi_{1322}-  \psi_{0202}\psi_{1323}\nonumber\\&+ 
     \psi_{2121}\psi_{3000}-  \psi_{2120}\psi_{3001}+  \psi_{2123}\psi_{3002}-  \psi_{2122}\psi_{3003}- 
     \psi_{2101}\psi_{3020}+  \psi_{2100}\psi_{3021}-  \psi_{2103}\psi_{3022}+  \psi_{2102}\psi_{3023}\nonumber\\&- 
     \psi_{2021}\psi_{3100}+  \psi_{2020}\psi_{3101}-  \psi_{2023}\psi_{3102}+  \psi_{2022}\psi_{3103}+ 
     \psi_{2001}\psi_{3120}-  \psi_{2000}\psi_{3121}+  \psi_{2003}\psi_{3122}-  \psi_{2002}\psi_{3123}\nonumber\\&+ 
     \psi_{2321}\psi_{3200}-  \psi_{2320}\psi_{3201}+  \psi_{2323}\psi_{3202}-  \psi_{2322}\psi_{3203}- 
     \psi_{2301}\psi_{3220}+  \psi_{2300}\psi_{3221}-  \psi_{2303}\psi_{3222}+  \psi_{2302}\psi_{3223}\nonumber\\&- 
     \psi_{2221}\psi_{3300}+  \psi_{2220}\psi_{3301}-  \psi_{2223}\psi_{3302}+  \psi_{2222}\psi_{3303}+ 
     \psi_{2201}\psi_{3320}-  \psi_{2200}\psi_{3321}+  \psi_{2203}\psi_{3322}- 
     \psi_{2202}\psi_{3323})\nonumber\\\times ( &\psi_{0131}\psi_{1010}-  \psi_{0130}\psi_{1011}+  \psi_{0133}\psi_{1012}- 
     \psi_{0132}\psi_{1013}-  \psi_{0111}\psi_{1030}+  \psi_{0110}\psi_{1031}-  \psi_{0113}\psi_{1032}+ 
     \psi_{0112}\psi_{1033}\nonumber\\&-  \psi_{0031}\psi_{1110}+  \psi_{0030}\psi_{1111}-  \psi_{0033}\psi_{1112}+ 
     \psi_{0032}\psi_{1113}+  \psi_{0011}\psi_{1130}-  \psi_{0010}\psi_{1131}+  \psi_{0013}\psi_{1132}- 
     \psi_{0012}\psi_{1133}\nonumber\\&+  \psi_{0331}\psi_{1210}-  \psi_{0330}\psi_{1211}+  \psi_{0333}\psi_{1212}- 
     \psi_{0332}\psi_{1213}-  \psi_{0311}\psi_{1230}+  \psi_{0310}\psi_{1231}-  \psi_{0313}\psi_{1232}+ 
     \psi_{0312}\psi_{1233}\nonumber\\&-  \psi_{0231}\psi_{1310}+  \psi_{0230}\psi_{1311}-  \psi_{0233}\psi_{1312}+ 
     \psi_{0232}\psi_{1313}+  \psi_{0211}\psi_{1330}-  \psi_{0210}\psi_{1331}+  \psi_{0213}\psi_{1332}- 
     \psi_{0212}\psi_{1333}\nonumber\\&+  \psi_{2131}\psi_{3010}-  \psi_{2130}\psi_{3011}+  \psi_{2133}\psi_{3012}- 
     \psi_{2132}\psi_{3013}-  \psi_{2111}\psi_{3030}+  \psi_{2110}\psi_{3031}-  \psi_{2113}\psi_{3032}+ 
     \psi_{2112}\psi_{3033}\nonumber\\&-  \psi_{2031}\psi_{3110}+  \psi_{2030}\psi_{3111}-  \psi_{2033}\psi_{3112}+ 
     \psi_{2032}\psi_{3113}+  \psi_{2011}\psi_{3130}-  \psi_{2010}\psi_{3131}+  \psi_{2013}\psi_{3132}- 
     \psi_{2012}\psi_{3133}\nonumber\\&+  \psi_{2331}\psi_{3210}-  \psi_{2330}\psi_{3211}+  \psi_{2333}\psi_{3212}- 
     \psi_{2332}\psi_{3213}-  \psi_{2311}\psi_{3230}+  \psi_{2310}\psi_{3231}-  \psi_{2313}\psi_{3232}+ 
     \psi_{2312}\psi_{3233}\nonumber\\&-  \psi_{2231}\psi_{3310}+  \psi_{2230}\psi_{3311}-  \psi_{2233}\psi_{3312}+ 
     \psi_{2232}\psi_{3313}+  \psi_{2211}\psi_{3330}-  \psi_{2210}\psi_{3331}+  \psi_{2213}\psi_{3332}- 
     \psi_{2212}\psi_{3333}),  
\end{align}
\normalsize

and
\footnotesize
\begin{align}
 Y_l=2 (&- \psi_{1323}\psi_{2010}+  \psi_{1322}\psi_{2011}-  \psi_{1321}\psi_{2012}+  \psi_{1320}\psi_{2013}+ 
     \psi_{1313}\psi_{2020}-  \psi_{1312}\psi_{2021}+  \psi_{1311}\psi_{2022}-  \psi_{1310}\psi_{2023}\nonumber\\&+ 
     \psi_{1223}\psi_{2110}-  \psi_{1222}\psi_{2111}+  \psi_{1221}\psi_{2112}-  \psi_{1220}\psi_{2113}- 
     \psi_{1213}\psi_{2120}+  \psi_{1212}\psi_{2121}-  \psi_{1211}\psi_{2122}+  \psi_{1210}\psi_{2123}\nonumber\\&- 
     \psi_{1123}\psi_{2210}+  \psi_{1122}\psi_{2211}-  \psi_{1121}\psi_{2212}+  \psi_{1120}\psi_{2213}+ 
     \psi_{1113}\psi_{2220}-  \psi_{1112}\psi_{2221}+  \psi_{1111}\psi_{2222}-  \psi_{1110}\psi_{2223}\nonumber\\&+ 
     \psi_{1023}\psi_{2310}-  \psi_{1022}\psi_{2311}+  \psi_{1021}\psi_{2312}-  \psi_{1020}\psi_{2313}- 
     \psi_{1013}\psi_{2320}+  \psi_{1012}\psi_{2321}-  \psi_{1011}\psi_{2322}+  \psi_{1010}\psi_{2323}\nonumber\\&+ 
     \psi_{0323}\psi_{3010}-  \psi_{0322}\psi_{3011}+  \psi_{0321}\psi_{3012}-  \psi_{0320}\psi_{3013}- 
     \psi_{0313}\psi_{3020}+  \psi_{0312}\psi_{3021}-  \psi_{0311}\psi_{3022}+  \psi_{0310}\psi_{3023}\nonumber\\&- 
     \psi_{0223}\psi_{3110}+  \psi_{0222}\psi_{3111}-  \psi_{0221}\psi_{3112}+  \psi_{0220}\psi_{3113}+ 
     \psi_{0213}\psi_{3120}-  \psi_{0212}\psi_{3121}+  \psi_{0211}\psi_{3122}-  \psi_{0210}\psi_{3123}\nonumber\\&+ 
     \psi_{0123}\psi_{3210}-  \psi_{0122}\psi_{3211}+  \psi_{0121}\psi_{3212}-  \psi_{0120}\psi_{3213}- 
     \psi_{0113}\psi_{3220}+  \psi_{0112}\psi_{3221}-  \psi_{0111}\psi_{3222}+  \psi_{0110}\psi_{3223}\nonumber\\&- 
     \psi_{0023}\psi_{3310}+  \psi_{0022}\psi_{3311}-  \psi_{0021}\psi_{3312}+  \psi_{0020}\psi_{3313}+ 
     \psi_{0013}\psi_{3320}-  \psi_{0012}\psi_{3321}+  \psi_{0011}\psi_{3322}- 
     \psi_{0010}\psi_{3323})^2\nonumber\\+ 
 2 (&- \psi_{1333}\psi_{2000}+  \psi_{1332}\psi_{2001}-  \psi_{1331}\psi_{2002}+  \psi_{1330}\psi_{2003}+ 
     \psi_{1303}\psi_{2030}-  \psi_{1302}\psi_{2031}+  \psi_{1301}\psi_{2032}-  \psi_{1300}\psi_{2033}\nonumber\\&+ 
     \psi_{1233}\psi_{2100}-  \psi_{1232}\psi_{2101}+  \psi_{1231}\psi_{2102}-  \psi_{1230}\psi_{2103}- 
     \psi_{1203}\psi_{2130}+  \psi_{1202}\psi_{2131}-  \psi_{1201}\psi_{2132}+  \psi_{1200}\psi_{2133}\nonumber\\&- 
     \psi_{1133}\psi_{2200}+  \psi_{1132}\psi_{2201}-  \psi_{1131}\psi_{2202}+  \psi_{1130}\psi_{2203}+ 
     \psi_{1103}\psi_{2230}-  \psi_{1102}\psi_{2231}+  \psi_{1101}\psi_{2232}-  \psi_{1100}\psi_{2233}\nonumber\\&+ 
     \psi_{1033}\psi_{2300}-  \psi_{1032}\psi_{2301}+  \psi_{1031}\psi_{2302}-  \psi_{1030}\psi_{2303}- 
     \psi_{1003}\psi_{2330}+  \psi_{1002}\psi_{2331}-  \psi_{1001}\psi_{2332}+  \psi_{1000}\psi_{2333}\nonumber\\&+ 
     \psi_{0333}\psi_{3000}-  \psi_{0332}\psi_{3001}+  \psi_{0331}\psi_{3002}-  \psi_{0330}\psi_{3003}- 
     \psi_{0303}\psi_{3030}+  \psi_{0302}\psi_{3031}-  \psi_{0301}\psi_{3032}+  \psi_{0300}\psi_{3033}\nonumber\\&- 
     \psi_{0233}\psi_{3100}+  \psi_{0232}\psi_{3101}-  \psi_{0231}\psi_{3102}+  \psi_{0230}\psi_{3103}+ 
     \psi_{0203}\psi_{3130}-  \psi_{0202}\psi_{3131}+  \psi_{0201}\psi_{3132}-  \psi_{0200}\psi_{3133}\nonumber\\&+ 
     \psi_{0133}\psi_{3200}-  \psi_{0132}\psi_{3201}+  \psi_{0131}\psi_{3202}-  \psi_{0130}\psi_{3203}- 
     \psi_{0103}\psi_{3230}+  \psi_{0102}\psi_{3231}-  \psi_{0101}\psi_{3232}+  \psi_{0100}\psi_{3233}\nonumber\\&- 
     \psi_{0033}\psi_{3300}+  \psi_{0032}\psi_{3301}-  \psi_{0031}\psi_{3302}+  \psi_{0030}\psi_{3303}+ 
     \psi_{0003}\psi_{3330}-  \psi_{0002}\psi_{3331}+  \psi_{0001}\psi_{3332}- 
     \psi_{0000}\psi_{3333})^2\nonumber\\+ 
 4 ( &\psi_{1323}\psi_{2000}-  \psi_{1322}\psi_{2001}+  \psi_{1321}\psi_{2002}-  \psi_{1320}\psi_{2003}- 
     \psi_{1303}\psi_{2020}+  \psi_{1302}\psi_{2021}-  \psi_{1301}\psi_{2022}+  \psi_{1300}\psi_{2023}\nonumber\\&- 
     \psi_{1223}\psi_{2100}+  \psi_{1222}\psi_{2101}-  \psi_{1221}\psi_{2102}+  \psi_{1220}\psi_{2103}+ 
     \psi_{1203}\psi_{2120}-  \psi_{1202}\psi_{2121}+  \psi_{1201}\psi_{2122}-  \psi_{1200}\psi_{2123}\nonumber\\&+ 
     \psi_{1123}\psi_{2200}-  \psi_{1122}\psi_{2201}+  \psi_{1121}\psi_{2202}-  \psi_{1120}\psi_{2203}- 
     \psi_{1103}\psi_{2220}+  \psi_{1102}\psi_{2221}-  \psi_{1101}\psi_{2222}+  \psi_{1100}\psi_{2223}\nonumber\\&- 
     \psi_{1023}\psi_{2300}+  \psi_{1022}\psi_{2301}-  \psi_{1021}\psi_{2302}+  \psi_{1020}\psi_{2303}+ 
     \psi_{1003}\psi_{2320}-  \psi_{1002}\psi_{2321}+  \psi_{1001}\psi_{2322}-  \psi_{1000}\psi_{2323}\nonumber\\&- 
     \psi_{0323}\psi_{3000}+  \psi_{0322}\psi_{3001}-  \psi_{0321}\psi_{3002}+  \psi_{0320}\psi_{3003}+ 
     \psi_{0303}\psi_{3020}-  \psi_{0302}\psi_{3021}+  \psi_{0301}\psi_{3022}-  \psi_{0300}\psi_{3023}\nonumber\\&+ 
     \psi_{0223}\psi_{3100}-  \psi_{0222}\psi_{3101}+  \psi_{0221}\psi_{3102}-  \psi_{0220}\psi_{3103}- 
     \psi_{0203}\psi_{3120}+  \psi_{0202}\psi_{3121}-  \psi_{0201}\psi_{3122}+  \psi_{0200}\psi_{3123}\nonumber\\&- 
     \psi_{0123}\psi_{3200}+  \psi_{0122}\psi_{3201}-  \psi_{0121}\psi_{3202}+  \psi_{0120}\psi_{3203}+ 
     \psi_{0103}\psi_{3220}-  \psi_{0102}\psi_{3221}+  \psi_{0101}\psi_{3222}-  \psi_{0100}\psi_{3223}\nonumber\\&+ 
     \psi_{0023}\psi_{3300}-  \psi_{0022}\psi_{3301}+  \psi_{0021}\psi_{3302}-  \psi_{0020}\psi_{3303}- 
     \psi_{0003}\psi_{3320}+  \psi_{0002}\psi_{3321}-  \psi_{0001}\psi_{3322}+ 
     \psi_{0000}\psi_{3323})\nonumber\\\times (&- \psi_{1333}\psi_{2010}+  \psi_{1332}\psi_{2011}-  \psi_{1331}\psi_{2012}+ 
     \psi_{1330}\psi_{2013}+  \psi_{1313}\psi_{2030}-  \psi_{1312}\psi_{2031}+  \psi_{1311}\psi_{2032}- 
     \psi_{1310}\psi_{2033}\nonumber\\&+  \psi_{1233}\psi_{2110}-  \psi_{1232}\psi_{2111}+  \psi_{1231}\psi_{2112}- 
     \psi_{1230}\psi_{2113}-  \psi_{1213}\psi_{2130}+  \psi_{1212}\psi_{2131}-  \psi_{1211}\psi_{2132}+ 
     \psi_{1210}\psi_{2133}\nonumber\\&-  \psi_{1133}\psi_{2210}+  \psi_{1132}\psi_{2211}-  \psi_{1131}\psi_{2212}+ 
     \psi_{1130}\psi_{2213}+  \psi_{1113}\psi_{2230}-  \psi_{1112}\psi_{2231}+  \psi_{1111}\psi_{2232}- 
     \psi_{1110}\psi_{2233}\nonumber\\&+  \psi_{1033}\psi_{2310}-  \psi_{1032}\psi_{2311}+  \psi_{1031}\psi_{2312}- 
     \psi_{1030}\psi_{2313}-  \psi_{1013}\psi_{2330}+  \psi_{1012}\psi_{2331}-  \psi_{1011}\psi_{2332}+ 
     \psi_{1010}\psi_{2333}\nonumber\\&+  \psi_{0333}\psi_{3010}-  \psi_{0332}\psi_{3011}+  \psi_{0331}\psi_{3012}- 
     \psi_{0330}\psi_{3013}-  \psi_{0313}\psi_{3030}+  \psi_{0312}\psi_{3031}-  \psi_{0311}\psi_{3032}+ 
     \psi_{0310}\psi_{3033}\nonumber\\&-  \psi_{0233}\psi_{3110}+  \psi_{0232}\psi_{3111}-  \psi_{0231}\psi_{3112}+ 
     \psi_{0230}\psi_{3113}+  \psi_{0213}\psi_{3130}-  \psi_{0212}\psi_{3131}+  \psi_{0211}\psi_{3132}- 
     \psi_{0210}\psi_{3133}\nonumber\\&+  \psi_{0133}\psi_{3210}-  \psi_{0132}\psi_{3211}+  \psi_{0131}\psi_{3212}- 
     \psi_{0130}\psi_{3213}-  \psi_{0113}\psi_{3230}+  \psi_{0112}\psi_{3231}-  \psi_{0111}\psi_{3232}+ 
     \psi_{0110}\psi_{3233}\nonumber\\&-  \psi_{0033}\psi_{3310}+  \psi_{0032}\psi_{3311}-  \psi_{0031}\psi_{3312}+ 
     \psi_{0030}\psi_{3313}+  \psi_{0013}\psi_{3330}-  \psi_{0012}\psi_{3331}+  \psi_{0011}\psi_{3332}- 
     \psi_{0010}\psi_{3333})\nonumber\\ + 
 4 (&- \psi_{1313}\psi_{2000}+  \psi_{1312}\psi_{2001}-  \psi_{1311}\psi_{2002}+  \psi_{1310}\psi_{2003}+ 
     \psi_{1303}\psi_{2010}-  \psi_{1302}\psi_{2011}+  \psi_{1301}\psi_{2012}-  \psi_{1300}\psi_{2013}\nonumber\\&+ 
     \psi_{1213}\psi_{2100}-  \psi_{1212}\psi_{2101}+  \psi_{1211}\psi_{2102}-  \psi_{1210}\psi_{2103}- 
     \psi_{1203}\psi_{2110}+  \psi_{1202}\psi_{2111}-  \psi_{1201}\psi_{2112}+  \psi_{1200}\psi_{2113}\nonumber\\&- 
     \psi_{1113}\psi_{2200}+  \psi_{1112}\psi_{2201}-  \psi_{1111}\psi_{2202}+  \psi_{1110}\psi_{2203}+ 
     \psi_{1103}\psi_{2210}-  \psi_{1102}\psi_{2211}+  \psi_{1101}\psi_{2212}-  \psi_{1100}\psi_{2213}\nonumber\\&+ 
     \psi_{1013}\psi_{2300}-  \psi_{1012}\psi_{2301}+  \psi_{1011}\psi_{2302}-  \psi_{1010}\psi_{2303}- 
     \psi_{1003}\psi_{2310}+  \psi_{1002}\psi_{2311}-  \psi_{1001}\psi_{2312}+  \psi_{1000}\psi_{2313}\nonumber\\&+ 
     \psi_{0313}\psi_{3000}-  \psi_{0312}\psi_{3001}+  \psi_{0311}\psi_{3002}-  \psi_{0310}\psi_{3003}- 
     \psi_{0303}\psi_{3010}+  \psi_{0302}\psi_{3011}-  \psi_{0301}\psi_{3012}+  \psi_{0300}\psi_{3013}\nonumber\\&- 
     \psi_{0213}\psi_{3100}+  \psi_{0212}\psi_{3101}-  \psi_{0211}\psi_{3102}+  \psi_{0210}\psi_{3103}+ 
     \psi_{0203}\psi_{3110}-  \psi_{0202}\psi_{3111}+  \psi_{0201}\psi_{3112}-  \psi_{0200}\psi_{3113}\nonumber\\&+ 
     \psi_{0113}\psi_{3200}-  \psi_{0112}\psi_{3201}+  \psi_{0111}\psi_{3202}-  \psi_{0110}\psi_{3203}- 
     \psi_{0103}\psi_{3210}+  \psi_{0102}\psi_{3211}-  \psi_{0101}\psi_{3212}+  \psi_{0100}\psi_{3213}\nonumber\\&- 
     \psi_{0013}\psi_{3300}+  \psi_{0012}\psi_{3301}-  \psi_{0011}\psi_{3302}+  \psi_{0010}\psi_{3303}+ 
     \psi_{0003}\psi_{3310}-  \psi_{0002}\psi_{3311}+  \psi_{0001}\psi_{3312}- 
     \psi_{0000}\psi_{3313})\nonumber\\\times (&- \psi_{1333}\psi_{2020}+  \psi_{1332}\psi_{2021}-  \psi_{1331}\psi_{2022}+ 
     \psi_{1330}\psi_{2023}+  \psi_{1323}\psi_{2030}-  \psi_{1322}\psi_{2031}+  \psi_{1321}\psi_{2032}- 
     \psi_{1320}\psi_{2033}\nonumber\\&+  \psi_{1233}\psi_{2120}-  \psi_{1232}\psi_{2121}+  \psi_{1231}\psi_{2122}- 
     \psi_{1230}\psi_{2123}-  \psi_{1223}\psi_{2130}+  \psi_{1222}\psi_{2131}-  \psi_{1221}\psi_{2132}+ 
     \psi_{1220}\psi_{2133}\nonumber\\&-  \psi_{1133}\psi_{2220}+  \psi_{1132}\psi_{2221}-  \psi_{1131}\psi_{2222}+ 
     \psi_{1130}\psi_{2223}+  \psi_{1123}\psi_{2230}-  \psi_{1122}\psi_{2231}+  \psi_{1121}\psi_{2232}- 
     \psi_{1120}\psi_{2233}\nonumber\\&+  \psi_{1033}\psi_{2320}-  \psi_{1032}\psi_{2321}+  \psi_{1031}\psi_{2322}- 
     \psi_{1030}\psi_{2323}-  \psi_{1023}\psi_{2330}+  \psi_{1022}\psi_{2331}-  \psi_{1021}\psi_{2332}+ 
     \psi_{1020}\psi_{2333}\nonumber\\&+  \psi_{0333}\psi_{3020}-  \psi_{0332}\psi_{3021}+  \psi_{0331}\psi_{3022}- 
     \psi_{0330}\psi_{3023}-  \psi_{0323}\psi_{3030}+  \psi_{0322}\psi_{3031}-  \psi_{0321}\psi_{3032}+ 
     \psi_{0320}\psi_{3033}\nonumber\\&-  \psi_{0233}\psi_{3120}+  \psi_{0232}\psi_{3121}-  \psi_{0231}\psi_{3122}+ 
     \psi_{0230}\psi_{3123}+  \psi_{0223}\psi_{3130}-  \psi_{0222}\psi_{3131}+  \psi_{0221}\psi_{3132}- 
     \psi_{0220}\psi_{3133}\nonumber\\&+  \psi_{0133}\psi_{3220}-  \psi_{0132}\psi_{3221}+  \psi_{0131}\psi_{3222}- 
     \psi_{0130}\psi_{3223}-  \psi_{0123}\psi_{3230}+  \psi_{0122}\psi_{3231}-  \psi_{0121}\psi_{3232}+ 
     \psi_{0120}\psi_{3233}\nonumber\\&-  \psi_{0033}\psi_{3320}+  \psi_{0032}\psi_{3321}-  \psi_{0031}\psi_{3322}+ 
     \psi_{0030}\psi_{3323}+  \psi_{0023}\psi_{3330}-  \psi_{0022}\psi_{3331}+  \psi_{0021}\psi_{3332}- 
     \psi_{0020}\psi_{3333}). 
\end{align}
\normalsize


\begin{thebibliography}{xx}
\bibitem{dirac2}P. A. M. Dirac, Proc. Royal Soc. A {\bf 117}, 610 (1928).
\bibitem{dirac} P. A. M. Dirac,
{\it Principles of Quantum Mechanics, Fourth edition} (Oxord University Press,
London, 1958).
\bibitem{bjorken} J. D. Bjorken and S. D. Drell,
{\it Relativistic Quantum Mechanics} (McGraw-Hill,
New York, 1964).
\bibitem{peskin} M. E. Peskin and D. V. Schroeder,
{\it An Introduction to Quantum Field Theory} (Perseus Books,
Reading, 1995).
\bibitem{schwartz} M. D. Schwartz,
{\it Quantum Field Theory and the Standard Model} (Cambridge University Press,
Cambridge, 2014).
\bibitem{pykk}P. Pykk\"o, Chem. Rev. {\bf 88}, 563 (1988).
\bibitem{yukawa}H. Yukawa, Proc. Phys. Math. Soc. Japan {\bf 17}, 48 (1935).

\bibitem{weyl}H. Weyl, Z. Phys. {\bf 56}, 330 (1929).
\bibitem{semenoff}G. W. Semenoff, Phys. Rev. Lett. {\bf 53}, 2449 (1984).
\bibitem{novos}K. S. Novoselov, A. Geim, S. Morozov, 
D. Jiang, M. I. Katsnelson, I. V. Grigorieva, S. V. Dubonos, and A. A. Firsov, Nature {\bf 438}, 197 (2005).
\bibitem{lu}L. Lu, L. Fu, J. Joannopoulos, and M. Solja\v{c}i\'c, Nature Photon. {\bf 7}, 294 (2013).
\bibitem{bohm}B-J. Yang and N. Nagaosa, Nat. Commun. {\bf 5} 4898 (2014).
\bibitem{liu}Z. K. Liu, B. Zhou, Y. Zhang, Z. J. Wang, 
 H. M. Weng, D. Prabhakaran, S.-K. Mo, Z. X. Shen, Z. Fang, X. Dai, Z. Hussain, and Y. L. Chen, Science {\bf 343}, 864 (2014).
 \bibitem{Xu}S-Y. Xu, I. Belopolski, N. Alidoust, M. Neupane, G. Bian, 
C. Zhang, R. Sankar, G. Chang, Z. Yuan, C-C. Lee, S-M. Huang, H. Zheng, J. Ma, D. S. Sanchez, B. Wang, A. Bansil, F. Chou, P. P. Shibayev, H. Lin, S. Jia, and M. Z. Hasan, Science {\bf 349}, 613 (2015).
\bibitem{pirie}H. Pirie, Y. Liu, A. Soumyanarayanan, 
 P. Chen, Y. He, M. M. Yee, P. F. S. Rosa, J. D. Thompson, D-J. Kim, Z. Fisk, X. Wang, J. Paglione, D. K. Morr, M. H. Hamidian, and J. E. Hoffman, Nat. Phys. {\bf 16}, 52 (2020).




\bibitem{epr} A. Einstein, B. Podolsky, and N. Rosen,  Phys. Rev. {\bf 47}, 777 (1935).
\bibitem{bell}J. S. Bell, Physics {\bf 1}, 195 (1964).
\bibitem{chsh}J. F. Clauser, M. A. Horne, A. Shimony, and R. A. Holt, Phys. Rev. Lett. {\bf 23}, 880 (1969).
\bibitem{bell2}J. S. Bell, Epistemol. Lett. {\bf 9}, 11 (1976).
\bibitem{bennett}C. H. Bennett, G. Brassard, C. Cr{\'e}peau, R. Jozsa, A. Peres, and
W. K. Wootters, Phys. Rev. Lett. {\bf 70}, 1895 (1993).

\bibitem{svet}G. Svetlichny, Phys. Rev. D {\bf 35}, 3066 (1987).

\bibitem{popescu}N. Linden and S. Popescu, Fortsch. Phys.  {\bf 46 }, 567 (1998).
\bibitem{carteret}H. A. Carteret, N. Linden, S. Popescu, and A. Sudbery, Foundations of Physics {\bf 29}, 527 (1999).

\bibitem{coffman}V. Coffman, J. Kundu, and W. K. Wootters, Phys. Rev. A {\bf 61}, 052306 (2000).
\bibitem{toni}A. Ac{\'i}n, A. Andrianov, L. Costa, E. Jan{\'e}, J. I. Latorre, and R. Tarrach,
Phys. Rev. Lett. {\bf 85}, 1560 (2000).
\bibitem{higuchi}H. A. Carteret, A. Higuchi, and A. Sudbery, J. Math. Phys. {\bf 41}, 7932 (2000).
\bibitem{dur}W. D\"ur, G. Vidal, and J. I. Cirac, Phys. Rev. A {\bf 62}, 062314 (2000).
\bibitem{sud} A. Sudbery, J. Phys. A: Math. Gen. {\bf 34}, 643 (2001).
\bibitem{wong}A. Wong and N. Christensen, Phys. Rev. A {\bf 63}, 044301 (2001).
\bibitem{tarrach} A. Ac{\'i}n, A. Andrianov, E. Jan{\'e} and R. Tarrach, J. Phys. A: Math. Gen. {\bf 34}, 6725 (2001).
\bibitem{verstraete2}F. Verstraete, J. Dehaene, B. De Moor, and H. Verschelde,
Phys. Rev. A {\bf 65}, 052112 (2002).

\bibitem{luque}J.-G. Luque and J.-Y. Thibon, Phys. Rev. A {\bf 67}, 042303 (2003).
\bibitem{moor}F. Verstraete, J. Dehaene, and B. De Moor,
Phys. Rev. A {\bf 68}, 012103 (2003).






\bibitem{toumazet}J.-G. Luque, J.-Y. Thibon and F. Toumazet, Math. Struct. Comp. Sci. {\bf 17}, 1133 (2007).

\bibitem{ghz} D. M. Greenberger, M. Horne, and A. Zeilinger, in {\it Bell's Theorem,  Quantum  Theory,  and  Conceptions  of  the  Universe}, edited  by  M. Kafatos (Kluwer, Dordrecht, 1989), p. 69.
\bibitem{ekert}A. Ekert and P. L. Knight, Am. J. Phys. {\bf 63}, 415 (1995).
\bibitem{wootters}S. A. Hill and W. K. Wootters, Phys. Rev. Lett. {\bf 78}, 5022 (1997).
\bibitem{wootters2}W. K. Wootters, Phys. Rev. Lett. {\bf 80}, 2245 (1998).

\bibitem{hilbert}D. Hilbert, Math. Ann. {\bf 36}, 473 (1890).
\bibitem{mumford}D. Mumford, J. Fogarty, and F. Kirwan, {\it Geometric Invariant Theory} (Springer-Verlag, Berlin, 1994).
\bibitem{grassl}M. Grassl, M. R\"otteler, and T. Beth, Phys. Rev. A {\bf 58}, 1833 (1998).


\bibitem{czachor}M. Czachor, Phys. Rev. A {\bf 55}, 72 (1997).
\bibitem{alsing}P. M. Alsing and G. J. Milburn, Quantum Inf. Comput. {\bf 2}, 487 (2002).
\bibitem{terno} A. Peres, P. F. Scudo, and D. R. Terno,
Phys. Rev. Lett. {\bf 88}, 230402 (2002).

\bibitem{adami}R. M. Gingrich and C. Adami,
Phys. Rev. Lett. {\bf 89}, 270402 (2002).

\bibitem{pachos}J. Pachos and E. Solano, Quantum Inf. Comput. {\bf 3}, 115 (2003).
\bibitem{ahn}D. Ahn, H.-j. Lee, Y. H. Moon, and S. W. Hwang,
Phys. Rev. A {\bf 67}, 012103 (2003).
\bibitem{terno2}D. R. Terno,
Phys. Rev. A {\bf 67}, 014102 (2003).
\bibitem{tera}H. Terashima and M. Ueda, Quantum Inf. Comput. {\bf 3}, 224 (2003).

\bibitem{tera2}H. Terashima and M. Ueda, Int. J. Quantum Inform. {\bf 1},  93 (2003).
\bibitem{mano}E. B. Manoukian and N. Yongram, Eur. Phys. J. D {\bf 31}, 137 (2004).
\bibitem{won}W. T. Kim and E. J. Son,
Phys. Rev. A {\bf  71}, 014102 (2005).
\bibitem{caban3}P. Caban and J. Rembieli{\'n}ski, Phys. Rev. A {\bf 72}, 012103 (2005).
\bibitem{leon}L. Lamata, J. Le{\'o}n, and E. Solano,
Phys. Rev. A {\bf 73}, 012335 (2006).

\bibitem{caban}P. Caban and J. Rembieli{\'n}ski, Phys. Rev. A {\bf 74}, 042103 (2006).
\bibitem{tessier}P. M. Alsing, I. Fuentes-Schuller, R. B. Mann, and T. E. Tessier,
Phys. Rev. A {\bf 74}, 032326 (2006).
\bibitem{geng}H-J. Wang and W. T. Geng,  J. Phys. A: Math. Theor. {\bf 40}, 11617 (2007).
\bibitem{delgado}A. Bermudez and M. A. Martin-Delgado, J. Phys. A: Math. Theor. {\bf 41}, 485302 (2008).
\bibitem{moradi}S. Moradi, Jetp Lett. {\bf 89}, 50 (2009).
\bibitem{caban2}P. Caban, J. Rembieli{\'n}ski, and M. W\l odarczyk, Phys. Rev. A {\bf 79}, 014102 (2009).


















\bibitem{spinorent}M. Johansson, Phys. Rev. A {\bf 105}, 032402 (2022).




\bibitem{verstraete}A. Miyake and F. Verstraete, 
Phys. Rev. A {\bf 69}, 012101  (2004).

\bibitem{pauli}W. Pauli, Ann. de l'Inst. Henri Poincar\'{e} {\bf 6}, 109 (1936).
\bibitem{messiah} A. Messiah,
{\it Quantum Mechanics Volume II} (North-Holland,
Amsterdam, 1965), Ch. III. 10.


\bibitem{kotov}V. N. Kotov, B. Uchoa, V. M. Pereira, F. Guinea, and A. H. Castro Neto,
Rev. Mod. Phys. {\bf 84}, 1067 (2012).
\bibitem{vass}I. V. Fialkovsky and D. V. Vassilevich, Int. J. Mod. Phys. A {\bf 27}, 1260007 (2012). 
\bibitem{fang}Z. Wang, Y. Sun, X.-Q. Chen, C. Franchini, G. Xu, H. Weng, X. Dai, and Z. Fang,
Phys. Rev. B {\bf 85}, 195320 (2012).
\bibitem{fang2}Z. Wang, H. Weng, Q. Wu, X. Dai, and Z. Fang,
Phys. Rev. B {\bf 88}, 125427 (2013).
\bibitem{neupane}M. Neupane, S-Y. Xu, R. Sankar, 
 N. Alidoust, G. Bian, C. Liu, I. Belopolski, T-R. Chang, H-T. Jeng, H. Lin, A. Bansil, F. Chou, and M. Z. Hasan, Nat. Commun. {\bf 5}, 3786 (2014).
\bibitem{cava}S. Borisenko, Q. Gibson, D. Evtushinsky, V. Zabolotnyy, B. B\"uchner, and R. J. Cava,
Phys. Rev. Lett. {\bf 113}, 027603 (2014).



\bibitem{gelfand}I. M. Gelfand and N. Y. Vilenkin, {\it Generalized Functions, Vol. IV} (Academic Press, New
York, 1964).
\bibitem{maurin} K. Maurin, {\it Generalized Eigenfunction Expansions and Unitary Representations
of Topological Groups} (Polish Scientific Publishers, Warsaw, 1968).
\bibitem{stein} E. M. Stein and R. Shakarchi, {\it Fourier Analysis: An Introduction (Princeton Lectures in Analysis I)} (Princeton University Press, Princeton, 2003), Ch. 5.1.3.

\bibitem{lee} J. M. Lee, {\it Introduction to Smooth Manifolds, Second Edition} (Springer, New York, 2012), Ch. 2.

\bibitem{mandl}F. Mandl and G. Shaw, {\it Quantum Field Theory} (Wiley, Chichester, 1984), Ch. 4.2.

\bibitem{wightman}A. S. Wightman, in {\it Carg$\grave{e}$se Lectures in Theoretical Physics}, edited by M. L{\'e}vy (Gordon and Breach, New York, 1967), p. 171.
\bibitem{duncan}A. Duncan, {\it The Conceptual Framework of Quantum Field Theory} (Oxford University Press, Oxford, 2012), Ch. 19.2.





\bibitem{wald}R. M. Wald, {\it General Relativity} (The University of Chicago Press,
Chicago, 1984).

\bibitem{navascues}M. Navascu{\'e}s, S. Pironio, and A. Ac{\'i}n,
Phys. Rev. Lett. {\bf 98}, 010401 (2007).
\bibitem{tsirelson}B. S. Tsirelson, {\it Bell inequalities and operator algebras:
http://www.imaph.tu-bs.de/qi/problems/33.html}, (2006).
\bibitem{werner}V. B. Scholz and R. F. Werner, arXiv:0812.4305 (2008).



\bibitem{all that} R. F. Streater and A. S. Wightman,
{\it PCT, Spin and Statistics, and All That} (Princeton University Press,
Princeton, 2000).
\bibitem{holl}S. Hollands and R. M. Wald, Commun. Math. Phys. {\bf 293}, 85, (2010).
\bibitem{spin}W. Pauli, Phys. Rev. {\bf 58}, 716 (1940).
\bibitem{fock}V. Fock, Z. Phys. {\bf 57}, 261 (1929).


\bibitem{zuber} C. Itzykson and J-B. Zuber,
{\it Quantum Field Theory} (Dover,
New York, 2006), Ch. 2-1-3.








\bibitem{hall}B. C. Hall, {\it Lie Groups, Lie Algebras, and Representations: An Elementary Introduction, Second Edition} (Springer,
Cham, 2015).



\bibitem{thaller}B. Thaller, {\it The Dirac Equation} (Springer,
Berlin, 1992), Ch. 4.2.


\bibitem{das}A. Das, {\it Lectures on Quantum Field Theory} (World Scientific,
Singapore, 2008), Ch. 3.9.

\bibitem{weinberg}S. Weinberg, Phys. Rev. Lett. {\bf 19}, 1264 (1967).



\bibitem{haldane}F. D. M. Haldane, Phys. Rev. Lett. {\bf 61}, 2015 (1988).

\bibitem{lomont}J. S. Lomont, {\it Applictions of Finite Groups} (Academic Press,
New York, 1959), Ch. II. 2. B.

\bibitem{ben}C. H. Bennett, S. Popescu, D. Rohrlich, J. A. Smolin, and A. V. Thapliyal,
Phys. Rev. A {\bf 63}, 012307 (2000).

\bibitem{kempe}J. Kempe, Phys. Rev. A {\bf 60}, 910 (1999).
\bibitem{lind2}N. Linden, S. Popescu, and A. Sudbery, Phys. Rev. Lett. {\bf 83}, 243 (1999).

\bibitem{car2} H. A. Carteret and A. Sudbery, J. Phys. A: Math. Gen. {\bf 33}, 4981 (2000).

\bibitem{luq4}J.-G. Luque and J.-Y. Thibon, J. Phys. A Math. Gen. {\bf 39}, 371 (2005).


\bibitem{vidal}G. Vidal, J. Mod. Opt. {\bf 47}, 355 (2000).
\bibitem{entmes}C. H. Bennett, D. P. DiVincenzo, J. A. Smolin, and W. K. Wootters,
Phys. Rev. A {\bf 54}, 3824 (1996).



\bibitem{omega}A. Cayley, Cambridge and Dublin Mathematical Journal {\bf 1}, 104 (1846).



\bibitem{uhlmann}A. Uhlmann,
Phys. Rev. A {\bf 62}, 032307 (2000).
\bibitem{rungta}P. Rungta, V. Bu\v zek, C. M. Caves, M. Hillery, and G. J. Milburn,
Phys. Rev. A {\bf 64}, 042315  (2001).

\bibitem{cayley}A. Cayley, Cambridge Math. J. {\bf 4},  16 (1845).




\bibitem{miyake1}A. Miyake, Phys. Rev. A {\bf 67}, 012108 (2003).

\bibitem{miyake}A. Miyake, 	Int. J. Quant. Info. {\bf 2}, 65 (2004).

\bibitem{briegel}H. J. Briegel and R. Raussendorf,
Phys. Rev. Lett. {\bf 86}, 910 (2001).

\bibitem{vinberg} A. L. Onishchik and E. B. Vinberg, {\it Lie Groups and Algebraic Groups} (Springer, Berlin, 1990).
\bibitem{nolan}N. R. Wallach, {\it Geometric Invariant Theory: Over the Real and Complex Numbers} (Springer, Cham, 2017).
\bibitem{hurw}A. Hurwitz, Nachrichten Ges. Wiss. G\"ottingen, {\bf 71} (1897).

\bibitem{weeyl}H. Weyl, {\it The classical groups, their invariants and representations} (Princeton University Press, Princeton, 1946).
\bibitem{dolgachev}I. Dolgachev, {\it Lectures on Invariant Theory} (Cambridge University Press, Cambridge, 2003), Ch. 3.2.


\bibitem{nagata}M. Nagata,  J. Math. Kyoto Univ. {\bf 3}, 369 (1964).

\bibitem{wall}R. Goodman and N. R. Wallach, {\it Symmetry, Representations and Invariants} (Springer, New York, 2009).







\bibitem{lima}\AA . Lima, Proc. London Math. Soc. {\bf s3-25}, 27 (1972).

\bibitem{wakker} H. J. M. Peters and P. P. Wakker, Econ. Lett. {\bf 22}, 251 (1986).

\bibitem{uhlmannn}A. Uhlmann, Open Syst. Inf. Dyn. {\bf 5}, 209 (1998).


\bibitem{cantor}G. Cantor, Math. Annalen {\bf 5}, 123 (1872).

\bibitem{reed}M. Reed and B. Simon, {\it Methods of Modern Mathematical Physics, II: Fourier Analysis, Self-Adjointness} (Academic Press,
San Diego, 1975), Ch. X.12.

\bibitem{griffiths}D. Griffiths, {\it Introduction to Elementary Particles} (Wiley,
New York, 1987), Ch. 11.3.


\end{thebibliography}
\end{document}